\documentclass[a4paper,fleqn]{cas-sc}

\usepackage[authoryear,longnamesfirst]{natbib}
\usepackage{graphicx}%
\usepackage{multirow}%
\usepackage{amsmath,amssymb,amsfonts}%
\usepackage{amsthm}%
\usepackage{mathrsfs}%
\usepackage[title]{appendix}%
\usepackage{xcolor}%
\usepackage{textcomp}%
\usepackage{manyfoot}%
\usepackage{booktabs}%
\usepackage{algorithm}%
\usepackage{algorithmicx}%
\usepackage{algpseudocode}%
\usepackage{listings}%

\usepackage{subfigure}

\newcommand{\bc}{\textbf{C}}

\newcommand{\p}{\partial}
\newcommand{\bb}{\boldsymbol}
\newcommand{\bu}{\textbf{v}}

\newtheorem{remark}{Remark}[section]
\newtheorem{thm}{Theorem}[section]

\newtheorem{lemma}[thm]{Lemma}

\newcommand{\jump}[1]{[\![#1]\!]}
\newcommand{\bp}{\bb{P}}
\newcommand{\B}{\bb{B}}

\newcommand{\DP}{\Delta P}

\newcommand{\bB}{\textbf{B}}

\newcommand{\pll}{p_\parallel}
\newcommand{\per}{p_{\perp}}
\newcommand{\bhat}{\boldsymbol{b}}
\newcommand{\con}{\textbf{U}}

\newcommand{\f}{\textbf{f}}
\newcommand{\F}{\textbf{f}}

\newcommand{\ent}{\mathcal{E}}
\newcommand{\entf}{\textbf{q}}
\newcommand{\evar}{\textbf{V}}
\newcommand{\s}{\textbf{S}}
\newcommand{\iph}{i+\frac{1}{2}}

\newcommand{\jph}{j+\frac{1}{2}}

\newcommand{\es}{\textbf{es}}
\newcommand{\explicit}{\textbf{exp}}
\newcommand{\implicit}{\textbf{imex}}

\newcommand{\ote}{\textbf{O2}^\es_\explicit}
\newcommand{\othe}{\textbf{O3}^\es_\explicit}
\newcommand{\ofe}{\textbf{O4}^\es_\explicit}
\newcommand{\oti}{\textbf{O2}^\es_\implicit}
\newcommand{\othi}{\textbf{O3}^\es_\implicit}
\newcommand{\ofi}{{\textbf{O4}^\es_\implicit}}

\usepackage{mdframed}
\usepackage[font=small]{caption}
\definecolor{darkgreen}{RGB}{0,100,0}
\newcommand{\revg}[1]{{\textcolor{black}{#1}}}
\newcommand{\revo}[1]{{\textcolor{black}{#1}}}

\def\tsc#1{\csdef{#1}{\textsc{\lowercase{#1}}\xspace}}
\tsc{WGM}
\tsc{QE}
\tsc{EP}
\tsc{PMS}
\tsc{BEC}
\tsc{DE}

\begin{document}
\let\WriteBookmarks\relax
\def\floatpagepagefraction{1}
\def\textpagefraction{.001}
\shorttitle{Entropy stable numerical schemes for GLM-CGL system}
\shortauthors{Singh et~al.}


\title [mode = title]{Entropy stable numerical schemes for divergence diminishing Chew, Goldberger \& Low equations for plasma flows}                      

\author[1]{Chetan Singh}[orcid=0009-0002-6822-8696]
\cormark[1]
\ead{maz218518@iitd.ac.in}

\credit{Formal analysis, Investigation, Methodology, Software, Visualization, Validation, Conceptualization, Writing – original draft}

\affiliation[1]{organization={Department of Mathematics, Indian Institute of Technology Delhi}, 
                country={India}}
                
\author[1,2]{Harish Kumar}[]
\ead{hkumar@iitd.ac.in}

\credit{Conceptualization, Methodology, Supervision, Writing – original draft, Writing – review \& editing, Funding acquisition }   

\affiliation[2]{organization={IITD-Abu Dhabi, Abu Dhabi}, 
	country={UAE}}             

\author[3]{Deepak Bhoriya}[]
\ead{dkbhoriya@gmail.com}

\credit{Formal analysis, Investigation, Visualization, Software, Writing – review \& editing}

\affiliation[3]{organization={Department of Mathematics, BITS-Pilani, Pilani, Rajasthan},
	country={India}}

\author[4,5]{Dinshaw S. Balsara}[]
\ead{dbalsara@nd.edu}

\credit{Conceptualization, Methodology, Supervision, Writing – review \& editing, Funding acquisition}

\affiliation[4]{organization={Physics Department, University of Notre Dame},
	country={USA}}

\affiliation[5]{organization={ACMS, University of Notre Dame},
	country={USA}}

\cortext[cor1]{Corresponding author.}

\begin{abstract}
Chew, Goldberger \& Low (CGL) equations are a set of hyperbolic PDEs with non-conservative products used to model the plasma flows, when the assumption of {\em local thermodynamic equilibrium} is not valid, and the pressure tensor is assumed to be rotated by the magnetic field. This results in the pressure tensor, which is described by the two scalar components. As the magnetic field also evolves, controlling the divergence of the magnetic field is important. In this work, we consider the generalized Lagrange multiplier (GLM) technique for the CGL model. The resulting model is referred to as the GLM-CGL system. To make the system suitable for entropy-stable schemes, we reformulate the GLM-CGL system by treating some conservative terms as non-conservative. The resulting system has a non-conservative part that does not affect entropy evolution.  We then propose entropy stable numerical methods for the GLM-CGL model. The numerical results for the GLM-CGL system are then compared with the CGL system without the GLM divergence diminishing approach to demonstrate that the GLM approach indeed leads to significant improvement in the magnetic field divergence diminishing.
\end{abstract}

\begin{keywords}
	CGL equations \sep GLM-CGL equations \sep GLM technique \sep Entropy stability \sep Divergence Cleaning \sep Entropy stable numerical schemes
\end{keywords}

\maketitle

\section{Introduction}
MHD has been proven remarkably successful in modelling a wide range of astrophysical phenomena (see~\cite{nakariakov2020magnetohydrodynamic,zanni2004mhd,kosovichev2013astrophysical}), including planetary magnetospheres (see~\cite{jia2012magnetospheric,jia2008three,zelenyi2019current}), heliosphere (see~\cite{linde1998heliosphere,adhikari2021transport,izmodenov2015three}), astrosphere (see~\cite{scherer2020mhd,meyer20213d,baalmann2022modelling}), turbulence (see~\cite{schekochihin2022mhd,sommeria1982and,biskamp2003magnetohydrodynamic}), and star formation (see~\cite{teyssier2019numerical,girma2024new,machida2008magnetohydrodynamics}). \revg{MHD assumes local thermodynamic equilibrium} for the plasma, which results in the isotropic pressure. However, for several applications, plasma is effectively collisionless; hence, the pressure cannot be assumed to be isotropic (see~\cite{quataert2003radiatively,ichimaru1977bimodal,chandran2011incorporating,hollweg1976collisionless,zouganelis2004transonic,burch2016magnetospheric,karimabadi2014link,dumin2002corotation,heinemann1999role,sen_entropy_2018,meena_positivity-preserving_2017,meena2019robust,hakim2008extended,Hirabayashi2016new,huang2019six}). To overcome this, Chew, Goldberger and Low's equations were first introduced in~\cite{chew1956boltzmann}. The model assumes a gyrotropic description of pressure based on the two scalar pressure components (see~\cite{hunana2019brief,webb2022action,meng2012classical,huang2019six}).

The CGL equations are hyperbolic in nature and also contain non-conservative terms. Hence, the solutions can be discontinuous and weak solutions need to be considered. For numerical solutions, the presence of a non-conservative product leads to additional complications as the solution needs to consider a particular path (see~\cite{dal1995definition,abgrall2010comment}). In practice, however, only linear paths are considered (see~\cite{huang2019six,meng2012classical}). In~\cite{bhoriya2024}, the authors present WENO finite volume schemes for the model. Also, in~\cite{singh2024entropy}, an entropy stable discretization is presented.  In~\cite{florinski2025anisotropic}, the suitability of the CGL model for the heliospheric interface simulations is discussed. More recently, in~\cite{singh2024eigen}, the eigenvalues and the right eigenvectors are presented for the model. These are then used to design HLL and HLLI \revo{approximate Riemann solvers} and AFD-WENO schemes~\cite{DSbalsara2024hp,DSbalsara2024hpwn}.

\revo{For two- and three-dimensional simulations}, the solutions of the CGL systems need to satisfy the divergence-free condition for the magnetic field. However, this is not guaranteed for numerical simulations even for MHD equations (see~\cite{brackbill1980effect,brackbill1985fluid,balsara1999staggered}). Hence, several strategies have been developed to overcome this difficulty (see~\cite{munz2000divergence,balsara2001divergence,dedner2002hyperbolic,balsara2021globally,li2005locally,balsara2004second,balsara2015divergence,balsara2009efficient,derigs2018ideal,rueda2025entropy,mignone2010high,balsara2025physical}). \revg{More recently, in \cite{wu2023provably}, authors design a central discontinuous Galerkin scheme for MHD, which is based on the geometric quasilinearization (GQL) techniques \cite{wu2023geometric}. Similar to the MHD case, the CGL model also has positivity constraints, which are much more complicated. In addition, the presence of a non-conservative term makes GQL-based positivity analysis rather non-standard.} In addition, \revo{similar to MHD}, in CGL equations, the entropy evolution equation involves divergence of the magnetic field (see~\cite{singh2024entropy}). As the CGL system is hyperbolic, the entropy stability of the numerical solution is also highly desirable. Hence, in addition to the positivity constraints, another stability criteria for numerical schemes of the CGL equations is entropy stability.

In this article, our aim is to address entropy stability of the algorithm, which is effected by the divergence evolution. We proceed using the following steps:
\begin{itemize}
	\item Following the divergence diminishing strategy based on the Generalized Lagrange Multiplier (GLM) technique from~\cite{munz2000divergence,dedner2002hyperbolic,derigs2018ideal,rueda2025entropy}, we first propose the GLM-CGL system. We then analyze the entropy evolution of GLM-CGL model.
	
	\item Following~\cite{singh2024entropy}, we reformulate the system by treating some conservative terms as non-conservative term, in such a way that the new non-conservative terms \revg{do not affect entropy production}. We also symmetrize the conservative part by following~\cite{chandrashekar2016entropy,singh2024entropy,derigs2018ideal,rueda2025entropy}.
	
	\item To construct divergence diminishing \revg{semi-discrete} entropy-stable schemes, we first design entropy-conservative numerical fluxes. Following~\cite{fjordholm2013eno}, higher-order numerical diffusion operators are designed using a sign-preserving reconstruction and entropy-scaled eigenvectors. The derivatives in the non-conservative terms are approximated by appropriate central difference schemes. The whole spatial discretization is proved to entropy stable \revg{ at semi-discrete level}.
	
	\item For the numerical validation, several one and two-dimensional problems are considered. In all 2D test cases, we show that the GLM approach is beneficial in controlling the divergence error of the magnetic field evolution.
	
\end{itemize}

The rest of the article is organized as follows: In Section~\ref{CGL_analysis}, we briefly describe the CGL equations. In Section~\ref{GLM_formulation}, we propose, analyze and reformulate the GLM-CGL system.  In Section~\ref{sec:semi_discrete}, we present numerical discretization of the system. In Section~\ref{sec:Fully_discrete}, we discuss the fully discrete scheme by presenting time discretization. In Section~\ref{sec:NR}, numerical tests are presented. Finally, we present concluding remarks in Section~\ref{sec:conc} 
\section{CGL Model}
\label{CGL_analysis}
The CGL model equations describing the plasma flow are given by (see~\cite{chew1956boltzmann,hunana2019brief,singh2024entropy,bhoriya2024,singh2024eigen}),
\begin{equation}\label{eq:cgl_con}
	\small
	\frac{\p}{\p t}\begin{pmatrix}
		\rho\\
		\rho \bu\\
		\pll\\
		e\\
		\B
	\end{pmatrix}+\nabla\cdot\begin{pmatrix}
		\rho \bu\\
		\rho \bu\bu^\top+\per \textbf{I}+(\pll-\per)\bhat\bhat^\top-\left(\B\B^\top-\frac{|\B|^2}{2}\textbf{I}\right)\\
		\pll\bu\\
		\bu\left(e+ \per+ \frac{|\B|^2}{2}\right)+\bu\cdot\left((\pll-\per)\bhat\bhat^\top-\B\B^\top\right)\\
		\bu\B^\top-\B\bu^\top
	\end{pmatrix}=\begin{pmatrix}
		0\\
		\textbf{0}\\
		\textcolor{black}{-2\pll\bhat^\top\cdot\nabla \bu\cdot\bhat}\\
		0\\
		\textbf{0}
	\end{pmatrix}.\quad \begin{aligned}
		&\text{(a)}\\
		&\text{(b)}\\
		&\text{(c)}\\
		&\text{(d)}\\
		&\text{(e)}
	\end{aligned}
\end{equation}
Here $\rho,\bu=(v_x,v_y,v_z)^\top$ and $\B=(B_x, B_y, B_z)^\top$ denote density, velocity, and magnetic fields for the plasma flow. We denote unit magnetic vector as  $\bhat=(b_{x},b_{y},b_{z})^\top=\frac{\B}{|\B|}$. Using the unit vector, the symmetric pressure tensor $\bp$ is described as,
$$
\bp=\pll\bhat\bhat^\top+\per(\textbf{I}-\bhat\bhat^\top).
$$
Here $\pll$  is the parallel component and $\per$ is the perpendicular component of the pressure tensor. We also define the average scalar pressure as $p=\frac{2 \per+\pll}{3}$. \revo{Throughout this work, we assume a heat capacity ratio $\gamma=5/3$}. Using this, we close the system by defining total energy $e$ as follows:
\begin{align}
	e=\frac{1}{2}\left(\rho |\bu|^2+ |\B|^2 + 3p \right),
	\label{eq:EOS}
\end{align}
In addition, we need to impose divergence-free condition on the magnetic field, i.e. $\nabla\cdot \B=0.$

%
\section{GLM-CGL system}
\label{GLM_formulation}

One possible approach to control the divergence error is based on the generalized Lagrange multiplier (GLM) method (see~\cite{munz2000divergence,dedner2002hyperbolic,mignone2010high}). For the MHD equations, the first GLM-based reformulation was proposed in~\cite{dedner2002hyperbolic}. However, this formulation is not consistent with the entropy condition. In the same work, the authors also introduced an extended version of the GLM-based MHD reformulation, referred to as the {\em EGLM-MHD system}. Several authors~\cite{domingues2013extended,mignone2010high} have also adopted this formulation. Although this system is consistent with the entropy condition, the ninth component of its entropy variable vanishes (see~\cite{derigs2018ideal}), which makes this formulation unsuitable for the construction of entropy-stable schemes. Several other GLM-based divergence-diminishing approaches are considered in~\cite{mackey2011effects,tricco2012constrained,tricco2016constrained}, but these are also not consistent with the entropy condition.

More recently, in~\cite{derigs2018ideal}, authors have proposed an entropy-consistent GLM reformulation of ideal MHD equations. We follow the same formulation to reformulate the CGL equations. The approach involves incorporating a Lagrange multiplier in the CGL system, namely, an auxiliary scalar field $\Psi$.  In the rest of the article, we use red text to indicate the GLM terms.

Following~\cite{derigs2018ideal}, we introduce the GLM-variable $\Psi$, which evolves as,
\begin{equation}\label{eq:glm_psi}
	\textcolor{red}{\frac{\p \Psi}{\p t} + c_h (\nabla\cdot\B) + \bu\cdot\nabla\Psi = 0},
\end{equation}
where $\textcolor{red}{c_h}$ is the hyperbolic divergence cleaning speed. 
The transport speed of $\Psi$ is taken as the velocity $\bu$, consistent with the velocity acting on $\bB$ in the~(\ref{eq:cgl_con}e). The magnetic field equation is now modified to,
\begin{equation}\label{eq:magnetic_glm}
	\frac{\p \B}{\p t} + \nabla\cdot(\bu\B^\top-\B\bu^\top + \textcolor{red}{c_h \Psi \textbf{I}}) = 0.
\end{equation}

Following~\cite{derigs2018ideal}, we note that the scalar $\Psi\rightarrow 0$ as $\nabla\cdot\B\to 0$, \revg{under appropriate initial and boundary conditions}. In addition, all the red-colored terms also vanish, allowing the model to return to the CGL equations~\eqref{eq:cgl_con}. Thus, the GLM modification to the CGL system is consistent. We also modify the total energy (Eqn.~\eqref{eq:EOS}) to reflect the contribution of the GLM-variable $\Psi$ and consider
\begin{align}
	E &= e + \textcolor{red}{\frac{1}{2}\Psi^{2}}
	= \frac{\rho |\bu|^2}{2}+\frac{|\B|^2}{2}+\frac{3 p}{2} + \textcolor{red}{\frac{1}{2}\Psi^{2}},\label{eq:AEOS}
\end{align}
as the new equation of state with the total energy equation as
\begin{equation}\label{eq:energy_glm}
	\frac{\p E}{\p t} + \nabla\cdot\left[\bu\left(e+ \per+ \frac{|\B|^2}{2}\right)+\bu\cdot\left((\pll-\per)\bhat\bhat^\top-\B\B^\top\right) + \textcolor{red}{c_h\Psi\B}\right] + \textcolor{red}{\Psi\bu\cdot\nabla\Psi}= 0.
\end{equation}
Finally, the  \textbf{GLM-CGL system} is given by,
\begin{equation}\label{eq:glm_cgl_con}
	\scalebox{0.9}{$\frac{\p}{\p t}\begin{pmatrix}
			\rho\\
			\rho \bu\\
			\pll\\
			E\\
			\B\\
			\Psi
		\end{pmatrix}+\nabla\cdot\left\{\underbrace{\begin{pmatrix}
				\rho \bu\\
				\rho \bu\bu^\top+\per \textbf{I}+(\pll-\per)\bhat\bhat^\top-\left(\B\B^\top-\frac{|\B|^2}{2}\textbf{I}\right)\\
				\pll\bu\\
				\bu\left(e+ \per+ \frac{|\B|^2}{2}\right)+\bu\cdot\left((\pll-\per)\bhat\bhat^\top-\B\B^\top\right) \\
				\bu\B^\top-\B\bu^\top\\
				0
		\end{pmatrix}}_{\f_{cgl}} + \textcolor{red}{\underbrace{\begin{pmatrix}
					0\\
					\textbf{0}\\
					0\\
					c_h\Psi\B\\
					c_h \Psi\textbf{I}\\
					c_h \B
			\end{pmatrix}}_{\f_{glm}}}\right\} \\
		+\textcolor{red}{\underbrace{\begin{pmatrix}
					0\\
					\textbf{0}\\
					0\\
					\Psi\bu\\
					\textbf{0}\\
					\bu
				\end{pmatrix}\cdot(\nabla\Psi)}_{\eta_{glm}}}=\underbrace{\begin{pmatrix}
				0\\
				\textbf{0}\\
				\textcolor{black}{-2\pll\bhat^\top\cdot\nabla \bu\cdot\bhat}\\
				0\\
				\textbf{0}\\
				0
		\end{pmatrix}}_{\eta_{cgl}}\quad \begin{aligned}
			&\text{(a)}\\
			&\text{(b)}\\
			&\text{(c)}\\
			&\text{(d)}\\
			&\text{(e)}\\
			&\text{(f)}
		\end{aligned}$}
\end{equation}
The eigenvalues of the GLM-CGL system~\eqref{eq:glm_cgl_con} are presented in Appendix~\ref{Eiegn_value_admissible domain}.
\subsection{Entropy analysis}
Let us denote $\con = (\rho, \rho \bu, \pll, E, \B, \Psi)^\top$ to be the \revo{working variable} for the GLM-CGL system. Following~\cite{singh2024entropy}, the entropy function is given by $\ent= -\rho s.$ Also, the entropy fluxes are given by $\entf_x=v_x\ent$ and $\entf_y=v_y\ent$ with,

\begin{align*}
	s=\ln \left( \dfrac{\det\bp}{\rho^5} \right)=\ln \left( \dfrac{\pll \per^2}{\rho^5} \right).
\end{align*}
Then \revo{the entropy variables $\evar=\frac{\p \ent}{\p \con}$} for the GLM-CGL system~\eqref{eq:glm_cgl_con} is,
\begin{align}
	\evar=\Bigg(5-s-\beta_\perp |\bu|^2,~2\beta_\perp\bu,~-\beta_\parallel+\beta_\perp,~-2\beta_\perp,~2\beta_\perp\B,~\textcolor{red}{2\beta_\perp\Psi}\Bigg)^\top,
	\label{eq:envar_glm}
\end{align} 
where $\beta_\perp=\frac{\rho}{\per}$ and $\beta_\parallel=\frac{\rho}{\pll}$. We have the following lemma.
\begin{lemma}\label{lemma:ent_eq_Glm}
	\revo{The smooth solutions of \eqref{eq:glm_cgl_con} satisfy}
	\begin{equation}
		\label{eq:ent_eql_glm_withdiv}
		\p_t \ent+ \p_x \entf_x + \p_y \entf_y +\frac{2\rho(\bu\cdot\B)}{\per} (\nabla\cdot\B)=0,
	\end{equation}
	which implies
	\begin{equation}
		\label{eq:ent_eql_glm}
		\p_t \ent+ \p_x \entf_x + \p_y \entf_y =0,
	\end{equation}
	if $\nabla\cdot\B=0$. 
\end{lemma}
\begin{proof}
	We first note that
	\begin{align}
		\evar^\top (\nabla\cdot\textcolor{red}{\f_{glm}}) & = \textcolor{red}{-2\beta_\perp \left(\nabla\cdot(c_h\Psi\B)\right) + 2\beta_\perp\B\cdot \left(\nabla\cdot(c_h \Psi\textbf{I})\right) + 2\beta_\perp\Psi(\nabla\cdot(c_h \B))}\nonumber\\
		& = \textcolor{red}{-2\beta_\perp \left(\nabla\cdot(c_h\Psi\B)\right) + 2\beta_\perp\B\cdot  \left(\nabla(c_h \Psi)\right) + 2\beta_\perp\Psi(\nabla\cdot(c_h \B))}\nonumber\\
		& = \textcolor{red}{-2\beta_\perp \left(\nabla\cdot(c_h\Psi\B)\right) + 2\beta_\perp\left(\nabla\cdot(c_h \Psi\B)\right)}\nonumber\\
		& = 0. \label{entropy_glm_condition_1}	
	\end{align}
	and
	\begin{align}
		\evar^\top \textcolor{red}{\eta_{glm}} = \textcolor{red}{-2\beta_\perp \left(\Psi\bu \cdot(\nabla\Psi)\right) + 2\beta_\perp\Psi\left(\bu \cdot(\nabla\Psi)\right)} = 0.
		\label{entropy_glm_condition_2}
	\end{align}
	Furthermore, from~\cite{singh2024entropy}, we have
	\begin{align}
		\evar^\top &(\partial_t \con + \nabla\cdot \f_{cgl} -\eta_{cgl})= 	\p_t \ent+ \p_x \entf_x + \p_y \entf_y +\frac{2\rho(\bu\cdot\B)}{\per} (\nabla\cdot\B) =0 \label{eq:ent_1}.
	\end{align}
	Combining \eqref{eq:ent_1}, \eqref{entropy_glm_condition_1} and \eqref{entropy_glm_condition_2}, we get
	\begin{align}
		\evar^\top  (\partial_t \con + \nabla\cdot (\f_{cgl} +\textcolor{red}{\f_{glm}}) + \textcolor{red}{\eta_{glm}}-\eta_{cgl} )
		&=\p_t \ent+ \p_x \entf_x + \p_y \entf_y +\frac{2\rho(\bu\cdot\B)}{\per}(\nabla\cdot\B) =0.
	\end{align}	
\end{proof}
We also note that  the equality \eqref{eq:ent_eql_glm}, turns into inequality 
\begin{equation}
	\label{eq:ent_ineql_glm}
	\p_t \ent+ \p_x \entf_x+\p_y\entf_y \le0.
\end{equation}
for the non-smooth solutions. Following~\cite{derigs2018ideal,rueda2025entropy}, the system~\eqref{eq:glm_cgl_con} is consistent with the entropy inequality~\eqref{eq:ent_ineql_glm}. 
\subsection{Reformulation of the GLM-CGL system}
\label{sec:reformulation_glm}
Following~\cite{yadav2023entropy,singh2024entropy}, we reformulate the GLM-CGL system by considering some conservative terms as non-conservative products. We get, 

\begin{equation}\label{eq:glm_cgl_noncon}
	\scalebox{1.0}{$\frac{\p}{\p t}\begin{pmatrix}
			\rho\\
			\rho \bu\\[1mm]
			\pll\\
			E\\
			\B\\
			\Psi
		\end{pmatrix}+\nabla\cdot\underbrace{\begin{pmatrix}
				\rho \bu\\
				\rho \bu\bu^\top+\per \textbf{I}-\left(\B\B^\top-\frac{|\B|^2}{2}\textbf{I}\right)\\[1mm]
				\pll\bu\\
				\bu\left(e+ \per+ \frac{|\B|^2}{2}\right)-\bu\cdot\left(\B\B^\top\right)+\textcolor{red}{c_h\Psi\B}\\
				\bu\B^\top-\B\bu^\top+\textcolor{red}{c_h \Psi\textbf{I}}\\
				\textcolor{red}{c_h \B}
		\end{pmatrix}}_{\f_{C}} +
		\textcolor{red}{\underbrace{\begin{pmatrix}
					0\\
					\textbf{0}\\[1mm]
					0\\
					\Psi\bu\\
					\textbf{0}\\
					\bu
				\end{pmatrix}\cdot(\nabla\Psi)}_{\eta_{glm}}}
		+\textcolor{blue}{\underbrace{\begin{pmatrix}
					0\\
					\nabla\cdot\left((\pll-\per)\bhat\bhat^\top\right)\\[1mm]
					2\pll\bhat^\top\cdot\nabla \bu\cdot\bhat\\
					\nabla\cdot\left(\bu\cdot\left((\pll-\per)\bhat\bhat^\top\right)\right) \\
					\textbf{0}\\
					0
			\end{pmatrix}}_{\eta_{NC}}}
		=\textbf{0}.\quad \begin{aligned}
			&\text{(a)}\\
			&\text{(b)}\\[1mm]
			&\text{(c)}\\
			&\text{(d)}\\
			&\text{(e)}\\
			&\text{(f)}
		\end{aligned}$}
\end{equation}
We note that  the choice of non-conservative terms $\textcolor{blue}{\eta_{NC}}$ (blue colored) is such that $\evar^\top\cdot\textcolor{blue}{\eta_{NC}}=0$ (see~\cite{singh2024entropy}). Also, from Eqn.~\eqref{entropy_glm_condition_2}, we already have $\evar^\top\cdot \textcolor{red}{\eta_{glm}}=0$. In two dimensions, we can now rewrite the GLM-CGL system~\eqref{eq:glm_cgl_noncon} as follows:
\begin{equation}\label{eq:glm_cgl_ref_noncons}
	\frac{\p \con}{\p t}+\frac{\p \f_{x}}{\p x} +\frac{\p \f_{y}}{\p y}+ \textcolor{red}{\Upsilon_{x} \frac{\p \Psi}{\p x}} + \textcolor{red}{\Upsilon_y \frac{\p \Psi}{\p y}} + \textcolor{blue}{\bc_{x}(\con)\frac{\p \con}{\p x} + \bc_{y}(\con)\frac{\p \con}{\p y}}=0.
\end{equation}
%
The flux functions are\\
\scalebox{1.0}{$\f_{x}=\begin{pmatrix}
		\rho v_x\\
		\rho v_x^2 + \per-\left(B_x^2-\frac{|\B|^2}{2}\right)\\
		\rho v_x v_y -  B_x B_y\\
		\rho v_x v_z -  B_x B_z\\
		\pll v_x\\
		v_x\left(e+ \per+ \frac{|\B|^2}{2}\right) -{B_x}(\B\cdot\bu)+\textcolor{red}{c_h\Psi B_x}\\
		\textcolor{red}{c_h\Psi}\\
		v_x B_y-v_y B_x\\
		v_x B_z-v_z B_x\\
		\textcolor{red}{c_h B_x}
	\end{pmatrix},~~\f_{y}= \begin{pmatrix}
		\rho v_y\\
		\rho v_x v_y - B_x B_y\\
		\rho v_y^2 +\per-  \left(B_y^2-\frac{|\B|^2}{2}\right)\\
		\rho v_y v_z-  B_y B_z\\
		\pll v_y\\
		v_y\left(e+ \per+ \frac{|\B|^2}{2}\right) -B_y(\B\cdot\bu)+\textcolor{red}{c_h\Psi B_y}\\
		v_y B_x - v_x B_y\\
		\textcolor{red}{c_h\Psi}\\
		v_y B_z - v_z B_y\\
		\textcolor{red}{c_h B_y}
	\end{pmatrix}$}.\\
The vectors $\textcolor{red}{\Upsilon_x, \Upsilon_y} $ are defined as $\textcolor{red}{\Upsilon_x = \left(0,~\textbf{0},~0,~\Psi v_x,~\textbf{0},~v_x\right)^\top}$ and $\textcolor{red}{\Upsilon_y = \left(0,~\textbf{0},~0,~\Psi v_y,~\textbf{0},~v_y\right)^\top}$. The matrices
$\textcolor{blue}{\bc_{x}(\con)}$  and $\textcolor{blue}{\bc_{y}(\con)}$ are given in Appendix~\ref{NC_CGL_Matrices}. The conservative fluxes $\f_x$ and $\f_y$ are similar to the GLM-MHD fluxes defined in~\cite{derigs2018ideal}. Following~\cite{yadav2023entropy,singh2024entropy}, we note that
\begin{equation}\label{eq:ent_def_glm}
	{\entf_x}'(\con) = \evar{\f_x}'(\con), \qquad {\entf_y}'(\con) = \evar{\f_y}'(\con),
\end{equation}
and
\begin{equation}\label{eq:psi_glm}
	\evar^\top \textcolor{red}{\Upsilon_x} = \evar^\top \textcolor{red}{\Upsilon_y} = 0,
\end{equation}
Following~\cite{singh2024entropy}, we also note that 
\begin{equation}
	\evar^\top \textcolor{blue}{\bc_x(\con)} = \evar^\top \textcolor{blue}{\bc_y(\con)} = 0.
\end{equation}

Hence, for the system\eqref{eq:glm_cgl_ref_noncons}, following~\cite{yadav2023entropy}, the entropy-entropy flux pair is given by ($\ent, \entf_x, \entf_y$). Furthermore, we make following remark about the symmetrizability of the GLM-CGL system:
\begin{remark}
	Following~\cite{singh2024entropy}, the GLM-CGL system is not symmetrizable.
\end{remark}
\subsection{Symmetrization}\label{sec:sym}
To make equations suitable for the construction of entropy-stable numerical schemes, we now follow Godunov's process from ~\cite{singh2024entropy,barth1999numerical,barth2006role,chandrashekar2016entropy,derigs2018ideal,derigs2016novel} for the conservative part.
We consider,
\begin{align}\label{eq:glm_cgl_con_part}
	\frac{\p \con}{\p t}+\frac{\p \f_{x}}{\p x} +\frac{\p \f_{y}}{\p y}=0.   
\end{align}
Following\cite{singh2024entropy}, we note that
\begin{align*}
	\frac{\p \f_{x}}{\p \evar}\neq\bigg(\frac{\p \f_{x}}{\p \evar}\bigg)^{\top}~~\text{and}~~~\frac{\p \f_{y}}{\p \evar}\neq\bigg(\frac{\p \f_{y}}{\p \evar}\bigg)^{\top},
\end{align*}
hence the conservative system is not symmetrizable, without assuming $\nabla\cdot\B=0$. 

Following ~\cite{godunov,barth1999numerical,barth2006role,chandrashekar2016entropy,derigs2018ideal,derigs2016novel,singh2024entropy}, we rewrite the term $\frac{2\rho(\bu\cdot\B)}{\per}$ using,
\begin{equation}\label{eq:Phi}
	\Phi(\evar)=\evar \Phi ' (\evar)=\frac{2\rho(\bu\cdot\B)}{\per}.
\end{equation}
Differentiating,
\begin{equation}\label{eq:Phi_derivative}
	\Phi'(\evar)=[0,~\B,~0,~\B\cdot\bu,~\bu,~\textcolor{red}{0}].
\end{equation}
We also define the entropy fluxes $\entf_x$ and $\entf_y$
\begin{align}
	\entf_x&=\evar \cdot \textbf{f}_x+\Phi B_x -\mathcal{F}_x\label{entropy_flux_x},\\
	\entf_y&=\evar \cdot \textbf{f}_y+\Phi B_y -\mathcal{F}_y\label{entropy_flux_y},
\end{align}
with entropy potentials $\mathcal{F}_x$ and $\mathcal{F}_y$
\begin{align}
	\mathcal{F}_x&=\evar \cdot \textbf{f}_x+\Phi B_x -\entf_x=2\rho v_x+\beta_{\perp} v_x|\B|^2 + \textcolor{red}{2 c_h \beta_\perp\Psi B_x} \label{entropy_potential_x},\\
	\mathcal{F}_y&=\evar \cdot \textbf{f}_y+\Phi B_y -\entf_y=2\rho v_y+\beta_{\perp} v_y |\B|^2 + \textcolor{red}{2 c_h \beta_\perp\Psi B_y} \label{entropy_potential_y}.
\end{align}
\revo{Finally, Eqn.~\eqref{eq:glm_cgl_con_part} is now symmetrizable if we consider the well-known Godunov-Powell source term $-\Phi'(\evar)^\top (\nabla\cdot \B)$ on the right-hand side (red colored), which results in}
\begin{equation}\label{eq:glm_cgl_god}
	\frac{\p}{\p t}\begin{pmatrix}
		\rho\\
		\rho \bu\\
		\pll\\
		E\\
		\B\\
		\Psi
	\end{pmatrix}+\nabla\cdot\begin{pmatrix}
		\rho \bu\\
		\rho \bu\bu+\per \textbf{I}-\left(\B\B^\top-\frac{|\B|^2}{2}\textbf{I}\right)\\
		\pll\bu\\
		\bu\left(e+ \per+ \frac{|\B|^2}{2}\right)-\bu\cdot\left(\B\B^\top\right)+\textcolor{red}{c_h\Psi\B}\\
		\bu\B^\top-\B\bu^\top+\textcolor{red}{c_h \Psi\textbf{I}}\\
		\textcolor{red}{c_h \B}
	\end{pmatrix}=\textcolor{red}{-\begin{pmatrix}
			0\\
			\B\\
			0\\
			\bu\cdot\B\\
			\bu\\
			0
		\end{pmatrix}( \nabla\cdot \B)}.\quad \begin{aligned}
		&\text{(a)}\\
		&\text{(b)}\\
		&\text{(c)}\\
		&\text{(d)}\\
		&\text{(e)}\\
		&\text{(f)}
	\end{aligned}
\end{equation}
Together with \revo{non-conservative terms of the GLM formulation}, we get,
\begin{equation}\label{eq:glm_cgl_mhd}
	\frac{\p \con}{\p t}+\frac{\p \f_{x}}{\p x} +\frac{\p \f_{y}}{\p y} +\textcolor{red}{ \Phi'(\evar)^\top ( \nabla\cdot \B)} + \textcolor{red}{\Upsilon_{x} \frac{\p \Psi}{\p x}}+ \textcolor{red}{\Upsilon_y \frac{\p \Psi}{\p y}} =0.
\end{equation}
The system is similar to the GLM-MHD system derived in \cite{derigs2018ideal}. We further remark that with GLM terms, the system \eqref{eq:glm_cgl_mhd} is also symmetrizable.
The $x$-direction eigenvalues of the system~\eqref{eq:glm_cgl_mhd} are,
\begin{align}
	\label{eq:eigenval_x}
	K_{x}=\left\{v_x,~ v_x, ~v_x\pm v_{ax},~v_x\pm c_h,~v_x\pm c_f,~v_x\pm c_s\right\},
\end{align}
where
\begin{align*}
	& v^2_{ax}= \frac{B_x^2}{\rho},~v^2_a=\frac{\B^2}{\rho}, ~a^2=\frac{2p_\perp}{\rho}, \text{ and  }
	c^2_{f,s}=\frac{1}{2}\bigg[(v_a^2+a^2)\pm\sqrt{(v_a^2+a^2)^2-4v^2_{ax}a^2}\bigg].
\end{align*}
Similarly, for $y$-direction, 
\begin{align}
	\label{eq:eigenval_y}
	& K_{y}=\left\{v_y,~ v_y, ~v_y\pm v_{ay},~v_y\pm c_h,~v_y\pm c_f,~v_y\pm c_s\right\}
\end{align}
where,
\begin{align*}
	& v^2_{ay}= \frac{B_y^2}{\rho},~v^2_a=\frac{\B^2}{\rho}, ~a^2=\frac{2p_\perp}{\rho}, \text{ and  }
	c^2_{f,s}=\frac{1}{2}\bigg[(v_a^2+a^2)\pm\sqrt{(v_a^2+a^2)^2-4v^2_{ay}a^2}\bigg].
\end{align*}
Following~\cite{barth1999numerical,chandrashekar2016entropy,singh2024entropy}, we derive and present the entropy-scaled right eigenvectors for the system~\eqref{eq:glm_cgl_mhd} in the $x$- and $y$-directions in Appendix~\ref{scaledrev}.
\section{Semi-discrete numerical schemes}
\label{sec:semi_discrete} 
Combining the Godunov terms, GLM terms and additional non-conservative terms, we get,
\begin{equation}\label{eq:glm_cgl_final}
	\frac{\p \con}{\p t}+\frac{\p \f_{x}}{\p x} +\frac{\p \f_{y}}{\p y}+\textcolor{red}{ \Phi'(\evar)^\top ( \nabla\cdot \B)}+ \textcolor{red}{\Upsilon_{x} \frac{\p \Psi}{\p x}} + \textcolor{red}{\Upsilon_y \frac{\p \Psi}{\p y}} + \textcolor{blue}{\bc_{x}(\con)\frac{\p \con}{\p x} + \bc_{y}(\con)\frac{\p \con}{\p y}}=0.
\end{equation}
We will now propose entropy stable numerical schemes for these equations using finite difference method. Consider a uniform mesh for the computational domain with cells $I_{ij}$ of size $\Delta x\times \Delta y.$ We assume that the centers of the cell $I_{ij}$ is denoted by $(x_i,y_j).$ The cell vertices are given by $(x_{i+\frac{1}{2}},y_{j+\frac{1}{2}})$, with  $x_{i+\frac{1}{2}} = \frac{x_{i} + x_{i+1}}{2}$ and $y_{j+\frac{1}{2}} = \frac{y_{j} + y_{j+1}}{2}$. A general semi-discrete finite difference scheme for~\eqref{eq:glm_cgl_final} on the uniform mesh is then  given by\\
\begin{align}\label{eq:semi-discrete_fd}
		\frac{d}{d t}\con_{i,j}&+\frac{\F_{x, i+\frac{1}{2}, j}-\F_{x, i-\frac{1}{2}, j}}{\Delta x}+\frac{\F_{y, i, j+\frac{1}{2}}-\F_{y, i, j-\frac{1}{2}}}{\Delta y} +\textcolor{red}{\Phi'(\evar_{i, j})^\top\left(\left(\frac{\p B_x}{\p x}\right)_{i,j} + \left(\frac{\p B_y}{\p y}\right)_{i,j}\right)}\nonumber\\
		&+ \textcolor{red}{\left((\Upsilon_{x})_{i,j}\left(\frac{\p \Psi}{\p x}\right)_{i,j} + (\Upsilon_{y})_{i,j}\left(\frac{\p \Psi}{\p y}\right)_{i,j}\right)} +\textcolor{blue}{\bc_x(\con_{i,j})\left(\frac{\p \con}{\p x}\right)_{i,j} +\bc_y(\con_{i,j}) \left(\frac{\p \con}{\p y}\right)_{i,j}}=0.
\end{align}
Here,  $\F_{x, i+\frac{1}{2}, j}$ and $\F_{ y,i, j+\frac{1}{2}}$ are the numerical fluxes, consistent with the continuous fluxes $\f_{x}$ and $\f_{y}$ in $x$- and $y$-directions, respectively. The derivatives in non-conservative terms are discretized with central difference of appropriate orders (see Appendix \ref{appendix:central}). For the grid function $a_{i,j}$, let us introduce the notations
\begin{align*}
	[\![a]\!]_{\iph,j} =a_{i+1,j}-a_{i,j},&\qquad \qquad 	[\![a]\!]_{i,\jph} =a_{i,j+1}-a_{i,j}
\end{align*}
for the jumps across the cell interfaces and
\begin{align*}
	\bar{a}_{\iph,j}= \frac{a_{i+1,j} + a_{i,j}}{2},&\qquad \qquad\bar{a}_{i,\jph}= \frac{a_{i,j+1} + a_{i,j}}{2},
\end{align*}
for the averages across the cell interfaces.
\subsection{Entropy conservative schemes}
Following~\cite{tadmor2003entropy,chandrashekar2016entropy,singh2024entropy}, to design entropy conservative scheme, we first need to satisfy the jump relations,
\begin{align}
	\quad \jump{\evar}_{\iph,j}  \cdot \tilde{\F}_{x,i+\frac{1}{2}, j}&=\jump{\mathcal{F}_x}_{\iph,j}- \jump{\Phi}_{\iph,j} \bar{B}_{x, i+\frac{1}{2}, j},\label{eq:consrvative_flux_x}&\\
	\quad\jump{\evar}_{i,\jph}  \cdot \tilde{\F}_{y,i,j+\frac{1}{2}}&=\jump{\mathcal{F}_y}_{i,\jph}- \jump{\Phi}_{i,\jph} \bar{B}_{y,i,j +\frac{1}{2}}.\label{eq:consrvative_flux_y}
\end{align}
Following \cite{singh2024entropy}, the entropy conservative flux 
$$\tilde{\F}_x=\left[ \tilde{f}_{x}^{(1)},\tilde{f}_{x}^{(2)}, \tilde{f}_{x}^{(3)}, \tilde{f}_{x}^{(4)}, \tilde{f}_{x}^{(5)}, \tilde{f}_{x}^{(6)}, \tilde{f}_{x}^{(7)},\tilde{f}_{x}^{(8)},\tilde{f}_{x}^{(9)},\textcolor{red}{\tilde{f}_{x}^{(10)}} \right]^{\top},$$
is given by,
$$
\begin{aligned}
	\tilde{f}_{x}^{(1)}=& {\rho}^{\ln} \bar{v}_{x},~\tilde{f}_{x}^{(2)}= \frac{\bar{\rho}}{ \bar{\beta}_\perp}+\bar{v}_{x} \tilde{f}_{x}^{(1)}+\frac{1}{2} \overline{\left|\boldsymbol{B}\right|^{2}}-\bar{B}_{x} \bar{B}_{x} \\
	\tilde{f}_{x}^{(3)}=& \bar{v}_{y} \tilde{f}_{x}^{(1)}-\bar{B}_{x} \bar{B}_{y},~\tilde{f}_{x}^{(4)}= \bar{v}_{z} \tilde{f}_{x}^{(1)}-\bar{B}_{x} \bar{B}_{z},~\tilde{f}_{x}^{(5)}= \frac{\tilde{f}_{x}^{(1)}}{{\beta}^{\ln}_\parallel},\\
	\tilde{f}_{x}^{(7)}=& \textcolor{red}{c_h\frac{\overline{\Psi \beta_\perp}}{\bar{\beta}_\perp}},~\tilde{f}_{x}^{(8)}= \frac{1}{\bar{\beta}_\perp}\left(\overline{\beta_\perp v_{x}} \bar{B}_{y}-\overline{\beta_\perp v_{y}} \bar{B}_{x}\right),~\tilde{f}_{x}^{(9)}= \frac{1}{\bar{\beta}_\perp}\left(\overline{\beta_\perp v_{x}} \bar{B}_{z}-\overline{\beta_\perp v_{z}}  \bar{B_{x}}\right),\\
	\textcolor{red}{\tilde{f}_{x}^{(10)}=}&\textcolor{red}{c_h \bar{B}_x},\\
	\tilde{f}_{x}^{(6)}=& \frac{1}{2}\left[\frac{2}{ {\beta}^{\ln}_\perp}-\overline{|\bu|^{2}}\right] \tilde{f}_{x}^{(1)}+\bar{v}_{x} \tilde{f}_{x}^{(2)}+\bar{v}_{y} \tilde{f}_{x}^{(3)}+\bar{v}_{z} \tilde{f}_{x}^{(4)} +\frac{\tilde{f}_{x}^{(5)}}{2}\\
	&+\bar{B}_{x} \tilde{f}_{x}^{(7)}+\bar{B}_{y} \tilde{f}_{x}^{(8)}+\bar{B}_{z} \tilde{f}_{x}^{(9)}-\frac{1}{2} \bar{v}_{x} \overline{\left|\boldsymbol{B}\right|^{2}}+\left(\bar{v}_{x} \bar{B}_{x}+\bar{v}_{y} \bar{B}_{y}+\bar{v}_{z} \bar{B}_{z}\right) \bar{B}_{x},
\end{aligned}
$$
\revg{To simplify the notations, we have ignored the spatial indexing and denote the left and right states by $\con_l$ and $\con_r$, respectively. For a scalar quantity $a$, the jump and arithmetic average across an cell interface are defined as
\begin{align*}
	[\![a]\!] =a_{r}-a_{l},\qquad \qquad 
	\bar{a}= \frac{a_{l} + a_{r}}{2}.
\end{align*}
Furthermore, the logarithmic average of $a$ is defined by 
\begin{equation*}
    a^{\ln} = \frac{\jump{a}}{\jump{\ln{a}}}
\end{equation*} 
In addition, for scalar quantities $a$ and $b$, and vector quantity $\mathbf d=(d_x,d_y,d_z)^\top$, we introduce the notation
\begin{align*}
    \overline{ab} = \frac{a_l b_l+ a_r b_r}{2},\quad \overline{\left(\frac{a}{b}\right)}=\frac{\frac{a_l}{b_l}+\frac{a_r}{b_r}}{2},\quad\text{and}~~ \overline{\left|\mathbf{d}\right|^{2}}=\overline{\left(d_x^2+d_y^2+d_z^2\right)}=\frac{\left(d_{x_l}^2+d_{y_l}^2+d_{z_l}^2\right)+\left(d_{x_r}^2+d_{y_r}^2+d_{z_r}^2\right)}{2}
\end{align*}
Similarly, $\tilde{\F}_{y}$, the entropy conservative flux in $y$-direction can be obtained.}

Therefore, the numerical scheme \eqref{eq:semi-discrete_fd} with entropy conservative numerical fluxes $\tilde{\F}_{x,\iph,j} =\tilde{\F}_x(\con_{i,j},\con_{i+1,j})$ and $\tilde{\F}_{y,i,\jph} =\tilde{\F}_y(\con_{i,j},\con_{i,j+1})$ is entropy conservative and second-order accurate. We have highlighted the additional terms in red, compared to the numerical flux in \cite{singh2024entropy}. These terms arise due to the GLM formulation.

To obtain the higher-order entropy conservative numerical schemes, we follow~\cite{leFloch2002} to construct $2p^{th}$-order accurate fluxes using second-order fluxes for any positive integer $p$. In particular for $(p = 2)$, the $4^{th}$-order $x$-directional entropy conservative flux $\tilde{\F}_{x,i+\frac{1}{2},j}^4$ is given below
\begin{align}
	\tilde{\F}_{x,i+\frac{1}{2},j}^4&=\frac{4}{3}\tilde{\F}_{x}(\con_{i,j},\con_{i+1,j})-\frac{1}{6} \bigg( \tilde{\F}_{x} (\con_{i-1,j},\con_{i+1,j})+
	\tilde{\F}_{x}(\con_{i,j},\con_{i+2,j}) \bigg).
	\label{eq:4thorder_numflux_x}
\end{align}

Similarly, we can derive the fourth-order entropy conservative flux $\tilde{\F}_{y,i,j+\frac{1}{2}}^4$ in $y$-direction.  Again, following~\cite{singh2024entropy}, the numerical scheme~\eqref{eq:semi-discrete_fd} with fourth order numerical fluxes described above and fourth order discretization of derivatives of the non-conservative derivative results in a fourth order entropy conservative scheme.

\subsection{Entropy stable schemes}\label{Sec:Entropy_SS}
To design entropy-stable numerical schemes, following~\cite{tadmor2003entropy}, we modify the entropy conservative numerical fluxes as follows
\begin{equation}
	\begin{aligned}
		{\hat{\F}}_{x,i+\frac{1}{2},j} =\tilde{\F}_{x,i+\frac{1}{2},j} - \frac{1}{2} \textbf{D}_{x,i+\frac{1}{2},j}[\![ \evar]\!]_{i+\frac{1}{2},j},~
		\quad
		{\hat{\F}}_{y,i,j+\frac{1}{2}} = \tilde{\F}_{y,i,j+\frac{1}{2}} - \frac{1}{2} \textbf{D}_{y,i,j+\frac{1}{2}}[\![ \evar]\!]_{i,j+\frac{1}{2}},
		\label{es_numflux}
	\end{aligned}
\end{equation}
where $\textbf{D}_{x,i+\frac{1}{2},j}$ and $\textbf{D}_{y,i,j+\frac{1}{2}}$ are symmetric positive definite matrices. We use Rusanov's type diffusion matrices, given by
\begin{equation}  \label{diffusiontype}
	\textbf{D}_{x,i+\frac{1}{2},j} = \tilde{R}_{x,i+\frac{1}{2},j} \Lambda_{x,i+\frac{1}{2},j} \tilde{R}_{x,i+\frac{1}{2},j}^{\top}~\text{and}~~\textbf{D}_{y,i,j+\frac{1}{2}} = \tilde{R}_{y,i,j+\frac{1}{2}} \Lambda_{y,i,j+\frac{1}{2}} \tilde{R}_{y,i,j+\frac{1}{2}}^{ \top},
\end{equation}
where, $\tilde{R}_d,$ $d\in\{x,y\}$  are the matrices of entropy-scaled right eigenvectors derived in Appendix~\ref{scaledrev}, following~\cite{barth1999numerical,singh2024entropy}. Also, ${\Lambda_d}$ are,
$${\Lambda_d}=\left( \max_{1 \leq k \leq 10} |\lambda^k_d|\right) \mathbf{I}_{10 \times 10},$$
\revg{where $\{\lambda^k_d: 1 \leq k \leq 10 \}$ are the eigenvalues of the
symmetrized conservative part of the GLM-CGL system~\eqref{eq:glm_cgl_mhd}, given by
Eqs.~\eqref{eq:eigenval_x} and \eqref{eq:eigenval_y} for the $x$- and
$y$-directions, respectively. }

\revo{An alternative is the Roe-type diffusion matrices are obtained by defining, 
$${\Lambda_d}= \text{diag}\left(|\lambda^1_d|,\cdots,|\lambda^{10}_d|\right).$$
in equation~\eqref{diffusiontype}. Following~\cite{nishikawa2008very,derigs2018ideal}, one may also define the hybrid Roe-Rusanov diffusion matrices as the convex combination
$$
\Lambda_d=(1-\beta)\,\mathrm{diag}\!\left(|\lambda_d^1|,\ldots,|\lambda_d^{10}|\right)
+\beta\left(\max_{1\le k\le10}|\lambda_d^k|\right)\mathbf{I}_{10\times10},
$$
in equation~\eqref{diffusiontype}, where $\beta=0.4$. Among these three choices, the Rusanov-type dissipation operator is adopted in this work due to its robustness, simplicity of implementation, and its ability to effectively suppress spurious oscillations near discontinuities. A numerical comparison of the Rusanov-type, Roe-type, and hybrid Roe-Rusanov dissipation operators is presented for the Brio-Wu shock tube problem in Section~\ref{test:bw}, demonstrating that the Rusanov-type operator provides a favorable balance between robustness and accuracy.}

\revo{For $x$-interface $\left(i+\frac{1}{2},j\right)$, we define the interface primitive state using the arithmetic average
\begin{equation*}
    \textbf{w}_{i+\frac{1}{2},j}=\frac{\textbf{w}_{i,j}+\textbf{w}_{i+1,j}}{2}
\end{equation*}
and similarly, for the $y$-interface $\left(i,j+\frac{1}{2}\right)$,
\begin{equation*}
    \textbf{w}_{i,j+\frac{1}{2}}=\frac{\textbf{w}_{i,j}+\textbf{w}_{i,j+1}}{2}
\end{equation*}
Here, $\textbf{w}$ denotes the vector of primitive variables defined in Appendix~\ref{scaledrev}. The corresponding interface working-variable states are then obtained via the primitive-to-working variable transformation. Subsequently, the entropy-scaled right eigenvector matrices $\tilde{R}_d$ and the diagonal eigenvalue matrices ${\Lambda_d};~~\forall~d\in\{x,y\}$, are evaluated using these interface states.}

To achieve higher order, we follow~\cite{fjordholm2012arbitrarily}, to reconstruct {\em scaled entropy variables}
$$\mathcal{W}^{\pm}_{x,i,j}= \tilde{R}^{{\top}}_{x,i\pm\frac{1}{2},j}\evar_{i,j},$$
using the ENO procedure, to get the reconstructed polynomials $\mathcal{P}^{\pm}_{x,i,j}(x_{i\pm\frac{1}{2}})$ of appropriate order. The ENO procedure has the {\em sign-preserving property}~\cite{fjordholm2013eno}, which ensures the entropy stability. \revo{For the second-order method, we employ the
MinMod limiter, which preserves the sign property}. The reconstructed values of scaled entropy variables are given by
\[\hat{\mathcal{W}}_{x,i,j}^{\pm}=\mathcal{P}^{\pm}_{x,i,j}(x_{i\pm\frac{1}{2}}).\]
Using these, we define the reconstructed scaled variables as
$$
\hat{\evar}^{\pm}_{x,i,j} =  \left\lbrace \tilde{R}^{{\top}}_{x,i\pm\frac{1}{2},j}\right\rbrace ^{(-1)}\hat{\mathcal{W}}_{x,i,j}^{\pm}.
$$
Finally, we denote,
\begin{equation}
	{\hat{\F}}^{k}_{x,i+\frac{1}{2},j}\,=\,\tilde{\F}^{2p}_{x,i+\frac{1}{2},j}\,-\,\frac{1}{2}\,\textbf{D}_{x,i+\frac{1}{2},j}[\![ \hat{\evar}]\!]_{x,i+\frac{1}{2},j}^k.
	\label{eq:entropy_stable_flux_x}
\end{equation}
where
$$
\jump{\hat{\evar}}^k_{x,i+\frac{1}{2},j} = \hat{\evar}^-_{x,i+1,j}-\hat{\evar}^+_{x,i,j}.
$$
Similarly, we get,
\begin{equation}
	{\hat{\F}}^{k}_{y,i,j+\frac{1}{2}}\,=\,\tilde{\F}^{2p}_{y,i,j+\frac{1}{2}}\,-\,\frac{1}{2}\,\textbf{D}_{y,i,j+\frac{1}{2}}[\![ \hat{\evar}]\!]_{y,i,j+\frac{1}{2}}^k.
	\label{eq:entropy_stable_flux_y}
\end{equation}
An appropriate value of $p$ is considered for the different orders of schemes. For the scheme~\eqref{eq:semi-discrete_fd}, we now get,

\begin{thm}[\textbf{see~\cite{singh2024entropy}}]\label{thm:ES_higher_order} The scheme using entropy stable fluxes~\eqref{eq:entropy_stable_flux_x},\eqref{eq:entropy_stable_flux_y} and with second-order (for $k=2$) and fourth-order (for $k=3,4$) central difference approximations for $\left(\frac{\p B_x}{\p x}\right)_{i,j}$, $\left(\frac{\p B_y}{\p y}\right)_{i,j}$, $\left(\frac{\p \Psi}{\p x}\right)_{i,j}$, $\left(\frac{\p \Psi}{\p y}\right)_{i,j}$, $\left(\frac{\p \con}{\p x}\right)_{i,j}$, and $\left(\frac{\p \con}{\p y}\right)_{i,j}$  is $k^{th}$-order accurate and entropy stable (where $k = 2,~3,~4$), i.e. it satisfies
	\begin{equation}
		\label{eq:semi-disc_ent_stab}
		\frac{d}{dt}  \ent(\con_{i,j})  +\frac{1}{\Delta x} \left( \hat{\entf}^k_{x,i+\frac{1}{2},j} - \hat{\entf}^k_{x,i-\frac{1}{2},j}\right)+\frac{1}{\Delta y}\left( \hat{\entf}^k_{y,i,j+\frac{1}{2}} - \hat{\entf}^k_{y,i,j-\frac{1}{2}}\right) \le 0,
	\end{equation}
	where $ \hat{\entf}^k_x$ and $ \hat{\entf}^k_y$ are given by
	\begin{equation*}
		\begin{aligned}
			\hat{\entf}^k_{x,i+\frac{1}{2},j}=  \tilde{\entf}^{2p}_{x,i+\frac{1}{2},j} - \frac{1}{2}\bar{\evar}_{i+\frac{1}{2},j}^{\top}  \textbf{D}_{x,i+\frac{1}{2},j}[\![ \hat{\evar}]\!]^k_{x,i+\frac{1}{2},j}
	\end{aligned} \end{equation*}
	
	and \begin{equation*}
		\begin{aligned}
			\hat{\entf}^k_{y,i,j+\frac{1}{2}}=    \tilde{\entf}^{2p}_{y,i,j+\frac{1}{2}} -  \frac{1}{2}\bar{\evar}_{i,j+\frac{1}{2}}^{\top}  \textbf{D}_{y,i,j+\frac{1}{2}}[\![ \hat{\evar}]\!]^k_{y,i,j+\frac{1}{2}}.\end{aligned} \end{equation*}
\end{thm}
\revo{As the CGL model allows an anisotropic pressure tensor, we can also use it to calculate the isotropic solutions, which are standard MHD solutions. This not only allows us to verify the numerical schemes, but also captures the additional physics in between the anisotropic and isotropic limits. }
To compute the isotropic limit solutions, we consider the following source term, 
\begin{equation}
	\label{eq:src_mhd}
	\s=\Biggl\{ \mathbf{0}_{1\times4},~\frac{\per - \pll}{\tau}, \mathbf{0}_{1\times4},~\textcolor{red}{0}\Biggr\}^\top.
\end{equation}
This source relaxes the solution to the isotropic limit. Here, $\tau$ is a constant.
\revg{\begin{remark}
    Using the admissible domain of the CGL system \eqref{eq:cgl_domain_new}, we obtain
    $$\evar^\top\cdot\s=-\frac{\rho}{\tau\pll\per}\left(\pll-\per\right)^2\leq0.$$
    Therefore, the semi-discrete scheme \eqref{eq:semi-discrete_fd} with the source term $\s$ is also entropy stable, i.e., it satisfies
	\begin{equation*}
		\frac{d}{dt}  \ent(\con_{i,j})  +\frac{1}{\Delta x} \left( \hat{\entf}^k_{x,i+\frac{1}{2},j} - \hat{\entf}^k_{x,i-\frac{1}{2},j}\right)+\frac{1}{\Delta y}\left( \hat{\entf}^k_{y,i,j+\frac{1}{2}} - \hat{\entf}^k_{y,i,j-\frac{1}{2}}\right) \leq \evar_{i,j}^\top\cdot\s_{i,j}\leq 0.
	\end{equation*}
\end{remark}}
\section{Time discretization}
\label{sec:Fully_discrete} 
Let $\con^{n}$ be the solution at $t=t^n$ and $\Delta t = t^{n+1} - t^n$.  Then the scheme~\eqref{eq:semi-discrete_fd} including the source term $\s(\con_{i,j})$,
can be represented as,
\begin{equation}
	\frac{d }{dt}\con_{i,j}(t) = \mathcal{L}_{i,j}(\con(t)) +\s(\con_{i,j}(t)),
	\label{fullydiscrete}
\end{equation}
where,\\
\scalebox{0.9}{%
	\begin{minipage}{1.1\linewidth} 
		\begin{align*}
			\small
			\mathcal{L}_{i,j}(\con(t))=&
			-\frac{\hat{\F}^{k}_{x, i+\frac{1}{2}, j}-\hat{\F}^{k}_{x, i-\frac{1}{2}, j}}{\Delta x}-\frac{\hat{\F}^{k}_{y, i, j+\frac{1}{2}}-\hat{\F}^{k}_{y, i, j-\frac{1}{2}}}{\Delta y} -\textcolor{red}{\Phi'(\evar_{i, j})^\top\left(\left(\frac{\p B_x}{\p x}\right)_{i,j} + \left(\frac{\p B_y}{\p y}\right)_{i,j}\right)}\nonumber\\
			&- \textcolor{red}{\left((\Upsilon_{x})_{i,j}\left(\frac{\p \Psi}{\p x}\right)_{i,j} + (\Upsilon_{y})_{i,j}\left(\frac{\p \Psi}{\p y}\right)_{i,j}\right)} -\textcolor{blue}{\bc_x(\con_{i,j})\left(\frac{\p \con}{\p x}\right)_{i,j} -\bc_y(\con_{i,j}) \left(\frac{\p \con}{\p y}\right)_{i,j}}.
		\end{align*}
	\end{minipage}%
}\vspace{0.5cm}\\
When we compute the anisotropic (CGL) solutions, we do not consider the source term $\s$. So, we use the explicit SSP-RK methods from~\cite{gottlieb2001strong}. These schemes are summarized as follows:
\begin{itemize}
	\item $\ote:$ We take $p=1$ for the entropy conservative numerical flux. For the diffusion operator, we use \textit{MinMod}-based sign-preserving reconstruction. The derivatives $\left(\frac{\p B_x}{\p x}\right)_{i,j}$, $\left(\frac{\p B_y}{\p y}\right)_{i,j}$, $\left(\frac{\p \Psi}{\p x}\right)_{i,j}$, $\left(\frac{\p \Psi}{\p y}\right)_{i,j}$, $\left(\frac{\p \con}{\p x}\right)_{i,j}$, and $\left(\frac{\p \con}{\p y}\right)_{i,j}$ are approximated by second order central difference. For the time update, we use the second-order SSP-RK method.
	
	\item $\othe:$ Here, we take $p=2$ for the entropy conservative part. For the diffusion operator, we use third-order ENO-based sign-preserving reconstruction. The derivatives $\left(\frac{\p B_x}{\p x}\right)_{i,j}$, $\left(\frac{\p B_y}{\p y}\right)_{i,j}$, $\left(\frac{\p \Psi}{\p x}\right)_{i,j}$, $\left(\frac{\p \Psi}{\p y}\right)_{i,j}$, $\left(\frac{\p \con}{\p x}\right)_{i,j}$, and $\left(\frac{\p \con}{\p y}\right)_{i,j}$, are approximated using fourth-order central difference. For the time update, we use the third-order SSP-RK method.
	
	\item $\ofe:$ Here, we take $p=2$ for the entropy conservative part.. For the diffusion operator, we use fourth-order ENO-based sign-preserving reconstruction. The derivatives $\left(\frac{\p B_x}{\p x}\right)_{i,j}$, $\left(\frac{\p B_y}{\p y}\right)_{i,j}$, $\left(\frac{\p \Psi}{\p x}\right)_{i,j}$, $\left(\frac{\p \Psi}{\p y}\right)_{i,j}$, $\left(\frac{\p \con}{\p x}\right)_{i,j}$, and $\left(\frac{\p \con}{\p y}\right)_{i,j}$, are approximated using fourth-order central difference. For the time update, we use the fourth-order SSP-RK method.	
\end{itemize}

To compute {\em isotropic} (MHD) solutions, we consider the source terms \eqref{eq:src_mhd}. As lower values of $\tau$ make the system stiff, we use ARK-IMEX schemes~\cite{pareschi2005implicit,kennedy2003additive}, where we treat $\mathcal{L}_{i,j}(\con(t))$ explicitly and $\s(\con_{i,j}(t))$ implicitly. We summarize these schemes as follows:

\begin{itemize}
	\item $\oti:$  The discretization of $\mathcal{L}_{i,j}(\con(t))$ is same as in the case of $\ote.$ For the ARK-IMEX time update, we follow, ~\cite{pareschi2005implicit} and consider second-order scheme which is L-stable, where  $\s(\con_{i,j}(t))$ is treated implicitly.
	
	\item $\othi:$ The discretization of $\mathcal{L}_{i,j}(\con(t))$ is same as in $\othe$. The source term $\s(\con_{i,j}(t))$ is treated implicitly, in the third-order ARK-IMEX scheme from~\cite{kennedy2003additive}.

	\item $\ofi:$ The term $\mathcal{L}_{i,j}(\con(t))$ is discretized as in $\ofe$. We use fourth order ARK-IMEX scheme from~\cite{kennedy2003additive}, where we treat $\s(\con_{i,j}(t))$, implicitly.	
\end{itemize}
\section{Numerical results}
\label{sec:NR}
In this section, we present the numerical results for the various one and two-dimensional test cases. \revg{Following~\cite{derigs2018ideal,rueda2025entropy}, we define the hyperbolic divergence cleaning speed based on the maximum non-GLM physical wave speed. For the one-dimensional cases we set
\[c_h = \alpha \left(\max\left\{ \bb{\Lambda}_x\setminus\{v_x + c_h, v_x-c_h\}\right\}\right),~~\alpha\in(0,1],\]
and for the two-dimensional cases,
\[
c_h = \alpha \left( \max\left\{ \bb{\Lambda}_x\setminus\{v_x + c_h, v_x-c_h\},  \bb{\Lambda}_y\setminus\{v_y + c_h, v_y-c_h\}\right\}\right),~~\alpha\in(0,1].
\]
where $\setminus$ denotes exclusion of the GLM waves from the eigenvalue sets and $\bb{\Lambda}_d$, $d\in\{x,y\}$, denotes the set of eigenvalues of system~\eqref{eq:glm_cgl_final} in the $d$-direction, as described in Appendix~\ref{Eiegn_value_admissible domain} (see Eqs.~\eqref{eq:fulleigen_glm_cgl_psi_x} and~\eqref{eq:fulleigen_glm_cgl_psi_y}). In both cases, $c_h$ is taken as a single global constant at each time step by taking the maximum over the entire computational domain and is therefore spatially constant. Therefore, no additional terms involving $\nabla c_h$ appear in the entropy analysis. We take $\alpha=1.0$.}

The time-step size is estimated as
$$
\Delta t =  \frac{\text{CFL}}{\max_{i,j} \left( \frac{\lambda_x(\con_{i,j})}{\Delta x} + \frac{\lambda_y(\con_{i,j})}{\Delta y} \right)},
$$
where $$\lambda_x=\max_{\lambda\in\mathbf{\Lambda}_x} |\lambda| \quad \text{and}\quad\lambda_y=\max_{\lambda\in\mathbf{\Lambda}_y} |\lambda|.$$

\revg{and $\mathbf{\Lambda}_x$ and $\mathbf{\Lambda}_y$ denote the eigenvalue sets of the system~\eqref{eq:glm_cgl_final} in the $x$- and $y$-directions, respectively, as given in Appendix~\ref{Eiegn_value_admissible domain} (see Eqs.~\eqref{eq:fulleigen_glm_cgl_psi_x} and~\eqref{eq:fulleigen_glm_cgl_psi_y}).}
We take $\text{CFL}=0.4$ for all test cases. For the source terms, we take $\tau=10^{-5}.$ We call these solutions as {\em isotropic GLM-CGL} solutions. The solution without the source terms is termed as {\em GLM-CGL} solutions. If no GLM divergence diminishing is used, the solution without source terms is called {\em CGL} solution. Similarly, solutions with source terms, but no GLM divergence diminishing, are termed as {\em isotropic CGL} solutions. To compare the isotropic solutions and MHD solutions, we have computed MHD solutions. For that, we use second-order finite volume methods with MinMod-based reconstruction. We use the Rusanov solver for the numerical flux. For 1D, we use $10000$ cells. For 2D, we use highly refined mesh of $1000 \times 1000$ cells.

For two-dimensional test cases, to show the effect of GLM terms on the simulations, we calculate the divergence errors for the magnetic field in $L_1-$norm,
\[\|\nabla\cdot\B\|_{1} = \frac{1}{N_x N_y}\sum_{i=1}^{N_x}\sum_{j=1}^{N_y}\left|(\nabla\cdot\B)_{i,j}\right|,\]
and $L_2$-norm,
\[\|\nabla\cdot\B\|_{2} = \left[\frac{1}{N_x N_y}\sum_{i=1}^{N_x}\sum_{j=1}^{N_y}\left|(\nabla\cdot\B)_{i,j}\right|^2\right]^\frac{1}{2}.\]
Here $N_x$ and $N_y$ denote the number of grid cells in the $x$ and $y$-directions. The discrete divergence error $(\nabla\cdot\B)_{i,j}$ is evaluated using
\[(\nabla\cdot\B)_{i,j} = \frac{B_{x,i+1,j}-B_{x,i-1,j}}{2\Delta x} + \frac{B_{y,i,j+1}-B_{y,i,j-1}}{2\Delta y}.\]

\revg{Although the proposed schemes do not ensure the positivity of the solutions, we do not employ any floor values to satisfy the positivity constraints. Furthermore, for each test case, we also present the minimum values of density and pressure components.}
\subsection{One-dimensional accuracy test}
\label{test:liu_2025}

\begin{table}[tbhp]
	\footnotesize
	\begin{center}
		\begin{tabular}{c|c|c|c|c|c|c|}
			\hline Number of cells  & \multicolumn{2}{|c}{$\ote$} & \multicolumn{2}{|c}{$\othe$} & \multicolumn{2}{|c}{$\ofe$}  \\
			\hline   & $L_1$ error  &  Order &  $L_1$ error      & Order & $L_1$ error      & Order \\
			\hline 12 & 4.86E-02 & -- & 7.14E-03 & -- & 1.98E-03 & -- \\
			24 & 1.51E-02 & 1.688145847 & 9.35E-04 & 2.932342894 & 2.09E-04 & 3.241997756 \\
			48 & 6.67E-03 & 1.176375673 & 1.18E-04 & 2.986256366 & 1.61E-05 & 3.704200874 \\
			96 & 1.93E-03 & 1.789540351 & 1.48E-05 & 2.996555203 & 1.16E-06 & 3.793435746 \\
			192 & 5.32E-04 & 1.85985482 & 1.85E-06 & 2.999133857 & 8.02E-08 & 3.852029024 \\
			384 & 1.46E-04 & 1.863664412 & 2.31E-07 & 2.999760591 & 5.40E-09 & 3.892946082 \\
			\hline
		\end{tabular}
		\caption{\textbf{\nameref{test:liu_2025}}: $L_1$ errors and order of accuracy for $\rho$.}
		\label{tab:1}
	\end{center}
\end{table} 

In this test case, we demonstrate the accuracy of the proposed numerical schemes for the GLM-CGL system. We adapt the MHD test problem from~\cite{liu2025entropy}. The computational domain is $[0,2\pi]$, with periodic boundary conditions. The initial data is given as,
\begin{align*}
	\rho(x,0) &= 1+0.2 \sin{(x)},\\
	\left(\bu,\pll,\per,\B,\Psi\right)& = \left(1,0,0,2,2,0.5,1.0,1.5,0\right).
\end{align*}
We do not consider the source term in this test case. The simulations are performed till time $t=1.3.$ The exact solution of the problem is $\rho(x,t) = 1+0.2 \sin{(x-t)}$, which is the advection of the density profile with unit velocity in the $x$-direction. In Table~\eqref{tab:1}, we have presented the $L_1$-errors and numerical order of convergence for density variables, using $\ote$,  $\othe$ and  $\ofe$ schemes for the GLM-CGL system. All the schemes have achieved the predicted order.


\subsection{One-dimensional artificial non-zero magnetic field divergence test}\label{test:1d_artificial}
\begin{figure}[!htbp]
	\begin{center}
		\subfigure[$B_x$ for $\ote$ scheme for CGL ]{\includegraphics[width=0.28\textwidth]{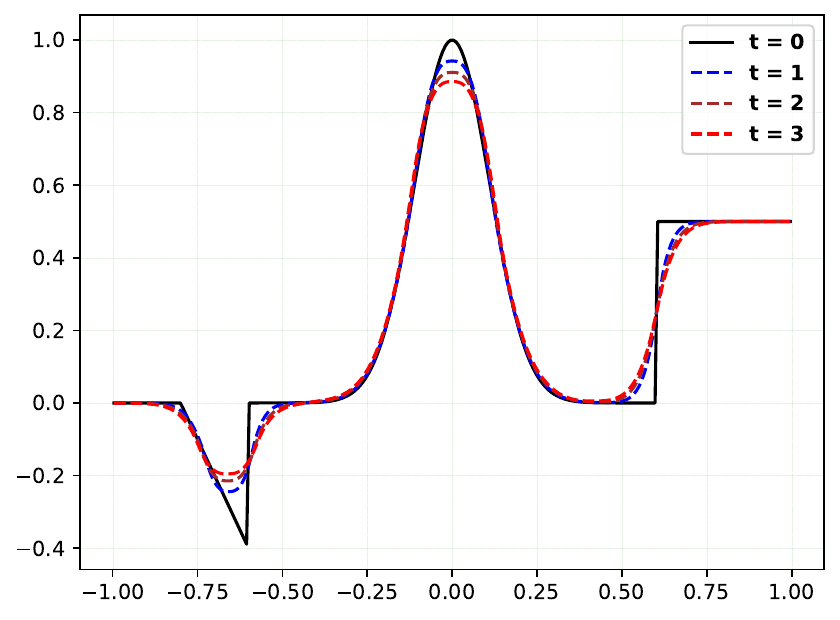}\label{fig:O2_wo_glm}}~
		\subfigure[$B_x$ for $\othe$ scheme for CGL ]{\includegraphics[width=0.28\textwidth]{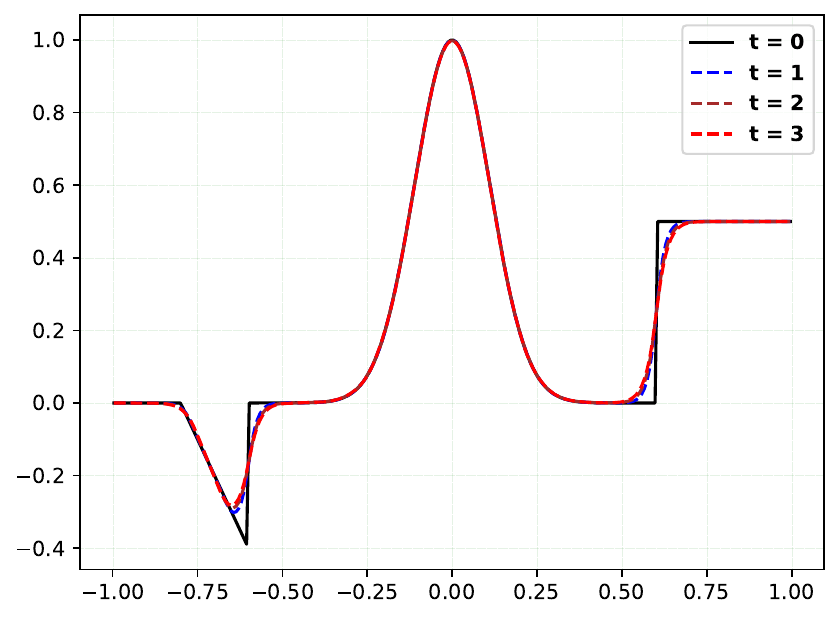}\label{fig:O3_wo_glm}}~
		\subfigure[$B_x$ for $\ofe$ scheme for CGL ]{\includegraphics[width=0.28\textwidth]{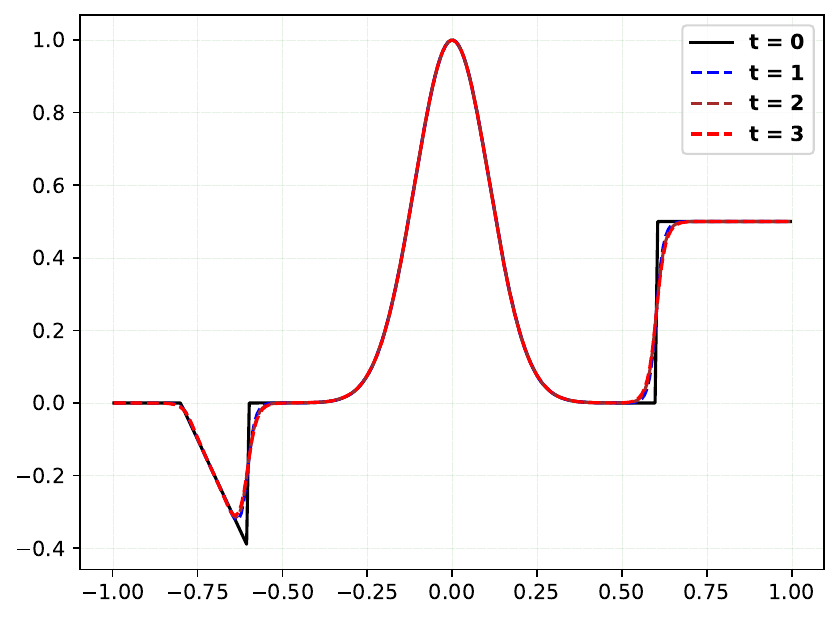}\label{fig:O4_wo_glm}}\\
		\subfigure[$B_x$ for $\ote$ scheme for GLM-CGL]{\includegraphics[width=0.28\textwidth]{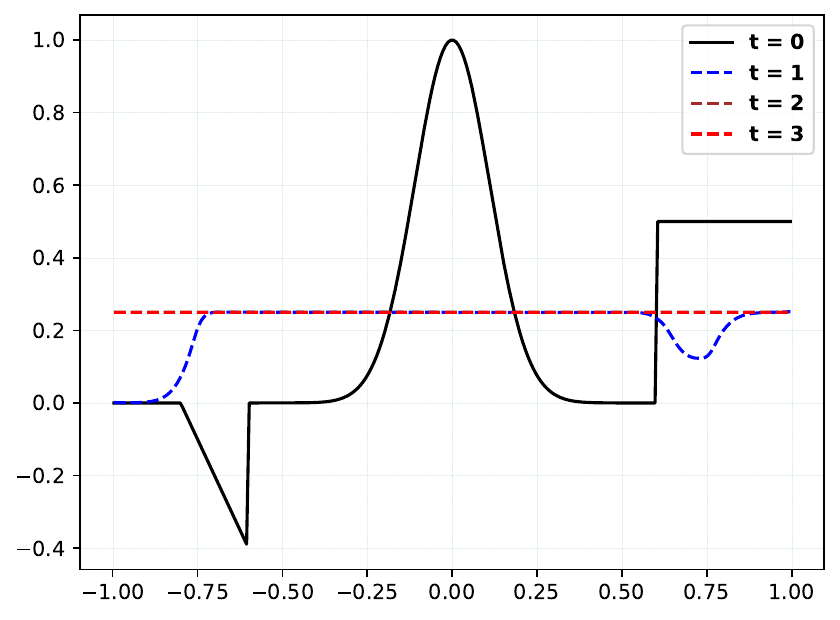}\label{fig:O2_w_glm}}~
		\subfigure[$B_x$ for $\othe$ scheme for GLM-CGL]{\includegraphics[width=0.28\textwidth]{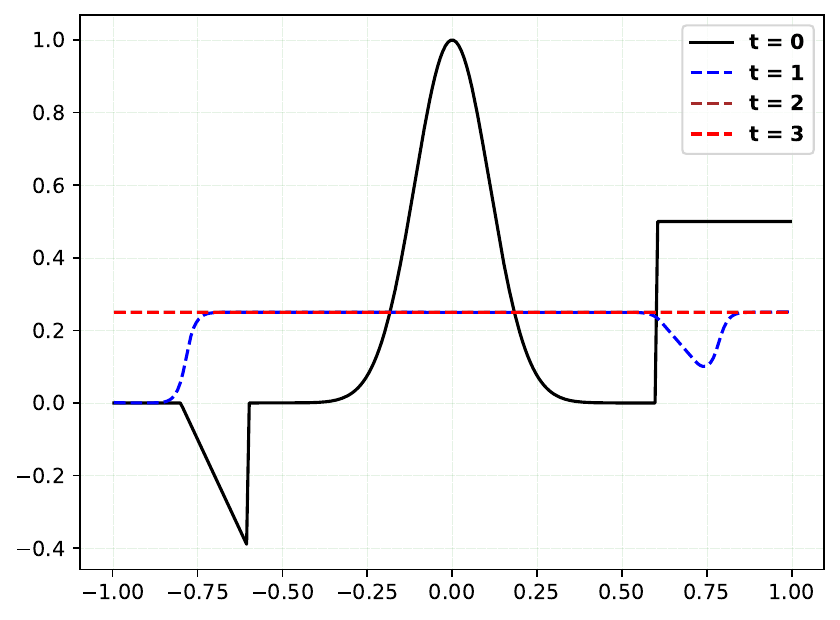}\label{fig:O3_w_glm}}~
		\subfigure[$B_x$ for $\ofe$ scheme for GLM-CGL]{\includegraphics[width=0.26\textwidth]{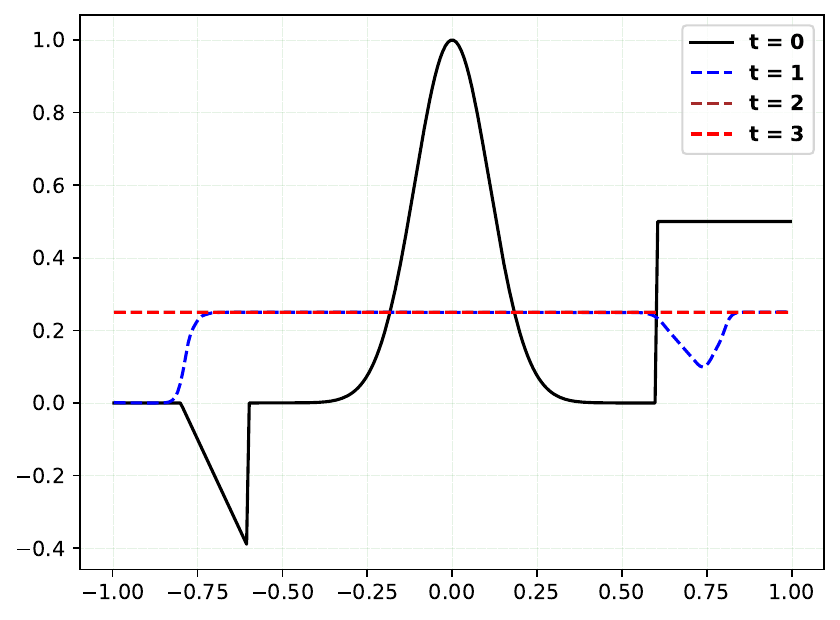}\label{fig:O4_w_glm}}\\
		\caption{\textbf{\nameref{test:1d_artificial}}: Plots of magnetic field evolution for $\ote$, $\othe$ and $\ofe$ schemes for CGL and GLM-CGL at four different time steps.}
		\label{fig:1d_artificial}
	\end{center}
\end{figure}
\revg{Following~\cite{derigs2018ideal}, we consider an initial condition with a non-zero discrete divergence of the magnetic field. In practical computations, numerically generated initial data may fail to satisfy the divergence-free constraint exactly, resulting in small divergence errors. We present the proposed GLM-CGL framework to demonstrate its ability to effectively handle such divergence errors.} The computational domain is taken to $[-1,1]$ with Neumann boundary conditions. The magnetic field component $B_x$ is given by,
\[B_x = \begin{cases}
	0.0, & \textrm{if } x \leq -0.8\\
	-2(x + 0.8), & \textrm{if } -0.8 < x \leq -0.6\\
	\exp{\left(-\frac{(x/0.11)^2}{2}\right)}, & \textrm{if } -0.6 < x \leq 0.6\\
	0.5, & \textrm{otherwise.}
\end{cases}\]
The rest of the variables are taken to be constants and given by,
\begin{align*}
	\left(\rho,\bu,\pll,\per,B_y,B_z,\Psi\right)&= \left(1,0,0,0,1,1,0,0,0\right).
\end{align*}
To demonstrate the effect of GLM-based formulation, we have compared the GLM-CGL formulation with the CGL equations without GLM. The numerical results for $\ote$, $\othe$ and $\ofe$ schemes for both the formulations are presented in Figure~\eqref{fig:1d_artificial} at time $t=0, 1, 2 $ and $3$. 

For the CGL equation without GLM, we observe that $\ote$, $\othe$ and $\ofe$ schemes (Figures~\eqref{fig:O2_wo_glm},~\eqref{fig:O3_wo_glm} and~\eqref{fig:O4_wo_glm}) preserve the initial profile of the $x-$magnetic field component $(B_x).$ So, the divergence of the initial profile is not diminished. We observe only slight diffusion in the $B_x$ profile over time. 

In Figures~\eqref{fig:O2_w_glm},~\eqref{fig:O3_w_glm} and~\eqref{fig:O4_w_glm}, we have plotted the magnetic field evolution $B_x$ for the $\ote$, $\othe$ and $\ofe$ schemes for GLM-CGL. In this case, we observe that as time increases, the divergence of the magnetic field decreases and becomes zero at $t=3$. This demonstrates that the GLM-CGL formulation deals with the divergence effectively, even in a one-dimensional case. We also observe that all the numerical schemes behave similarly, with higher-order schemes being less diffusive. Furthermore, the results are comparable to those obtained in~\cite{derigs2018ideal} for MHD equations.

\subsection{Brio-Wu shock tube problem}
\label{test:bw}
\begin{figure}[!htbp]
	\begin{center}
		\subfigure[Density for explicit schemes for CGL and GLM-CGL]{\includegraphics[width=0.28\textwidth]{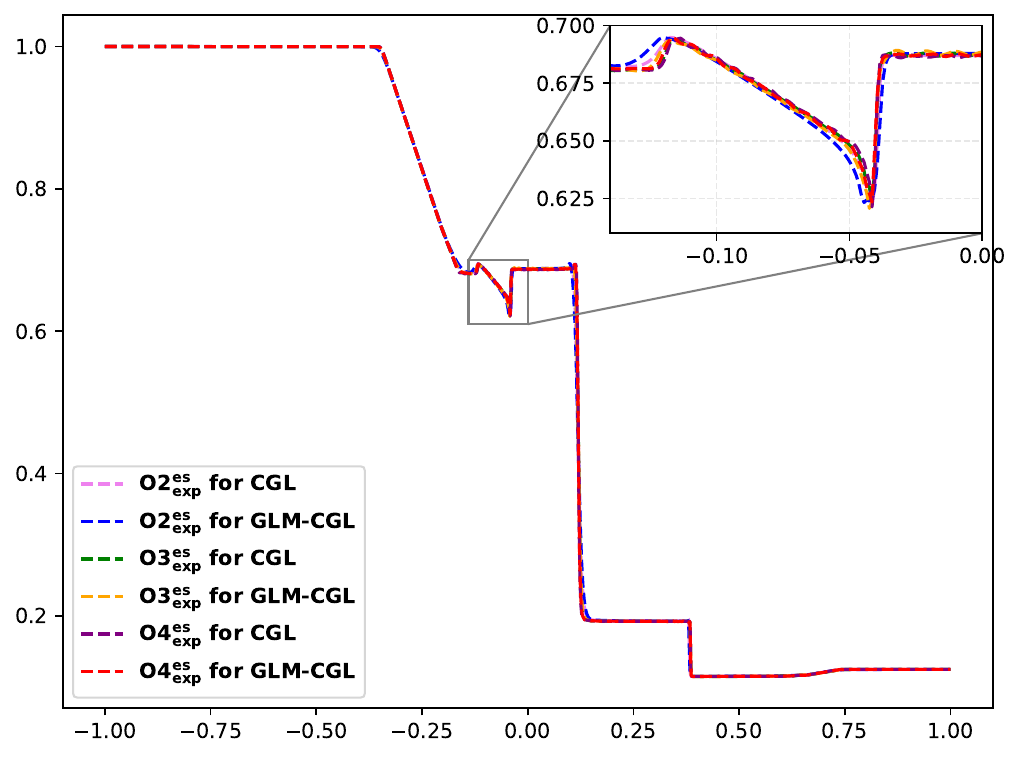} \label{fig:bw_exp_rho}}~
		\subfigure[$\pll$ for explicit schemes for CGL and GLM-CGL]{\includegraphics[width=0.28\textwidth]{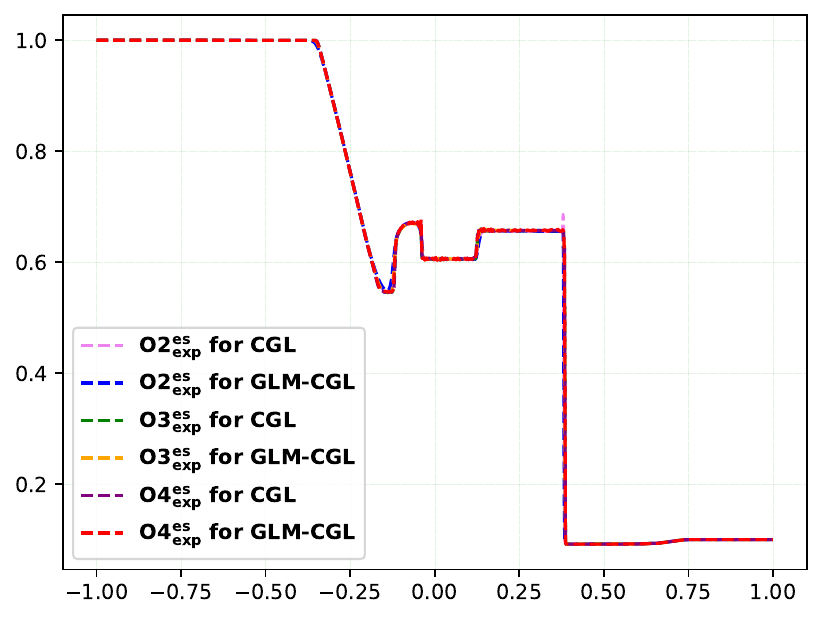}\label{fig:bw_exp_pll}}~
		\subfigure[$\per$ for explicit schemes for CGL and GLM-CGL]{\includegraphics[width=0.28\textwidth]{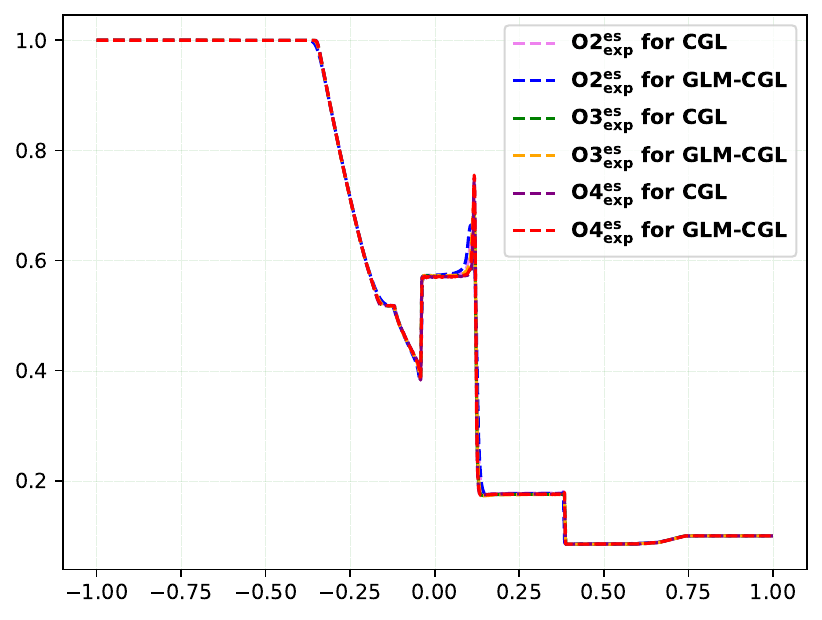}\label{fig:bw_exp_perp}}\\
		\subfigure[Density for IMEX schemes for isotropic CGL and isotropic GLM-CGL]{\includegraphics[width=0.28\textwidth]{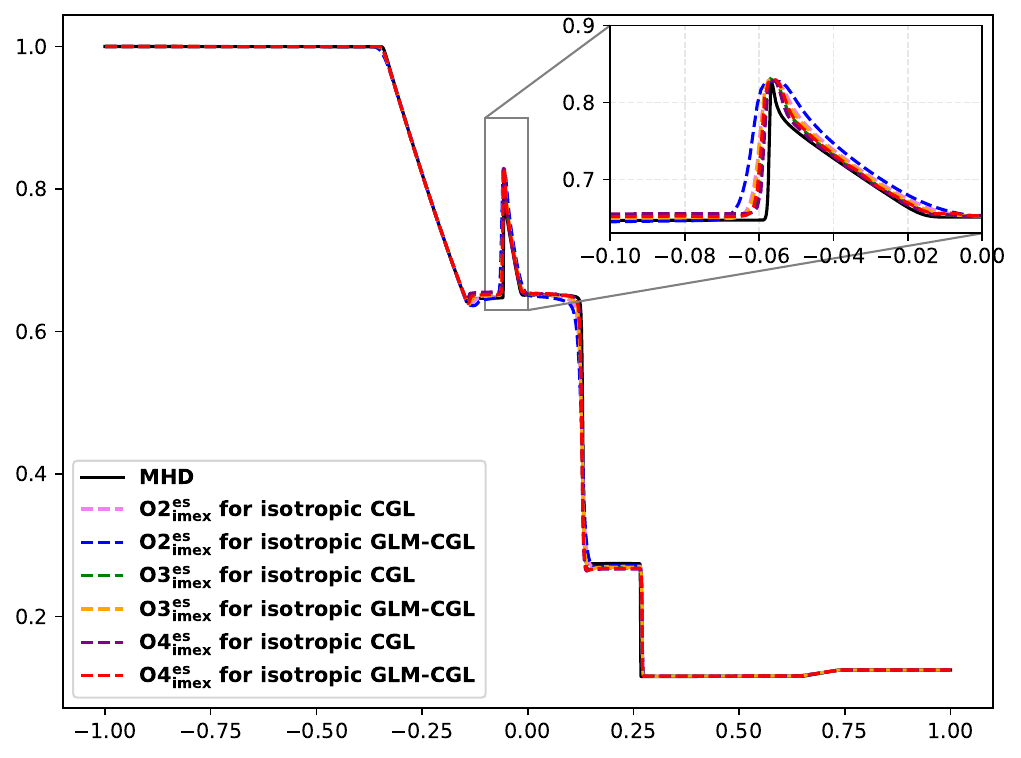}\label{fig:bw_imp_rho}}~
		\subfigure[$\pll$ for IMEX schemes for isotropic CGL and isotropic GLM-CGL]{\includegraphics[width=0.28\textwidth]{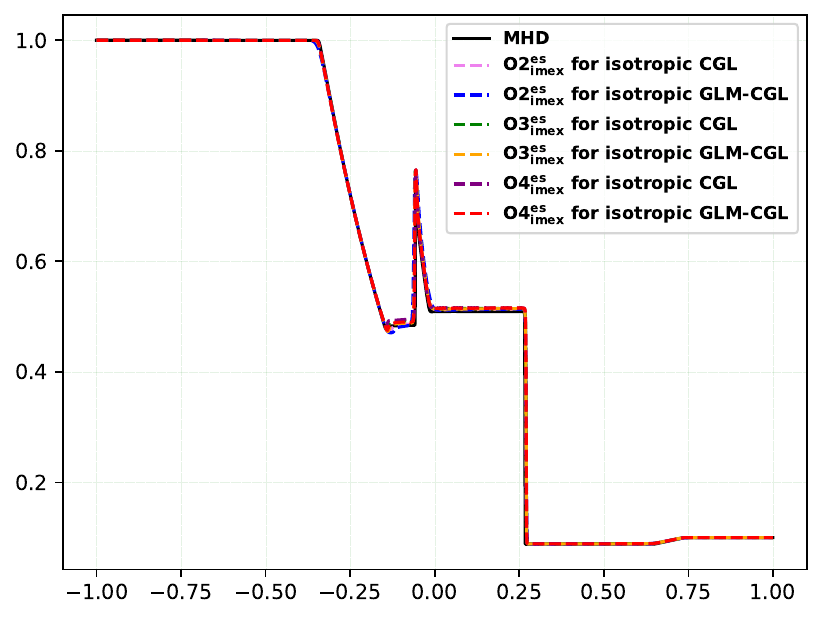}\label{fig:bw_imp_pll}}~
		\subfigure[$\per$ for IMEX schemes for isotropic CGL and isotropic GLM-CGL]{\includegraphics[width=0.28\textwidth]{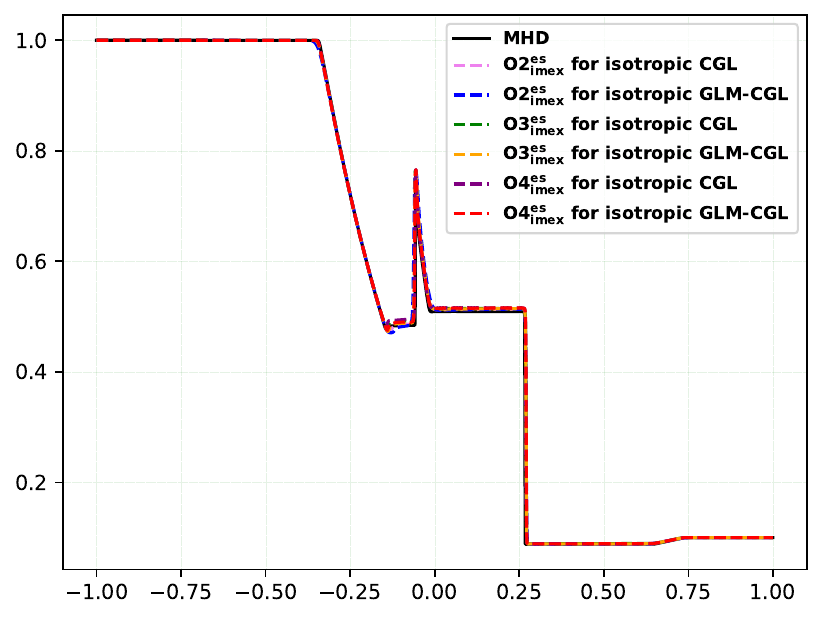}\label{fig:bw_imp_perp}}
		\caption{\textbf{\nameref{test:bw}}: Plots of density, parallel and perpendicular pressure components for explicit schemes and IMEX schemes using $2000$ cells at final time $t = 0.2$.}
		\label{fig:bw}
	\end{center}    
\end{figure}

\begin{table}[ht]
\centering
\renewcommand{\arraystretch}{1.4}
\begin{tabular}{|c|c|c|c|c|}
\hline
\multirow{2}{*}{Scheme} &
\multirow{2}{*}{System} &
$\rho_{\min}$ &
$(p_{\parallel})_{\min}$ &
$(p_{\perp})_{\min}$ \\
&&&&\\
\hline

\multirow{2}{*}{$\ote$}
& CGL
& 1.152440E-01 & 9.183145E-02 &  8.517001E-02 \\ \cline{2-5}

& GLM--CGL
& 1.152106E-01 & 9.180354E-02 &8.512142E-02   \\  \hline

\multirow{2}{*}{$\othe$}
& CGL
& 1.151120E-01 & 9.172041E-02 & 8.497582E-02  \\ \cline{2-5}

& GLM--CGL
& 1.150731E-01 & 9.168778E-02 & 8.491903E-02  \\  \hline

\multirow{2}{*}{$\ofe$}
& CGL
& 1.151282E-01  & 9.171415E-02 & 8.499921E-02 \\ \cline{2-5}

& GLM--CGL
& 1.150758E-01 & 9.166919E-02 &  8.492299E-02 \\  \hline

\multirow{2}{*}{$\oti$}
& Isotropic CGL
& 1.161367E-01 & 8.846605E-02 & 8.846605E-02 \\ \cline{2-5}

& Isotropic GLM--CGL
& 1.160663E-01 &  8.837670E-02& 8.837670E-02  \\  \hline

\multirow{2}{*}{$\othi$}
& Isotropic CGL
& 1.163002E-01 & 8.867169E-02 &   8.867169E-02\\ \cline{2-5}

& Isotropic GLM--CGL
& 1.162022E-01 & 8.854720E-02 &  8.854720E-02  \\  \hline

\multirow{2}{*}{$\ofi$}
& Isotropic CGL
&  1.163883E-01& 8.878357E-02 &  8.878357E-02 \\ \cline{2-5}

& Isotropic GLM--CGL
& 1.162905E-01 & 8.865944E-02 & 8.865944E-02 \\ \hline

\end{tabular}
\caption{\textbf{\nameref{test:bw}}: Minimum values of density $\rho$, parallel pressure $p_{\parallel}$, and perpendicular pressure $p_{\perp}$.}
\label{tab:min_bw}
\end{table}

\begin{figure}[!htbp]
	\begin{center}
		\subfigure[Density for $\ofe$ schemes for GLM-CGL]{\includegraphics[width=\textwidth]{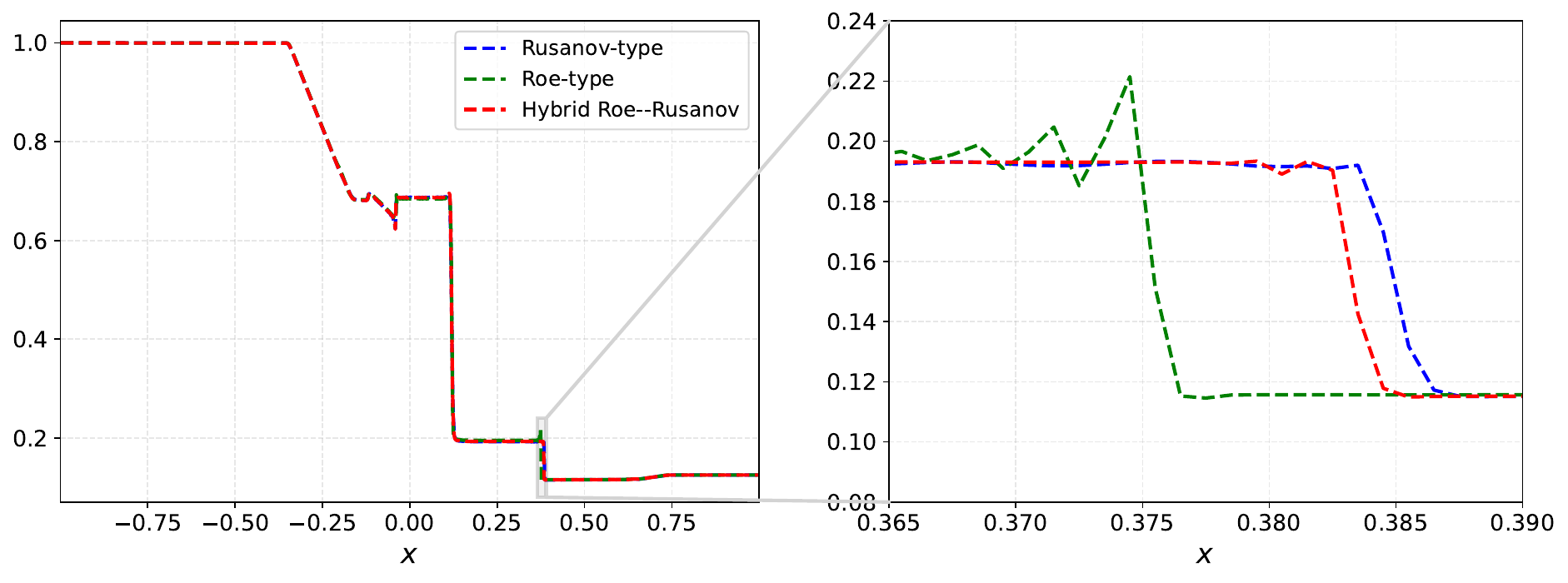} \label{fig:bw_com_rho}}\\
		\subfigure[$\pll$ for $\ofe$ schemes for GLM-CGL]{\includegraphics[width=\textwidth]{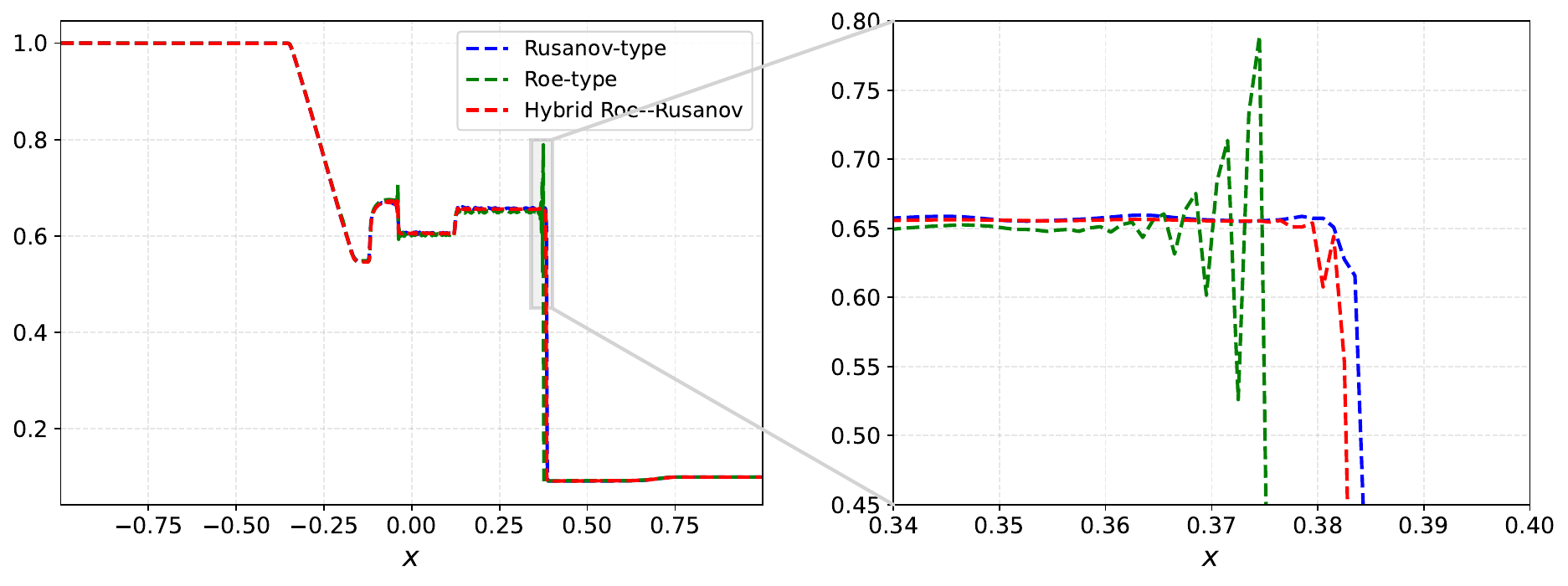}\label{fig:bw_com_pll}}
		\caption{\textbf{\nameref{test:bw}}: Comparison of the density and parallel pressure profiles obtained using the Rusanov-type, Roe-type, and hybrid Roe-Rusanov dissipation operators with $\ofe$ scheme for GLM-CGL using $2000$ cells at the final time $t=0.2$.}
		\label{fig:Comparison}
	\end{center}    
\end{figure}
In this test, we generalize the Brio-Wu shock tube test case for MHD equations~\cite{Brio1988upwind}. The test was also considered in~\cite{singh2024entropy,singh2024eigen}. Following~\cite{singh2024entropy}, we take the computational domain to be $[-1,1]$ with Neumann boundary conditions. The initial conditions are:
\[(\rho, \bu, \pll, \per, B_x,B_{y}, B_{z},\Psi) = \begin{cases}
	(1, 0, 0, 0, 1, 1, 0.75,1, 0,0), & \textrm{if } x\leq 0\\
	(0.125, 0, 0, 0, 0.1, 0.1, 0.75,-1, 0,0), & \textrm{otherwise}
\end{cases}\]

In Figure \eqref{fig:bw}, we present the numerical results at time $t=0.2$ for CGL and GLM-CGL formulations. We also consider the isotropic CGL and isotropic GLM-CGL by considering the source term and using IMEX schemes $\oti$, $\othi$ and $\ofi$. The simulations use $2000$ cells. We have also compared MHD and isotropic solutions. We have plotted profiles of density, $\pll$ and $\per$ for all the test cases.

For the CGL and GLM-CGL, we observe that all schemes successfully resolve the waves. We observe small-scale oscillations (similar to those in~\cite{singh2024entropy}), but they decrease with further refinements. \revo{We also observe that the GLM-CGL formulation exhibits slightly more dissipative behavior than the baseline CGL formulation for the same order of scheme. This behavior is expected because the GLM formulation introduces an auxiliary scalar field $\Psi$, which generates two additional hyperbolic cleaning waves, $v_d\pm c_h$. These waves continuously propagate and damp magnetic divergence errors, and the associated GLM correction introduces additional numerical dissipation that is absent in the baseline CGL formulation. Therefore, the improved divergence control comes with a small increase in numerical dissipation.}

For the isotropic cases, we note that both $\pll$ and $\per$ have the same profile (which matches the MHD pressure profile), as expected. Furthermore, the density profile matches the MHD density profile. Furthermore, we again see that isotropic GLM-CGL solutions are slightly more diffusive than the isotropic CGL solution, as discussed above. In all the cases, we also observe that higher-order schemes are more accurate.

\revg{Table~\eqref{tab:min_bw} lists the minimum values of density, $\pll$, and $\per$ at the final time. All values remain positive for all schemes and formulations, indicating that physically admissible solutions are maintained throughout the computation. For the isotropic cases, we additionally observe that $\pll=\per$, as expected.}

\revo{To further investigate the effect of the dissipation operator, we compare the Rusanov-type, Roe-type, and hybrid Roe-Rusanov dissipation operators introduced in Section~\ref{Sec:Entropy_SS}. Figure~\ref{fig:Comparison} compares the density and parallel pressure profiles obtained using the three dissipation operators. The left panels show the complete solution, while the right panels present zoomed views near the discontinuity. All three operators capture the overall wave structure well. However, the zoomed plots show that the Roe-type operator develops noticeable oscillations near the discontinuity. These oscillations are reduced by the hybrid Roe-Rusanov operator, whose solution is very close to that obtained with the Rusanov-type operator. These results demonstrate that the Rusanov-type operator provides a good balance between robustness and accuracy, while the hybrid Roe--Rusanov operator offers no significant improvement for this test case.}
\subsection{Two-dimensional accuracy test}
\label{test:liu_2025_2D} 
\begin{table}[tbhp]
	\footnotesize
	\begin{center}
		\begin{tabular}{c|c|c|c|c|c|c|}
			\hline Number of cells  & \multicolumn{2}{|c}{$\ote$} & \multicolumn{2}{|c}{$\othe$} & \multicolumn{2}{|c}{$\ofe$}  \\
			\hline   & $L_1$ error  &  Order &  $L_1$ error      & Order & $L_1$ error      & Order \\
			\hline 24 $\times$ 24 & 2.22E-02 & -- &  1.47E-03 & -- & 3.29E-04 & -- \\
			48 $\times$ 48 & 9.66E-03 & 1.197102245 & 1.88E-04 & 2.97345722 & 2.60E-05 & 3.665723172 \\
			96 $\times$ 96 & 3.00E-03 & 1.685415758 & 2.35E-05 & 2.995912704 & 1.87E-06 & 3.79493394 \\
			192 $\times$ 192 & 8.43E-04 & 1.834041352 & 2.94E-06 & 2.999439213 & 1.30E-07 & 3.845890159 \\
			384$\times$ 384 & 2.32E-04 & 1.858969497 & 3.68E-07 & 2.999783442 & 8.78E-09 & 3.888987053 \\
			\hline
		\end{tabular}
		\caption{\textbf{\nameref{test:liu_2025_2D}}: $L_1$ errors and order of accuracy for $\rho$.}
		\label{tab:3}
	\end{center}
\end{table} 
For the two-dimensional accuracy test, we generalize the one-dimensional accuracy test in Section \ref{test:liu_2025} to two dimensions for the GLM-CGL system. We consider the computational domain  $[0,2\pi]\times[0,2\pi]$ with periodic boundary conditions. 
The initial conditions are given below,
\begin{align*}
	\rho(x,y,0) &= 1+0.2 \sin{(x+y)},\\
	\left(\bu,\pll,\per,\B,\Psi\right)& = \left(0.5,0.5,0,2,2,0.5,1.0,1.5,0\right).
\end{align*}
The exact solution is then given by $\rho(x,y,t) = 1+0.2 \sin{(x+y-t)}$. For the $\ote$,  $\othe$ and  $\ofe$ schemes with GLM, the $L_1$-errors of the density are shown in Table\eqref{tab:3} at the final time of $t=1.3$. We note that all the schemes converge with the theoretically predicted order of accuracy.
\subsection{Two-dimensional Circularly polarized Alfv\'en waves}\label{test:2d_alfven}
To test the accuracy of the magnetic field components, in this test we consider the generalization of the two-dimensional MHD circularly polarized alfv\'en waves problem, which was introduced in~\cite{toth2000b}, for the GLM-CGL system.  We consider the computational domain is $\left[0,\frac{1}{\cos{\alpha}}\right]\times\left[0,\frac{1}{\sin{\alpha}}\right]$, with periodic boundary conditions. The initial condition is given as
\begin{align*}
	\left(\rho,\pll,\per,\Psi \right)& = \left(1, 0.1, 0.1, 0\right),\\
	\bu &= \left(-v_\perp \sin{\alpha},~v_\perp \cos{\alpha},~0.1 \cos{(2\pi x_{||})}\right),\\
	\B &= \left(B_{||}\cos{\alpha}-B_{\perp}\sin{\alpha},~B_{||}\sin{\alpha}+B_{\perp}\cos{\alpha},~v_z\right),
\end{align*}
where $B_{||}=1,~B_{\perp}=v_{\perp}=0.1 \sin{(2\pi x_{||})}$ and $x_{||}=x\cos{\alpha}+y\sin{\alpha}$. Circularly polarized alfv\'en wave have a unit wavelength in the direction of the $x$-axis and propagate at an angle $\alpha = 30^o$. For the GLM-CGL with $\ote$,  $\othe$ and  $\ofe$ schemes, the $L_1$-errors of $B_y$ are shown in Table\eqref{tab:4} at the final time of $t=5.0$. We note that every scheme has the desired order of accuracy.
\begin{table}[tbhp]
	\footnotesize
	\begin{center}
		\begin{tabular}{c|c|c|c|c|c|c|}
			\hline Number of cells  & \multicolumn{2}{|c}{$\ote$} & \multicolumn{2}{|c}{$\othe$} & \multicolumn{2}{|c}{$\ofe$}  \\
			\hline   & $L_1$ error  &  Order &  $L_1$ error      & Order & $L_1$ error      & Order \\
			\hline 32 $\times$ 32 & 2.53E-02 & -- & 1.25E-03 & -- & 2.92E-04 & -- \\
			64 $\times$ 64 & 8.06E-03 & 1.649122178 & 1.57E-04 & 2.997015955 & 2.39E-05 & 3.606907328 \\
			128 $\times$ 128 & 3.49E-03 & 1.208907716 & 2.01E-05 & 2.964144554 & 1.69E-06 & 3.825756677 \\
			256 $\times$ 256 & 9.84E-04 & 1.824933152 & 2.55E-06 & 2.978762804 & 1.13E-07 & 3.898604125 \\
			512 $\times$ 512 & 2.75E-04 & 1.840806962 & 3.26E-07 & 2.965870547 & 6.36E-09 & 4.153053838 \\
			\hline
		\end{tabular}
		\caption{\textbf{\nameref{test:2d_alfven}}: $L_1$ errors and order of accuracy for $B_y$.}
		\label{tab:4}
	\end{center}
\end{table} 
\subsection{Orszag-Tang Problem}\label{test:ot}
\begin{figure}[!htbp]
	\begin{center}
		\subfigure[Density for $\oti$ scheme for isotropic CGL]{\includegraphics[width=0.26\textwidth]{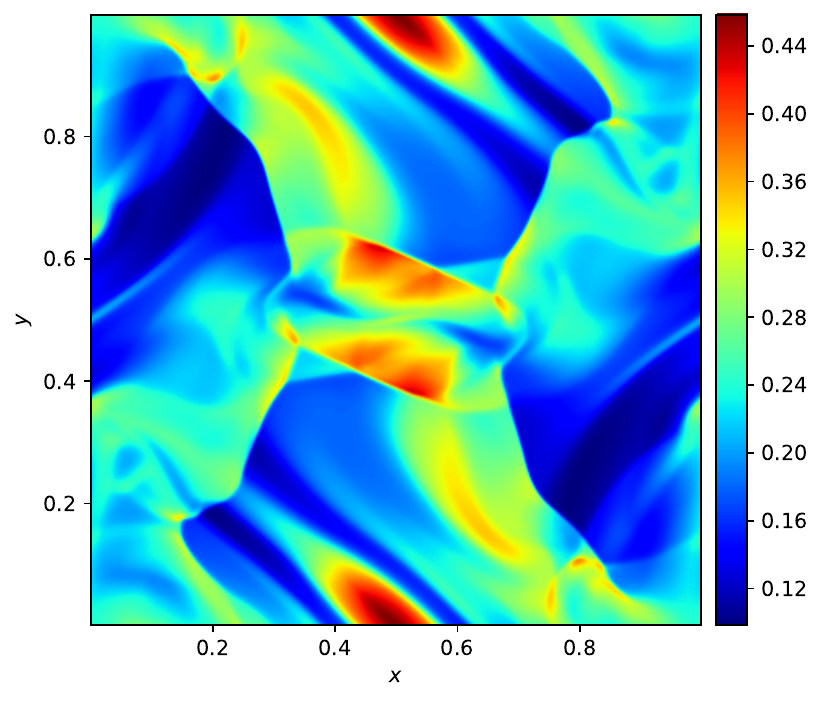}\label{fig:ot_rho_wo_o2}}~
		\subfigure[Density for $\othi$ scheme for isotropic CGL]{\includegraphics[width=0.26\textwidth]{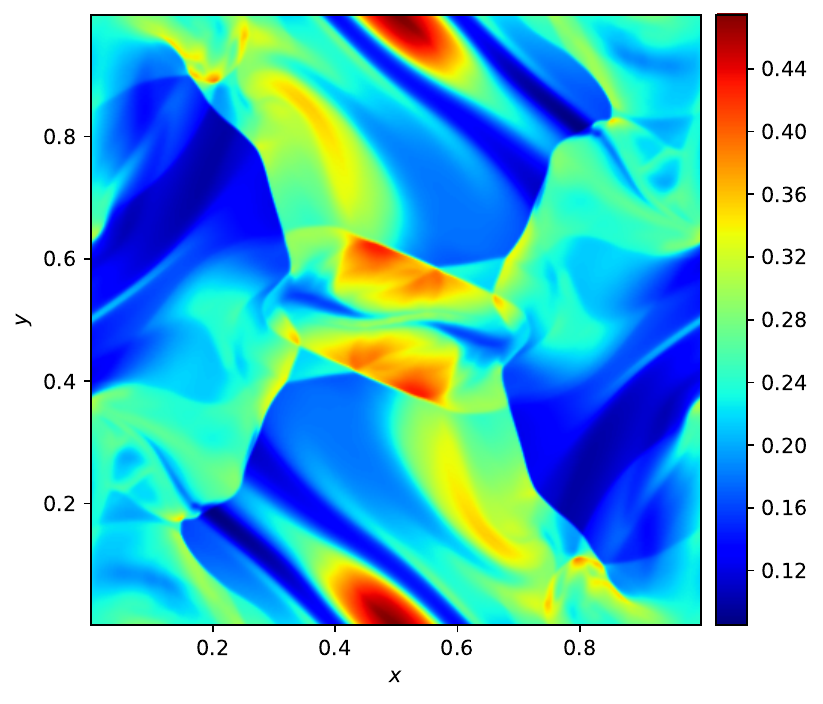}\label{fig:ot_rho_wo_o3}}~\
		\subfigure[Density for $\ofi$ scheme for isotropic CGL]{\includegraphics[width=0.26\textwidth]{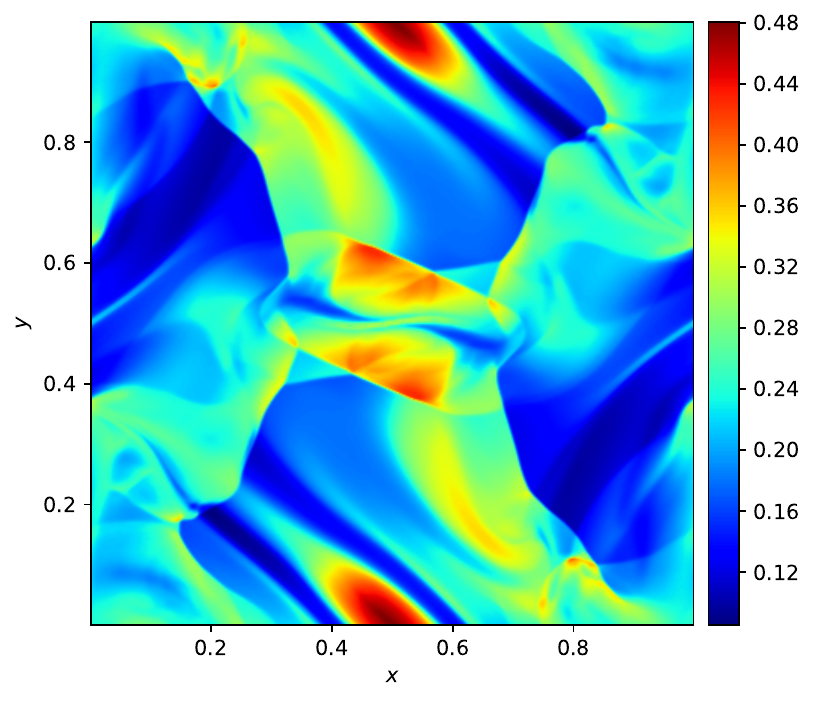}\label{fig:ot_rho_wo_o4}}\\
		\subfigure[Density for $\oti$ scheme for isotropic GLM-CGL]{\includegraphics[width=0.26\textwidth]{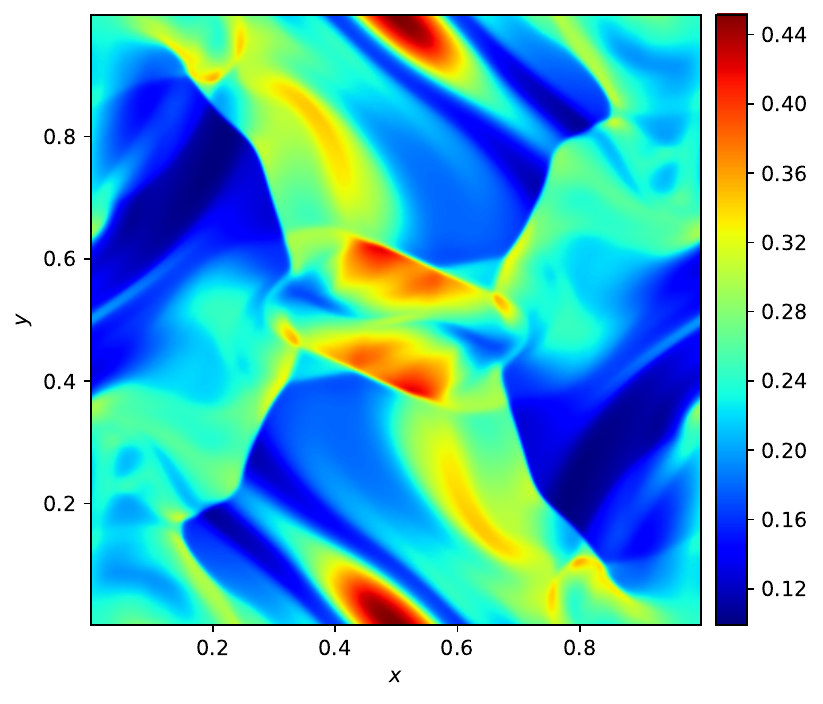}\label{fig:ot_rho_w_o2}}~
		\subfigure[Density for $\othi$ scheme for isotropic GLM-CGL]{\includegraphics[width=0.26\textwidth]{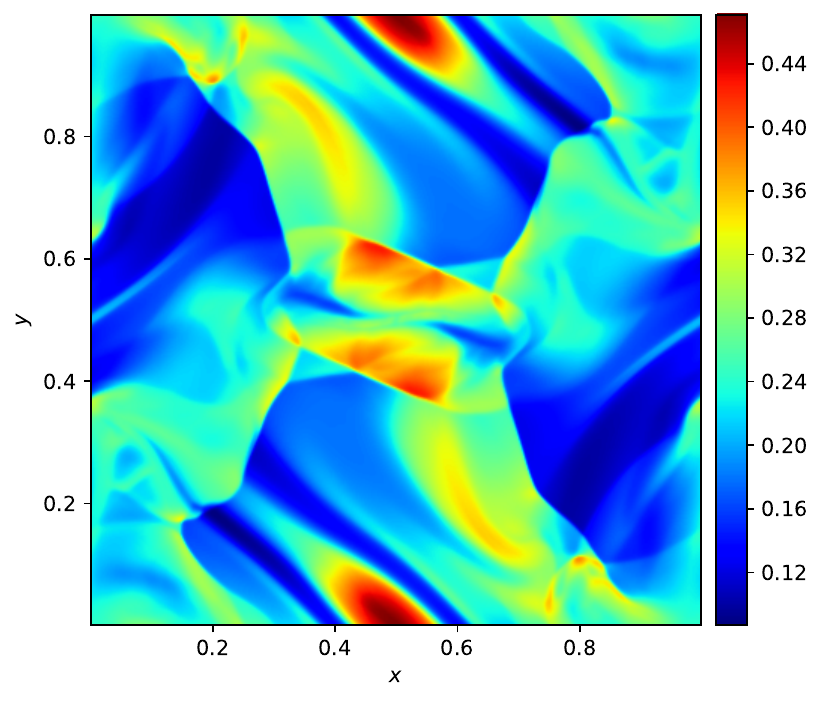}\label{fig:ot_rho_w_o3}}~
		\subfigure[Density for $\ofi$ scheme for isotropic GLM-CGL]{\includegraphics[width=0.26\textwidth]{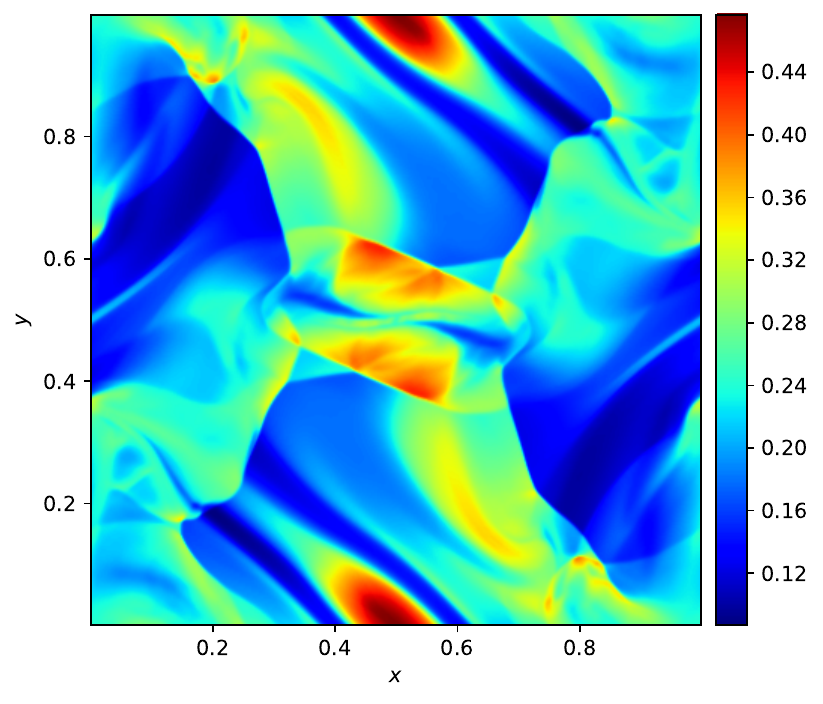}\label{fig:ot_rho_w_o4}}\\
		\subfigure[$|(\nabla\cdot\B)_{i,j}|$ for $\oti$ scheme for isotropic CGL]{\includegraphics[width=0.26\textwidth]{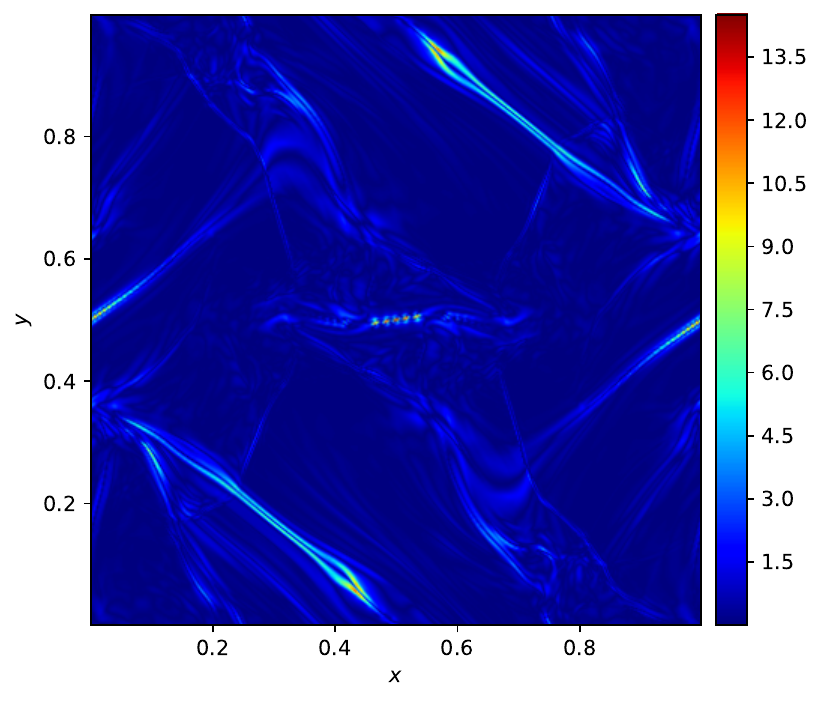}\label{fig:ot_db_wo_o2}}~
		\subfigure[$|(\nabla\cdot\B)_{i,j}|$ for $\othi$ scheme for isotropic CGL]{\includegraphics[width=0.26\textwidth]{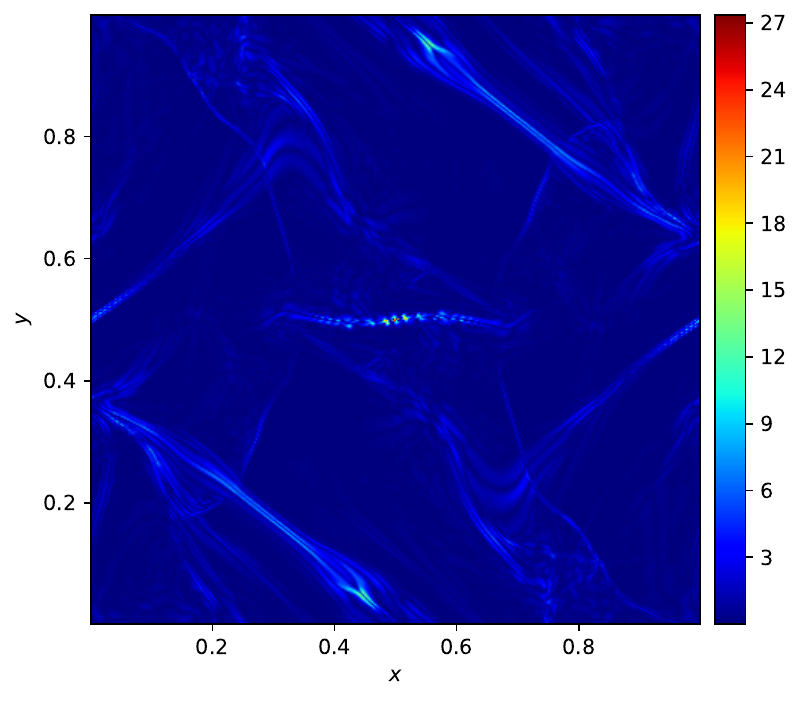}\label{fig:ot_db_wo_o3}}~
		\subfigure[$|(\nabla\cdot\B)_{i,j}|$ for $\ofi$ scheme for isotropic CGL]{\includegraphics[width=0.26\textwidth]{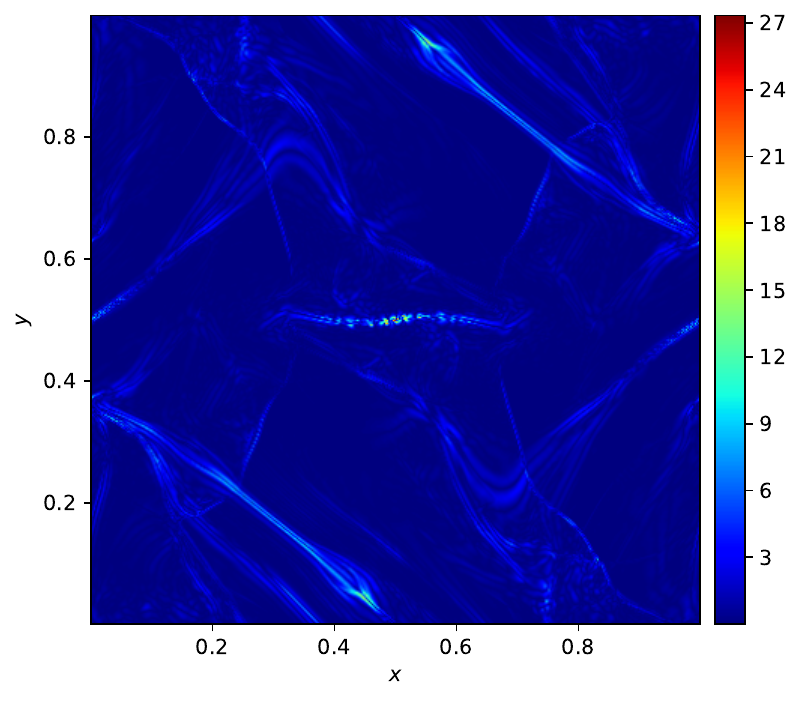}\label{fig:ot_db_wo_o4}}\\
		\subfigure[$|(\nabla\cdot\B)_{i,j}|$ for $\oti$ scheme for isotropic GLM-CGL]{\includegraphics[width=0.26\textwidth]{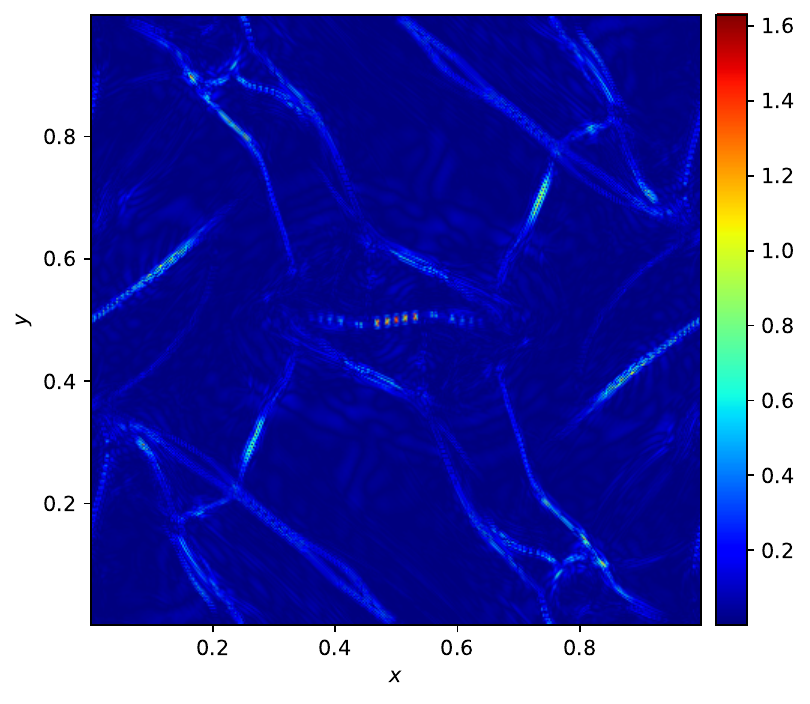}\label{fig:ot_db_w_o2}}~
		\subfigure[$|(\nabla\cdot\B)_{i,j}|$ for $\othi$ scheme for isotropic GLM-CGL]{\includegraphics[width=0.26\textwidth]{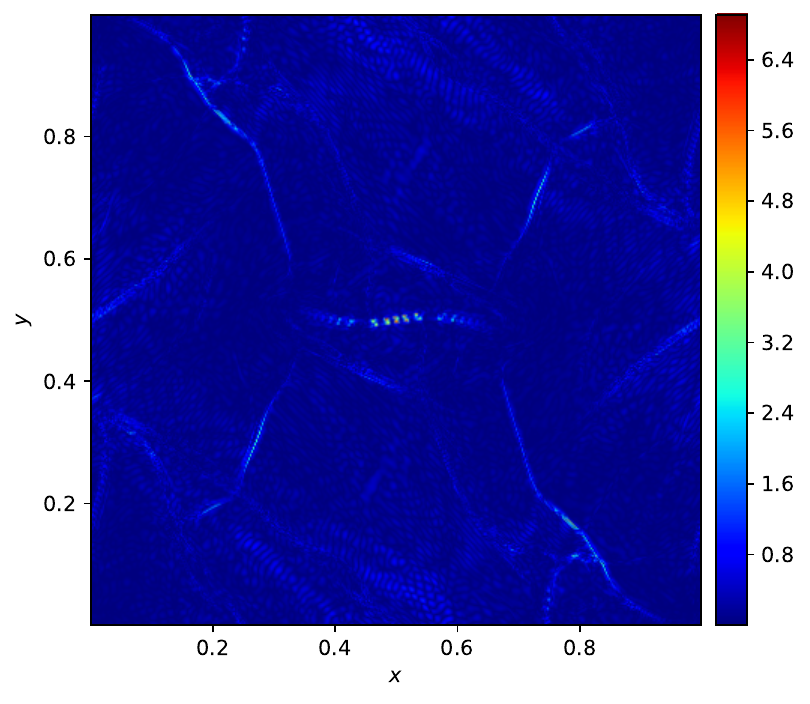}\label{fig:ot_db_w_o3}}~
		\subfigure[$|(\nabla\cdot\B)_{i,j}|$ for $\ofi$ scheme for isotropic GLM-CGL]{\includegraphics[width=0.26\textwidth]{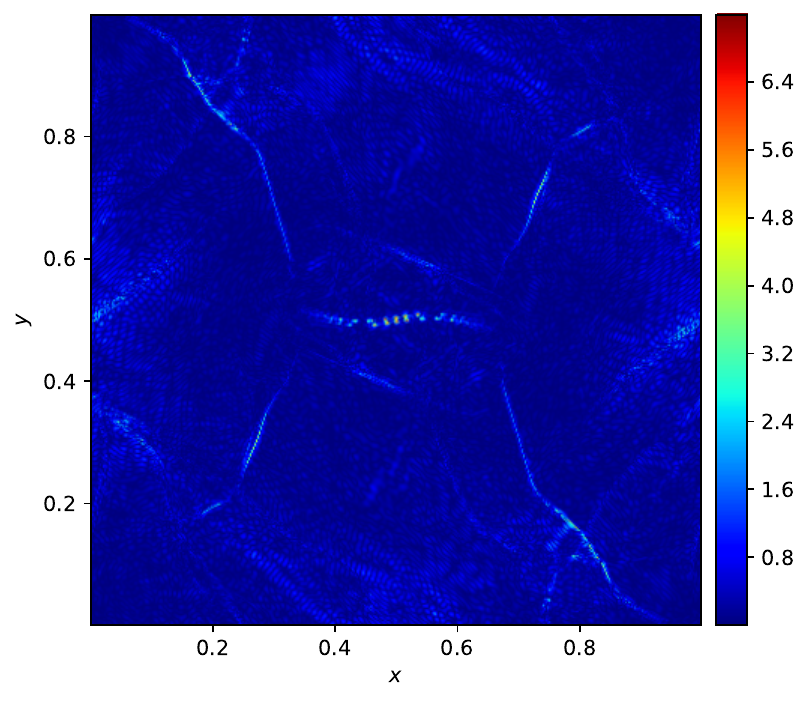}\label{fig:ot_db_w_o4}}\\
		\caption{\textbf{\nameref{test:ot}}: Plots of density and $|(\nabla\cdot\B)_{i,j}|$ for $\oti$, $\othi$ and $\ofi$ schemes for isotropic CGL and isotropic GLM-CGL at time $t=0.5$.}
		\label{fig:ot_mhd_rho_divb}
	\end{center}
\end{figure}
\begin{figure}[!htbp]
	\begin{center}
		\subfigure[$\|\nabla\cdot\B\|_{1}$ and $\|\nabla\cdot\B\|_{2}$ for $\oti$ scheme for isotropic CGL and isotropic GLM-CGL ]{\includegraphics[width=0.3\textwidth]{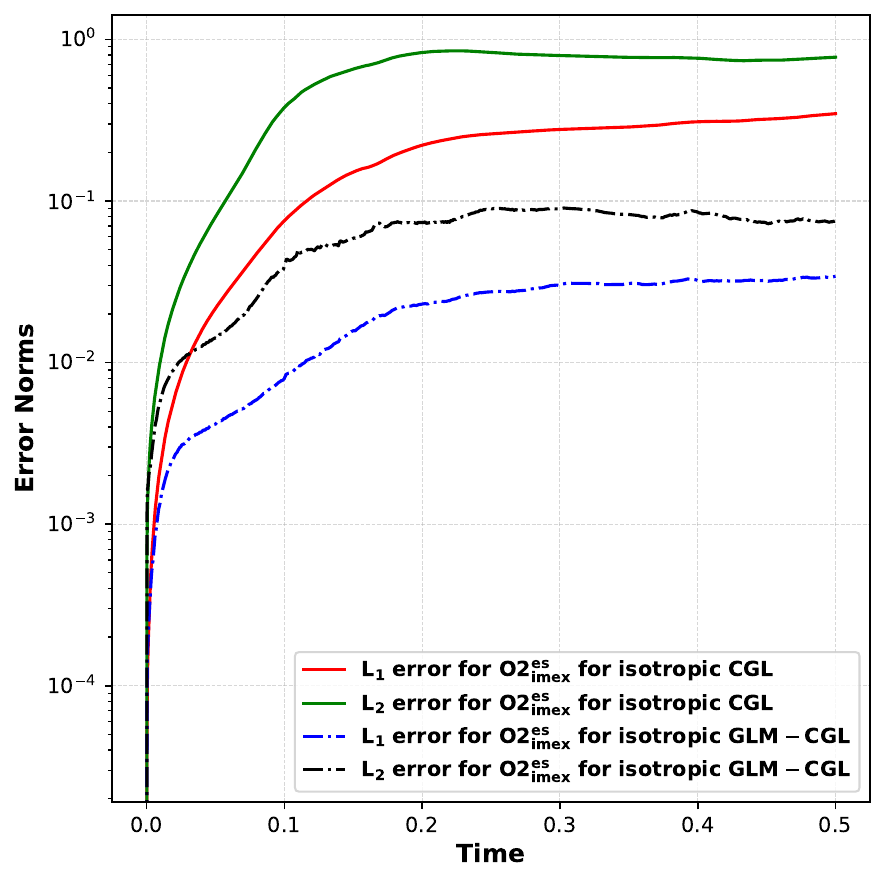}\label{fig:ot_error_o2}}~
		\subfigure[$\|\nabla\cdot\B\|_{1}$ and $\|\nabla\cdot\B\|_{2}$ for $\othi$ scheme for isotropic CGL and isotropic GLM-CGL ]{\includegraphics[width=0.3\textwidth]{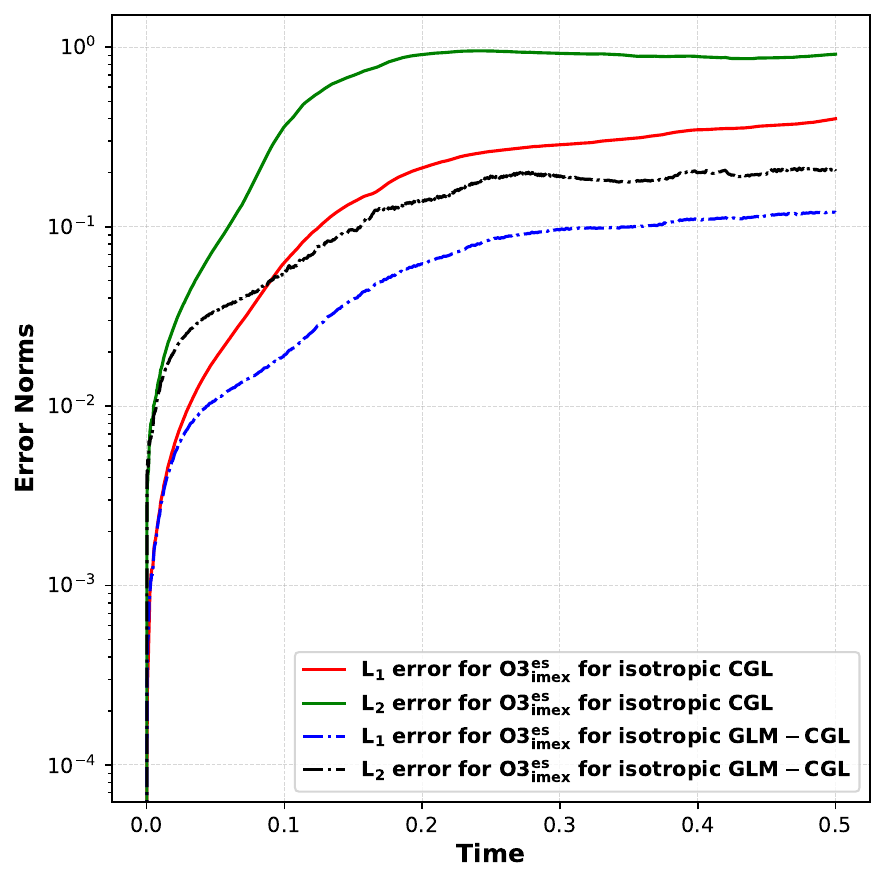}\label{fig:ot_error_o3}}~
		\subfigure[$\|\nabla\cdot\B\|_{1}$ and $\|\nabla\cdot\B\|_{2}$ for $\ofi$ scheme for isotropic CGL and isotropic GLM-CGL ]{\includegraphics[width=0.3\textwidth]{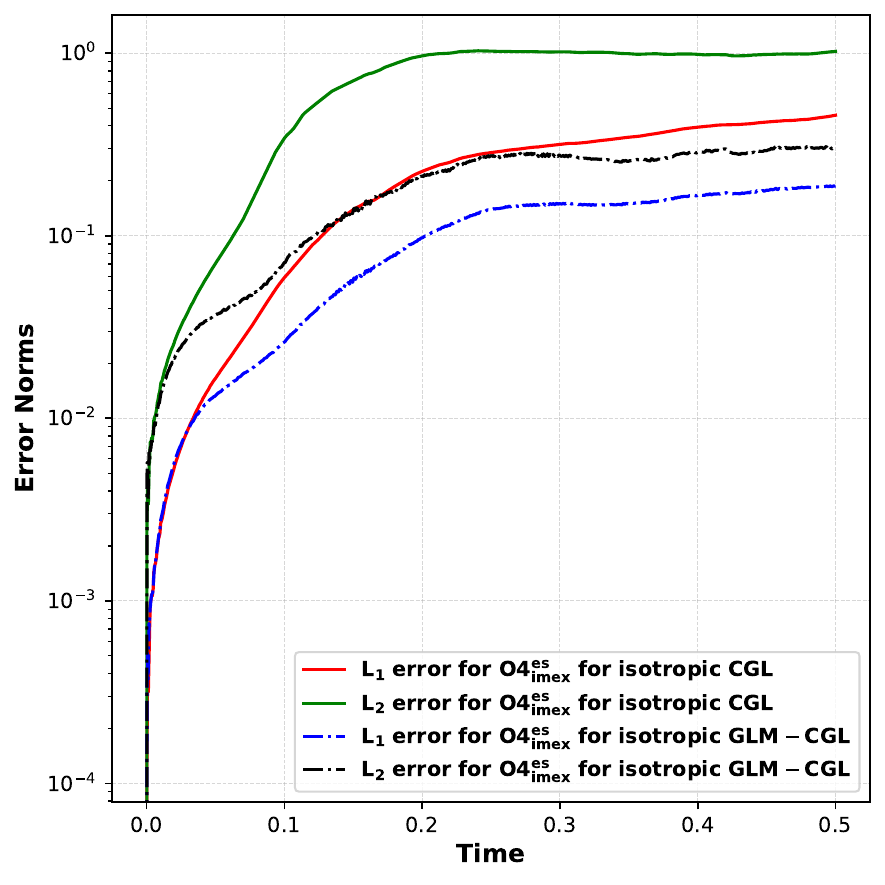}\label{fig:ot_error_o4}}
		\caption{\textbf{\nameref{test:ot}}: Evolution of the magnetic field divergence constraint errors till time $t=0.5$.}
		\label{fig:ot_error}
	\end{center}
\end{figure}
\begin{figure}[!htbp]
	\begin{center}
		\subfigure[Cut of pressure component $\pll$ for isotropic CGL and isotropic GLM-CGL along $y = 0.3125$]{\includegraphics[width=\textwidth]{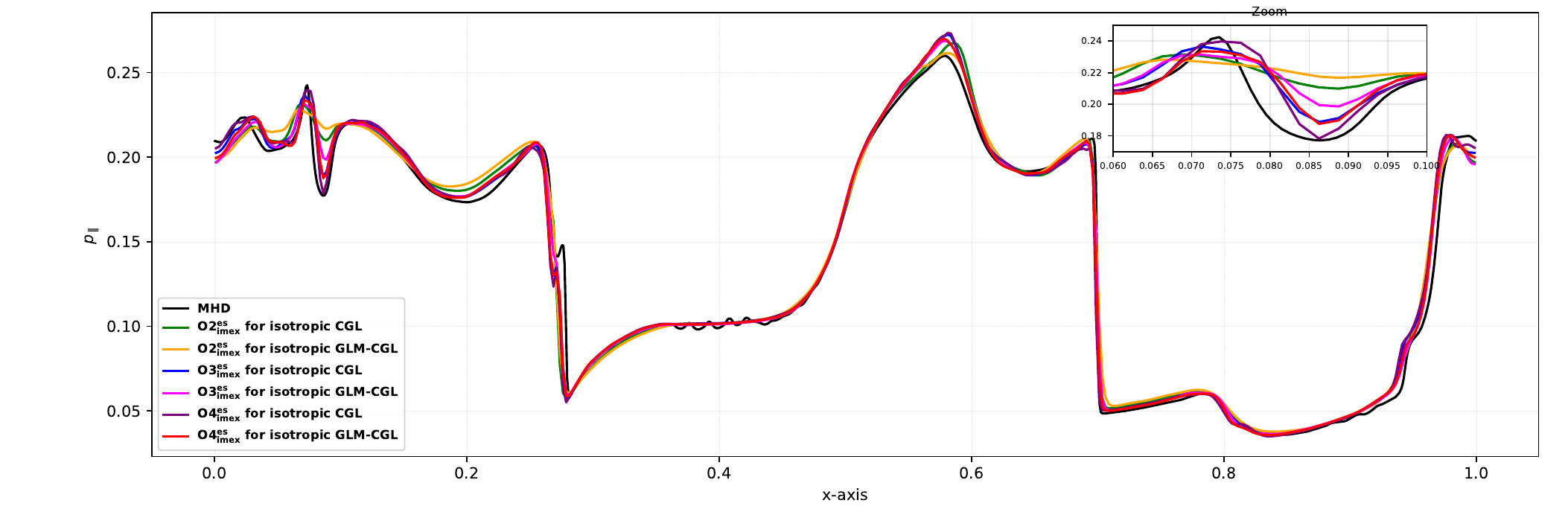}\label{fig:ot_cut_pll}}\\[1ex]
		\subfigure[Cut of pressure component $\per$ for isotropic CGL and isotropic GLM-CGL along $y = 0.3125$]{\includegraphics[width=\textwidth]{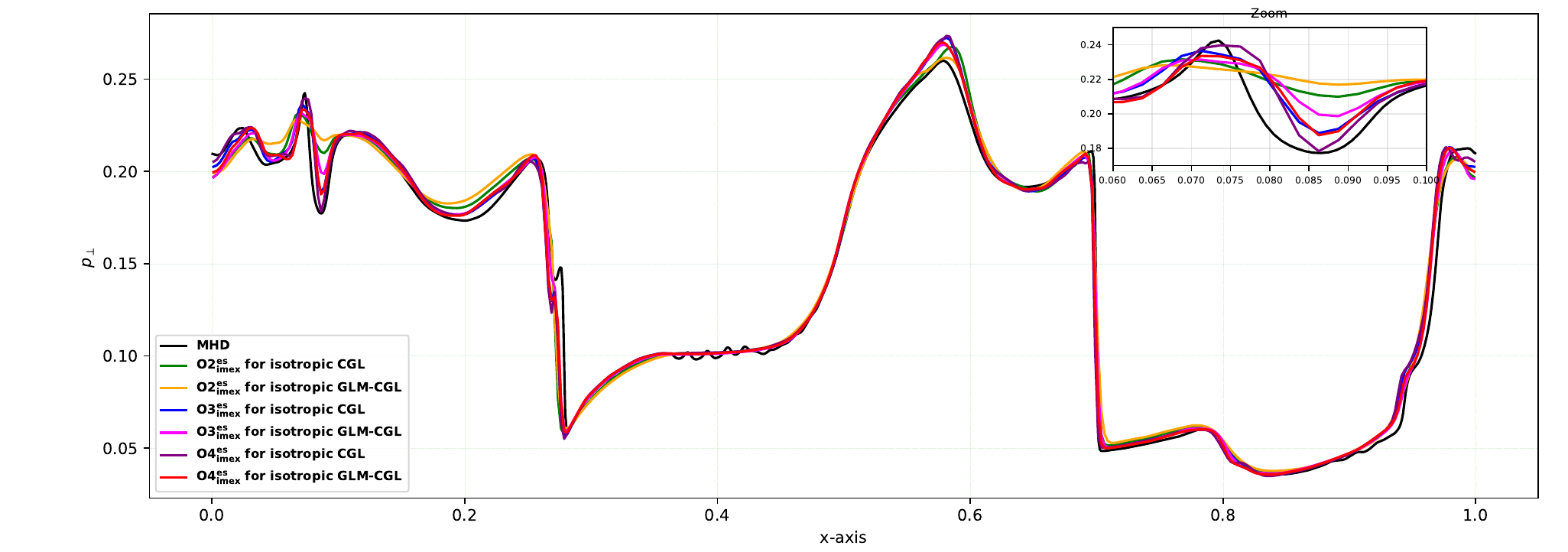}\label{OT/OT_perp.pdf}}
		\caption{\textbf{\nameref{test:ot}}: Cut plots of pressure components along $y = 0.3125$ at time $t=0.5$.}
		\label{fig:ot_cut_plot}
	\end{center}
\end{figure}

\begin{table}[ht]
\centering
\renewcommand{\arraystretch}{1.4}
\begin{tabular}{|c|c|c|c|c|c|c|}
\hline
\multirow{2}{*}{Scheme} &
\multirow{2}{*}{System} &
$\rho_{\min}$ &
$(p_{\parallel})_{\min}$ &
$(p_{\perp})_{\min}$ &
$\|\nabla\!\cdot\!\mathbf{B}\|_{1}$ &
$\|\nabla\!\cdot\!\mathbf{B}\|_{2}$ \\
&&&&&&\\
\hline

\multirow{2}{*}{$\oti$}
& Isotropic CGL
& 9.84216E-02 & 3.508819E-02 & 3.508814E-02 & 3.470541E-01 & 7.762731E-01\\ \cline{2-7}

& Isotropic GLM--CGL
& 9.899383E-02 & 3.569912E-02 & 3.569648E-02 & 3.419513E-02 & 7.513952E-02 \\  \hline

\multirow{2}{*}{$\othi$}
& Isotropic CGL
& 8.530408E-02 & 2.718044E-02 & 2.717974E-02 & 3.994897E-01& 9.149954E-01 \\ \cline{2-7}

& Isotropic GLM--CGL
& 8.687611E-02 & 2.829207E-02 & 2.829227E-02 & 1.200279E-01 & 2.046241E-01 \\  \hline

\multirow{2}{*}{$\ofi$}
& Isotropic CGL
& 8.554454E-02 & 2.723577E-02 & 2.723303E-02 & 4.573710E-01 & 1.023018E+00 \\ \cline{2-7}

& Isotropic GLM--CGL
& 8.69596E-02 & 2.812515E-02 & 2.812516E-02 & 1.873823E-01 & 3.020084E-01 \\ \hline

\end{tabular}
\caption{\textbf{\nameref{test:ot}}: Minimum values of density $\rho$, parallel pressure $p_{\parallel}$, and perpendicular pressure $p_{\perp}$, together with the value of $L_1$ and $L_2$ norm of $\nabla\cdot\B$ at final time.}
\label{tab:div_OT}
\end{table}

\begin{figure}[!htbp]
	\begin{center}
		\subfigure[Density for isotropic CGL at time $t=0.5$]{\includegraphics[width=0.24\textwidth]{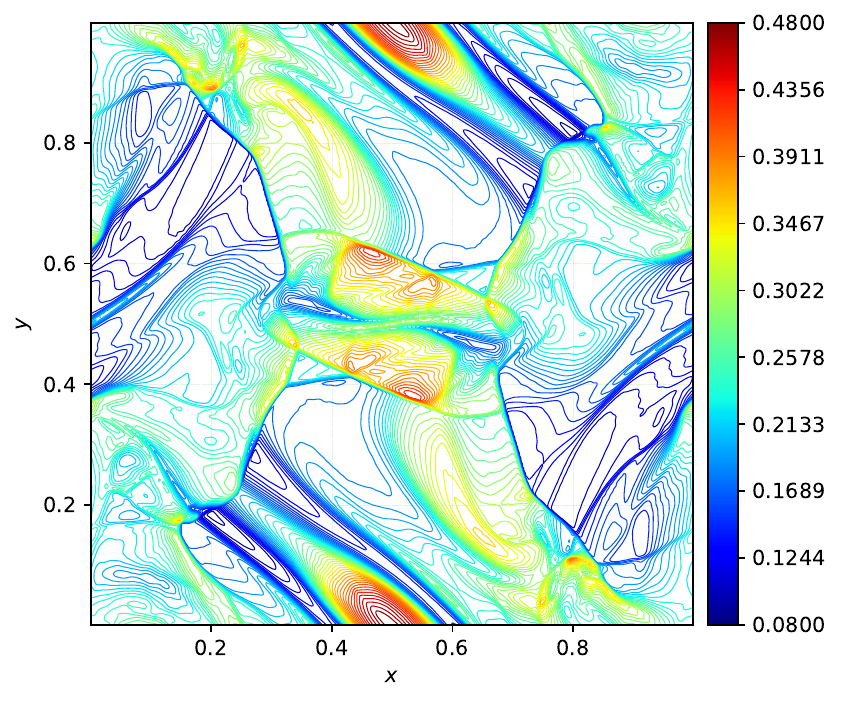}\label{fig:ot_t_0_5_wo}}~
		\subfigure[Density for isotropic CGL at time $t=1$]{\includegraphics[width=0.24\textwidth]{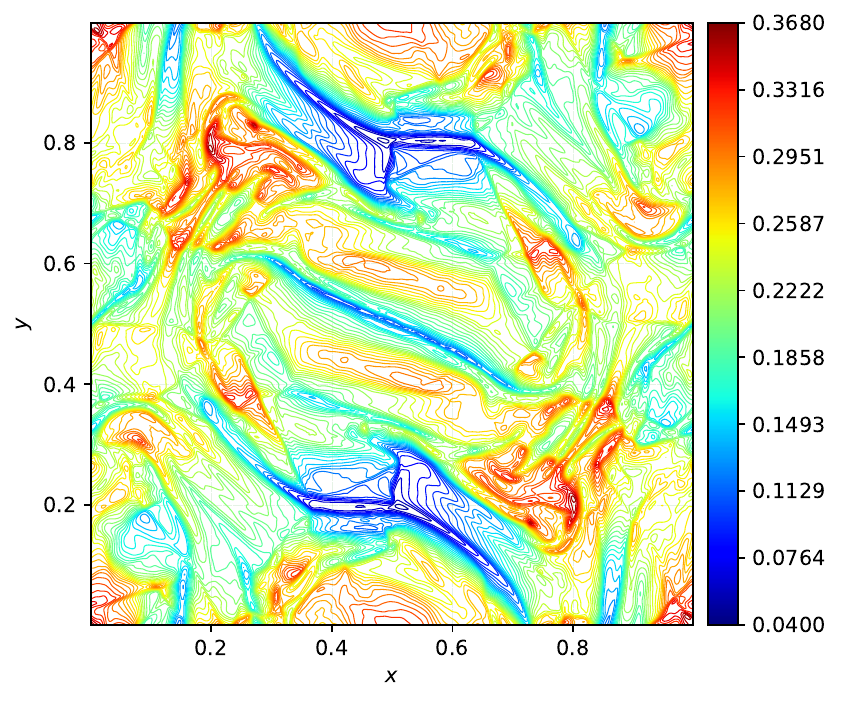}\label{fig:ot_t_1_wo}}~
        \subfigure[Density for isotropic CGL at time $t=2$]{\includegraphics[width=0.24\textwidth]{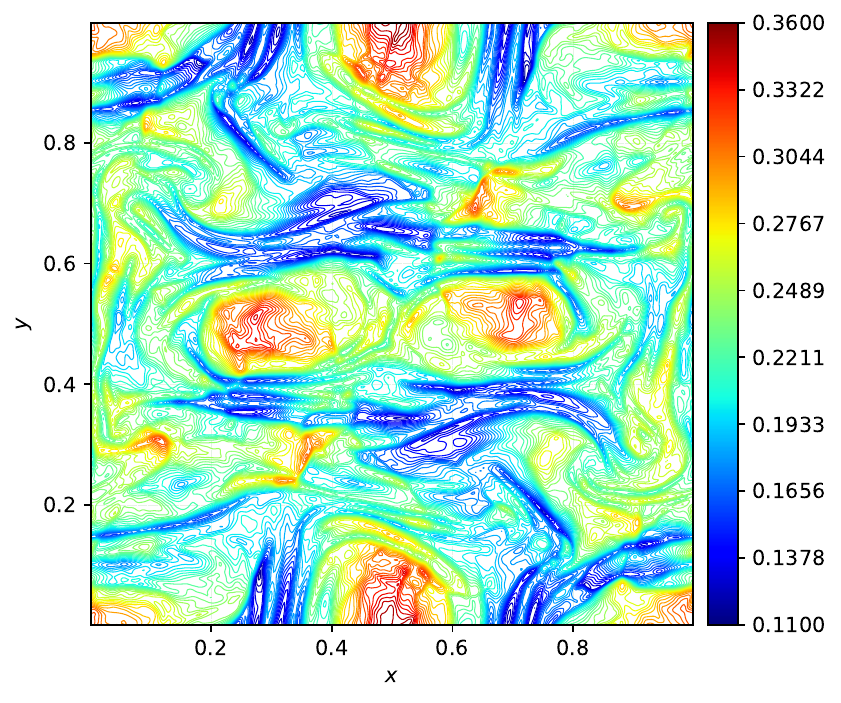}\label{fig:ot_t_2_wo}}~
        \subfigure[Density for isotropic CGL at time $t=5$]{\includegraphics[width=0.24\textwidth]{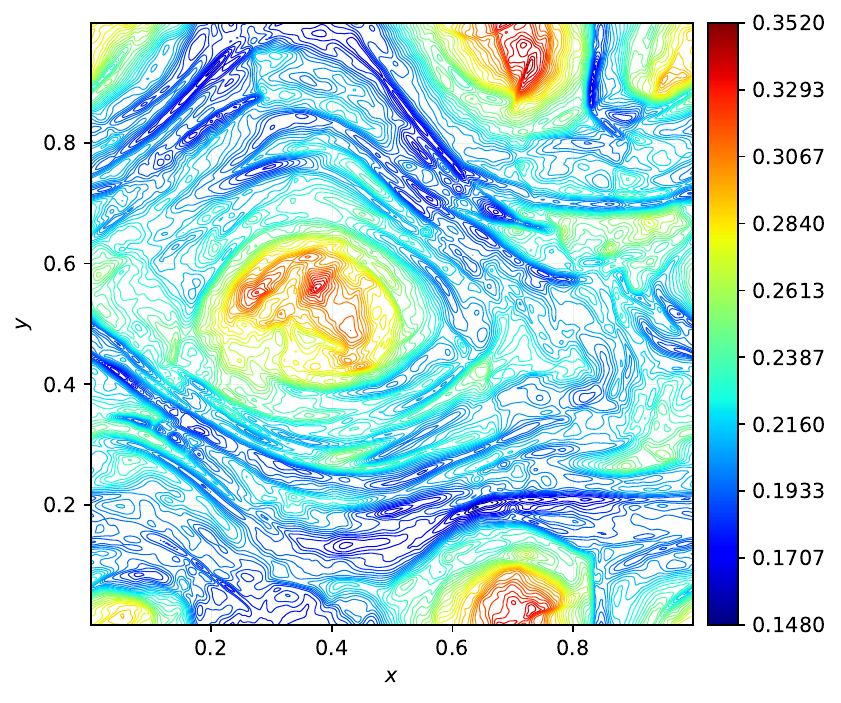}\label{fig:ot_t_5_wo}}\\
        \subfigure[Density for isotropic GLM-CGL at time $t=0.5$]{\includegraphics[width=0.24\textwidth]{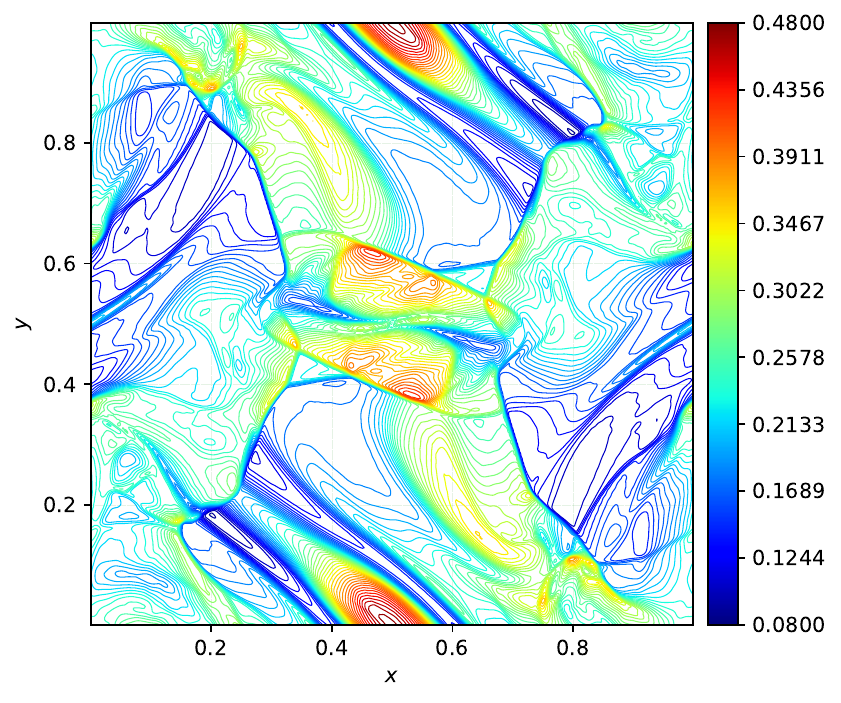}\label{fig:ot_t_0_5_w}}~
		\subfigure[Density for isotropic GLM-CGL at time $t=1$]{\includegraphics[width=0.24\textwidth]{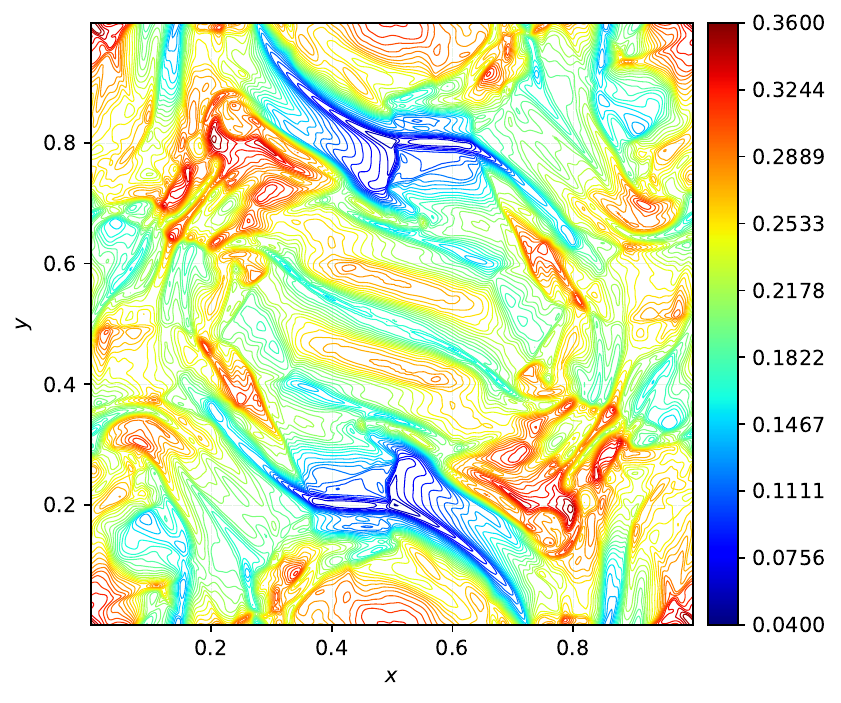}\label{fig:ot_t_1_w}}~
        \subfigure[Density for isotropic GLM-CGL at time $t=2$]{\includegraphics[width=0.24\textwidth]{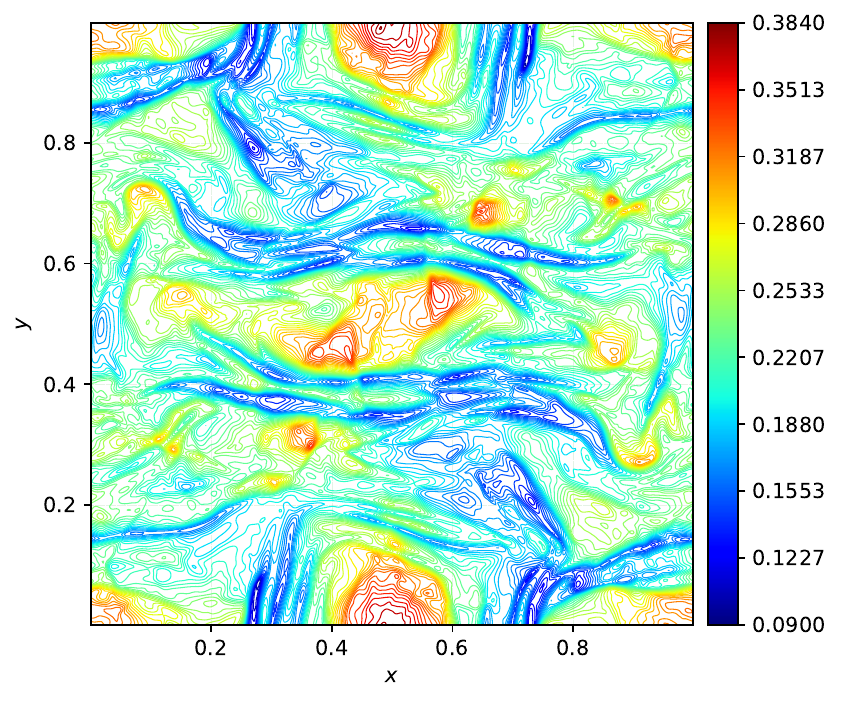}\label{fig:ot_t_2_w}}~
        \subfigure[Density for isotropic GLM-CGL at time $t=5$]{\includegraphics[width=0.24\textwidth]{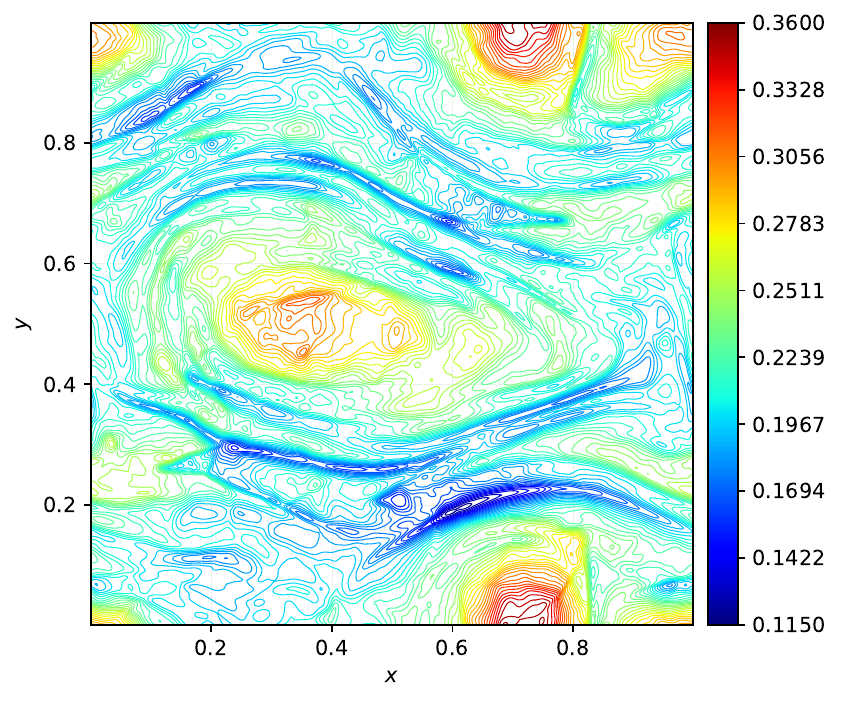}\label{fig:ot_t_5_w}}
		\caption{\textbf{\nameref{test:ot}}: Plots of density for $\ofi$ scheme  at different times using $400\times400$ cells.}
		\label{fig:ot_time}
	\end{center}
\end{figure}

Following the Orszag-Tang problem for MHD ~\cite{orszag1979small,toth2000b,chandrashekar2016entropy}, we present a generalized version of the problem for the GLM-CGL (see~\cite{singh2024entropy}).  The domain for the problem is $[0,1]\times[0,1]$, with periodic boundary conditions. The initial state is given as, 
\begin{align*}
	\left(\rho,\pll,\per,\Psi \right)& = \left(\frac{25}{36\pi},\frac{5}{12\pi},\frac{5}{12\pi}, 0 \right),\\
	\bu &= \left(-\sin{(2\pi y)}, \sin{(2\pi x)}, 0\right),\\
	\B &= \frac{1}{\sqrt{4\pi}}\left(-\sin{(2\pi y)}, \sin{(4\pi x)}, 0\right).
\end{align*}
It was observed in~\cite{singh2024entropy,bhoriya2024} that the CGL model without the source term (anisotropic case) has difficulty in simulating this test case. \revo{We found that similar stability issues also occur in the GLM-CGL model without the source term. Therefore, we consider the isotropic CGL and isotropic GLM-CGL cases only.} This also makes the solution suitable to compare with the MHD solution.

\revg{The numerical results are presented in Figures \eqref{fig:ot_mhd_rho_divb}, \eqref{fig:ot_error} and~\eqref{fig:ot_cut_plot}. We have used $400\times400$ cells and simulated till the final time $t=0.5$, using the numerical schemes $\oti$, $\othi$ and $\ofi$. From the 2D density plots, we do not see any significant difference between the solutions. However, to observe more closely, we have plotted a one-dimensional cut plot at $y=0.3125$ of $\pll$ and $\per$ for different schemes and compared with the MHD solutions. All the solutions are similar to the MHD solutions. We observe that all schemes produce solutions that are close to the MHD solution. We further observe that GLM-CGL solutions are slightly more diffusive when compared with CGL solutions of the same order, consistent with the observations discussed above.}

In Figure~\eqref{fig:ot_mhd_rho_divb}, we observed that the divergence errors for the GLM-CGL model are significantly lower than the CGL model (almost one-third) for all the schemes. \revo{The colorbars in all divergence plots are dynamically scaled using the absolute local minimum and maximum values of $\nabla\cdot\B$ in each plot.}

This is further observed in the Figure~\eqref{fig:ot_error}, where we have plotted the time evolution of the $L_1$ and $L_2$ errors for divergence of the magnetic field. We again observe that the GLM-CGL divergence errors are significantly lower than the CGL case, for all orders.

\revg{Table~\eqref{tab:div_OT} lists the minimum values of density, $\pll$, and $\per$ as well as the value of $L_1$ and $L_2$ norm of $\nabla\cdot\B$ at the final time. All reported values of density and pressure remain positive for all schemes and formulations, indicating that physically admissible solutions are maintained throughout the computation. The table also demonstrates the significant reduction in divergence errors achieved by the GLM formulation.}

\revo{To test robustness, we extend the Orszag-Tang vortex simulation to a longer final time $t=5$. We plot the contour solutions for the fourth-order scheme only, however all other schemes also produce stable solutions. We plot the contours of the solutions to show the waves more effectively. As shown in Figure~\eqref{fig:ot_time}, $\ofi$ scheme for both isotropic CGL and isotropic GLM-CGL remains stable throughout the computation, with no numerical breakdown observed. In addition, density contours at $t=0.5,~1,~2,~5$ obtained using the $\ofi$ scheme on a $400\times400$ grid, are presented. These contours demonstrate the progressive development of shock-driven MHD turbulence. At $t=1$, the interaction of shocks and vortices produces more complex flow structures. By $t=2$, additional small-scale features appear throughout the domain, and at $t=5$, the solution exhibits a fully developed shock-driven MHD turbulence. The isotropic GLM-CGL solution remains stable and captures the same large-scale turbulent structures as the isotropic CGL solution. However, at later times $t=2$ and $t=5$, the isotropic GLM-CGL solution appears slightly smoother, whereas the isotropic CGL solution retains more fine-scale filamentary structures. This difference is likely due to the additional numerical dissipation introduced by the hyperbolic divergence-cleaning mechanism.}
\subsection{Rotor problem}\label{test:rt}
\begin{figure}[!htbp]
	\begin{center}
		\subfigure[$\per$ for $\ote$ scheme for CGL]{\includegraphics[width=0.26\textwidth]{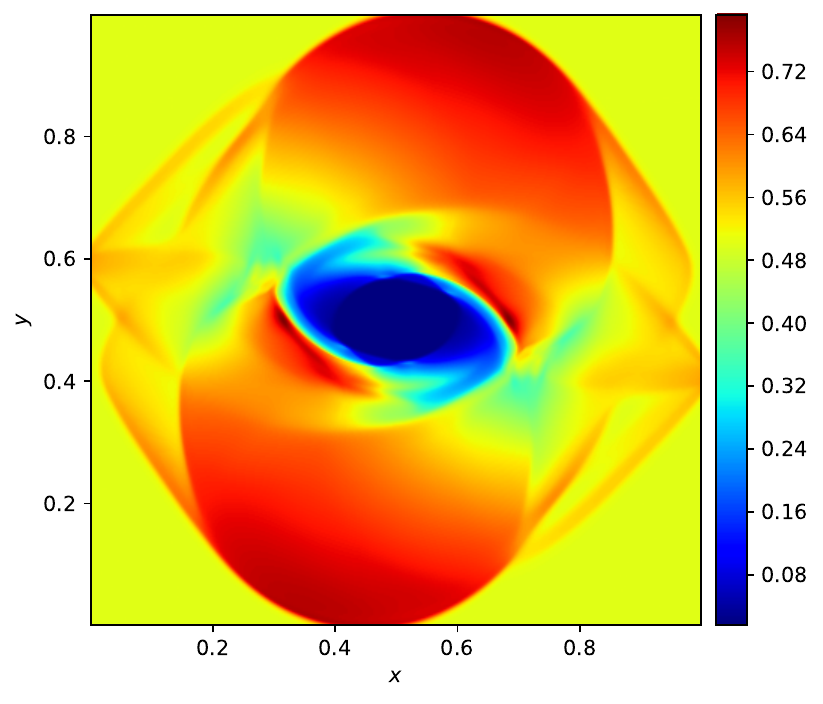}\label{fig:rt_rho_wo_o2}}~
		\subfigure[$\per$ for $\othe$ scheme for CGL]{\includegraphics[width=0.26\textwidth]{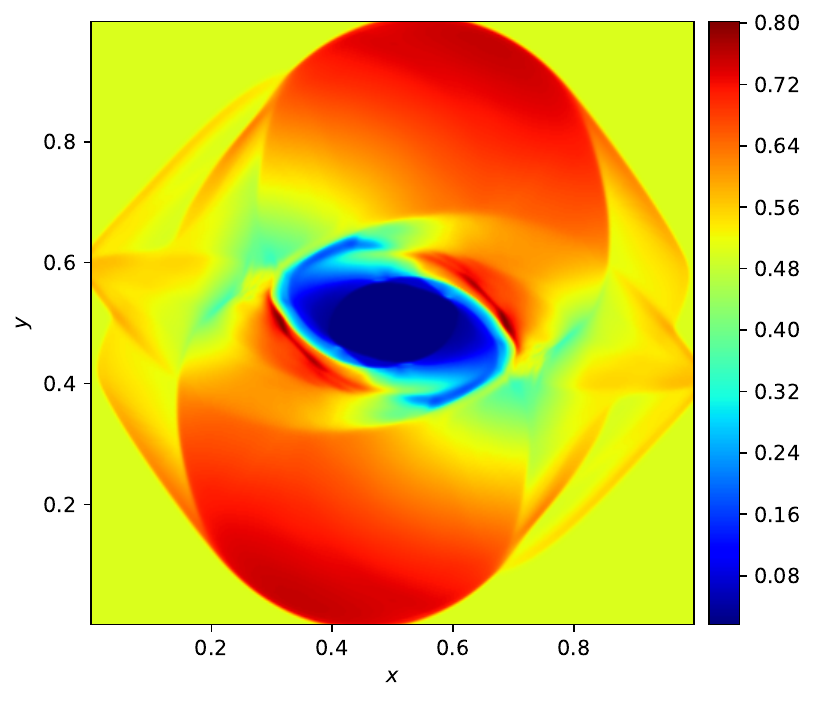}\label{fig:rt_rho_wo_o3}}~
		\subfigure[$\per$ for $\ofe$ scheme for CGL]{\includegraphics[width=0.26\textwidth]{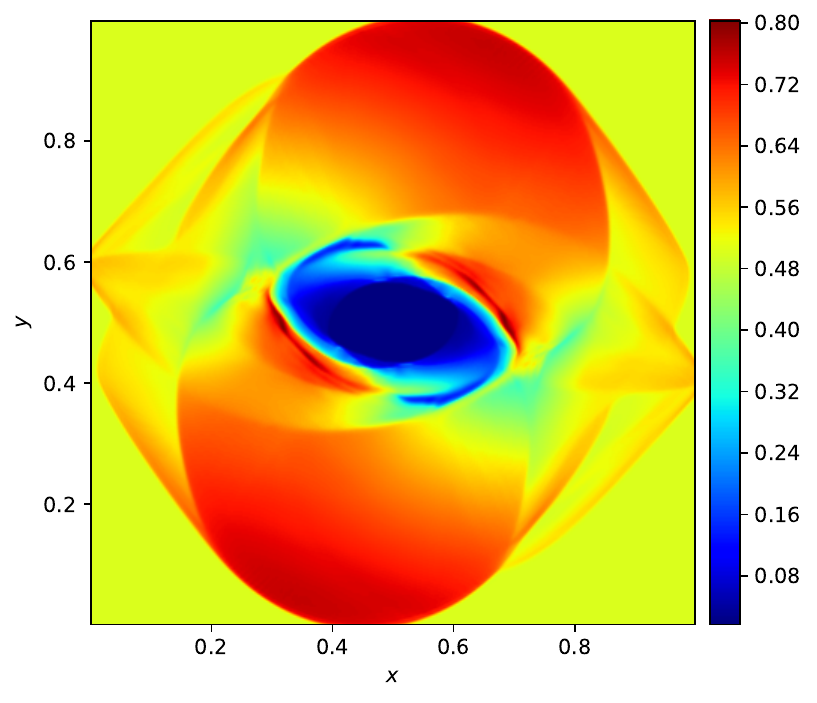}\label{fig:rt_rho_wo_o4}}\\
		\subfigure[$\per$ for $\ote$ scheme for GLM-CGL]{\includegraphics[width=0.26\textwidth]{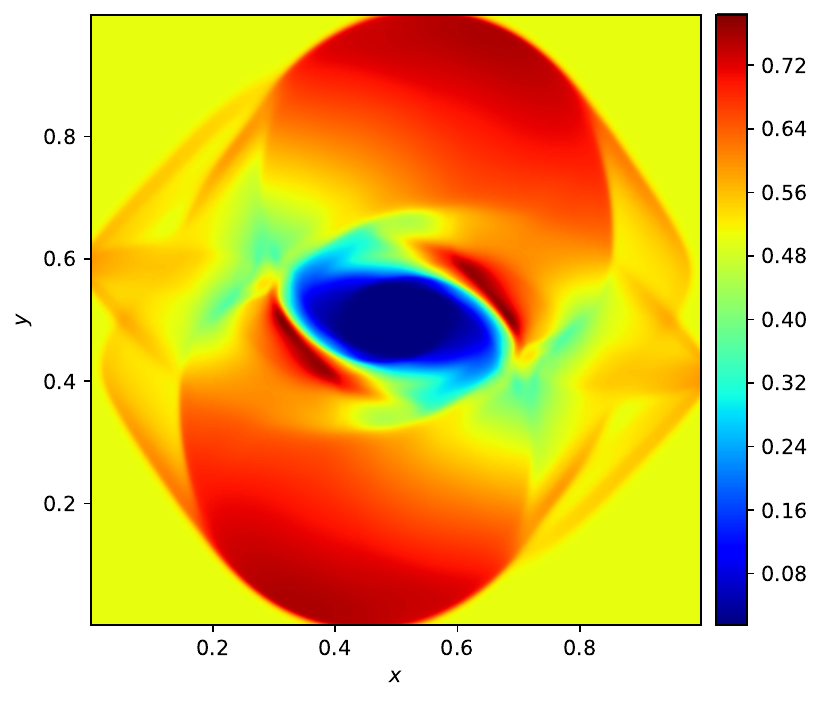}\label{fig:rt_rho_w_o2}}~
		\subfigure[$\per$ for $\othe$ scheme for GLM-CGL]{\includegraphics[width=0.26\textwidth]{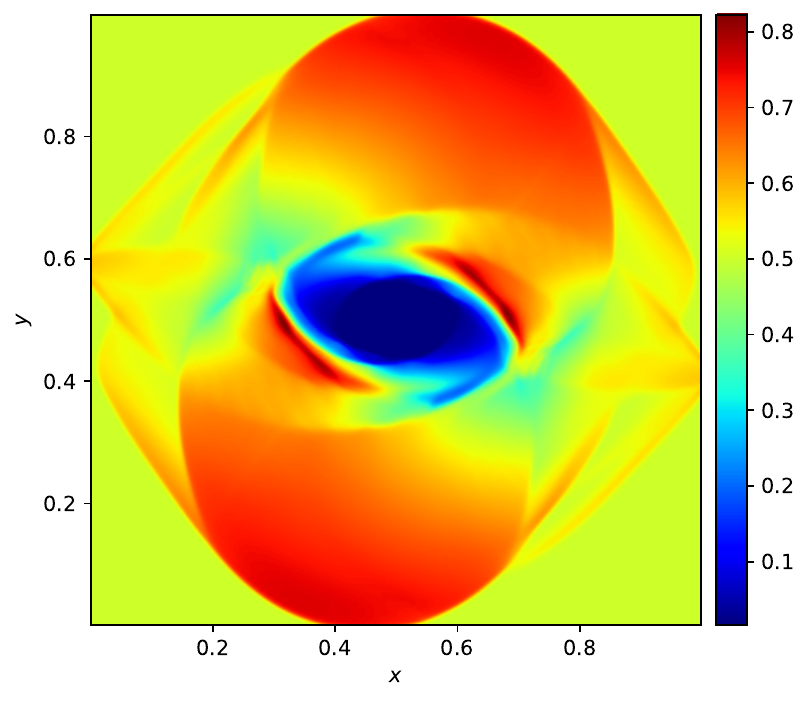}\label{fig:rt_rho_w_o3}}~
		\subfigure[$\per$ for $\ofe$ scheme for GLM-CGL]{\includegraphics[width=0.26\textwidth]{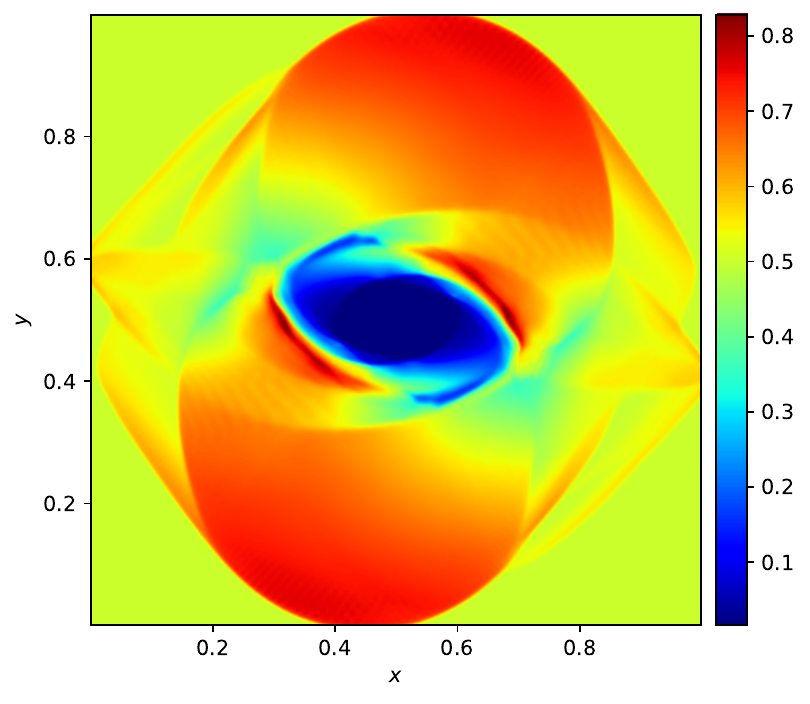}\label{fig:rt_rho_w_o4}}\\
		\subfigure[$|(\nabla\cdot\B)_{i,j}|$ for $\ote$ scheme for CGL]{\includegraphics[width=0.26\textwidth]{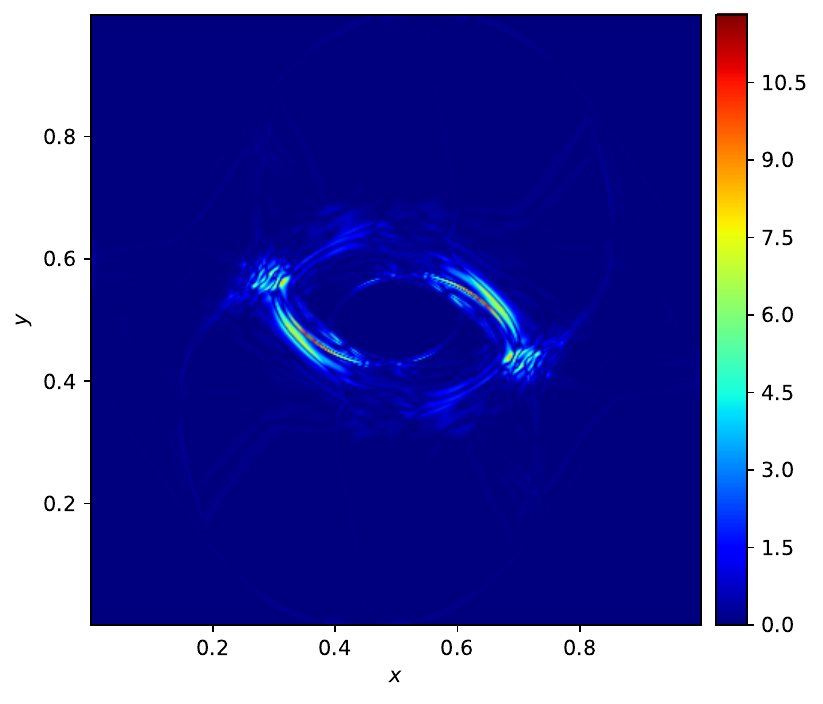}\label{fig:rt_db_wo_o2}}~
		\subfigure[$|(\nabla\cdot\B)_{i,j}|$ for $\othe$ scheme for CGL]{\includegraphics[width=0.26\textwidth]{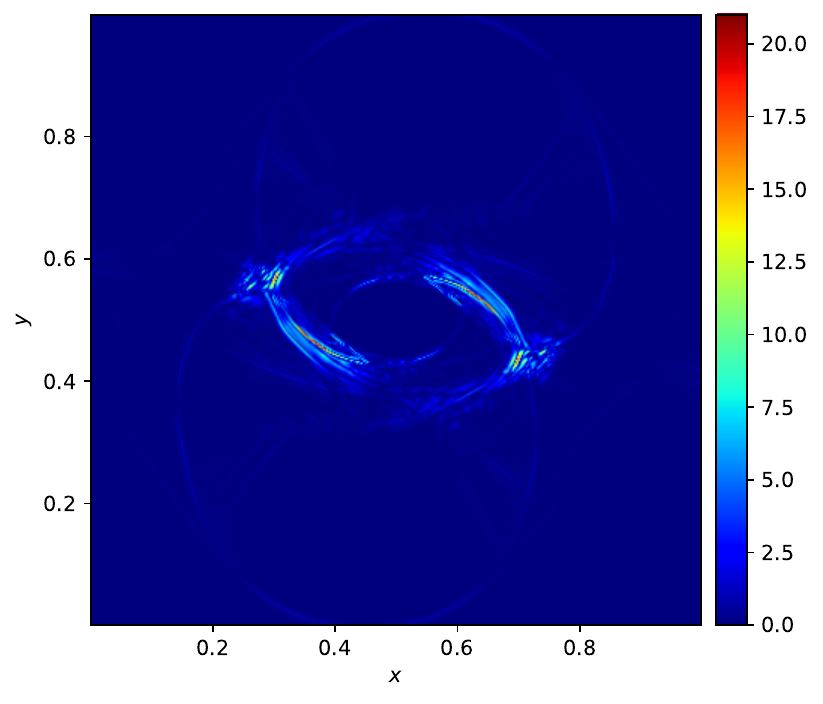}\label{fig:rt_db_wo_o3}}~
		\subfigure[$|(\nabla\cdot\B)_{i,j}|$ for $\ofe$ scheme for CGL]{\includegraphics[width=0.26\textwidth]{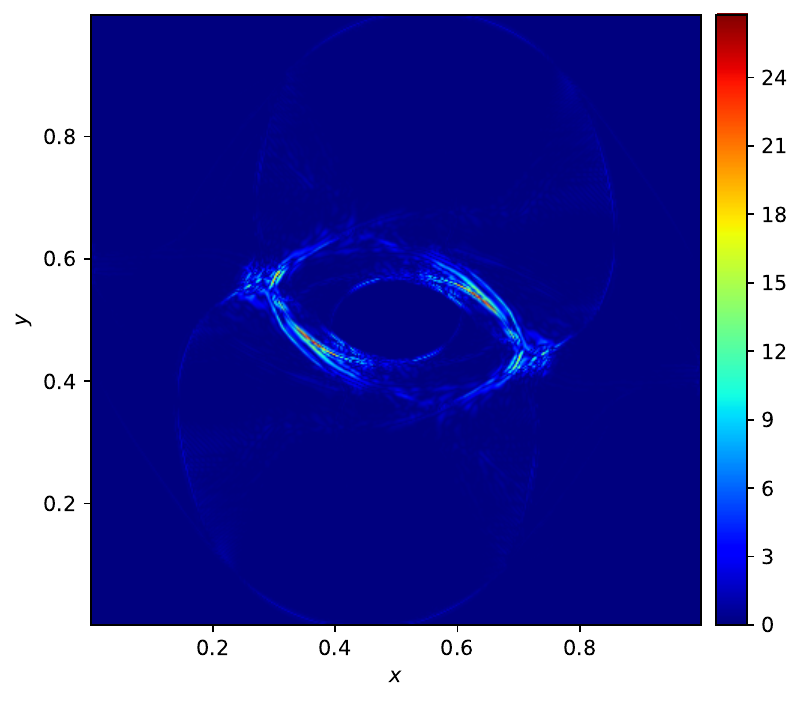}\label{fig:rt_db_wo_o4}}\\
		\subfigure[$|(\nabla\cdot\B)_{i,j}|$ for $\ote$ scheme for GLM-CGL]{\includegraphics[width=0.26\textwidth]{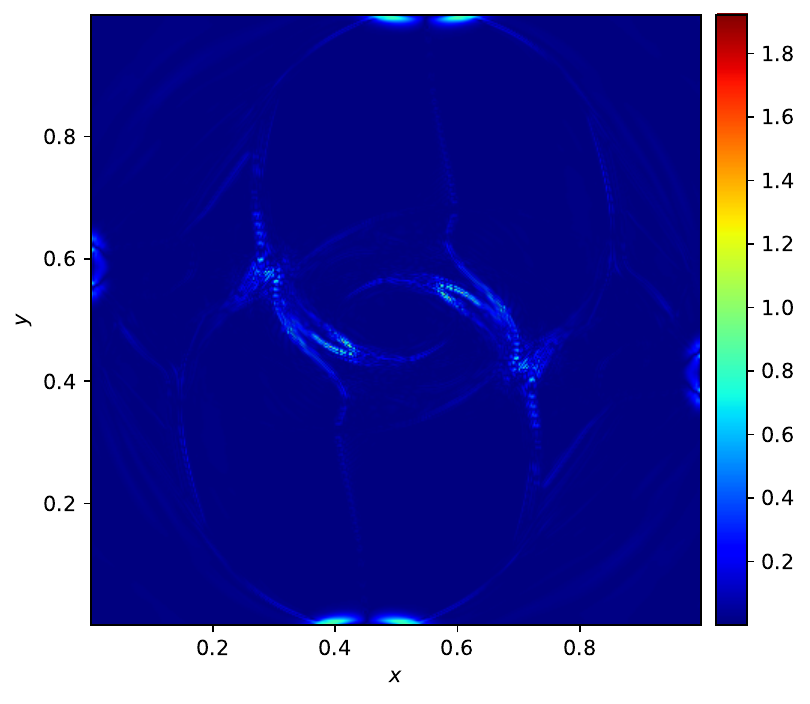}\label{fig:rt_db_w_o2}}~
		\subfigure[$|(\nabla\cdot\B)_{i,j}|$ for $\othe$ scheme for GLM-CGL]{\includegraphics[width=0.26\textwidth]{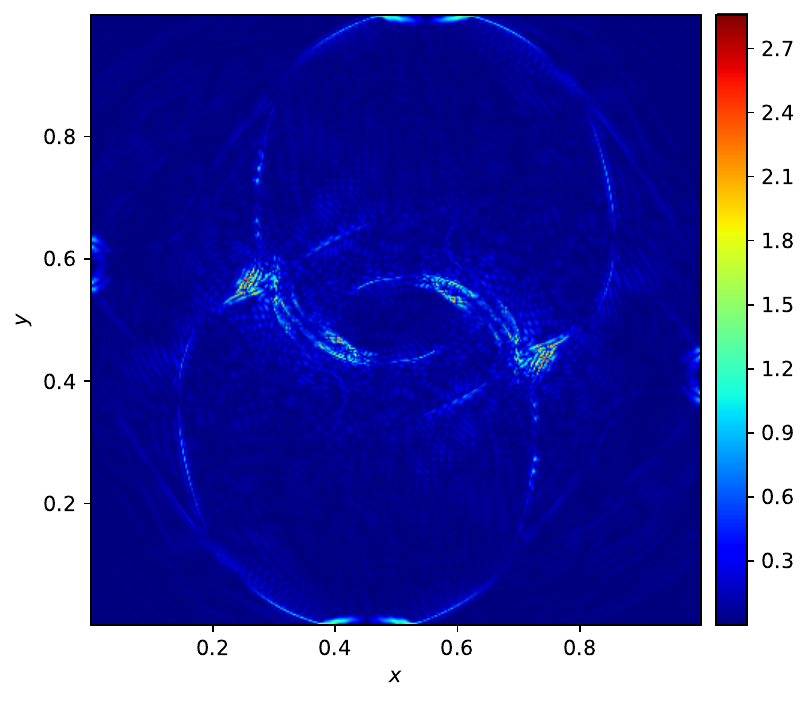}\label{fig:rt_db_w_o3}}~
		\subfigure[$|(\nabla\cdot\B)_{i,j}|$ for $\ofe$ scheme for GLM-CGL]{\includegraphics[width=0.26\textwidth]{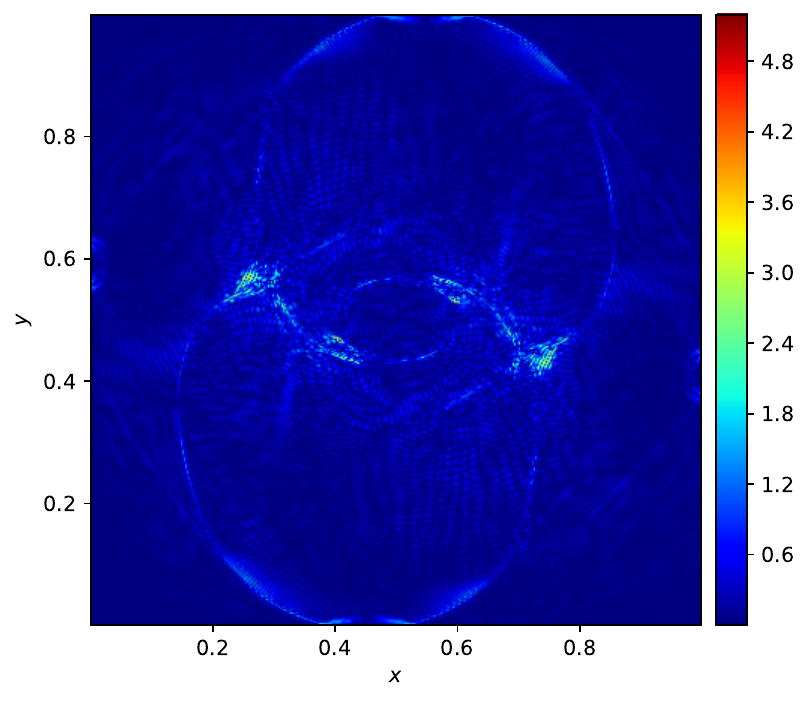}\label{fig:rt_db_w_o4}}\\
		\caption{\textbf{\nameref{test:rt}}: Plots of $\per$ and $|(\nabla\cdot\B)_{i,j}|$ for $\ote$, $\othe$ and $\ofe$ schemes for CGL and GLM-CGL at time $t=0.295$.}
		\label{fig:rt_cgl_rho_divb}
	\end{center}
\end{figure}
\begin{figure}[!htbp]
	\begin{center}
		\subfigure[$\per$ for $\oti$ scheme for isotropic CGL]{\includegraphics[width=0.26\textwidth]{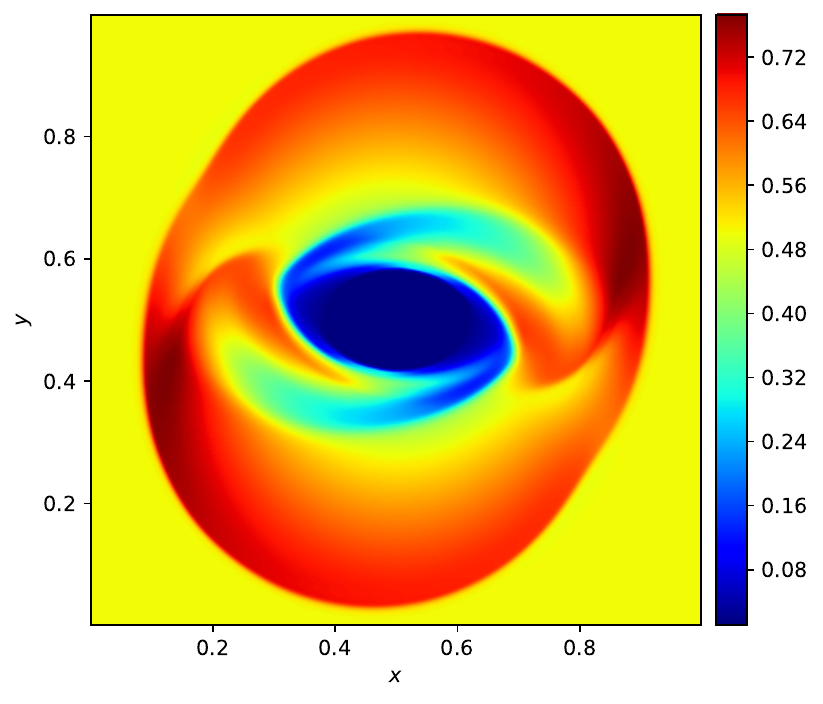}\label{fig:rt_par_wo_m_o2}}~
		\subfigure[$\per$ for $\othi$ scheme for isotropic CGL]{\includegraphics[width=0.26\textwidth]{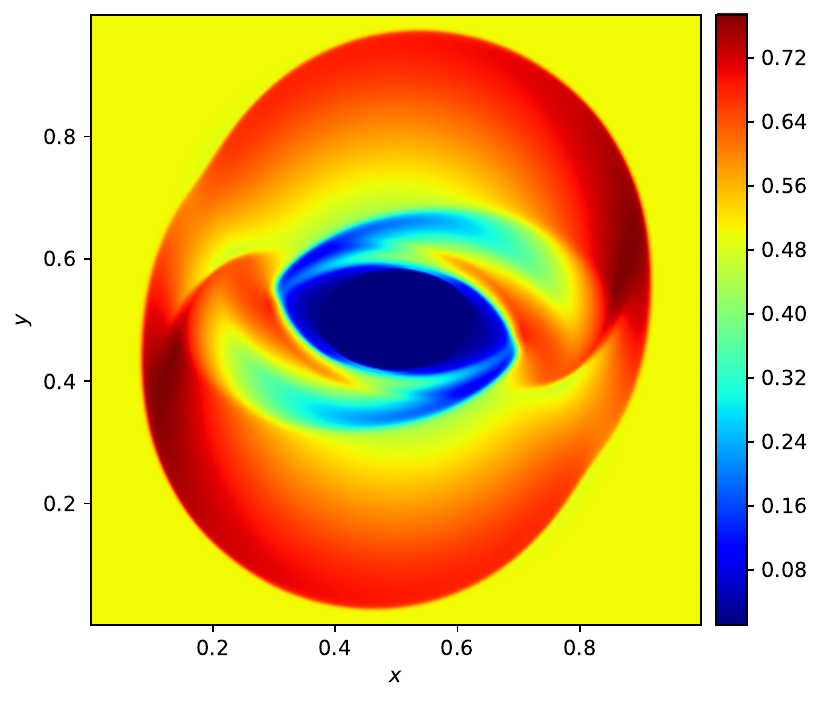}\label{fig:rt_par_wo_m_o3}}~
		\subfigure[$\per$ for $\ofi$ scheme for isotropic CGL]{\includegraphics[width=0.26\textwidth]{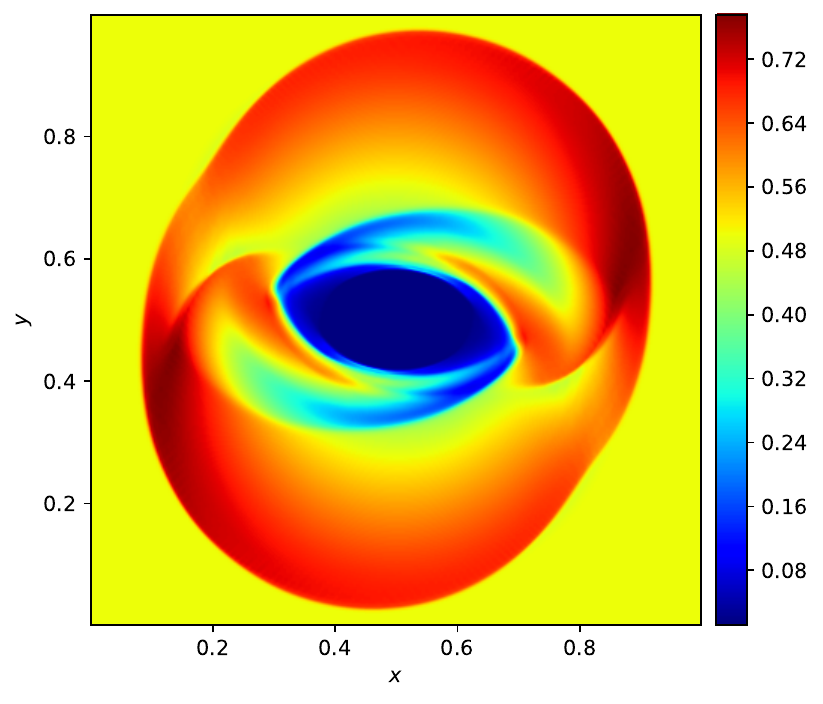}\label{fig:rt_par_wo_m_o4}}\\
		\subfigure[$\per$ for $\oti$ scheme for isotropic GLM-CGL]{\includegraphics[width=0.26\textwidth]{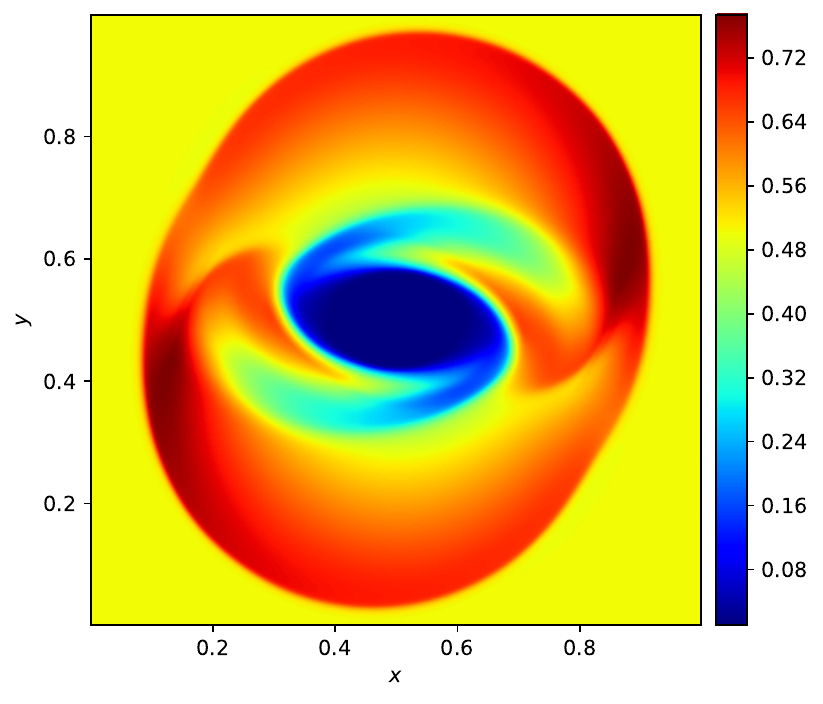}\label{fig:rt_par_w_m_o2}}~
		\subfigure[$\per$ for $\othi$ scheme for isotropic GLM-CGL]{\includegraphics[width=0.26\textwidth]{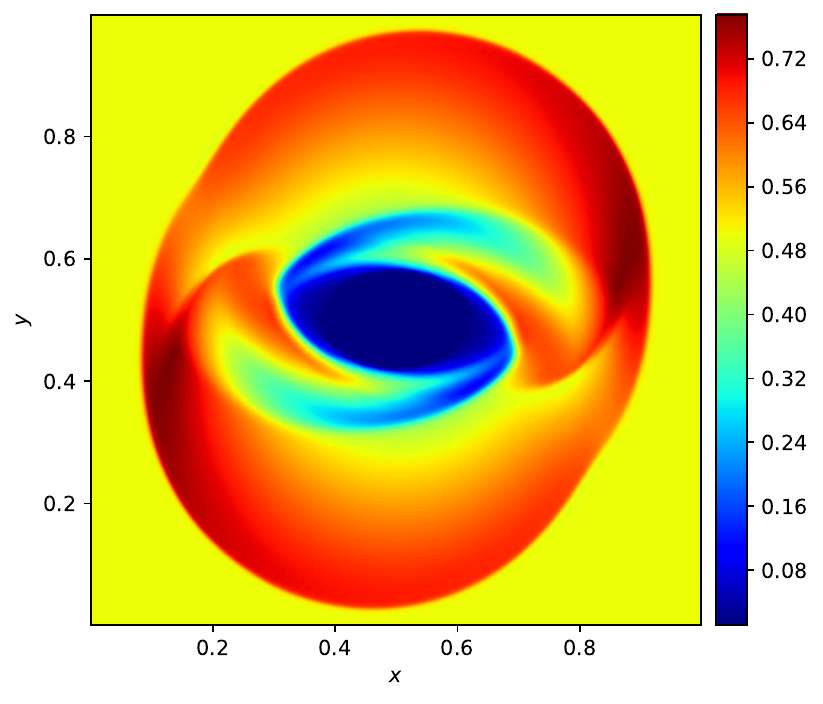}\label{fig:rt_par_w_m_o3}}~
		\subfigure[$\per$ for $\ofi$ scheme for isotropic GLM-CGL]{\includegraphics[width=0.26\textwidth]{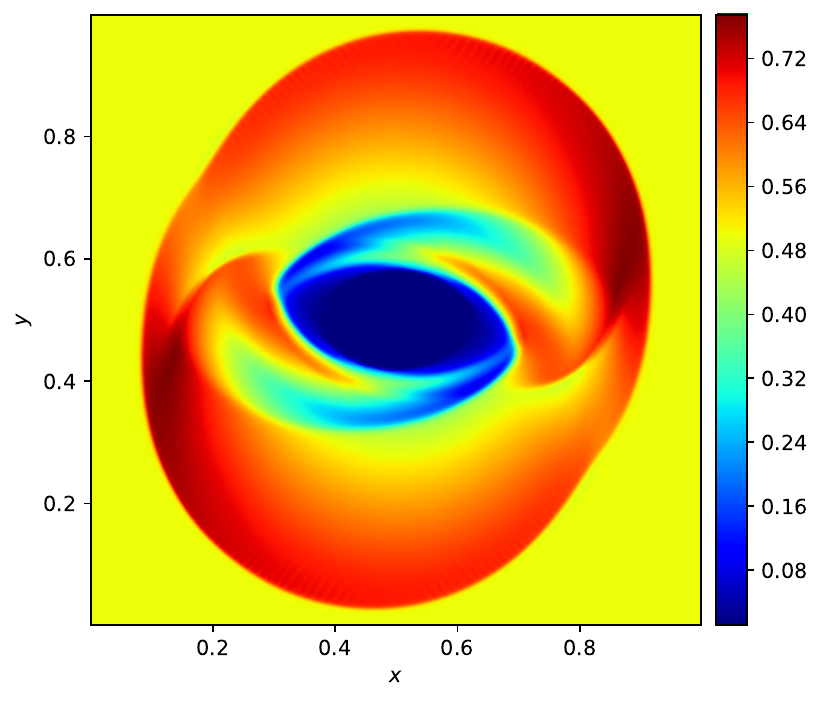}\label{fig:rt_par_w_m_o4}}\\
		\subfigure[$|(\nabla\cdot\B)_{i,j}|$ for $\oti$ scheme for isotropic CGL]{\includegraphics[width=0.26\textwidth]{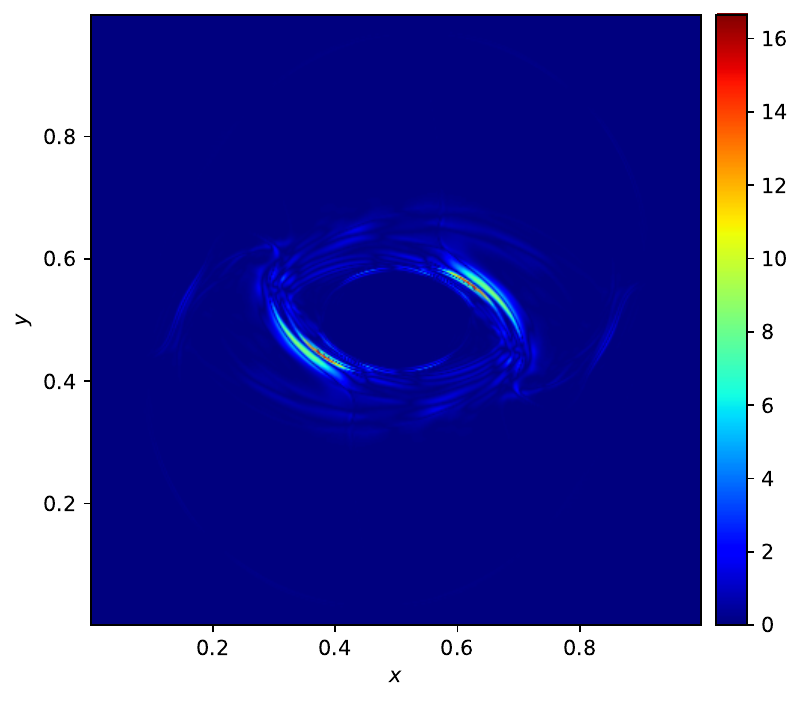}\label{fig:rt_db_wo_m_o2}}~
		\subfigure[$|(\nabla\cdot\B)_{i,j}|$ for $\othi$ scheme for isotropic CGL]{\includegraphics[width=0.26\textwidth]{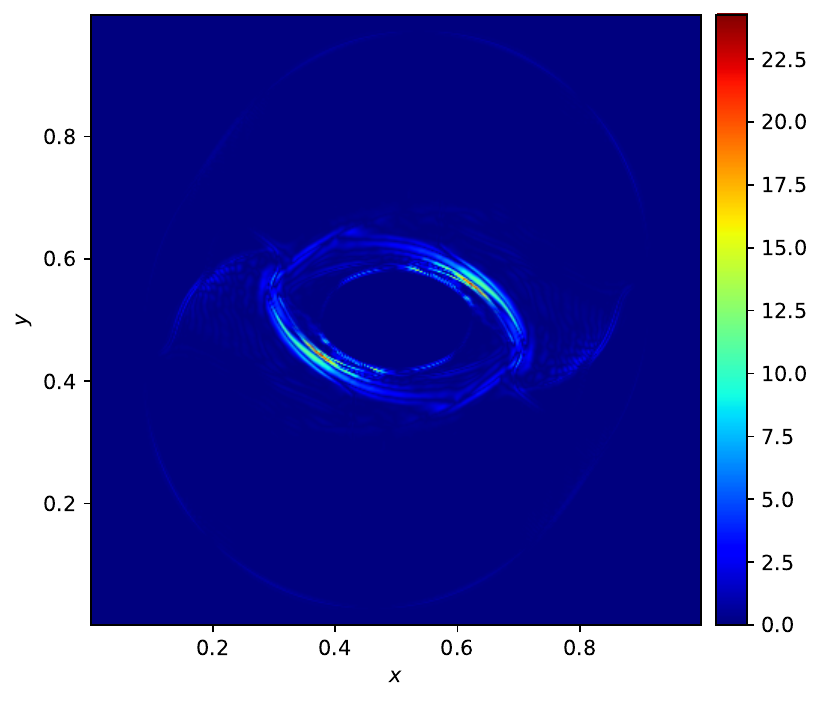}\label{fig:rt_db_wo_m_o3}}~
		\subfigure[$|(\nabla\cdot\B)_{i,j}|$ for $\ofi$ scheme for isotropic CGL]{\includegraphics[width=0.26\textwidth]{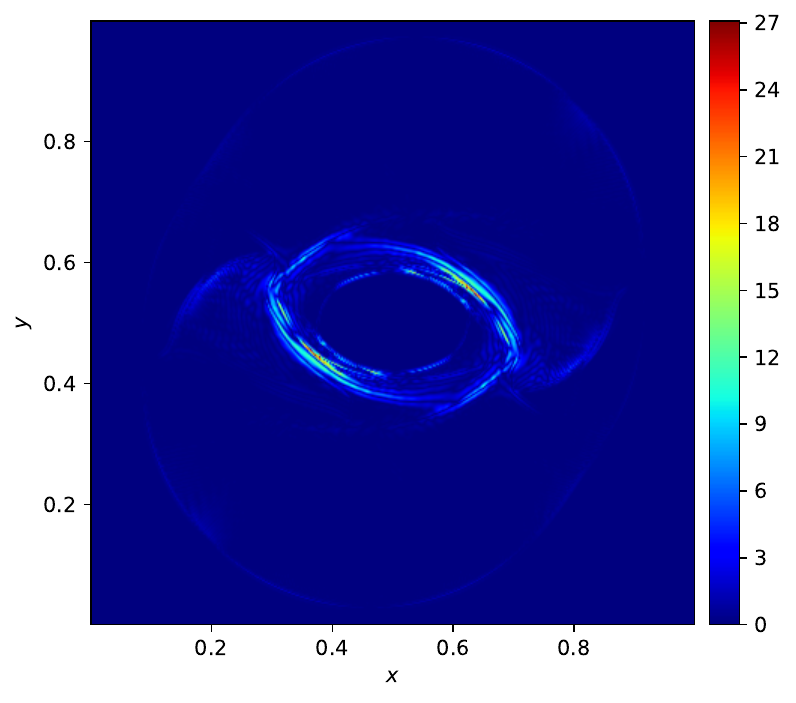}\label{fig:rt_db_wo_m_o4}}\\
		\subfigure[$|(\nabla\cdot\B)_{i,j}|$ for $\oti$ scheme for isotropic GLM-CGL]{\includegraphics[width=0.26\textwidth]{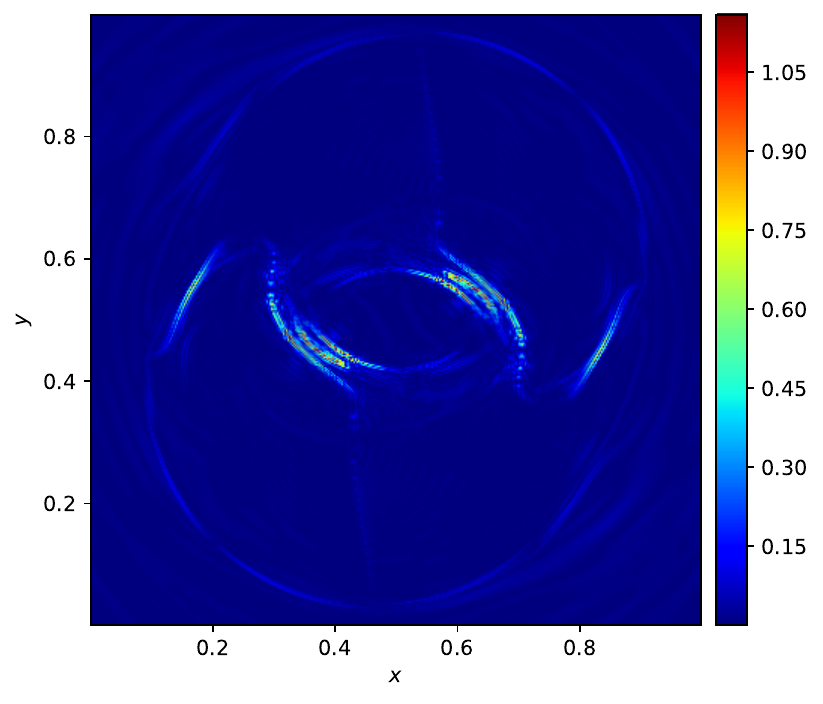}\label{fig:rt_db_w_m_o2}}~
		\subfigure[$|(\nabla\cdot\B)_{i,j}|$ for $\othi$ scheme for isotropic GLM-CGL]{\includegraphics[width=0.26\textwidth]{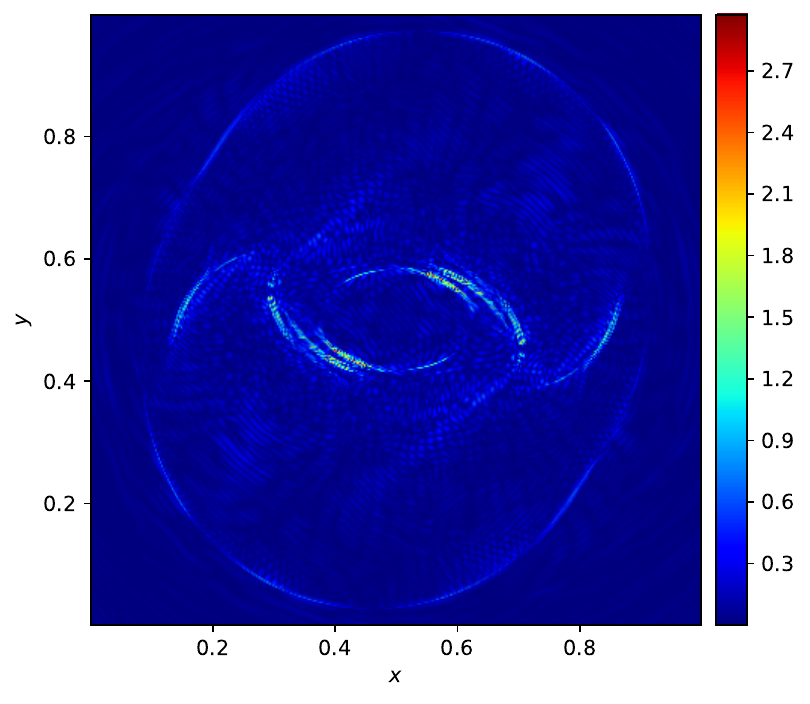}\label{fig:rt_db_w_m_o3}}~
		\subfigure[$|(\nabla\cdot\B)_{i,j}|$ for $\ofi$ scheme for isotropic GLM-CGL]{\includegraphics[width=0.26\textwidth]{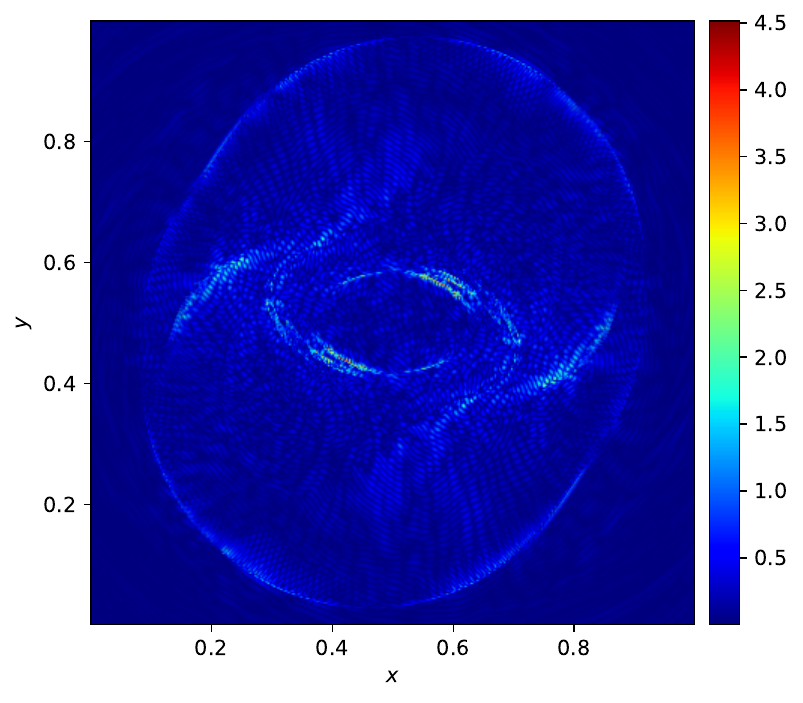}\label{fig:rt_db_w_m_o4}}\\
		\caption{\textbf{\nameref{test:rt}}: Plots of $\per$ and $|(\nabla\cdot\B)_{i,j}|$ for $\oti$, $\othi$ and $\ofi$ schemes for isotropic CGL and isotropic GLM-CGL at time $t=0.295$.}
		\label{fig:rt_mhd_par_divb}
	\end{center}
\end{figure}
\begin{figure}[!htbp]
	\begin{center}	
		\subfigure[Cut of pressure component $\pll$ for isotropic CGL and isotropic GLM-CGL along $y = 0.489$ using $400 \times400$ cells  ]{\includegraphics[width=0.4\textwidth]{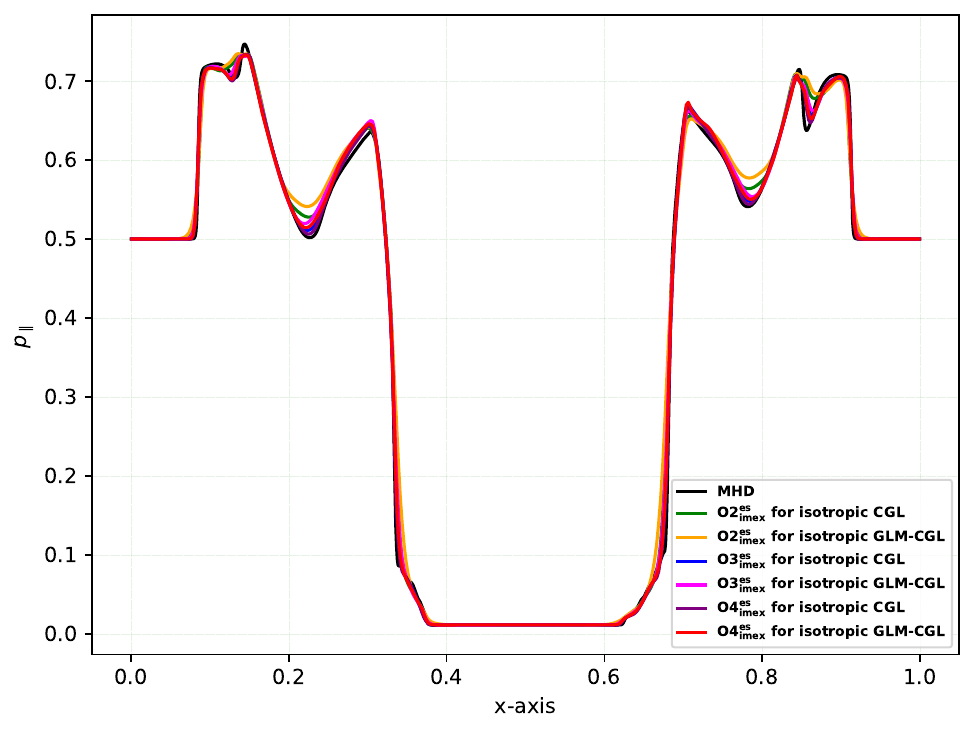}\label{fig:rt_cut_pll}}~
		\subfigure[Cut of pressure component $\per$ for isotropic CGL and isotropic GLM-CGL along $y = 0.489$ using $400 \times400$ cells  ]{\includegraphics[width=0.4\textwidth]{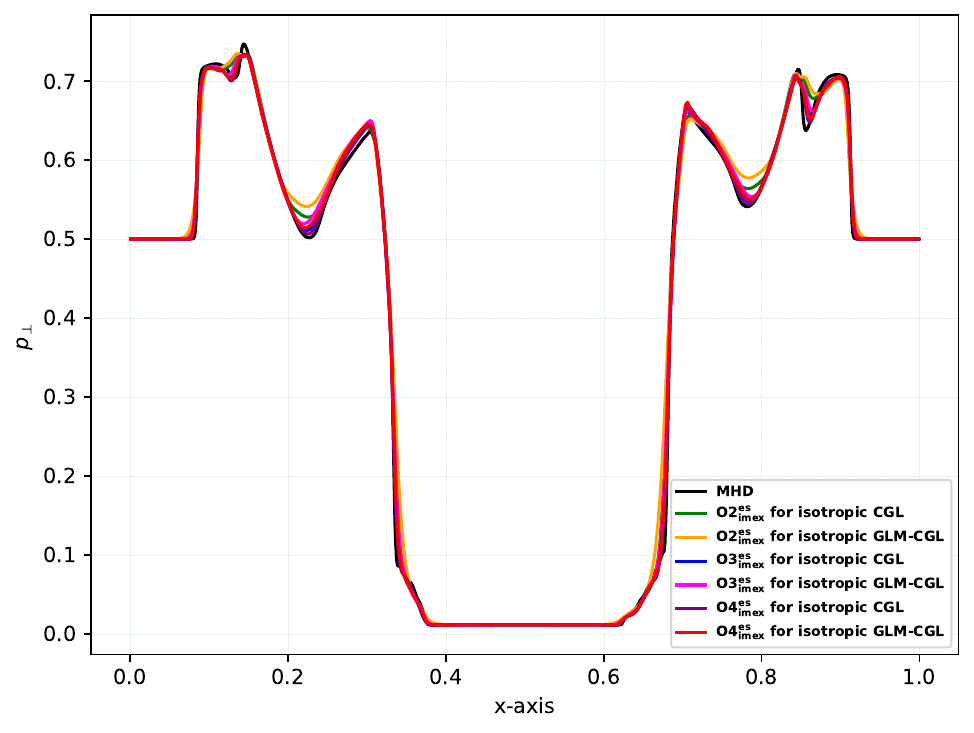}\label{fig:rt_cut_per}}\\
		\caption{\textbf{\nameref{test:rt}}: Cut plot of pressure components $\pll$ and $\per$ at time $t=0.295$.}
		\label{fig:rt_cut}
	\end{center}
\end{figure}

\begin{figure}[!htbp]
	\begin{center}	
		\subfigure[$\|\nabla\cdot\B\|_{1}$ and $\|\nabla\cdot\B\|_{2}$ for $\ote$ scheme for CGL and GLM-CGL ]{\includegraphics[width=0.28\textwidth]{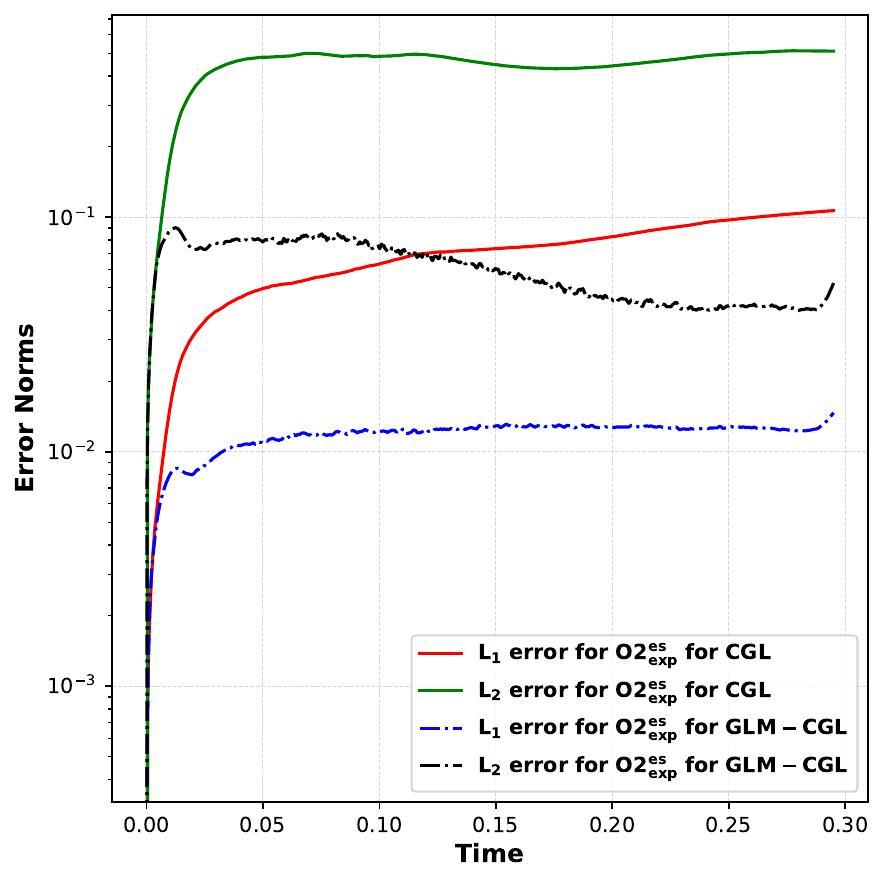}\label{fig:rt_error_cgl_o2}}~
		\subfigure[$\|\nabla\cdot\B\|_{1}$ and $\|\nabla\cdot\B\|_{2}$ for $\othe$ scheme for CGL and GLM-CGL ]{\includegraphics[width=0.28\textwidth]{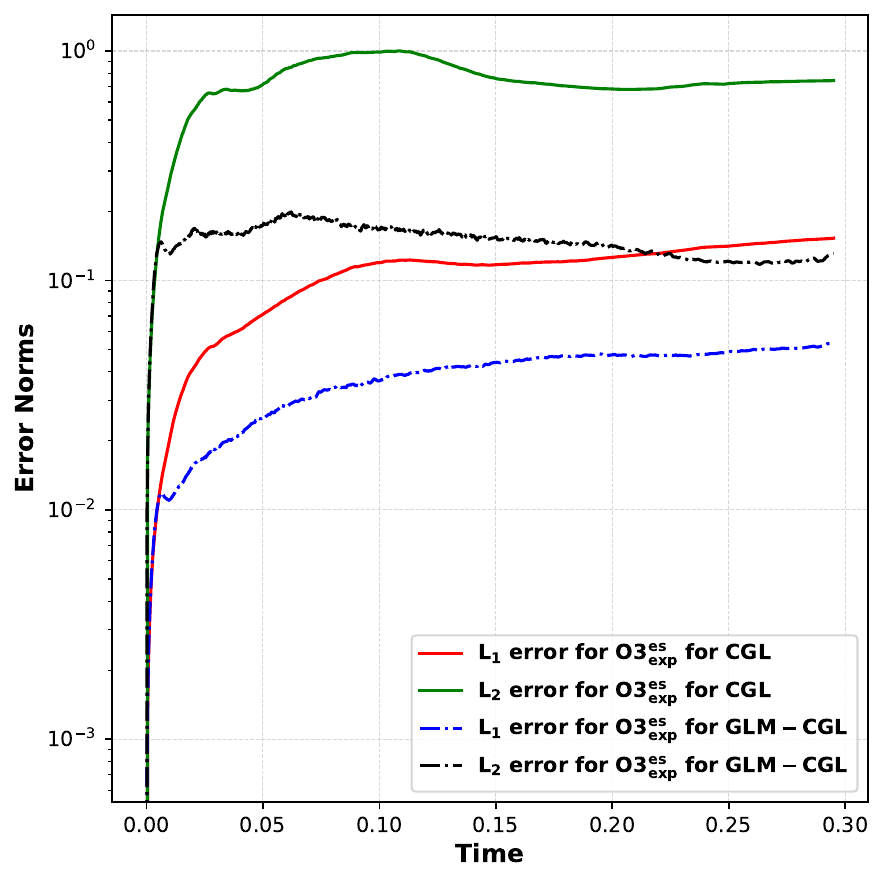}\label{fig:rt_error_cgl_o3}}~
		\subfigure[$\|\nabla\cdot\B\|_{1}$ and $\|\nabla\cdot\B\|_{2}$ for $\ofe$ scheme for CGL and GLM-CGL ]{\includegraphics[width=0.28\textwidth]{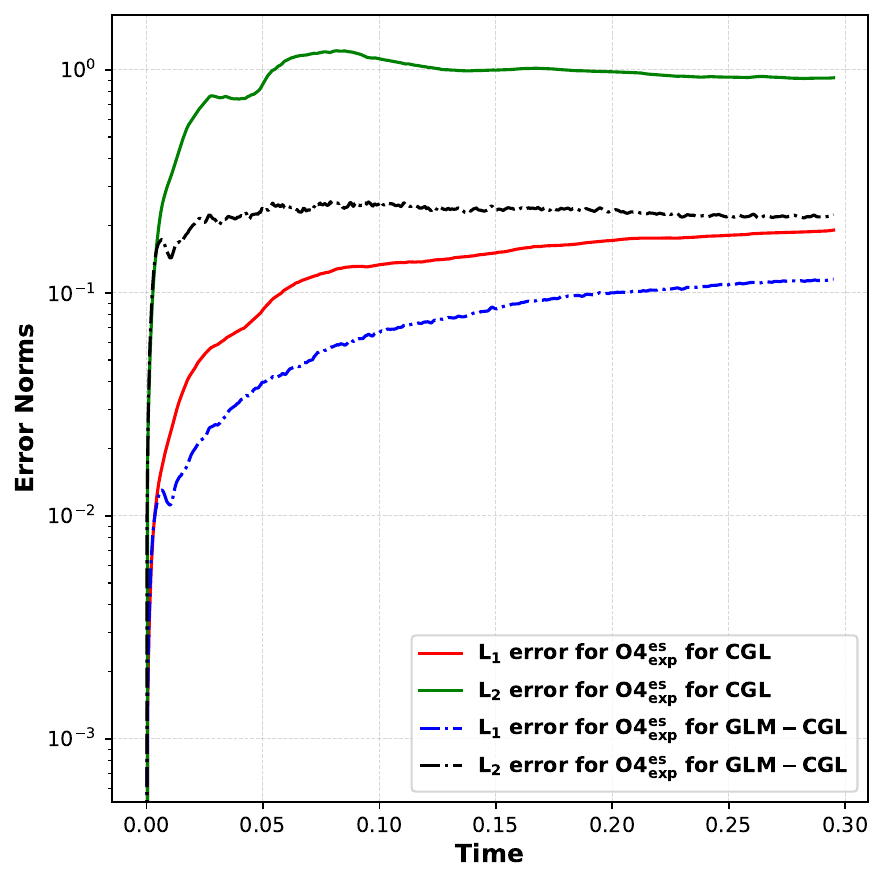}\label{fig:rt_error_cgl_o4}}\\
		\subfigure[$\|\nabla\cdot\B\|_{1}$ and $\|\nabla\cdot\B\|_{2}$ for $\oti$ scheme for isotropic CGL and isotropic GLM-CGL ]{\includegraphics[width=0.28\textwidth]{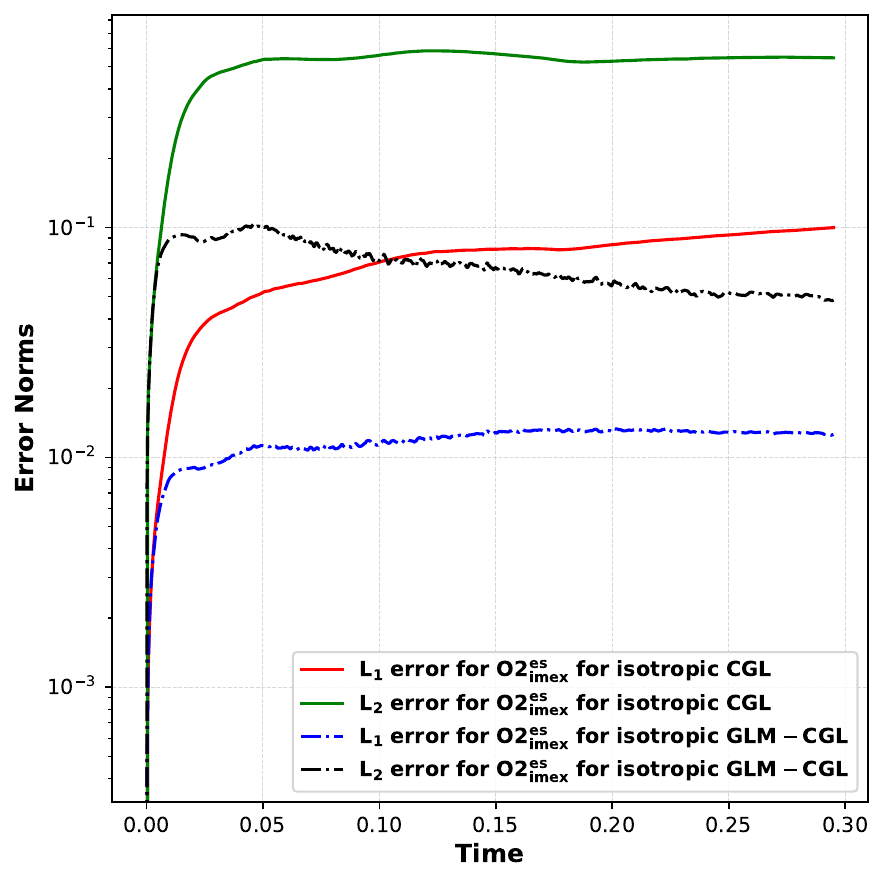}\label{fig:rt_error_mhd_o2}}~
		\subfigure[$\|\nabla\cdot\B\|_{1}$ and $\|\nabla\cdot\B\|_{2}$ for $\othi$ scheme for isotropic CGL and isotropic GLM-CGL ]{\includegraphics[width=0.28\textwidth]{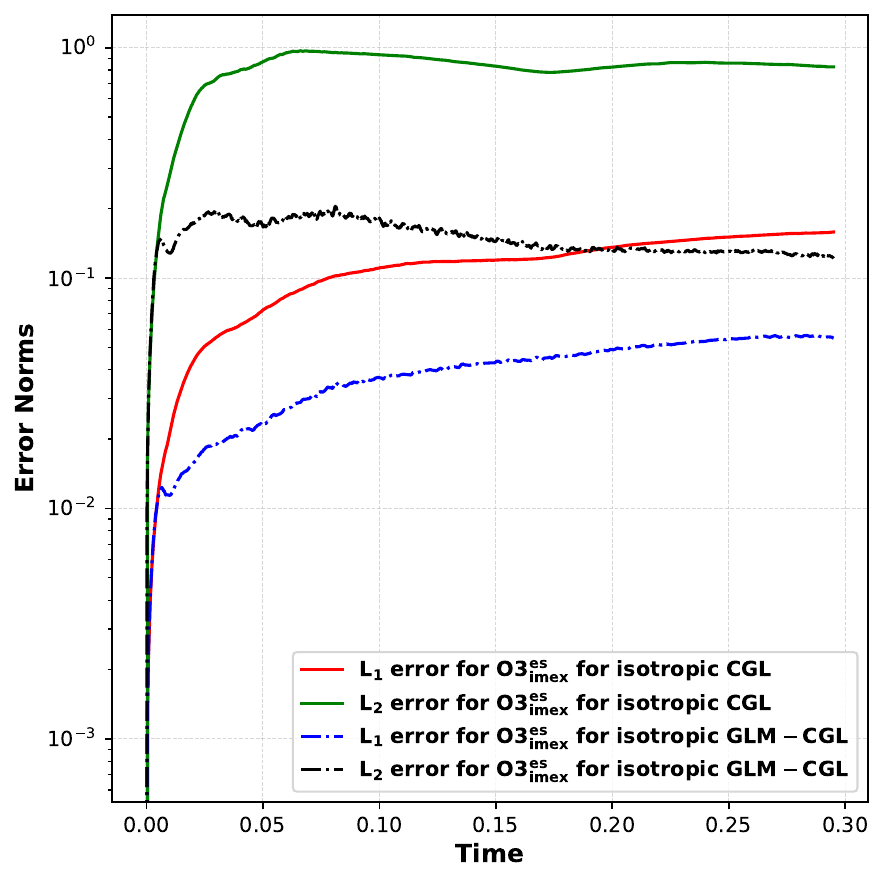}\label{fig:rt_error_mhd_o3}}~
		\subfigure[$\|\nabla\cdot\B\|_{1}$ and $\|\nabla\cdot\B\|_{2}$ for $\ofi$ scheme for isotropic CGL and isotropic GLM-CGL ]{\includegraphics[width=0.28\textwidth]{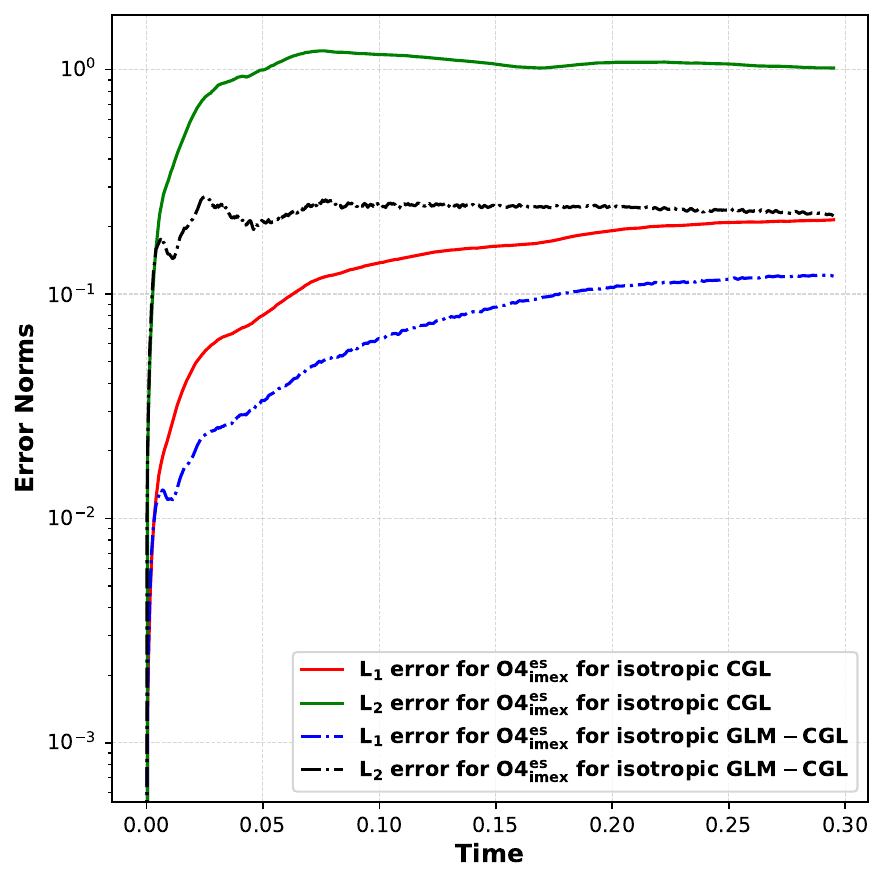}\label{fig:rt_error_mhd_o4}}\\
		\caption{\textbf{\nameref{test:rt}}: Evolution of the magnetic field divergence constraint errors till time $t=0.295$.}
		\label{fig:rt_cut_error}
	\end{center}
\end{figure}

\begin{table}[ht]
\centering
\renewcommand{\arraystretch}{1.4}
\begin{tabular}{|c|c|c|c|c|c|c|}
\hline
\multirow{2}{*}{Scheme} &
\multirow{2}{*}{System} &
$\rho_{\min}$ &
$(p_{\parallel})_{\min}$ &
$(p_{\perp})_{\min}$ &
$\|\nabla\!\cdot\!\mathbf{B}\|_{1}$ &
$\|\nabla\!\cdot\!\mathbf{B}\|_{2}$ \\
&&&&&&\\
\hline

\multirow{2}{*}{$\ote$}
& CGL
& 6.966201E-01 & 5.184962E-03 & 1.586086E-02 & 1.067062E-01 & 5.110535E-01 \\ \cline{2-7}

& GLM--CGL
& 6.946072E-01 & 5.004655E-03 & 1.524349E-02 &  1.466224E-02 & 5.243560E-02  \\  \hline

\multirow{2}{*}{$\othe$}
& CGL
& 5.922451E-01 & 3.268449E-03 & 1.634908E-02 &  1.527916E-01 & 7.431989E-01 \\ \cline{2-7}

& GLM--CGL
& 5.987161E-01 & 5.244759E-03 & 1.62247E-02 &  5.363814E-02 & 1.309093E-01  \\  \hline

\multirow{2}{*}{$\ofe$}
& CGL
& 5.531599E-01 & 1.466932E-03 & 1.645908E-02 &  1.913728E-01 & 9.218534E-01 \\ \cline{2-7}

& GLM--CGL
& 5.781969E-01 & 2.033486E-03 & 1.636593E-02 &  1.153811E-01 & 2.244592E-01  \\  \hline

\multirow{2}{*}{$\oti$}
& Isotropic CGL
& 6.961634E-01 & 1.088873E-02 & 1.088872E-02 &  9.985200E-02 & 5.471052E-01 \\ \cline{2-7}

& Isotropic GLM--CGL
&  7.247776E-01 & 1.05835E-02 & 1.058388E-02 &  1.257881E-02 & 4.836722E-02  \\  \hline

\multirow{2}{*}{$\othi$}
& Isotropic CGL
& 5.924475E-01 & 1.120533E-02 & 1.120554E-02 &  1.585083E-01 & 8.256589E-01 \\ \cline{2-7}

& Isotropic GLM--CGL
& 6.101004E-01 & 1.114166E-02 & 1.114187E-02 &  5.486058E-02 & 1.216828E-01  \\  \hline

\multirow{2}{*}{$\ofi$}
& Isotropic CGL
& 5.623716E-01 & 1.124085E-02 & 1.124092E-02 &  2.141761E-01 & 1.016381E+00  \\ \cline{2-7}

& Isotropic GLM--CGL
& 5.771091E-01 & 1.122265E-02 & 1.122283E-02 &  1.200924E-01 & 2.230905E-01  \\ \hline

\end{tabular}
\caption{\textbf{\nameref{test:rt}}: Minimum values of density $\rho$, parallel pressure $p_{\parallel}$, and perpendicular pressure $p_{\perp}$, together with the value of $L_1$ and $L_2$ norm of $\nabla\cdot\B$ at final time.}
\label{tab:div_rt}
\end{table}
Another interesting test case for MHD equations is the rotor test case, which is considered in ~\cite{balsara1999staggered,balsara2004second,toth2000b,derigs2018ideal,fuchs2011approximate}. Here, we generalize this test to the CGL model. We consider a computational domain of $[0,1]\times[0,1]$ with Neumann boundary conditions.  The initial density profile is given by,
%
\[\rho = \begin{cases}
	10.0, & \textrm{if } r< 0.1,\\
	1 + 9f(r), & \textrm{if } 0.1\leq r<0.115,\\
	1.0, & \textrm{otherwise,}
\end{cases}\]
where, $r(x,y) = |(x,y) - (0.5,0.5)|$ and $f(r) = \frac{23 - 200r}{3}$. The initial velocity vector is taken to be, 
\[\bu= \begin{cases}
	\left(-(10y-5),~(10x-5),~0\right), & \textrm{if } r< 0.1,\\
	\left(-(10y-5)f(r),~(10x-5)f(r),~0\right), & \textrm{if } 0.1\leq r<0.115,\\
	\left(0,~0,~0\right), & \textrm{otherwise.}
\end{cases}\]
The rest of the states are taken to be,
\[\left(\pll, \per, B_{x}, B_{y}, B_{z},\Psi\right) = 
\left(0.5,0.5,\frac{2.5}{\sqrt{4\pi}},0,0,0\right).\]
The numerical simulations are performed using $400\times400$ cells till the final time $t=0.295.$  The results for the CGL and GLM-CGL systems are plotted in Figure~\eqref{fig:rt_cgl_rho_divb}, using $\ote$, $\othe$ and $\ofe$  schemes. We see that all the schemes have resolved complicated shock structures in $\per$ variable. We have also plotted the absolute value of the magnetic field divergence. When comparing the results of CGL and GLM-CGL, we observe that the GLM-CGL model produces significantly lower divergence error when compared with the CGL model. 

In Figure~\eqref{fig:rt_mhd_par_divb}, we have plotted the results for the isotropic case using $\oti$, $\othi$ and $\ofi$ schemes at $400\times400$ cells at final time $t=0.295.$ We again see that the solutions are very similar to the MHD case, and all the schemes are able to resolve wave structures. Furthermore, similar to the anisotropic case, the GLM-CGL model produces significantly lower divergence errors. \revo{The colorbars in all divergence plots are dynamically scaled using the absolute local minimum and maximum values of $\nabla\cdot\B$ in each plot.}

To compare the isotropic solutions with the MHD solution, in Figure~\eqref{fig:rt_cut}, we have plotted a one-dimensional cut of the isotropic solutions and compared it with the MHD solution. We observe that both pressure components match the MHD pressure. Again, we also observe that the GLM-CGL solutions are slightly more diffusive than the CGL solutions for the same scheme, consistent with the observations discussed above. 

To monitor the divergence error over time, we have plotted the evolution of the $L_1$ and $L_2$ errors of the divergence of the magnetic field. We note that the errors are stable; however, GLM-CGL errors are significantly less than the CGL errors for both isotropic and anisotropic cases.

\revg{Table~\eqref{tab:div_rt} lists the minimum values of density, $\pll$, and $\per$ as well as the value of $L_1$ and $L_2$ norm of $\nabla\cdot\B$ at the final time. All reported values of density and pressure remain positive for all schemes and formulations, indicating that physically admissible solutions are maintained throughout the computation. The table also demonstrates the significant reduction in divergence errors achieved by the GLM formulation.}
\subsection{Field loop advection problem}\label{test:fl}
Following~\cite{gardiner2005unsplit,Hirabayashi2016new}, we consider the field loop advection problem for the CGL model. The computation domain for the test case is $[-1,1]\times[-0.5,0.5]$ with periodic boundary conditions. The solutions are computed till the final time $t=2.0$. The initial conditions are given by,
\[\left(\rho, \bu, \pll, \per, \Psi\right) = 
\left(2\times10^6, 1,2,0,2\times10^6,2\times10^6,0\right)\]
The magnetic field loop is given in the form of a vector potential as
\[A_z(x,y) = \begin{cases}
	(R-r), & \textrm{if } r\leq R,\\
	0, & \textrm{otherwise}.
\end{cases}\]
where $r=\sqrt{x^2+y^2}$ is the distance from the origin and $R=0.3$. The numerical results are presented in Figures~\eqref{fig:fl_cgl_mbs_divb}, \eqref{fig:fl_mhd_par_divb} and \eqref{fig:fl_error} using $400\times200$ cells. We have plotted $|\B|^2$ for CGL, GLM-CGL, isotropic CGL and isotropic GLM-CGL. We observe that the third and fourth order schemes are much less diffusive than the second order schemes. From the $|(\nabla\cdot\B)_{i,j}|$ plots, we observed that the GLM-CGL model produces much lower values than the CGL model. \revo{The colorbars in all divergence plots are dynamically scaled using the absolute local minimum and maximum values of $\nabla\cdot\B$ in each plot.} The time evolution of the $L_1$ and $L_2$ divergence of the magnetic field shows that CGL and GLM-CGL models have similar evolutions. 

\revg{Table~\eqref{tab:div_fl} lists the minimum values of density, $\pll$, and $\per$ as well as the value of $L_1$ and $L_2$ norm of $\nabla\cdot\B$ at the final time. All reported values of density and pressure remain positive for all schemes and formulations, indicating that physically admissible solutions are maintained throughout the computation. The table also demonstrates the significant reduction in divergence errors achieved by the GLM formulation.}
\begin{figure}[!htbp]
	\begin{center}
		\subfigure[$|\B|^2$ for $\ote$ scheme for CGL]{\includegraphics[width=0.32\textwidth]{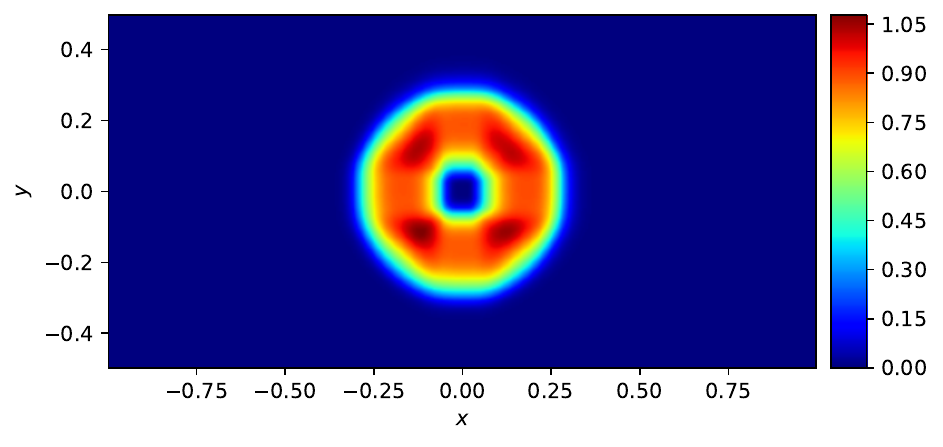}\label{fig:mbs_wo_glm_O2}}~
		\subfigure[$|\B|^2$ for $\othe$ scheme for CGL]{\includegraphics[width=0.32\textwidth]{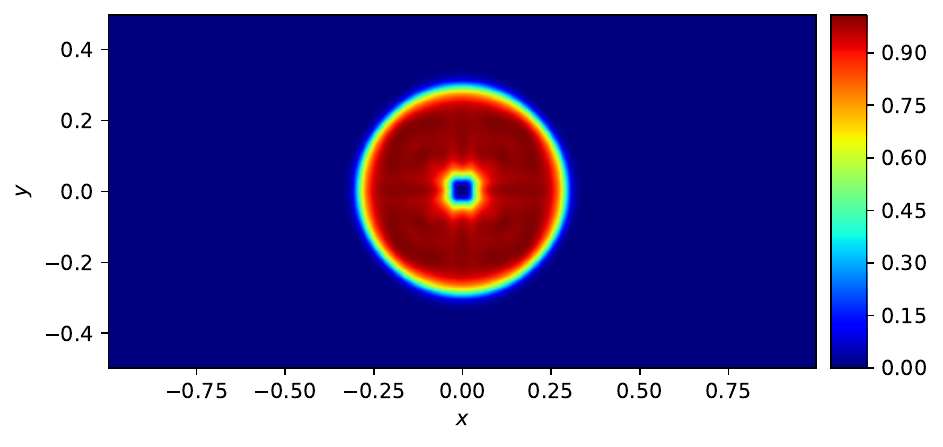}\label{fig:mbs_wo_glm_O3}}~
		\subfigure[$|\B|^2$ for $\ofe$ scheme for CGL]{\includegraphics[width=0.32\textwidth]{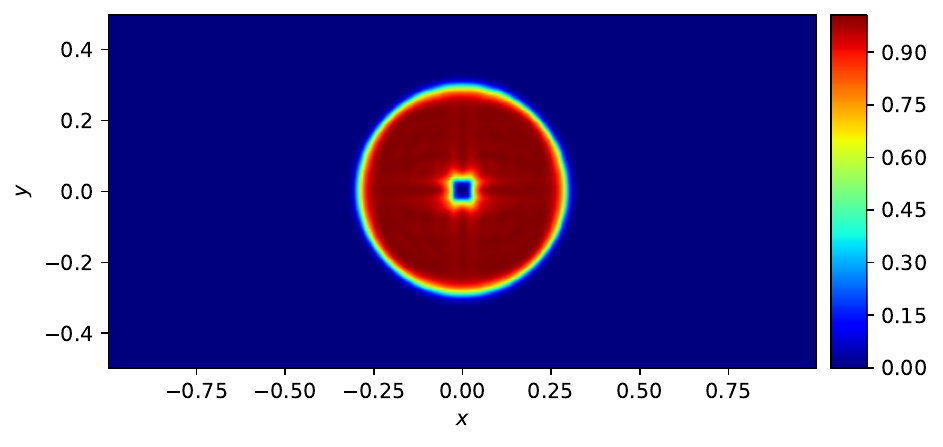}\label{fig:mbs_wo_glm_O4}}\\
		\subfigure[$|\B|^2$ for $\ote$ scheme for GLM-CGL]{\includegraphics[width=0.32\textwidth]{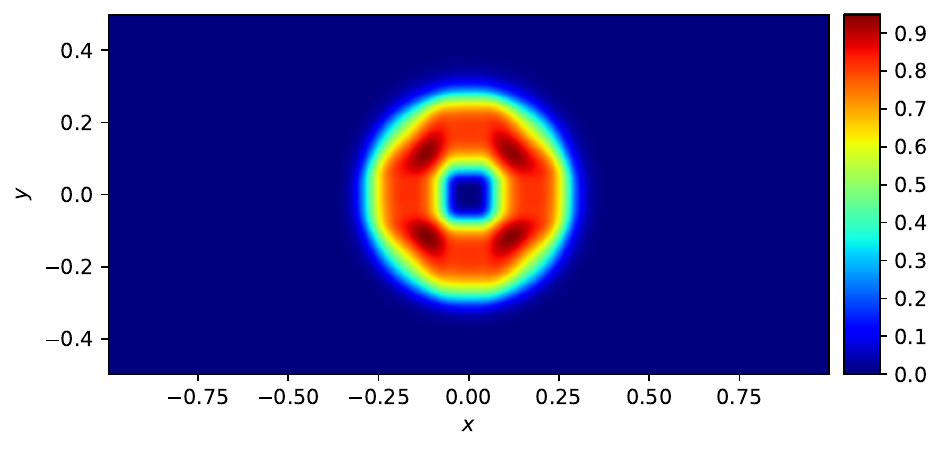}\label{fig:mbs_w_glm_O2}}~
		\subfigure[$|\B|^2$ for $\othe$ scheme for GLM-CGL]{\includegraphics[width=0.32\textwidth]{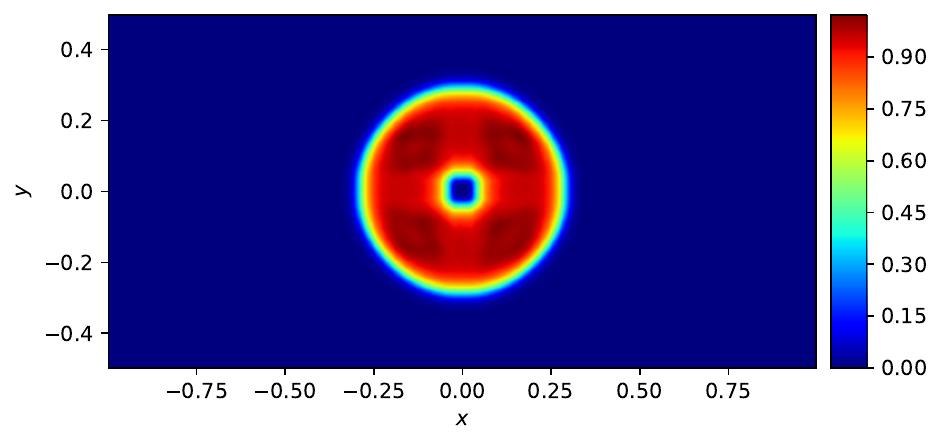}\label{fig:mbs_w_glm_O3}}~
		\subfigure[$|\B|^2$ for $\ofe$ scheme for GLM-CGL]{\includegraphics[width=0.32\textwidth]{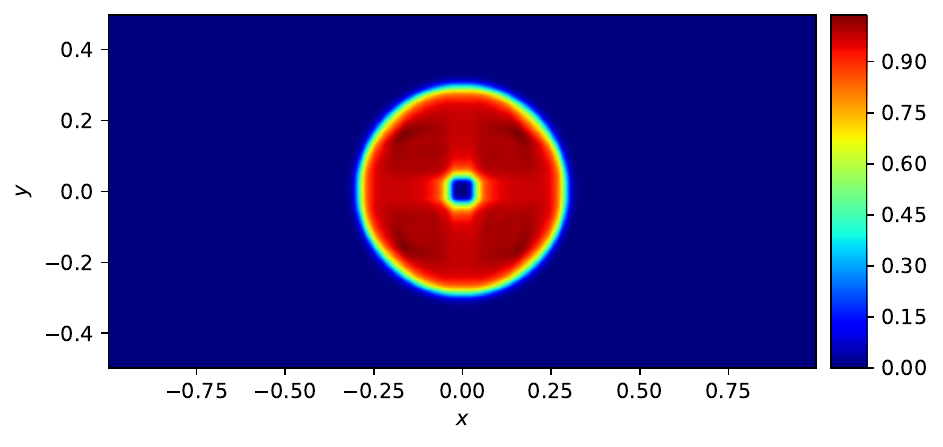}\label{fig:mbs_w_glm_O4}}\\
		\subfigure[$|(\nabla\cdot\B)_{i,j}|$ for $\ote$ scheme for CGL]{\includegraphics[width=0.32\textwidth]{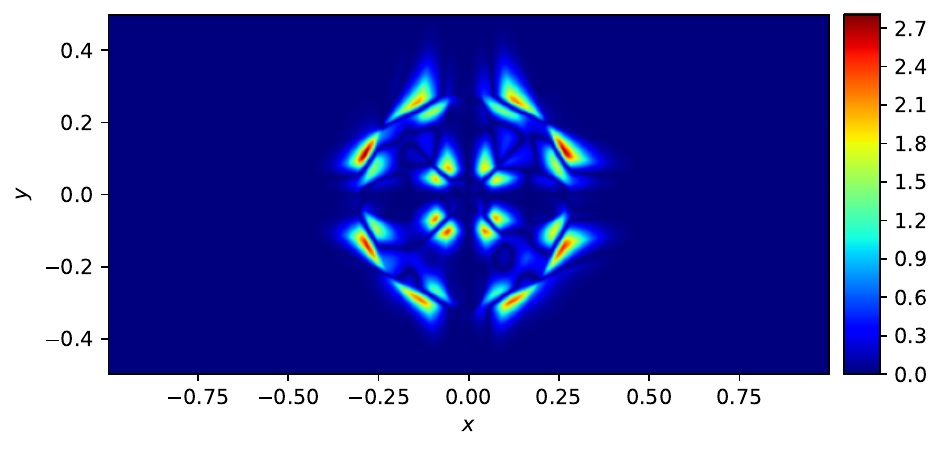}\label{fig:db_wo_glm_O2}}~
		\subfigure[$|(\nabla\cdot\B)_{i,j}|$ for $\othe$ scheme for CGL]{\includegraphics[width=0.32\textwidth]{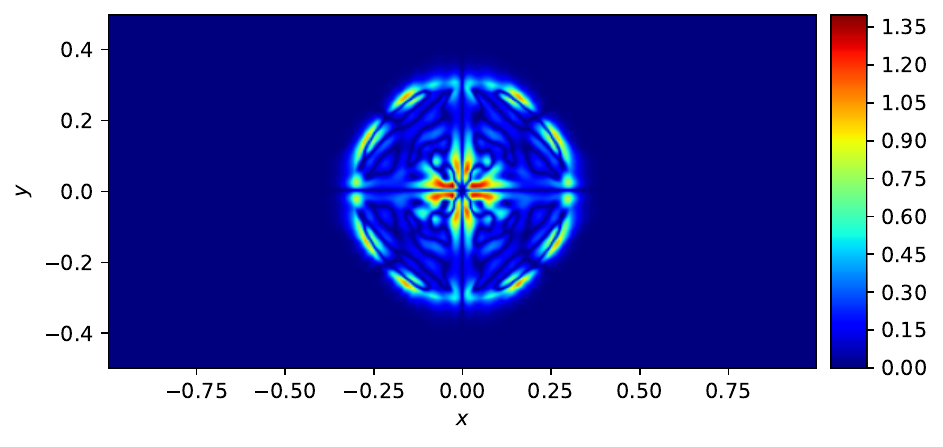}\label{fig:db_wo_glm_O3}}~
		\subfigure[$|(\nabla\cdot\B)_{i,j}|$ for $\ofe$ scheme for CGL]{\includegraphics[width=0.32\textwidth]{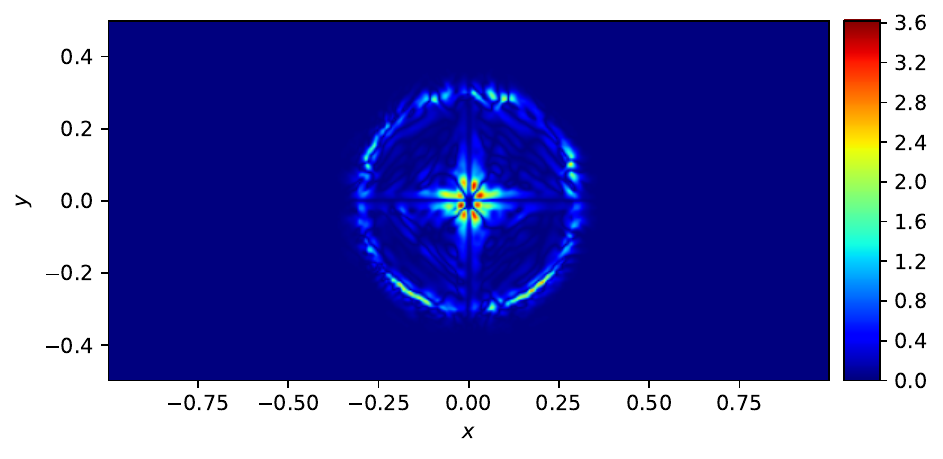}\label{fig:db_wo_glm_O4}}\\
		\subfigure[$|(\nabla\cdot\B)_{i,j}|$ for $\ote$ scheme for GLM-CGL]{\includegraphics[width=0.32\textwidth]{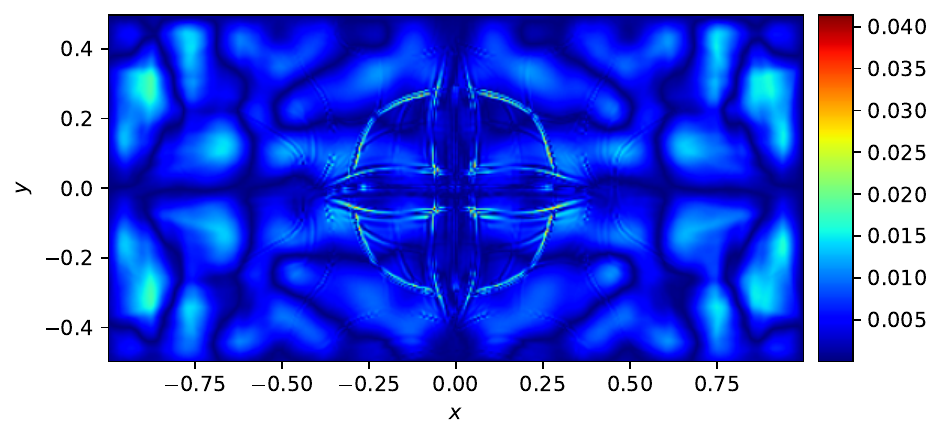}\label{fig:db_w_glm_O2}}~
		\subfigure[$|(\nabla\cdot\B)_{i,j}|$ for $\othe$ scheme for GLM-CGL]{\includegraphics[width=0.32\textwidth]{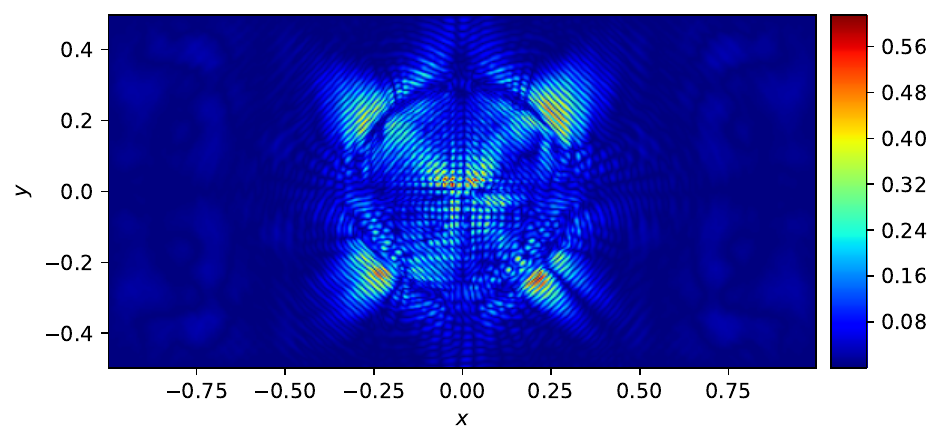}\label{fig:db_w_glm_O3}}~
		\subfigure[$|(\nabla\cdot\B)_{i,j}|$ for $\ofe$ scheme for GLM-CGL]{\includegraphics[width=0.32\textwidth]{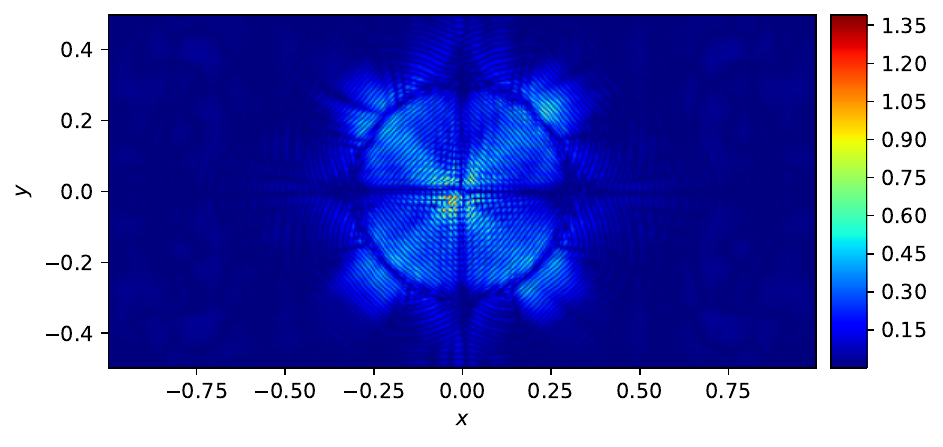}\label{fig:db_w_glm_O4}}\\
		\caption{\textbf{\nameref{test:fl}}: Plots of $|\B|^2$ and $|(\nabla\cdot\B)_{i,j}|$ for $\ote$, $\othe$ and $\ofe$ schemes for CGL and GLM-CGL at time $t=2.0$.}
		\label{fig:fl_cgl_mbs_divb}
	\end{center}
\end{figure}
\begin{figure}[!htbp]
	\begin{center}
		\subfigure[$|\B|^2$ for $\oti$ scheme for isotropic CGL]{\includegraphics[width=0.32\textwidth]{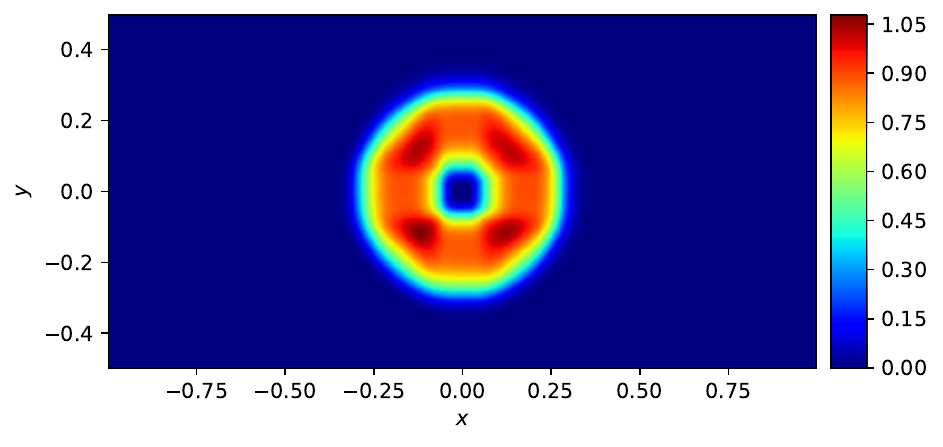}\label{fig:mbs_wo_glm_m_O2}}~
		\subfigure[$|\B|^2$ for $\othi$ scheme for isotropic CGL]{\includegraphics[width=0.32\textwidth]{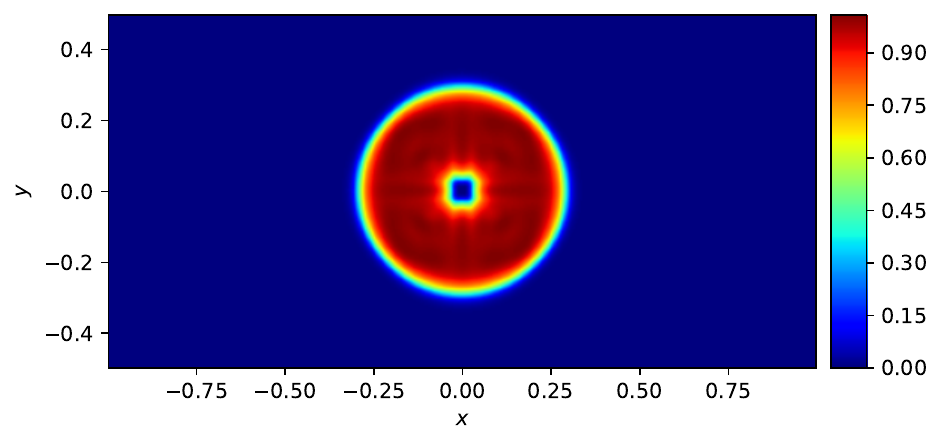}\label{fig:mbs_wo_glm_m_O3}}~
		\subfigure[$|\B|^2$ for $\ofi$ scheme for isotropic CGL]{\includegraphics[width=0.32\textwidth]{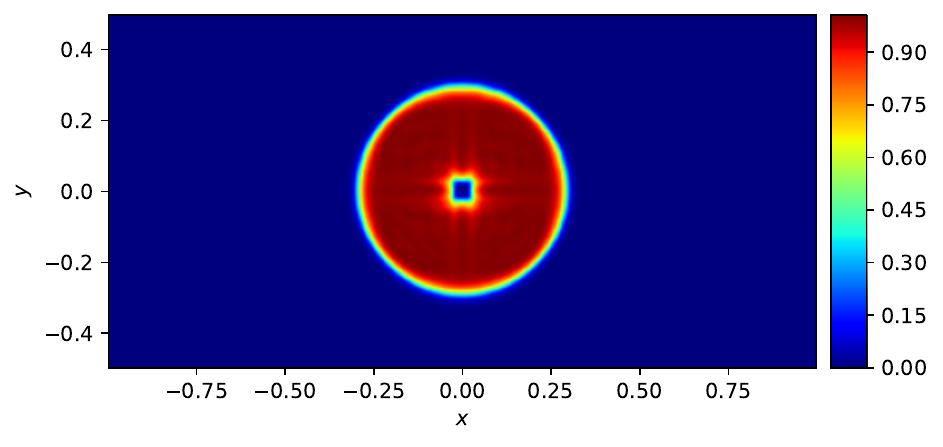}\label{fig:mbs_wo_glm_m_O4}}\\
		\subfigure[$|\B|^2$ for $\oti$ scheme for isotropic GLM-CGL]{\includegraphics[width=0.32\textwidth]{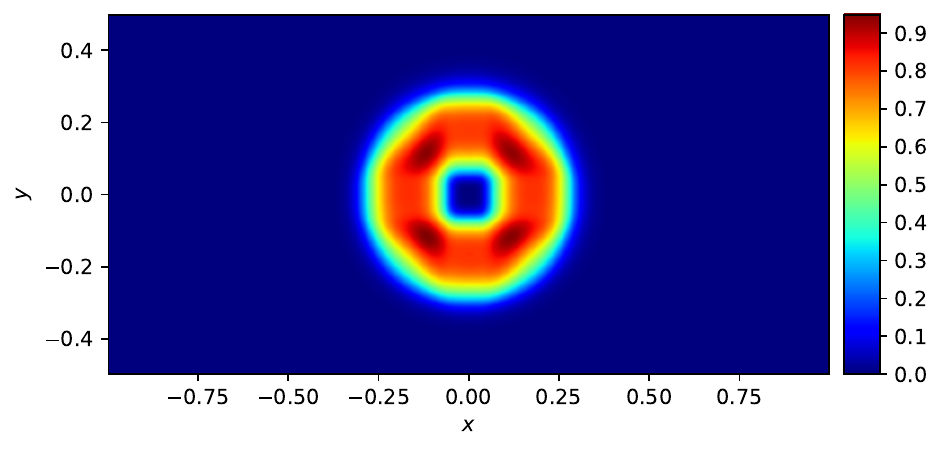}\label{fig:mbs_w_glm_m_O2}}~
		\subfigure[$|\B|^2$ for $\othi$ scheme for isotropic GLM-CGL]{\includegraphics[width=0.32\textwidth]{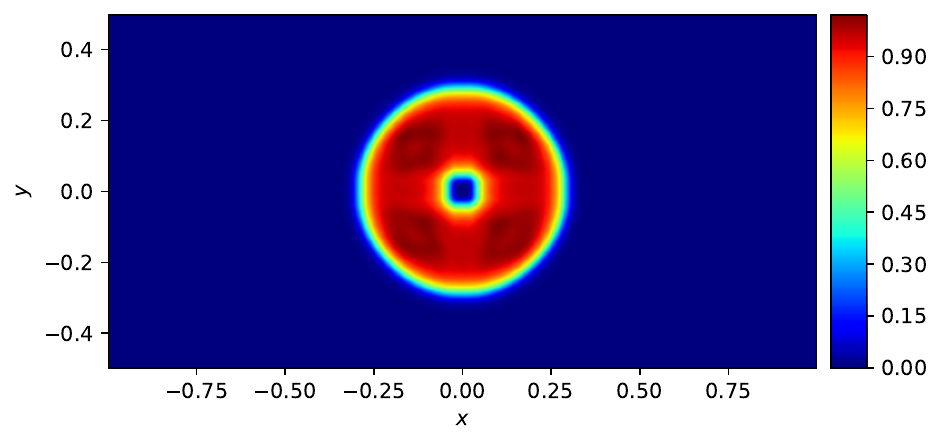}\label{fig:mbs_w_glm_m_O3}}~
		\subfigure[$|\B|^2$ for $\ofi$ scheme for isotropic GLM-CGL]{\includegraphics[width=0.32\textwidth]{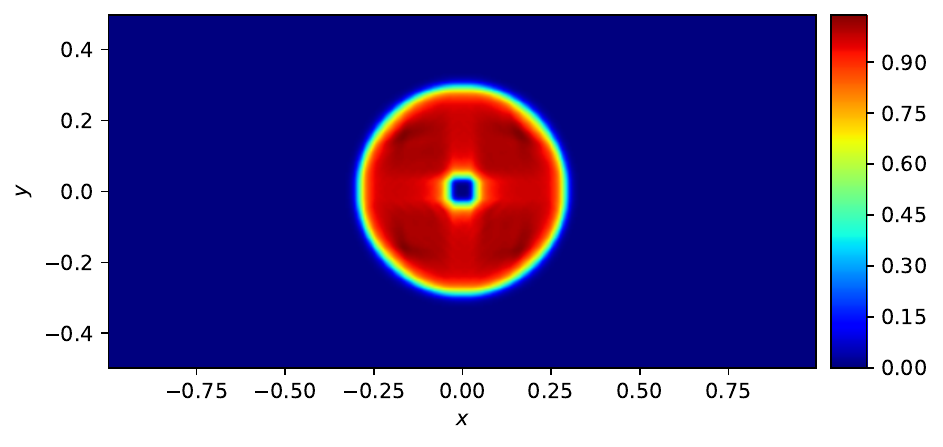}\label{fig:mbs_w_glm_m_O4}}\\
		\subfigure[$|(\nabla\cdot\B)_{i,j}|$ for $\oti$ scheme for isotropic CGL]{\includegraphics[width=0.32\textwidth]{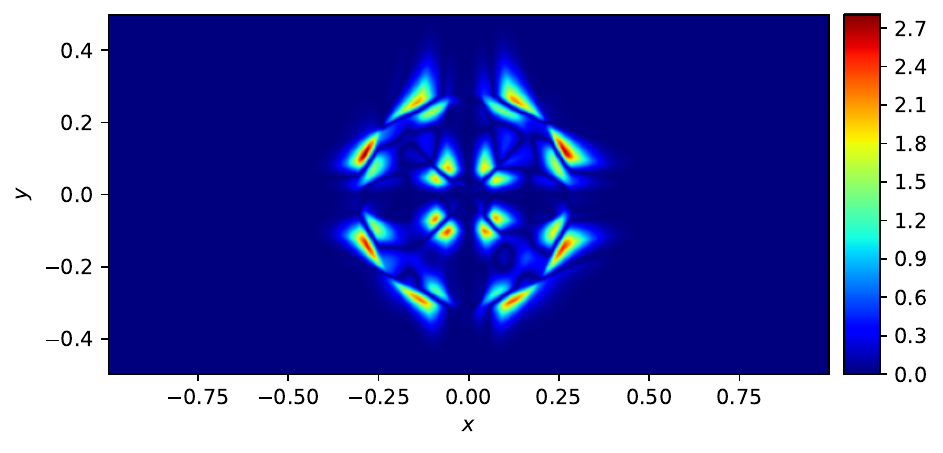}\label{fig:db_wo_glm_m_O2}}~
		\subfigure[$|(\nabla\cdot\B)_{i,j}|$ for $\othi$ scheme for isotropic CGL]{\includegraphics[width=0.32\textwidth]{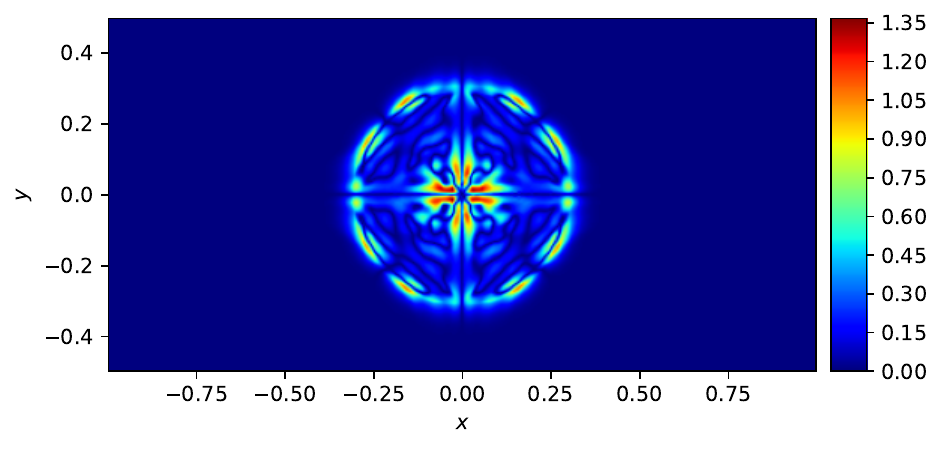}\label{fig:db_wo_glm_m_O3}}~
		\subfigure[$|(\nabla\cdot\B)_{i,j}|$ for $\ofi$ scheme for isotropic CGL]{\includegraphics[width=0.32\textwidth]{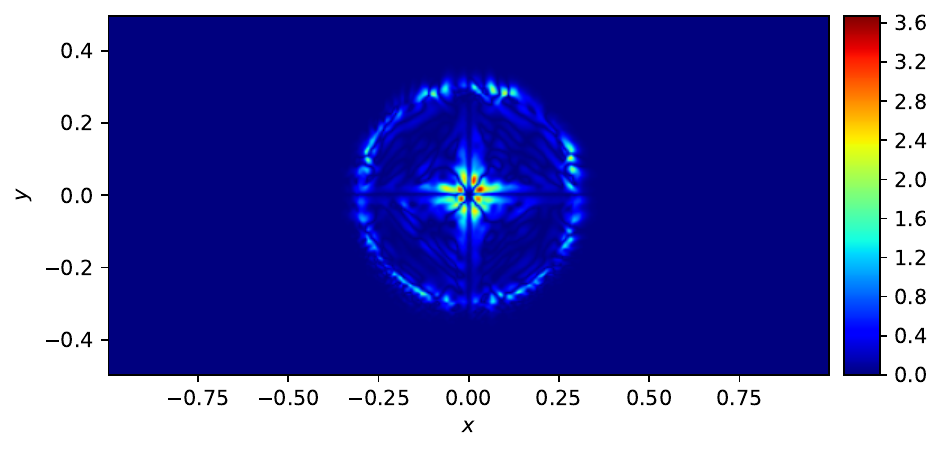}\label{fig:db_wo_glm_m_O4}}\\
		\subfigure[$|(\nabla\cdot\B)_{i,j}|$ for $\oti$ scheme for isotropic GLM-CGL]{\includegraphics[width=0.32\textwidth]{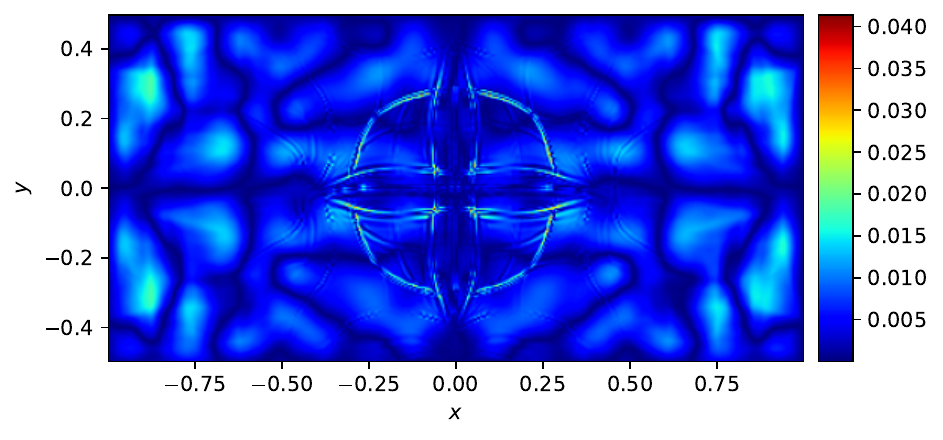}\label{fig:db_w_glm_m_O2}}~
		\subfigure[$|(\nabla\cdot\B)_{i,j}|$ for $\othi$ scheme for isotropic GLM-CGL]{\includegraphics[width=0.32\textwidth]{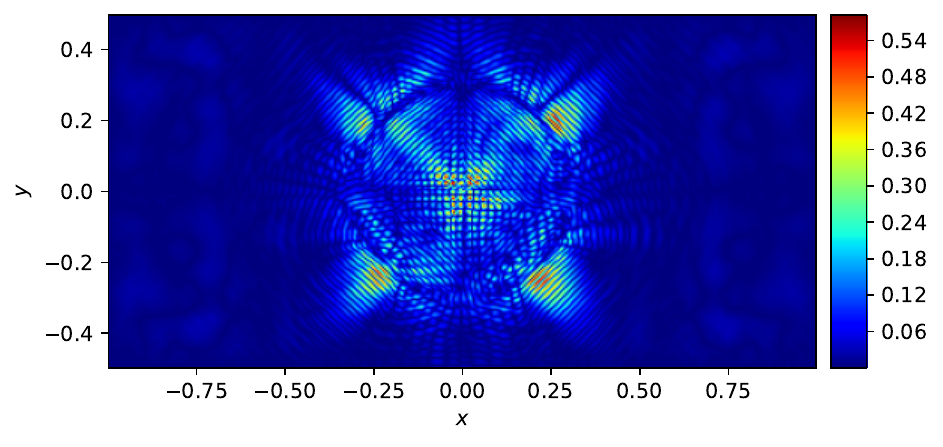}\label{fig:db_w_glm_m_O3}}~
		\subfigure[$|(\nabla\cdot\B)_{i,j}|$ for $\ofi$ scheme for isotropic GLM-CGL]{\includegraphics[width=0.32\textwidth]{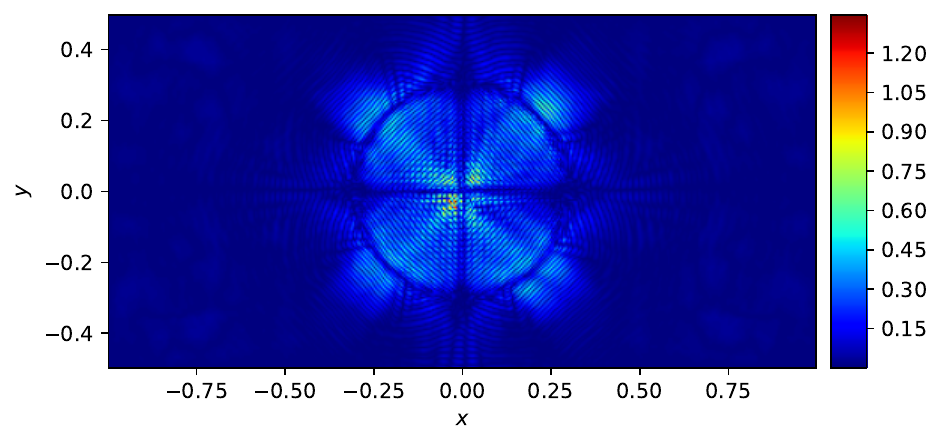}\label{fig:db_w_glm_m_O4}}\\
		\caption{\textbf{\nameref{test:fl}}: Plots of $|\B|^2$ and $|(\nabla\cdot\B)_{i,j}|$ for $\oti$, $\othi$ and $\ofi$ schemes for isotropic CGL and isotropic GLM-CGL at time $t=2.0$.}
		\label{fig:fl_mhd_par_divb}
	\end{center}
\end{figure}
\begin{figure}[!htbp]
	\begin{center}	
		\subfigure[$\|\nabla\cdot\B\|_{1}$ and $\|\nabla\cdot\B\|_{2}$ for $\ote$ scheme for CGL and GLM-CGL ]{\includegraphics[width=0.28\textwidth]{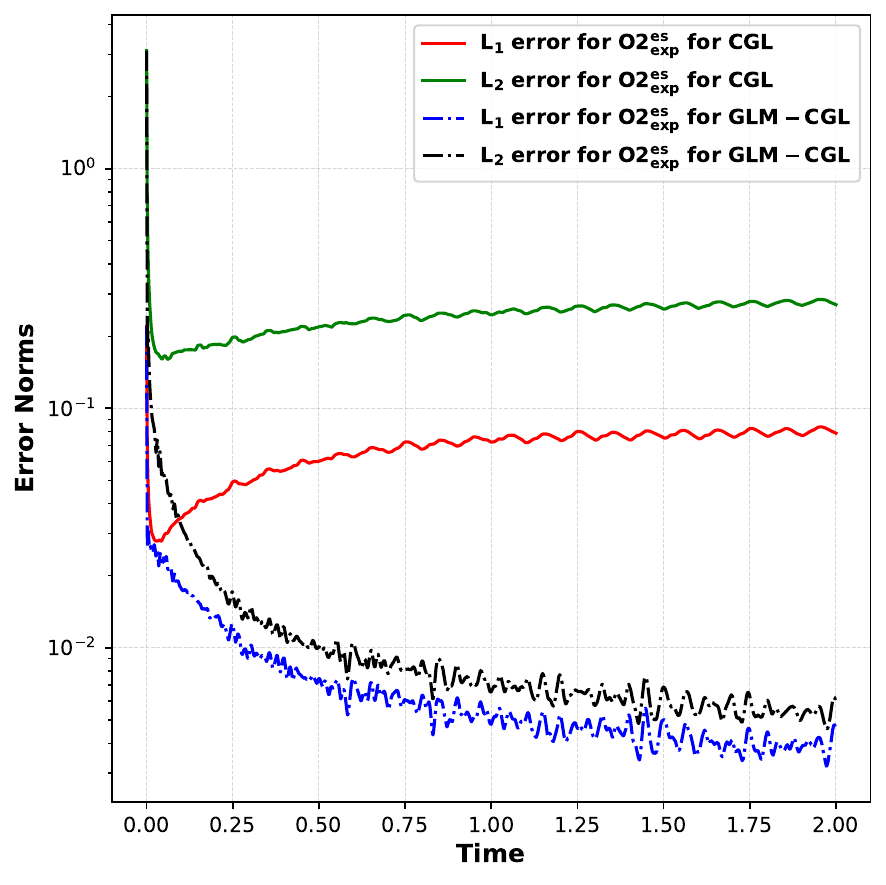}\label{fig:error_fl_O2}}~
		\subfigure[$\|\nabla\cdot\B\|_{1}$ and $\|\nabla\cdot\B\|_{2}$ for $\othe$ scheme for CGL and GLM-CGL ]{\includegraphics[width=0.28\textwidth]{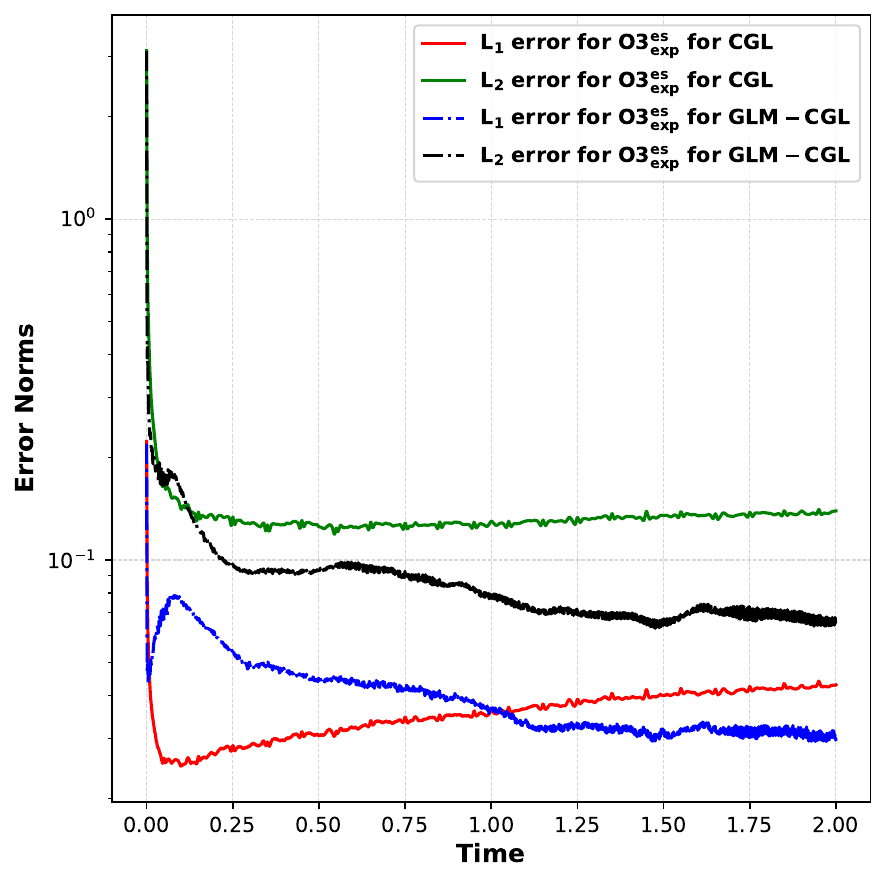}\label{fig:error_fl_O3}}~
		\subfigure[$\|\nabla\cdot\B\|_{1}$ and $\|\nabla\cdot\B\|_{2}$ for $\ofe$ scheme for CGL and GLM-CGL ]{\includegraphics[width=0.28\textwidth]{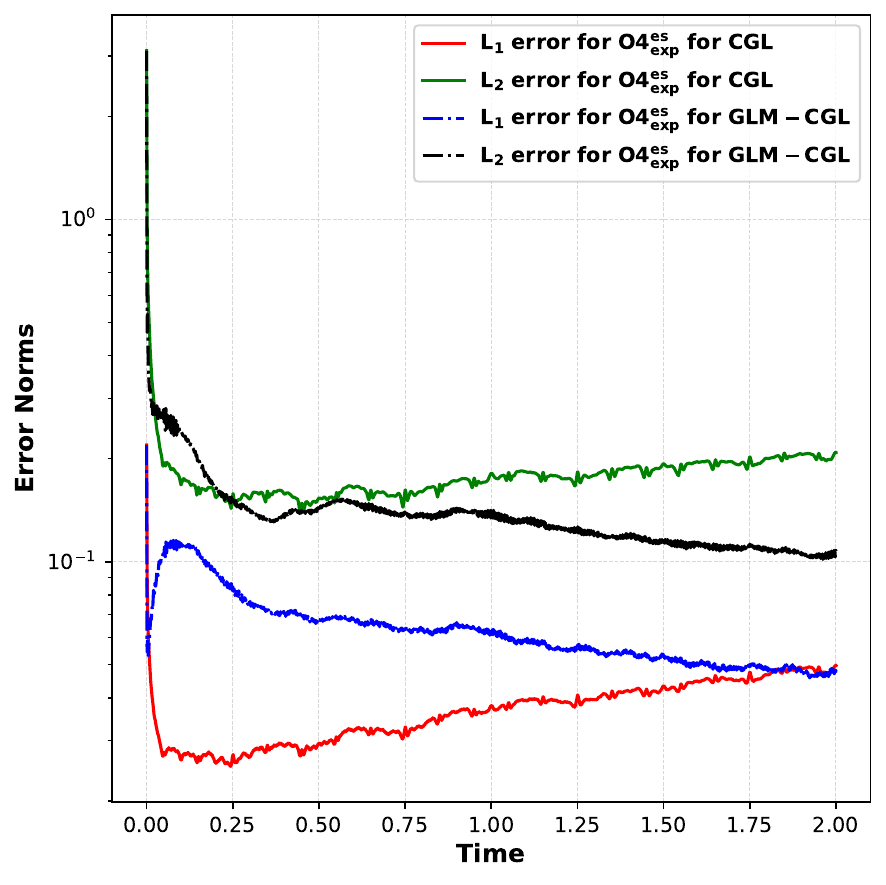}\label{fig:error_fl_O4}}\\
		\subfigure[$\|\nabla\cdot\B\|_{1}$ and $\|\nabla\cdot\B\|_{2}$ for $\oti$ scheme for isotropic CGL and isotropic GLM-CGL ]{\includegraphics[width=0.28\textwidth]{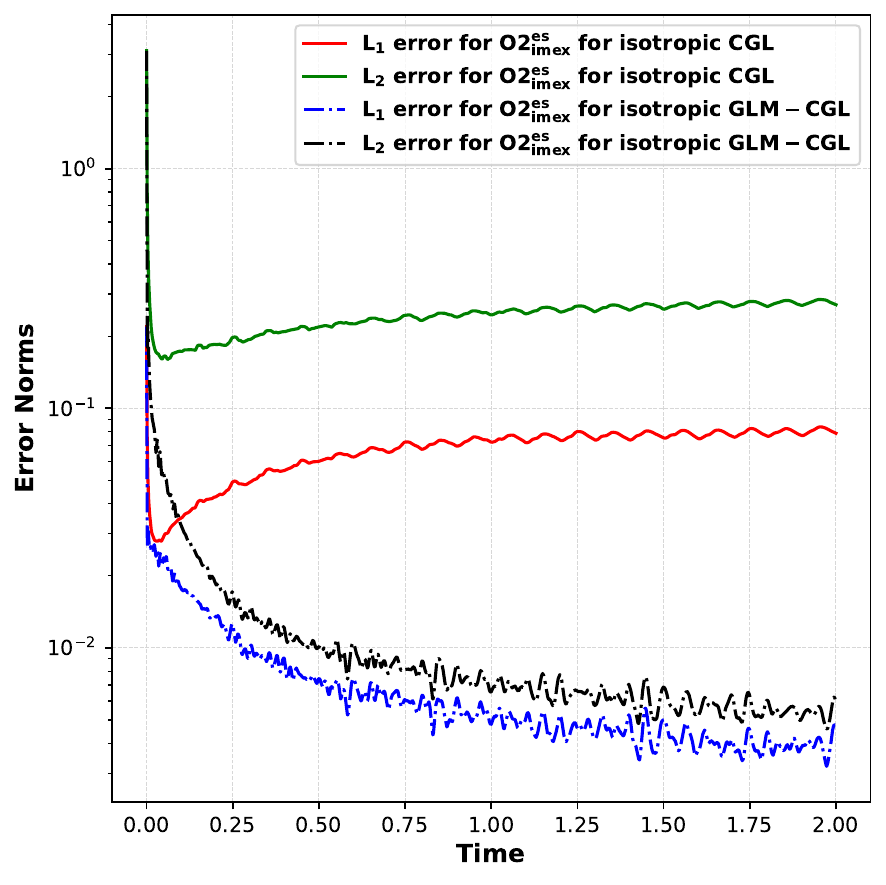}\label{fig:error_fl_m_O2}}~
		\subfigure[$\|\nabla\cdot\B\|_{1}$ and $\|\nabla\cdot\B\|_{2}$ for $\othi$ scheme for isotropic CGL and isotropic GLM-CGL ]{\includegraphics[width=0.28\textwidth]{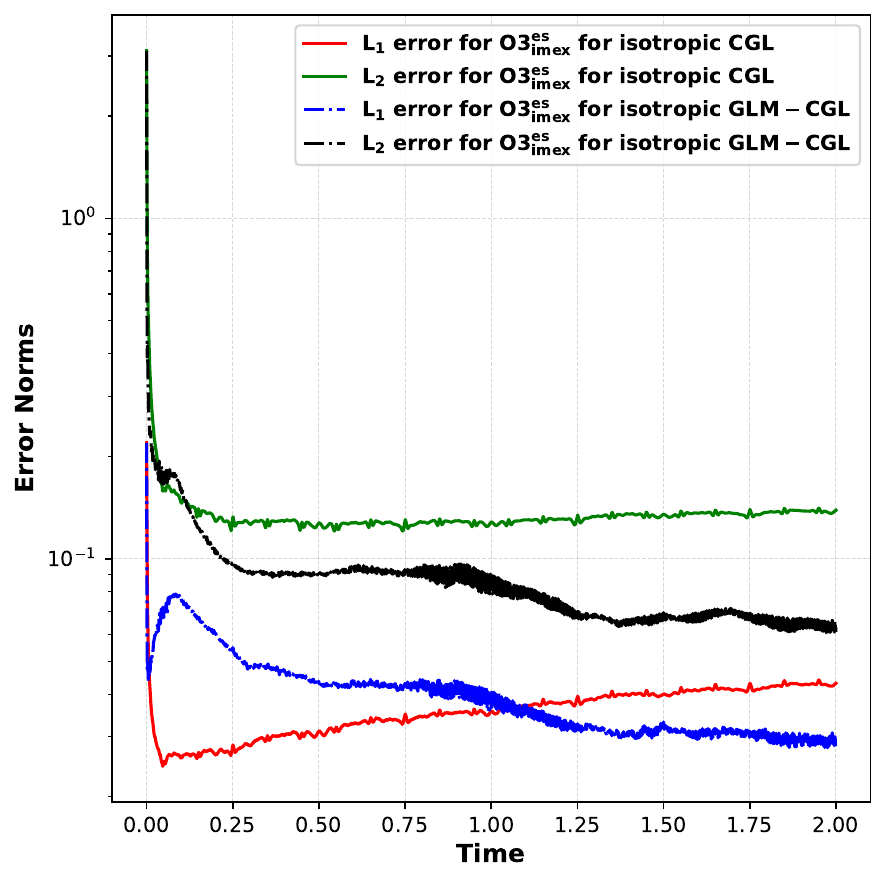}\label{fig:error_fl_m_O3}}~
		\subfigure[$\|\nabla\cdot\B\|_{1}$ and $\|\nabla\cdot\B\|_{2}$ for $\ofi$ scheme for isotropic CGL and isotropic GLM-CGL ]{\includegraphics[width=0.28\textwidth]{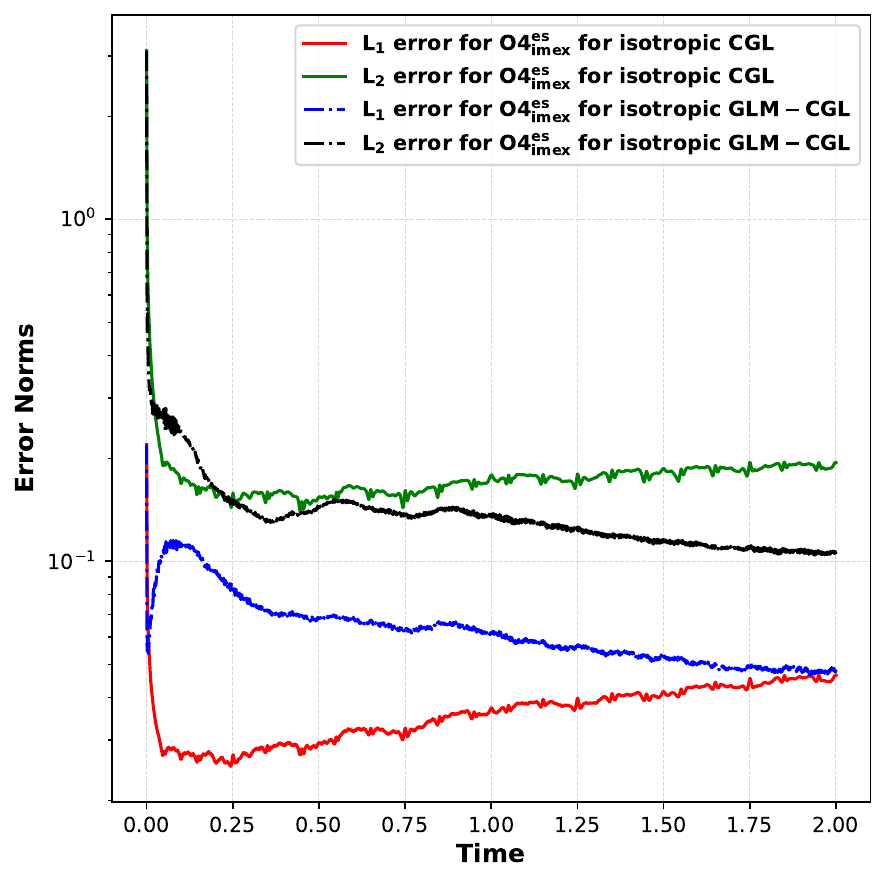}\label{fig:error_fl_m_O4}}\\
		\caption{\textbf{\nameref{test:fl}}: Evolution of the magnetic field divergence constraint errors till time $t=2.0$.}
		\label{fig:fl_error}
	\end{center}
\end{figure}

\begin{table}[ht]
\centering
\renewcommand{\arraystretch}{1.4}
\begin{tabular}{|c|c|c|c|c|c|c|}
\hline
\multirow{2}{*}{Scheme} &
\multirow{2}{*}{System} &
$\rho_{\min}$ &
$(p_{\parallel})_{\min}$ &
$(p_{\perp})_{\min}$ &
$\|\nabla\!\cdot\!\mathbf{B}\|_{1}$ &
$\|\nabla\!\cdot\!\mathbf{B}\|_{2}$ \\
&&&&&&\\
\hline

\multirow{2}{*}{$\ote$}
& CGL
& 2.000000E+06 & 2.000000E+06 & 2.000000E+06 &  7.870649E-02 & 2.706245E-01 \\ \cline{2-7}

& GLM--CGL
& 2.000000E+06 & 2.000000E+06 & 2.000000E+06 &  4.556005E-03 & 5.918808E-03  \\  \hline

\multirow{2}{*}{$\othe$}
& CGL
& 2.000000E+06 & 2.000000E+06 & 2.000000E+06 &  4.301938E-02 & 1.393251E-01 \\ \cline{2-7}

& GLM--CGL
& 2.000000E+06 & 2.000000E+06 & 2.000000E+06 &  2.982112E-02 & 6.506547E-02  \\  \hline

\multirow{2}{*}{$\ofe$}
& CGL
& 1.999999E+06 & 1.999999E+06 & 1.999999E+06 &  4.975039E-02 & 2.085295E-01 \\ \cline{2-7}

& GLM--CGL
& 2.000000E+06 & 2.000000E+06 & 2.000000E+06 & 4.800345E-02 & 1.075252E-01  \\  \hline
\multirow{2}{*}{$\oti$}
& Isotropic CGL
& 2.000000E+06 & 2.000000E+06 & 2.000000E+06 &  7.870441E-02 & 2.706190E-01 \\ \cline{2-7}

& Isotropic GLM--CGL
& 2.000000E+06 & 2.000000E+06 & 2.000000E+06 &  4.557633E-03 & 5.920931E-03  \\  \hline

\multirow{2}{*}{$\othi$}
& Isotropic CGL
& 2.000000E+06 & 2.000000E+06 & 2.000000E+06 &  4.306339E-02 & 1.388938E-01  \\ \cline{2-7}

& Isotropic GLM--CGL
& 2.000000E+06 & 2.000000E+06 & 2.000000E+06 &  2.866595E-02 & 6.276774E-02  \\  \hline

\multirow{2}{*}{$\ofi$}
& Isotropic CGL
& 1.999999E+06 & 1.999999E+06 & 1.999999E+06 &  4.641830E-02 & 1.943417E-01  \\ \cline{2-7}

& Isotropic GLM--CGL
& 2.000000E+06 & 2.000000E+06 & 2.000000E+06 &  4.768871E-02 & 1.063972E-01  \\ \hline

\end{tabular}
\caption{\textbf{\nameref{test:fl}}: Minimum values of density $\rho$, parallel pressure $p_{\parallel}$, and perpendicular pressure $p_{\perp}$, together with the value of $L_1$ and $L_2$ norm of $\nabla\cdot\B$ at final time.}
\label{tab:div_fl}
\end{table}
\subsection{CGL Riemann problem}\label{test:CRP}
Following the MHD Riemann problem in ~\cite{torrilhon2005locally,artebrant2008increasing}, we present the corresponding CGL test case. The computational domain is taken to be $[-0.4,0.4]\times[-0.4,0.4]$ with Neumann boundary conditions. The initial conditions are given below: 
\[\left(\rho, \pll, \per\right)  = \begin{cases}
	(10, 15, 15)), & \textrm{if } x<0,~y<0,\\
	(1,0.5,0.5), & \textrm{otherwise.}
\end{cases}\]
\[\left(\bu, B_{x}, B_{y}, B_{z},\Psi\right) = 
\left(0,0,0,\frac{1}{\sqrt{2}},\frac{1}{\sqrt{2}},0,0\right).\]

The numerical solutions for $B_y$ are plotted in Figures~\eqref{fig:crp_cgl_by_divb}, \eqref{fig:crp_mhd_by_divb}, and  \eqref{fig:crp_error}. We observe that the higher-order schemes are much more accurate than the second-order schemes. We also note that the results of the isotropic cases are similar to the MHD case in~\cite{torrilhon2005locally}. Furthermore, CGL and GLM-CGL results are different from the isotropic case. 

From the plots of $|(\nabla\cdot\B)_{i,j}|$, in Figures~\eqref{fig:crp_cgl_by_divb} and \eqref{fig:crp_mhd_by_divb},  we observe that in each case, GLM-CGL produces much lower values than the CGL case. \revo{The colorbars in all divergence plots are dynamically scaled using the absolute local minimum and maximum values of $\nabla\cdot\B$ in each plot.} This is also evident from the time evolution of the $L_1$ and $L_2$ norms of the divergence of the magnetic field in Figure~\eqref{fig:crp_error}.

\revg{Table~\eqref{tab:div_CRP} lists the minimum values of density, $\pll$, and $\per$ as well as the value of $L_1$ and $L_2$ norm of $\nabla\cdot\B$ at the final time. All reported values of density and pressure remain positive for all schemes and formulations, indicating that physically admissible solutions are maintained throughout the computation. The table also demonstrates the significant reduction in divergence errors achieved by the GLM formulation.}
\begin{figure}[!htbp]
	\begin{center}
		\subfigure[$B_y$ for $\ote$ scheme for CGL]{\includegraphics[width=0.26\textwidth]{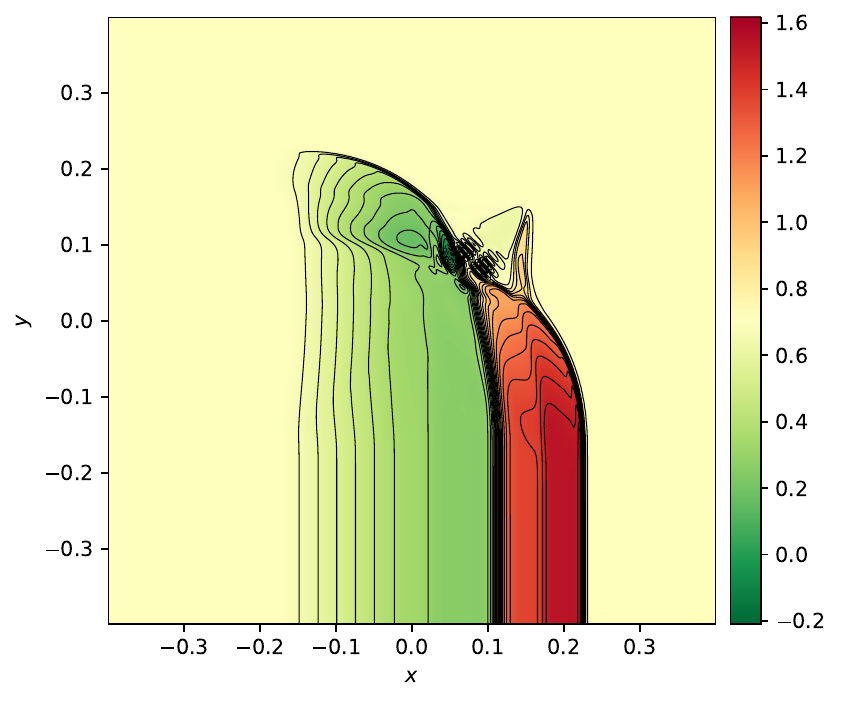}\label{fig:crp_by_wo_o2}}~
		\subfigure[$B_y$ for $\othe$ scheme for CGL]{\includegraphics[width=0.26\textwidth]{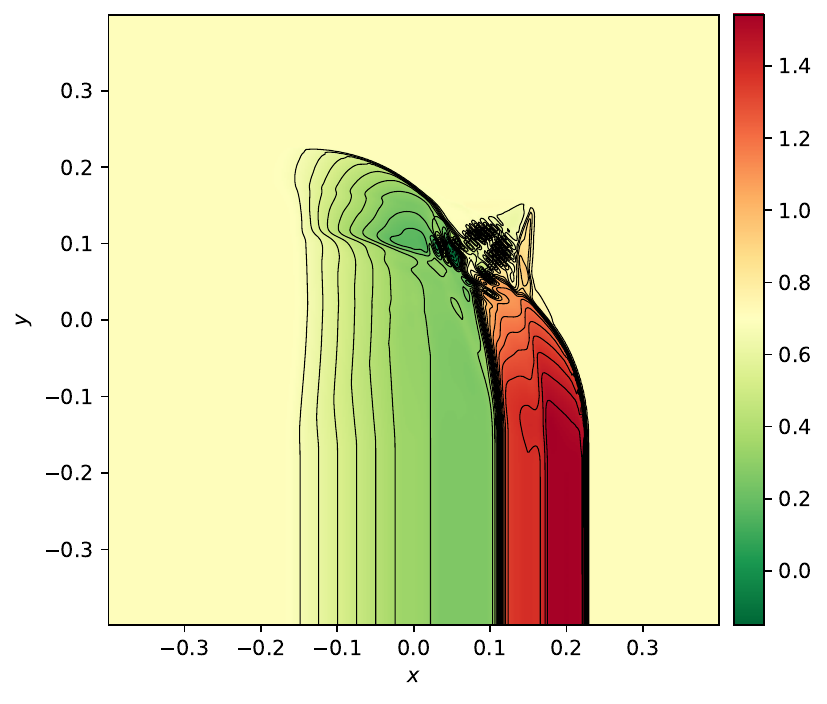}\label{fig:crp_by_wo_o3}}~
		\subfigure[$B_y$ for $\ofe$ scheme for CGL]{\includegraphics[width=0.26\textwidth]{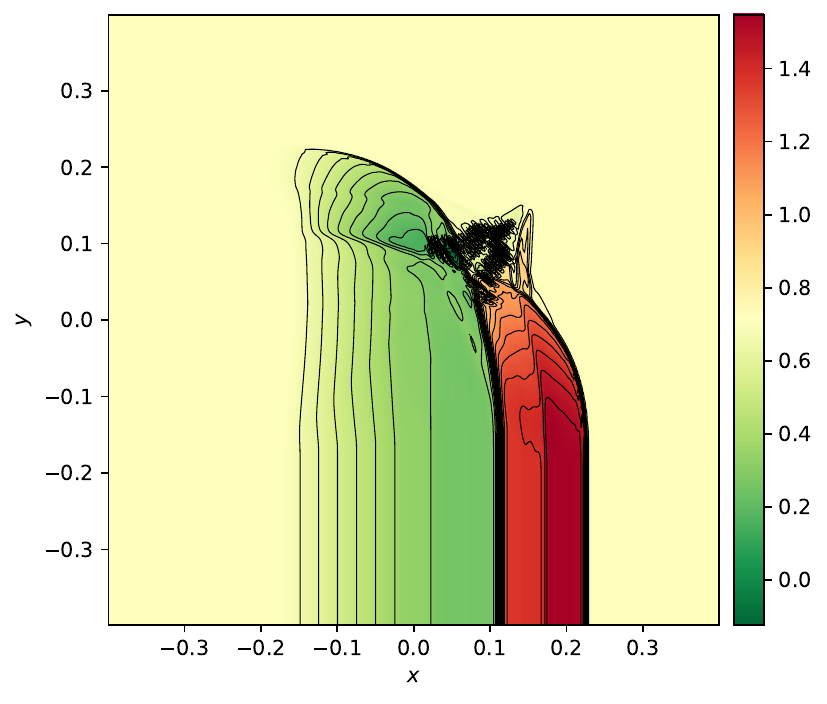}\label{fig:crp_by_wo_o4}}\\
		\subfigure[$B_y$ for $\ote$ scheme for GLM-CGL]{\includegraphics[width=0.26\textwidth]{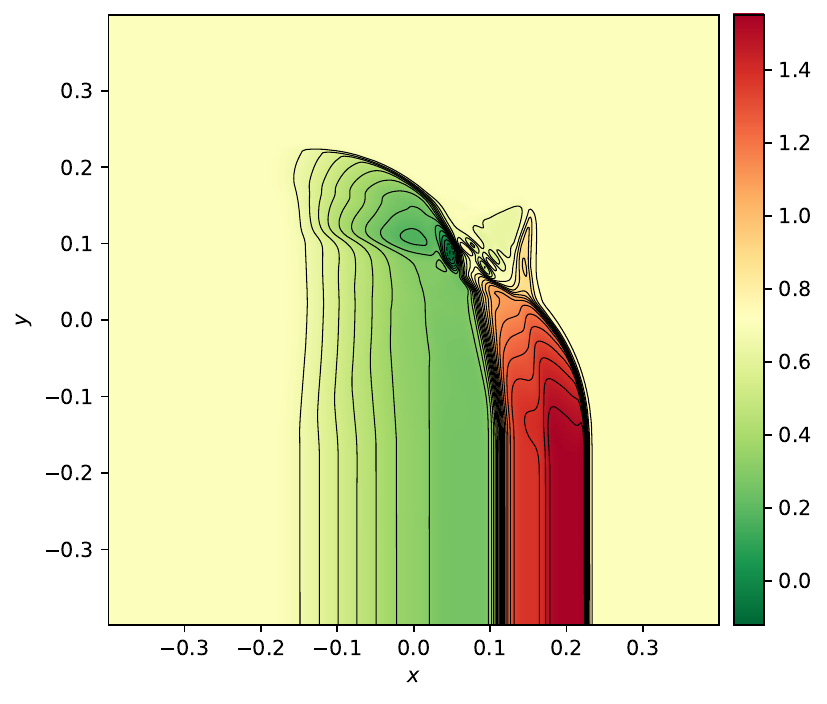}\label{fig:crp_by_w_o2}}~
		\subfigure[$B_y$ for $\othe$ scheme for GLM-CGL]{\includegraphics[width=0.26\textwidth]{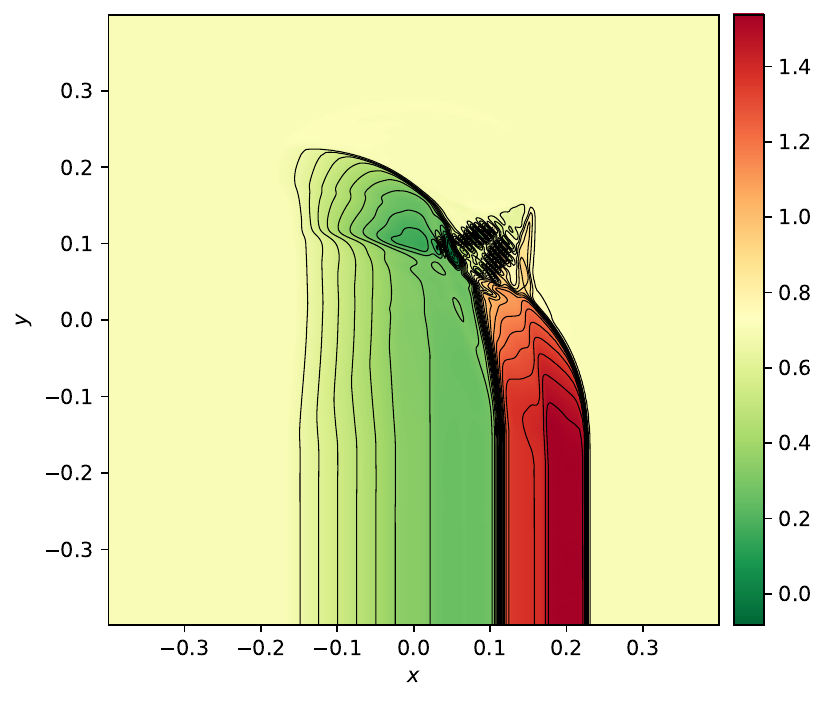}\label{fig:crp_by_w_o3}}~
		\subfigure[$B_y$ for $\ofe$ scheme for GLM-CGL]{\includegraphics[width=0.26\textwidth]{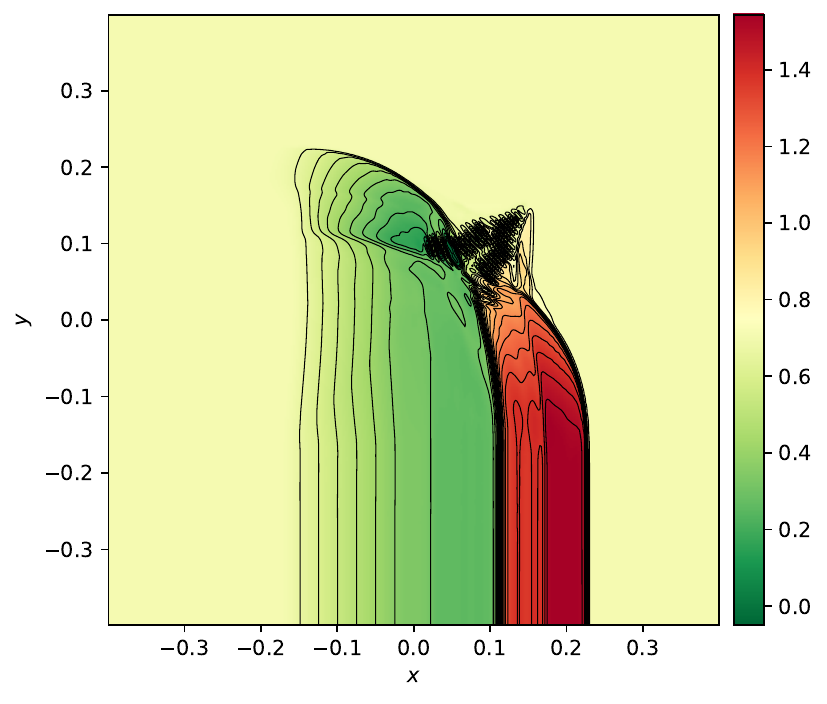}\label{fig:crp_by_w_o4}}\\
		\subfigure[$|(\nabla\cdot\B)_{i,j}|$ for $\ote$ scheme for CGL]{\includegraphics[width=0.26\textwidth]{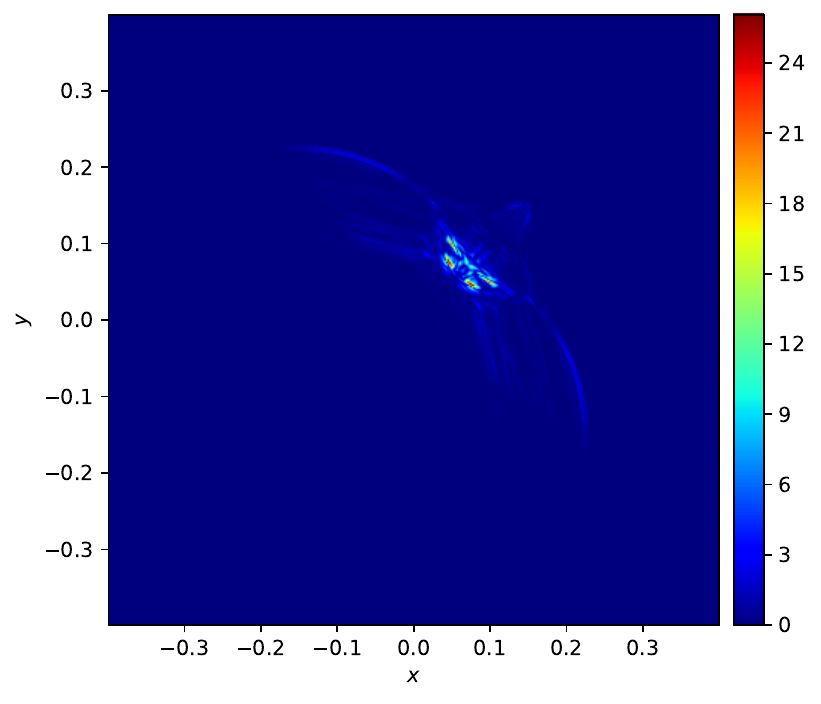}\label{fig:crp_db_wo_o2}}~
		\subfigure[$|(\nabla\cdot\B)_{i,j}|$ for $\othe$ scheme for CGL]{\includegraphics[width=0.26\textwidth]{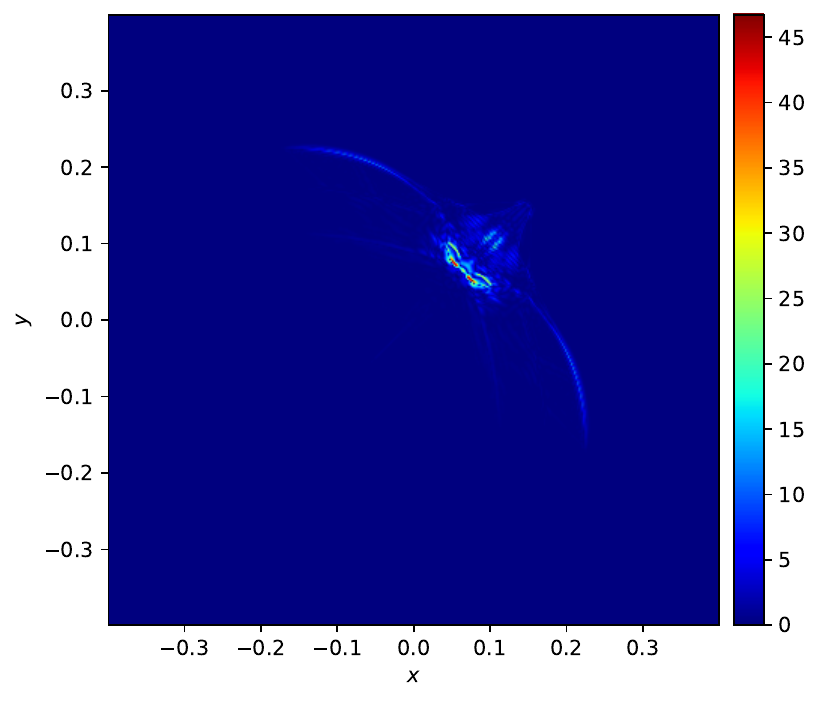}\label{fig:crp_db_wo_o3}}~
		\subfigure[$|(\nabla\cdot\B)_{i,j}|$ for $\ofe$ scheme for CGL]{\includegraphics[width=0.26\textwidth]{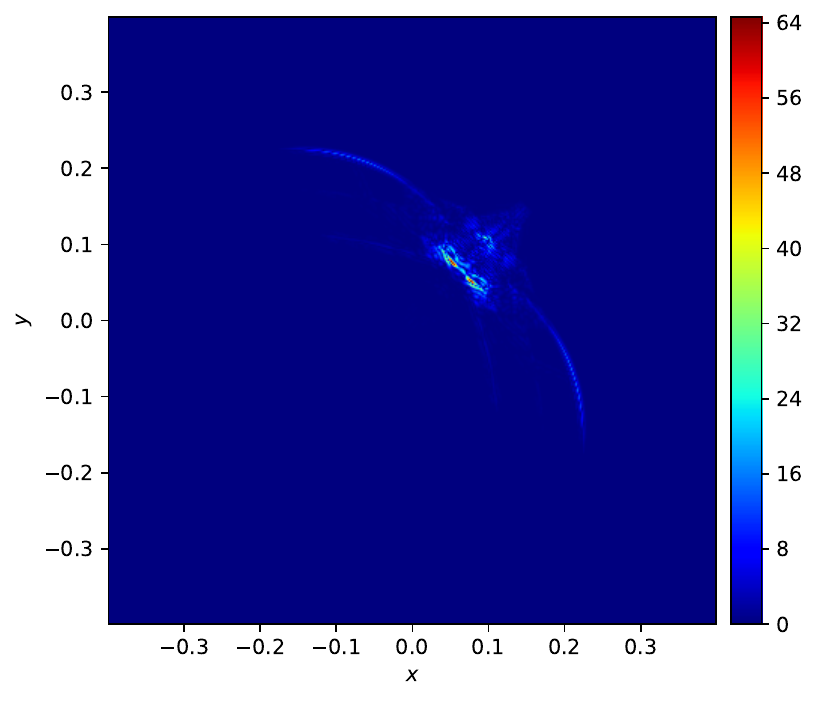}\label{fig:crp_db_wo_o4}}\\
		\subfigure[$|(\nabla\cdot\B)_{i,j}|$ for $\ote$ scheme for GLM-CGL]{\includegraphics[width=0.26\textwidth]{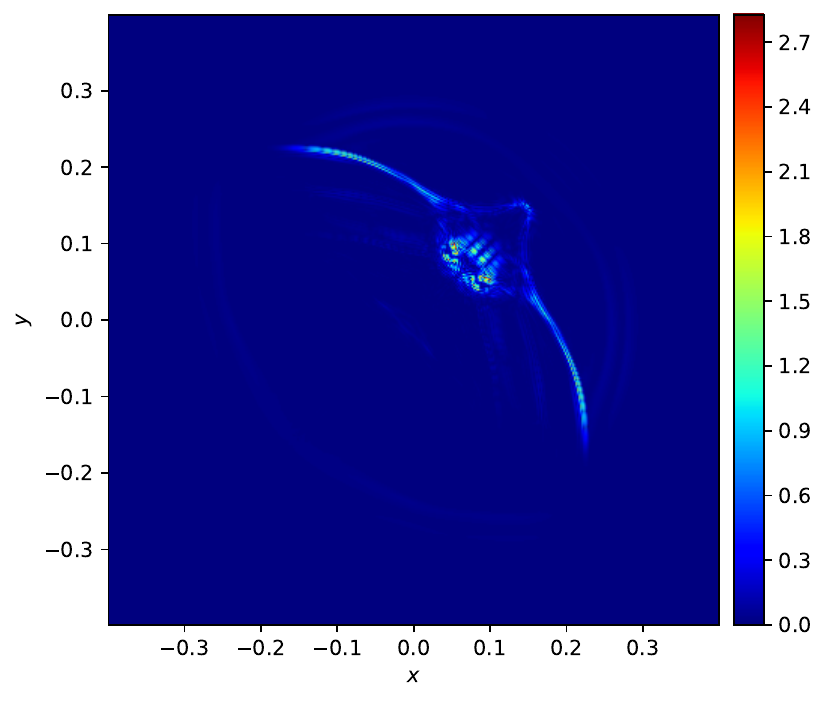}\label{fig:crp_db_w_o2}}~
		\subfigure[$|(\nabla\cdot\B)_{i,j}|$ for $\othe$ scheme for GLM-CGL]{\includegraphics[width=0.26\textwidth]{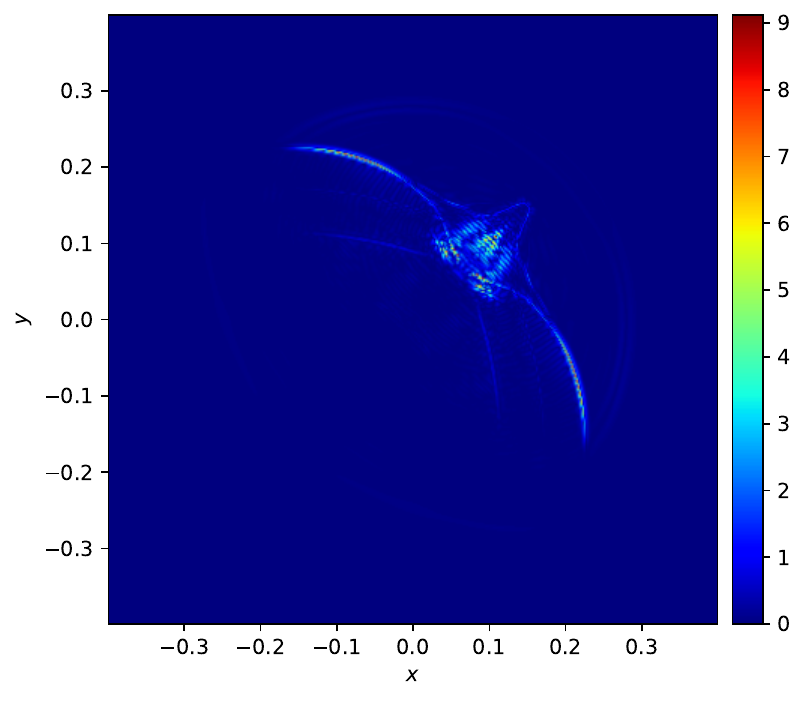}\label{fig:crp_db_w_o3}}~
		\subfigure[$|(\nabla\cdot\B)_{i,j}|$ for $\ofe$ scheme for GLM-CGL]{\includegraphics[width=0.26\textwidth]{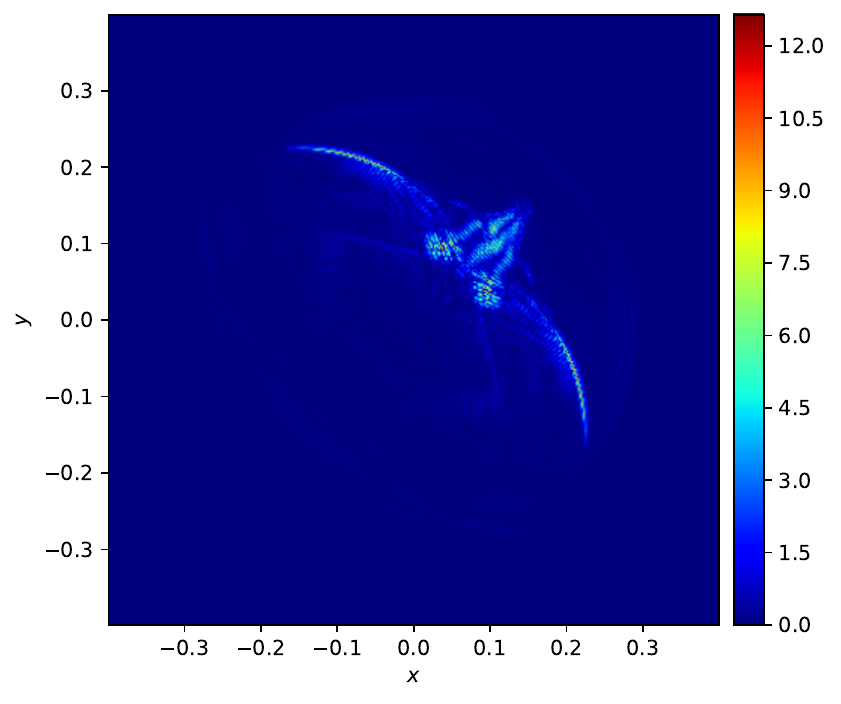}\label{fig:crp_db_w_o4}}\\
		\caption{\textbf{\nameref{test:CRP}}: Plots of $B_y$ and $|(\nabla\cdot\B)_{i,j}|$ for $\ote$, $\othe$ and $\ofe$ schemes for CGL and GLM-CGL at time $t=0.1$.}
		\label{fig:crp_cgl_by_divb}
	\end{center}
\end{figure}
\begin{figure}[!htbp]
	\begin{center}
		\subfigure[$B_y$ for $\oti$ scheme for isotropic CGL]{\includegraphics[width=0.26\textwidth]{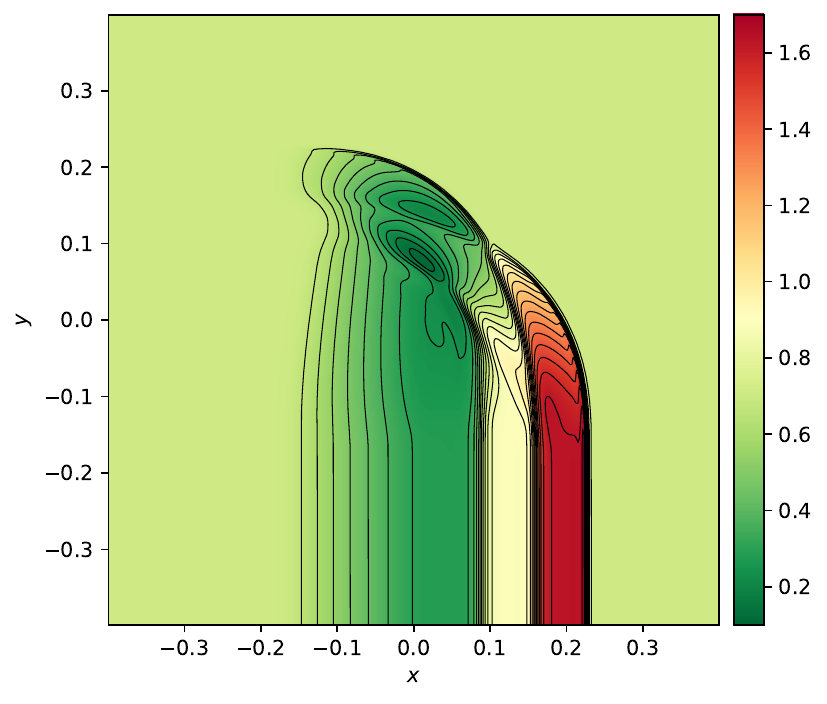}\label{fig:crp_by_wo_m_o2}}~
		\subfigure[$B_y$ for $\othi$ scheme for isotropic CGL]{\includegraphics[width=0.26\textwidth]{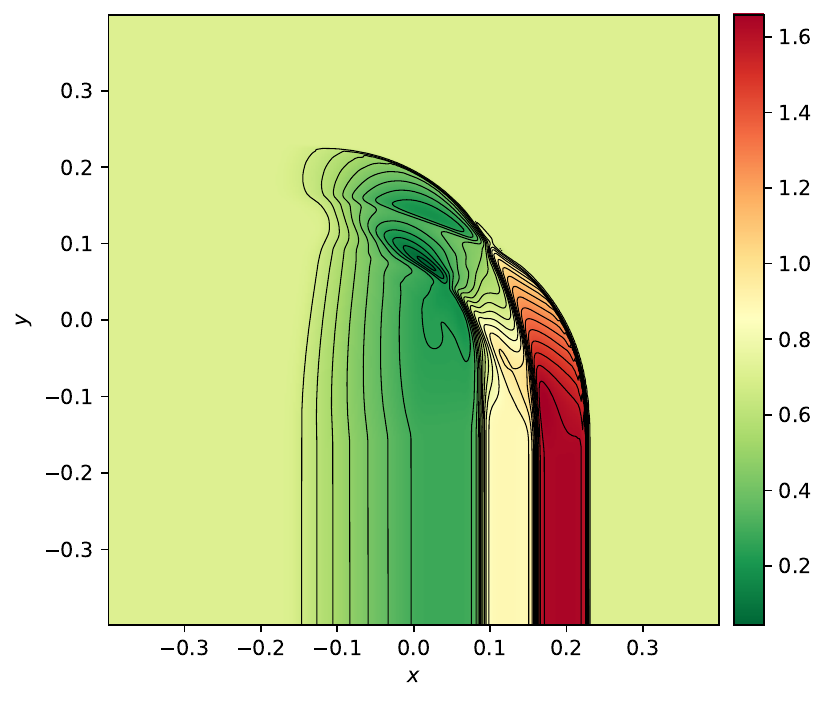}\label{fig:crp_by_wo_m_o3}}~
		\subfigure[$B_y$ for $\ofi$ scheme for isotropic CGL]{\includegraphics[width=0.26\textwidth]{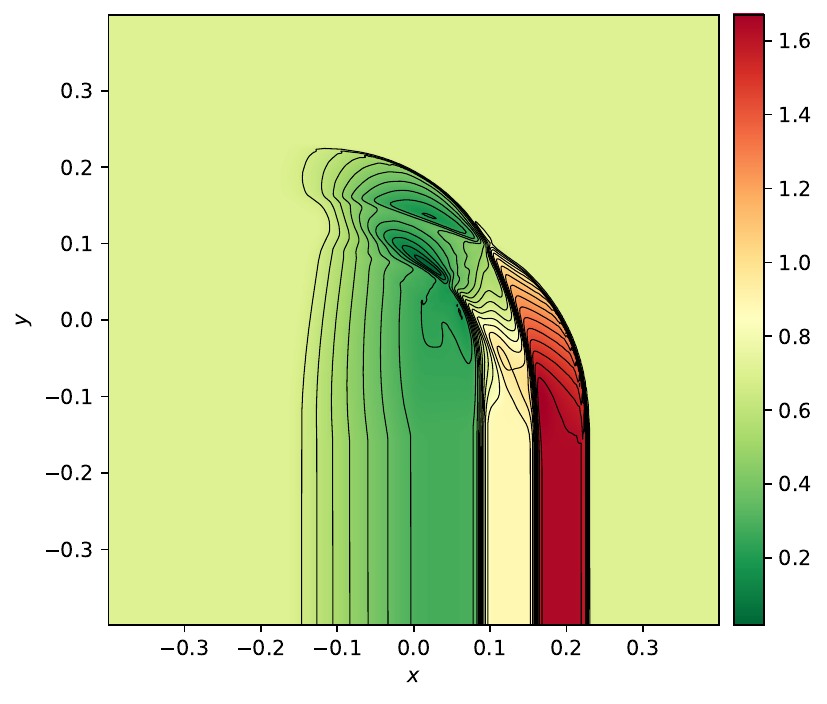}\label{fig:crp_by_wo_m_o4}}\\
		\subfigure[$B_y$ for $\oti$ scheme for isotropic GLM-CGL]{\includegraphics[width=0.26\textwidth]{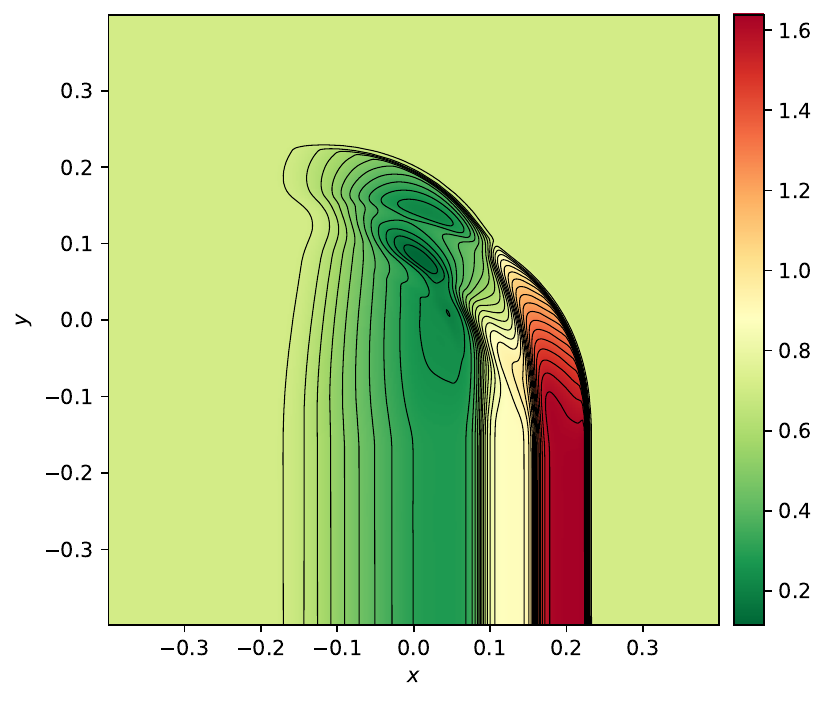}\label{fig:crp_by_w_m_o2}}~
		\subfigure[$B_y$ for $\othi$ scheme for isotropic GLM-CGL]{\includegraphics[width=0.26\textwidth]{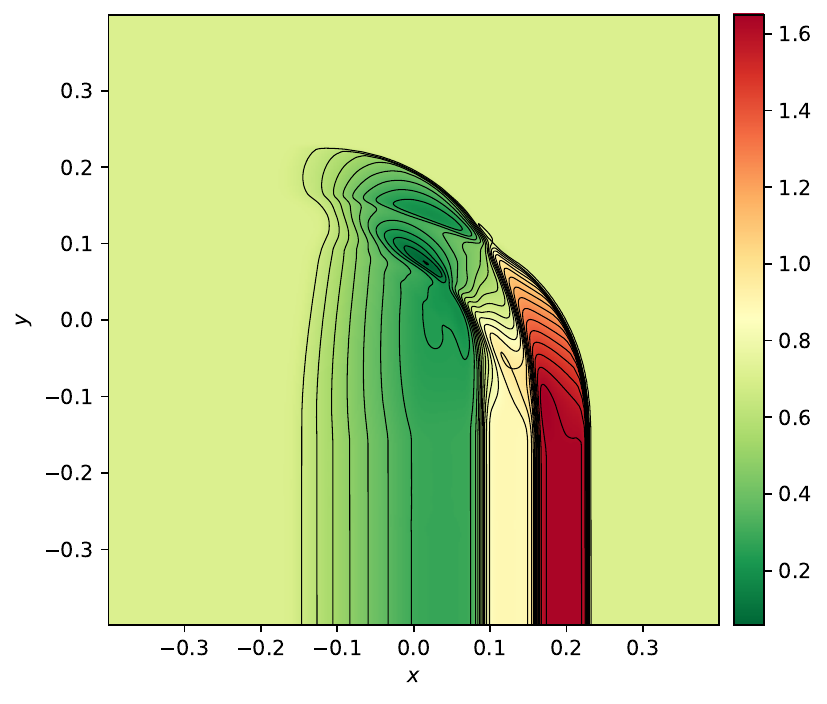}\label{fig:crp_by_w_m_o3}}~
		\subfigure[$B_y$ for $\ofi$ scheme for isotropic GLM-CGL]{\includegraphics[width=0.26\textwidth]{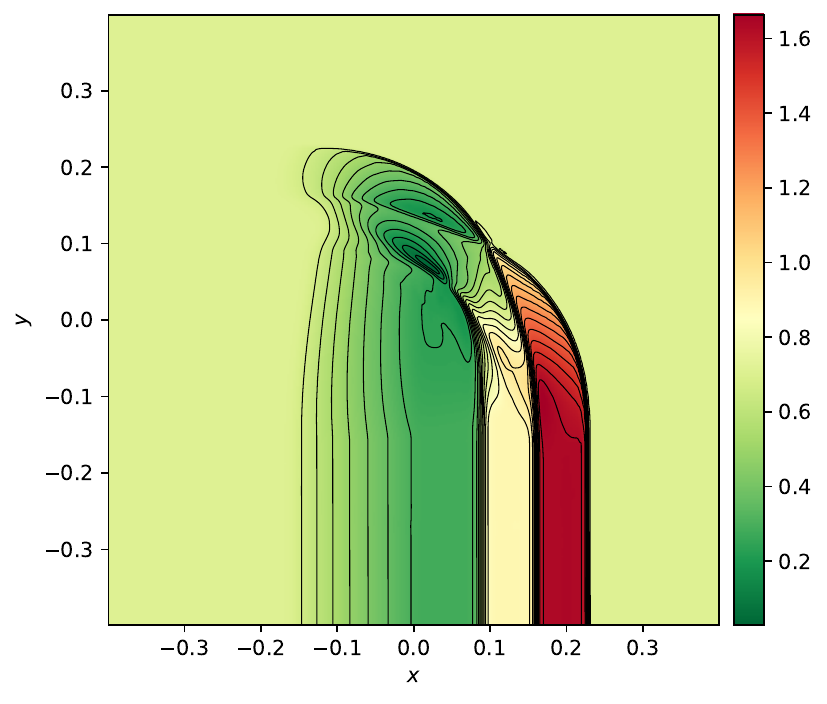}\label{fig:crp_by_w_m_o4}}\\
		\subfigure[$|(\nabla\cdot\B)_{i,j}|$ for $\oti$ scheme for isotropic CGL]{\includegraphics[width=0.26\textwidth]{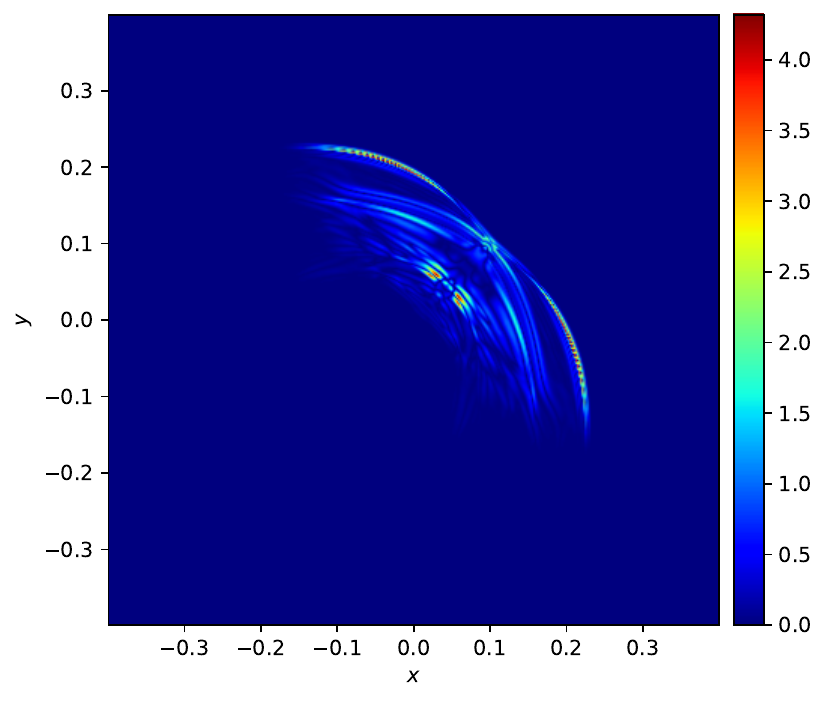}\label{fig:crp_db_wo_m_o2}}~
		\subfigure[$|(\nabla\cdot\B)_{i,j}|$ for $\othi$ scheme for isotropic CGL]{\includegraphics[width=0.26\textwidth]{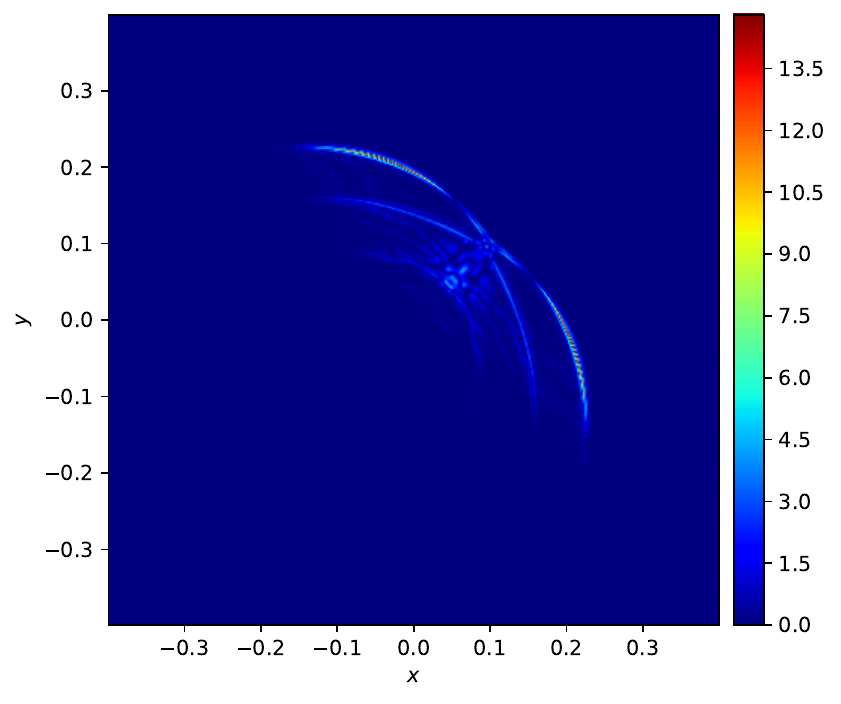}\label{fig:crp_db_wo_m_o3}}~
		\subfigure[$|(\nabla\cdot\B)_{i,j}|$ for $\ofi$ scheme for isotropic CGL]{\includegraphics[width=0.26\textwidth]{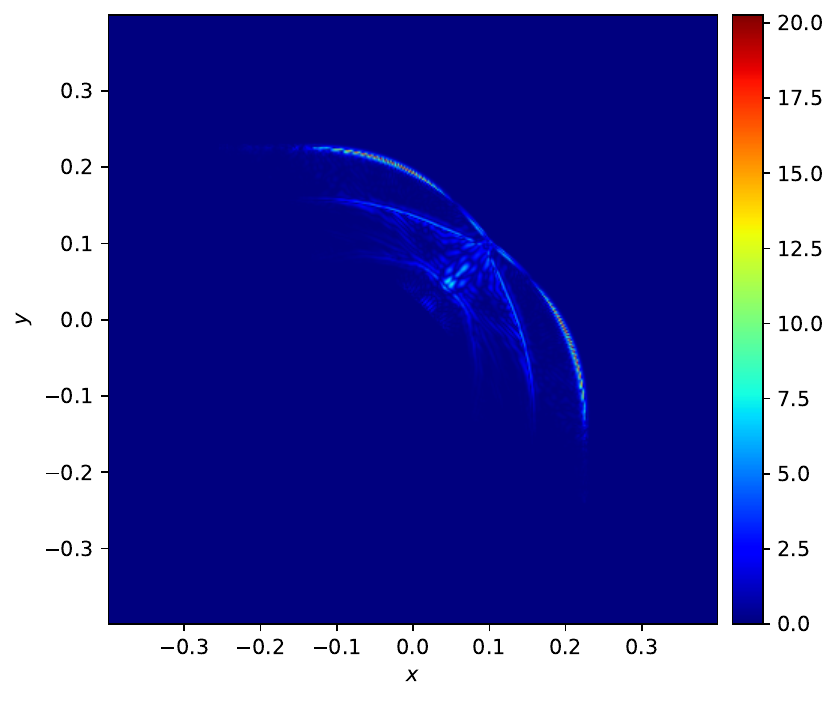}\label{fig:crp_db_wo_m_o4}}\\
		\subfigure[$|(\nabla\cdot\B)_{i,j}|$ for $\oti$ scheme for isotropic GLM-CGL]{\includegraphics[width=0.26\textwidth]{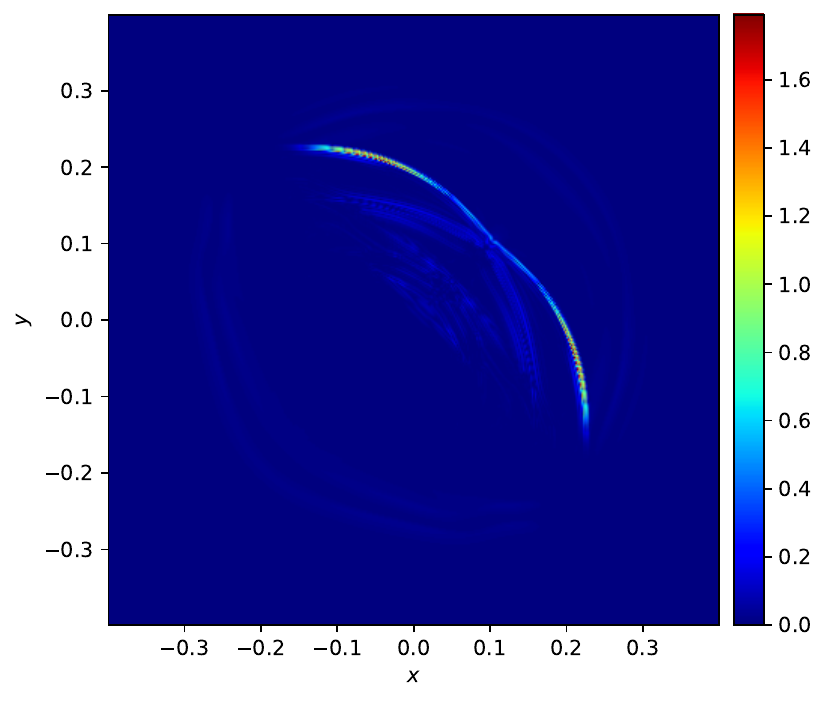}\label{fig:crp_db_w_m_o2}}~
		\subfigure[$|(\nabla\cdot\B)_{i,j}|$ for $\othi$ scheme for isotropic GLM-CGL]{\includegraphics[width=0.26\textwidth]{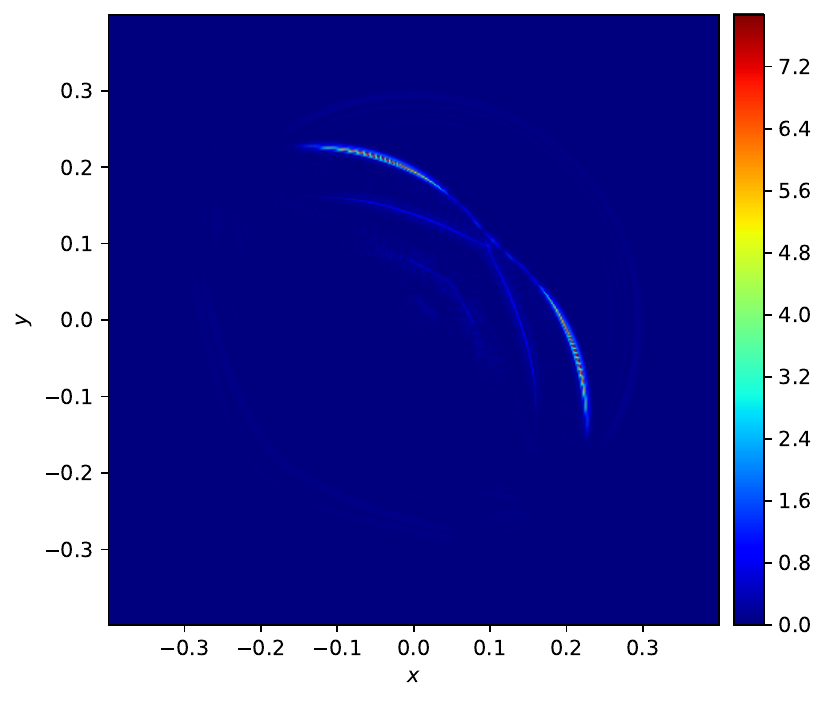}\label{fig:crp_db_w_m_o3}}~
		\subfigure[$|(\nabla\cdot\B)_{i,j}|$ for $\ofi$ scheme for isotropic GLM-CGL]{\includegraphics[width=0.26\textwidth]{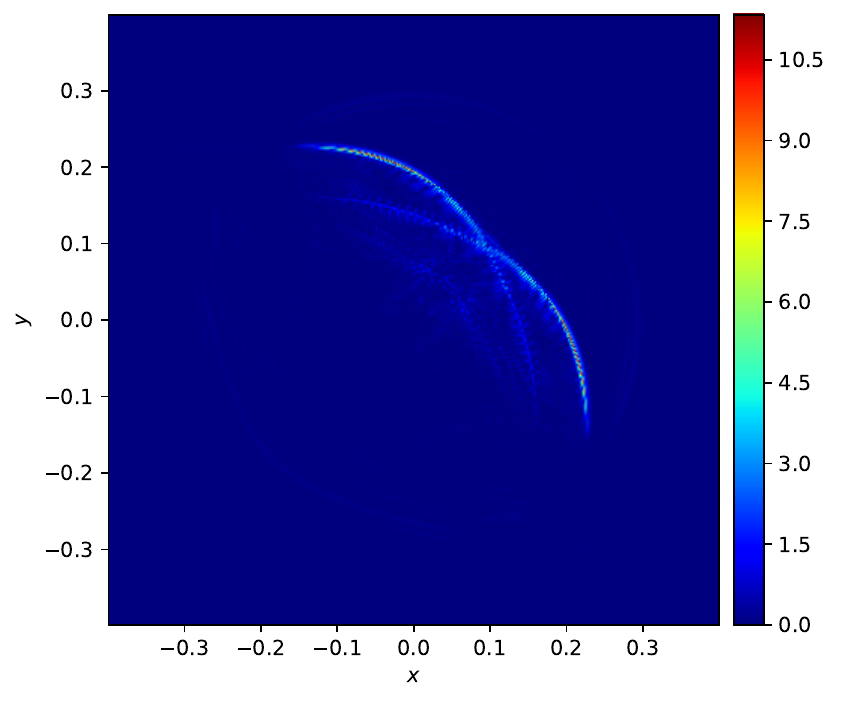}\label{fig:crp_db_w_m_o4}}\\
		\caption{\textbf{\nameref{test:CRP}}: Plots of $B_y$ and $|(\nabla\cdot\B)_{i,j}|$ for $\oti$, $\othi$ and $\ofi$ schemes for isotropic CGL and isotropic GLM-CGL at time $t=0.1$.}
		\label{fig:crp_mhd_by_divb}
	\end{center}
\end{figure}
\begin{figure}[!htbp]
	\begin{center}	
		\subfigure[$\|\nabla\cdot\B\|_{1}$ and $\|\nabla\cdot\B\|_{2}$ for $\ote$ scheme for CGL and GLM-CGL ]{\includegraphics[width=0.28\textwidth]{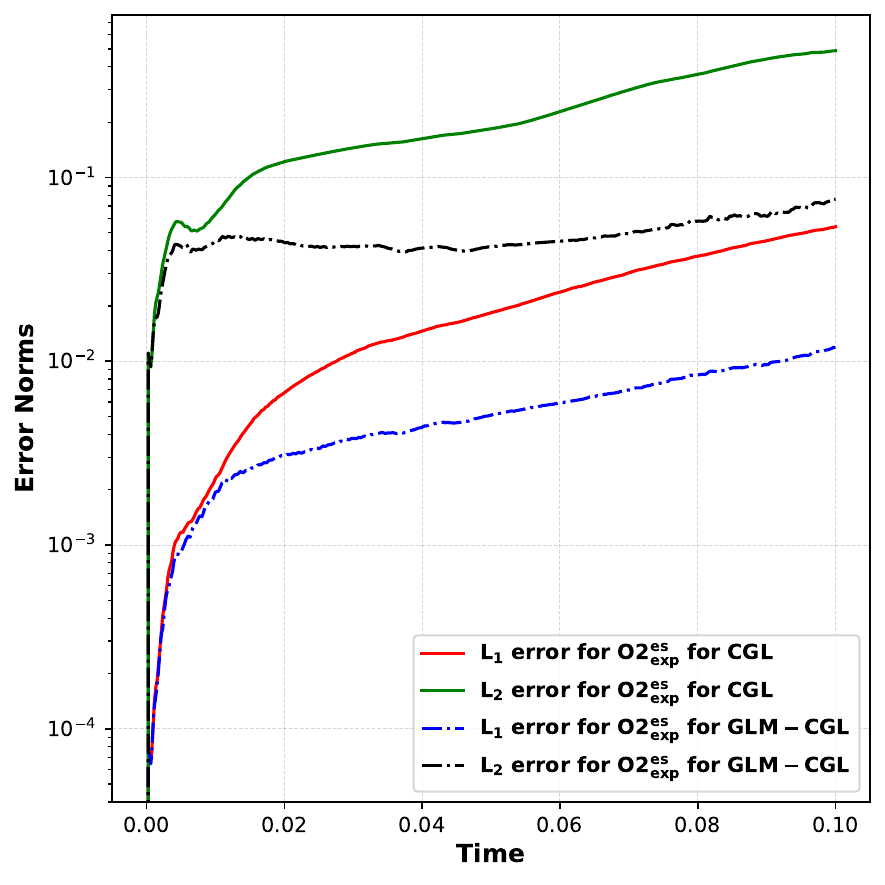}\label{fig:crp_error_cgl_o2}}~
		\subfigure[$\|\nabla\cdot\B\|_{1}$ and $\|\nabla\cdot\B\|_{2}$ for $\othe$ scheme for CGL and GLM-CGL ]{\includegraphics[width=0.28\textwidth]{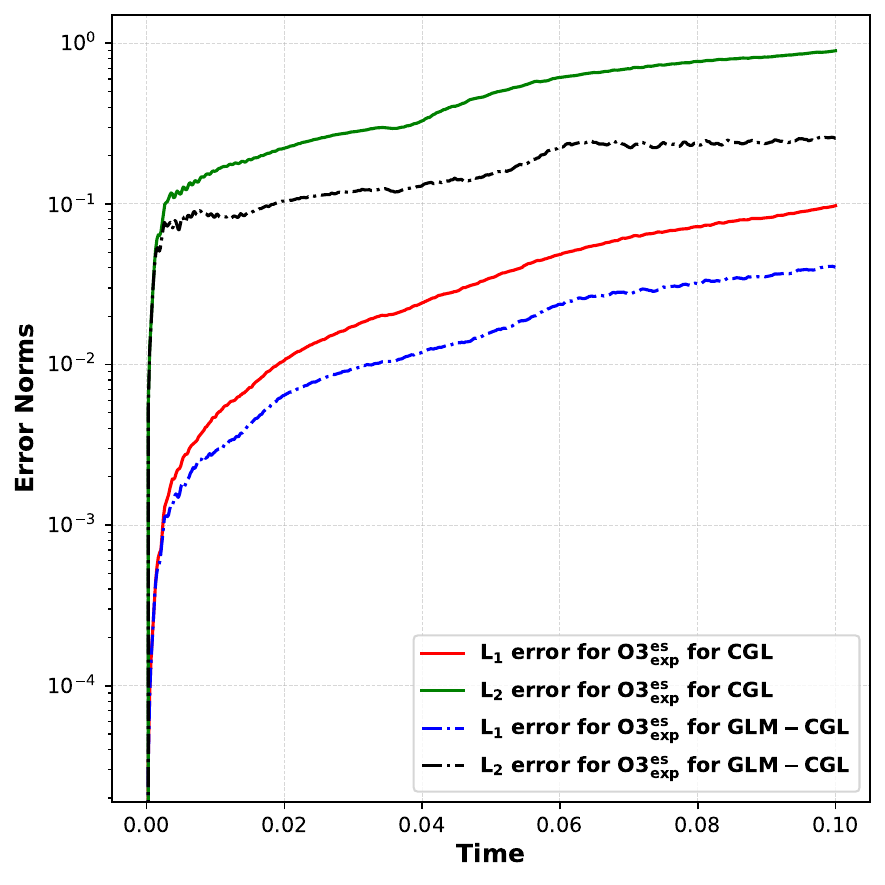}\label{fig:crp_error_cgl_o3}}~
		\subfigure[$\|\nabla\cdot\B\|_{1}$ and $\|\nabla\cdot\B\|_{2}$ for $\ofe$ scheme for CGL and GLM-CGL ]{\includegraphics[width=0.28\textwidth]{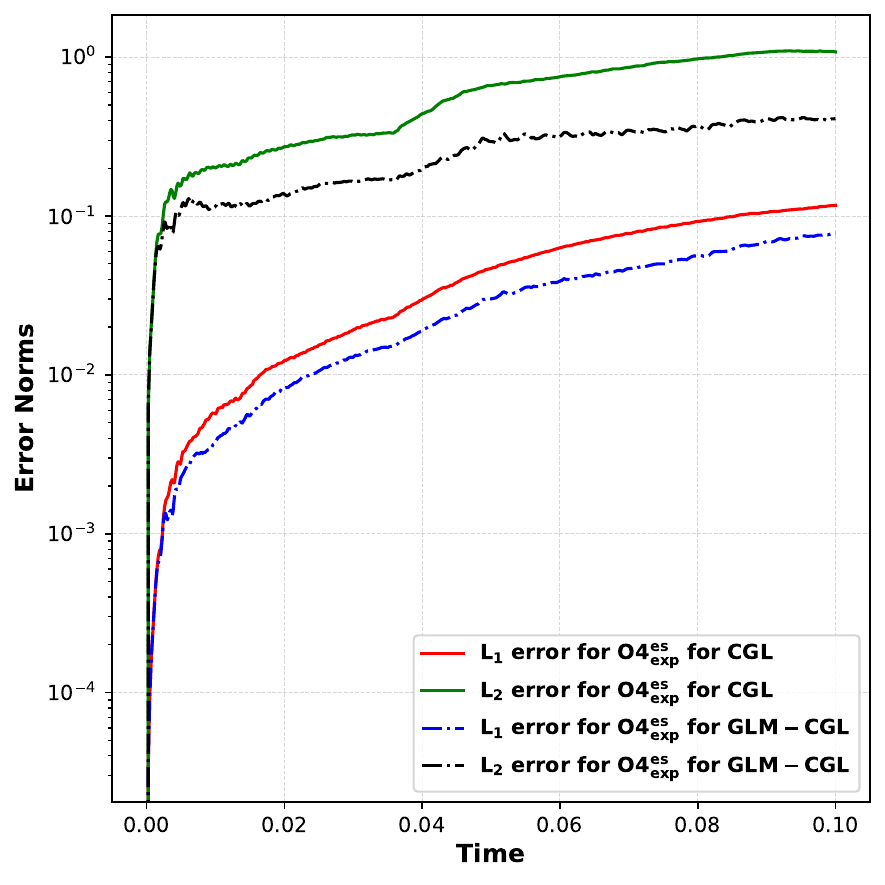}\label{fig:crp_error_cgl_o4}}\\
		\subfigure[$\|\nabla\cdot\B\|_{1}$ and $\|\nabla\cdot\B\|_{2}$ for $\oti$ scheme for isotropic CGL and isotropic GLM-CGL ]{\includegraphics[width=0.28\textwidth]{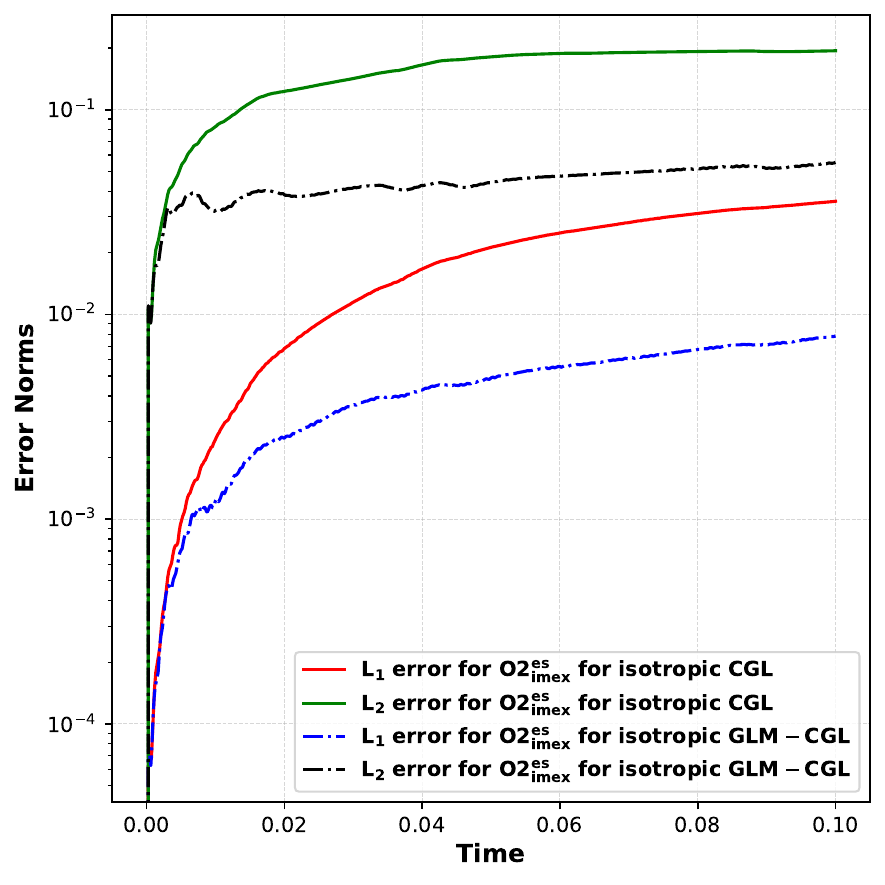}\label{fig:crp_error_mhd_o2}}~
		\subfigure[$\|\nabla\cdot\B\|_{1}$ and $\|\nabla\cdot\B\|_{2}$ for $\othi$ scheme for isotropic CGL and isotropic GLM-CGL ]{\includegraphics[width=0.28\textwidth]{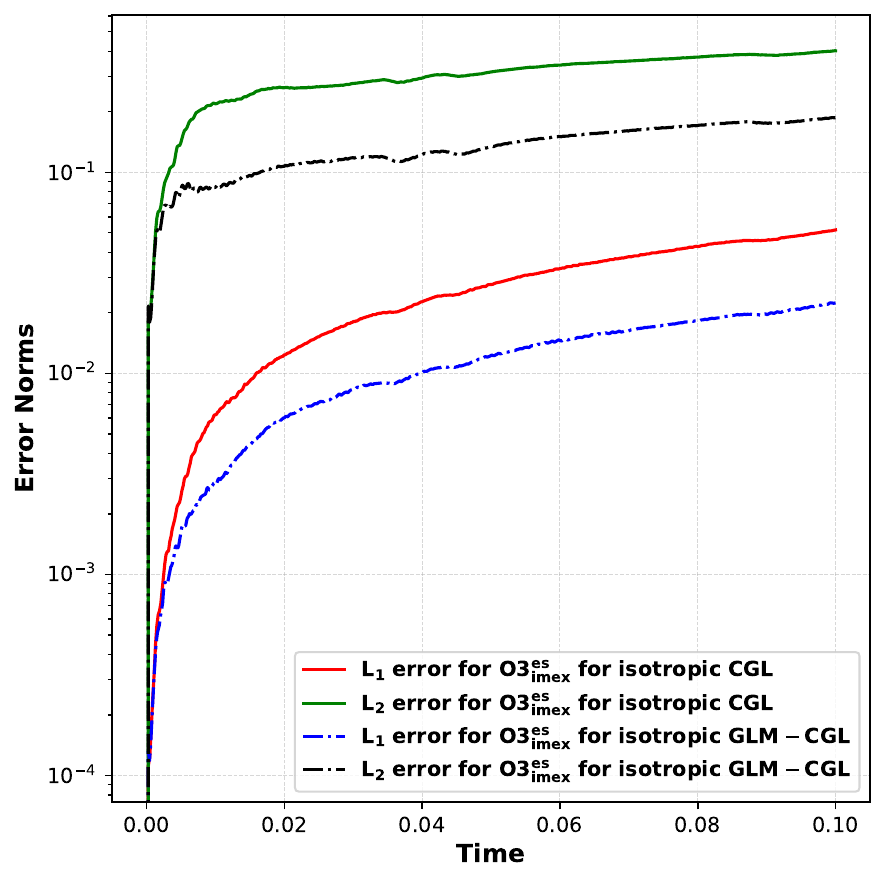}\label{fig:crp_error_mhd_o3}}~
		\subfigure[$\|\nabla\cdot\B\|_{1}$ and $\|\nabla\cdot\B\|_{2}$ for $\ofi$ scheme for isotropic CGL and isotropic GLM-CGL ]{\includegraphics[width=0.28\textwidth]{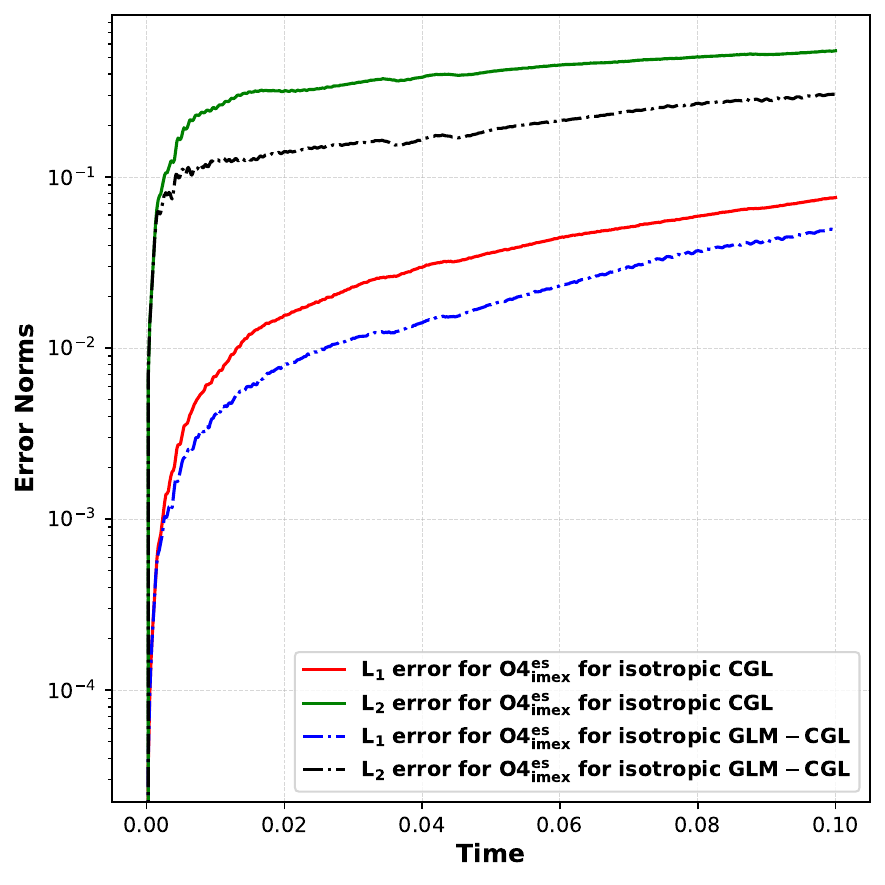}\label{fig:crp_error_mhd_o4}}\\
		\caption{\textbf{\nameref{test:CRP}}: Evolution of the magnetic field divergence constraint errors till time $t=0.1$.}
		\label{fig:crp_error}
	\end{center}
\end{figure}

\begin{table}[ht]
\centering
\renewcommand{\arraystretch}{1.4}
\begin{tabular}{|c|c|c|c|c|c|c|}
\hline
\multirow{2}{*}{Scheme} &
\multirow{2}{*}{System} &
$\rho_{\min}$ &
$(p_{\parallel})_{\min}$ &
$(p_{\perp})_{\min}$ &
$\|\nabla\!\cdot\!\mathbf{B}\|_{1}$ &
$\|\nabla\!\cdot\!\mathbf{B}\|_{2}$ \\
&&&&&&\\
\hline

\multirow{2}{*}{$\ote$}
& CGL
& 1.000000E+00 & 5.000000E-01 & 5.000000E-01 &  5.396165E-02 & 4.893650E-01 \\ \cline{2-7}

& GLM--CGL
& 9.999987E-01 & 4.999992E-01 & 4.999987E-01 &  1.190260E-02 & 7.621534E-02  \\  \hline

\multirow{2}{*}{$\othe$}
& CGL
& 9.965900E-01 & 4.965189E-01 & 4.963302E-01 &  9.750131E-02 & 8.984835E-01 \\ \cline{2-7}

& GLM--CGL
& 9.997391E-01 & 4.998337E-01 & 4.997380E-01 &  4.041956E-02 & 2.566261E-01  \\  \hline

\multirow{2}{*}{$\ofe$}
& CGL
& 9.952653E-01 & 4.953468E-01 & 4.942906E-01 &  1.167868E-01 & 1.079055E+00 \\ \cline{2-7}

& GLM--CGL
& 9.990223E-01 & 4.993633E-01 & 4.989795E-01 &  7.788214E-02 & 4.096734E-01  \\  \hline

\multirow{2}{*}{$\oti$}
& Isotropic CGL
& 1.000000E+00 & 5.000000E-01 & 5.000000E-01 &  3.570736E-02 & 1.939253E-01 \\ \cline{2-7}

& Isotropic GLM--CGL
& 9.999929E-01 & 4.999941E-01 & 4.999941E-01 &  7.848564E-03 & 5.530822E-02  \\  \hline

\multirow{2}{*}{$\othi$}
& Isotropic CGL
&  9.940200E-01& 4.941600E-01 & 4.940666E-01 &  5.170539E-02 & 4.023279E-01  \\ \cline{2-7}

& Isotropic GLM--CGL
& 9.995013E-01 & 4.995770E-01 & 4.995738E-01 &  2.241917E-02 & 1.868455E-01  \\  \hline

\multirow{2}{*}{$\ofi$}
& Isotropic CGL
& 9.902849E-01 & 4.904222E-01 & 4.903417E-01 &  7.598519E-02 & 5.475922E-01  \\ \cline{2-7}

& Isotropic GLM--CGL
& 9.981122E-01 & 4.983336E-01 & 4.983153E-01 &  4.976239E-02 & 3.056925E-01  \\ \hline

\end{tabular}
\caption{\textbf{\nameref{test:CRP}}: Minimum values of density $\rho$, parallel pressure $p_{\parallel}$, and perpendicular pressure $p_{\perp}$, together with the value of $L_1$ and $L_2$ norm of $\nabla\cdot\B$ at final time.}
\label{tab:div_CRP}
\end{table}
\subsection{Two-Dimensional Riemann problem}
\label{test:2DRP}
In the last test case, we consider a two-dimensional Riemann problem with four states, motivated by the similar MHD test case in ~\cite{dai1998simple,artebrant2008increasing}. We consider a rectangular computational domain given by $[-1.5,1.5]\times[-1.5,1.5]$. We enforce Neumann boundary conditions. The initial condition consists of four states, which are given as follows:
\[\left(\rho, \bu, \pll, \per\right),  = \begin{cases}
	(1, 0.75, -0.5, 0, 1, 1), & \textrm{if } x>0,~y>0\\
	(2, 0.75, 0.5, 0, 1, 1), & \textrm{if } x<0,~y>0\\
	(1, -0.75, 0.5, 0, 1, 1), & \textrm{if } x<0,~y<0\\
	(3, -0.75, -0.5, 0, 1, 1), & \textrm{if } x>0,~y<0\\
\end{cases}\]
\[\left(B_{x}, B_{y}, B_{z},\Psi\right) = 
\left(\frac{2}{\sqrt{4\pi}},0,\frac{1}{\sqrt{4\pi}},0\right).\]
The solutions are computed using $400\times400$ cells till the final time of $t=1$. In Fig.\eqref{fig:2DRP_cgl_by_divb}, we have plotted the schlieren images of $B_y$ and $|(\nabla\cdot\B)_{i,j}|$ for CGL and GLM-CGL using $\ote$, $\othe$ and $\ofe$ schemes, and in Fig.\eqref{fig:2DRP_mhd_by_divb}, we have plotted the schlieren image of $B_y$ and $|(\nabla\cdot\B)_{i,j}|$ for isotropic CGL and isotropic GLM-CGL using $\oti$, $\othi$ and $\ofi$ schemes. We observe that the results for isotropic CGL are similar to those presented in~\cite{artebrant2008increasing}. Furthermore, higher-order schemes produce much sharper solutions than the second-order schemes.

We also compare the $|(\nabla\cdot\B)_{i,j}|$ for CGL and GLM-CGL models. In each case, we observe that GLM-CGL produces much lower values of the divergence errors. \revo{The colorbars in all divergence plots are dynamically scaled using the absolute local minimum and maximum values of $\nabla\cdot\B$ in each plot.} This can also be observed when we compare the time evolution of the $L_1$ and $L_2$ norms divergence of the magnetic fields. We consistently observe that the GLM-CGL produces much lower values than the CGL model in all cases.

\revg{Table~\eqref{tab:div_2DRP} lists the minimum values of density, $\pll$, and $\per$ as well as the value of $L_1$ and $L_2$ norm of $\nabla\cdot\B$ at the final time. All reported values of density and pressure remain positive for all schemes and formulations, indicating that physically admissible solutions are maintained throughout the computation. The table also demonstrates the significant reduction in divergence errors achieved by the GLM formulation.}
\begin{figure}[!htbp]
	\begin{center}
		\subfigure[Schlieren image of $B_y$ for $\ote$ scheme for CGL]{\includegraphics[width=0.26\textwidth]{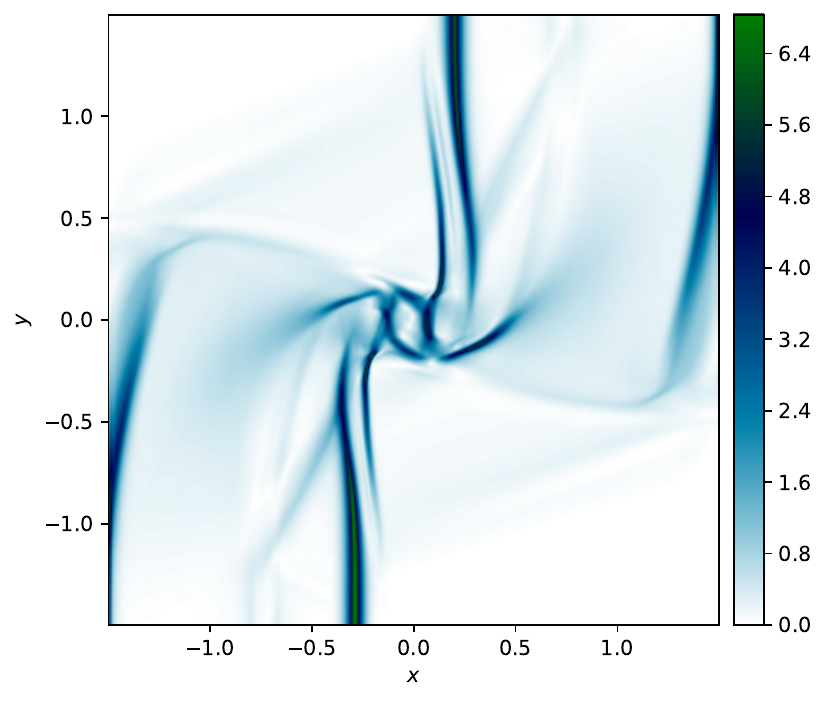}\label{fig:2DRP_by_wo_o2}}~
		\subfigure[Schlieren image of $B_y$ for $\othe$ scheme for CGL]{\includegraphics[width=0.26\textwidth]{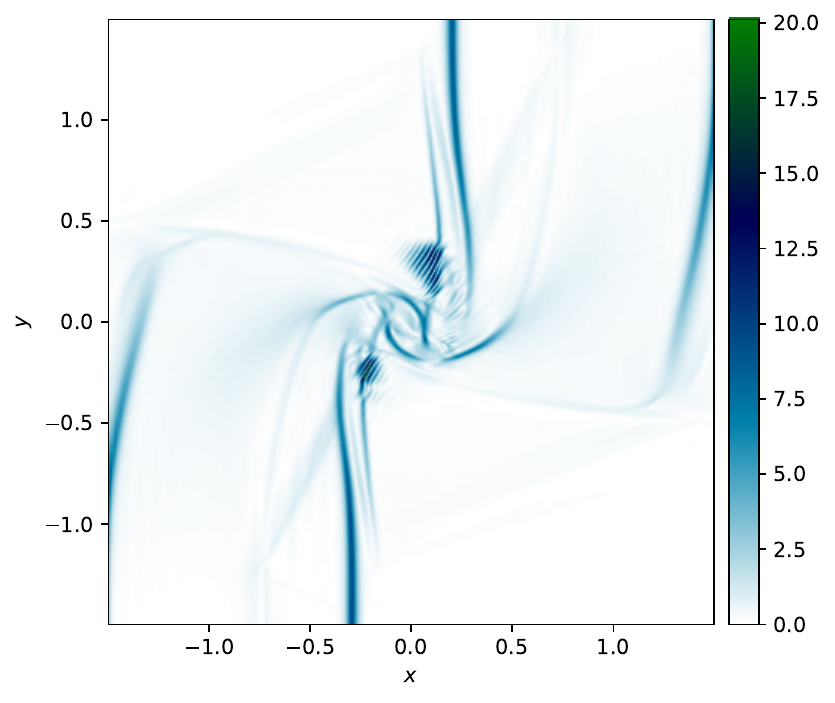}\label{fig:2DRP_by_wo_o3}}~
		\subfigure[Schlieren image of $B_y$ for $\ofe$ scheme for CGL]{\includegraphics[width=0.26\textwidth]{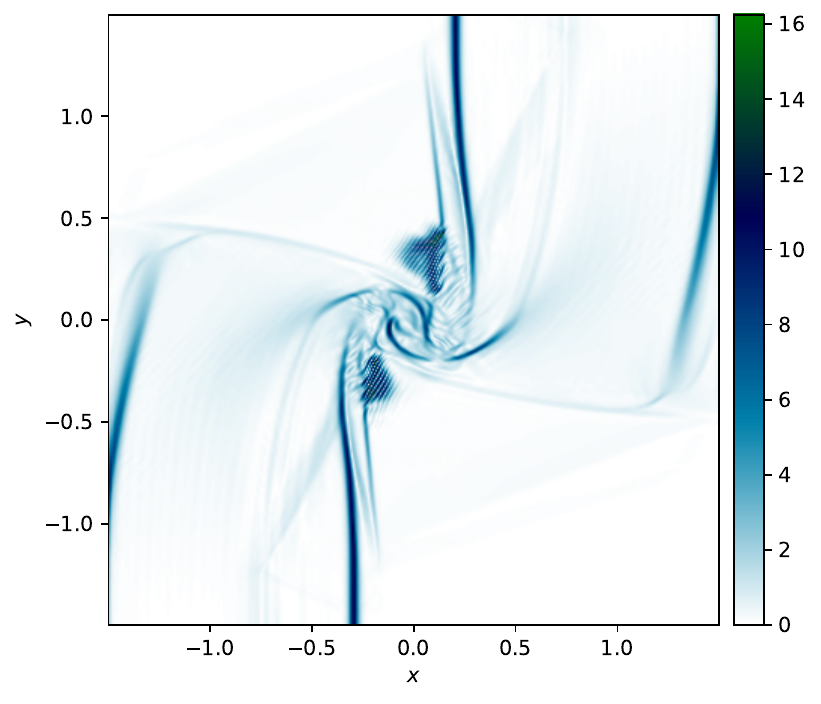}\label{fig:2DRP_by_wo_o4}}\\
		\subfigure[Schlieren image of $B_y$ for $\ote$ scheme for GLM-CGL]{\includegraphics[width=0.26\textwidth]{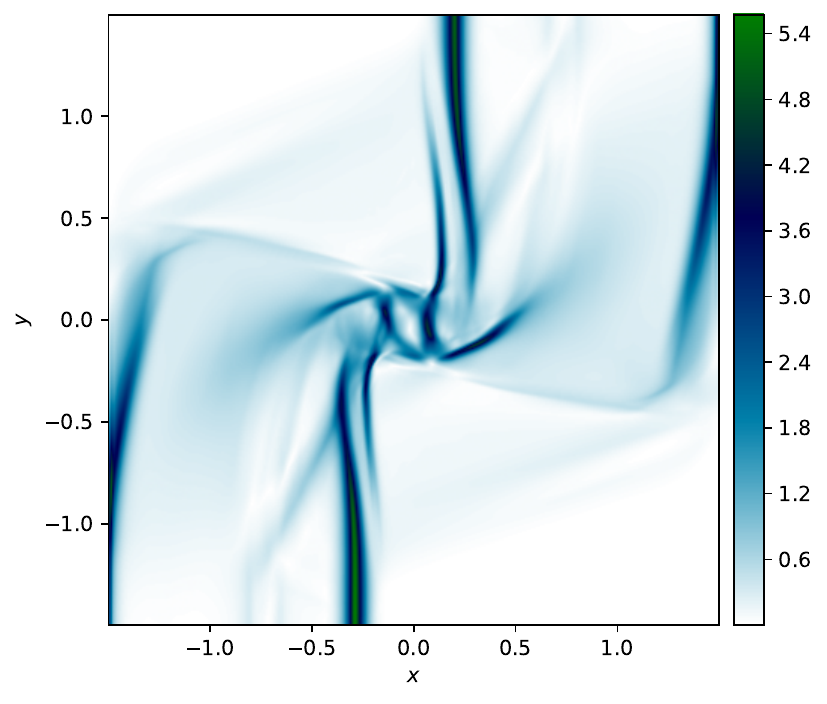}\label{fig:2DRP_by_w_o2}}~
		\subfigure[Schlieren image of $B_y$ for $\othe$ scheme for GLM-CGL]{\includegraphics[width=0.26\textwidth]{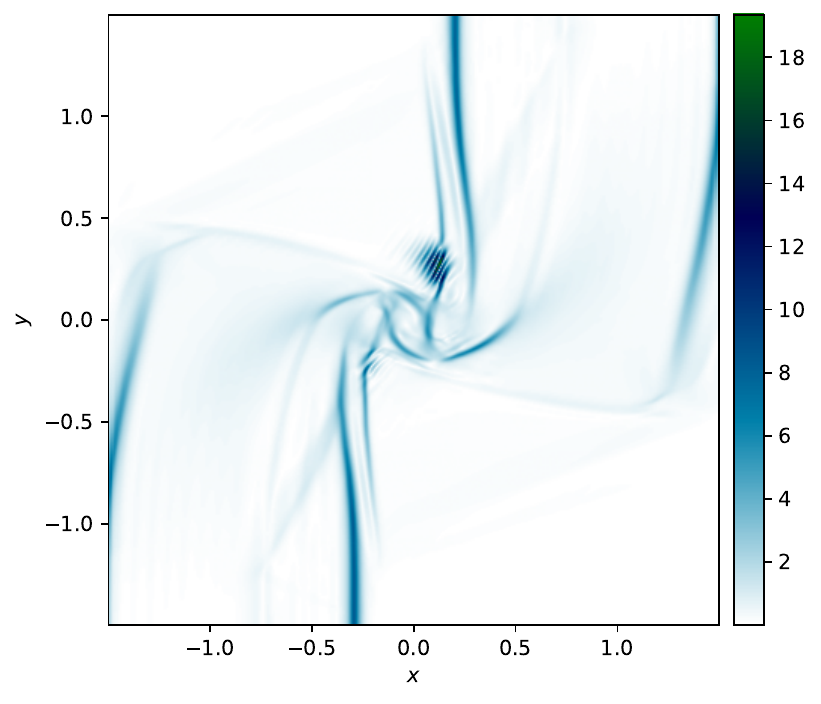}\label{fig:2DRP_by_w_o3}}~
		\subfigure[Schlieren image of $B_y$ for $\ofe$ scheme for GLM-CGL]{\includegraphics[width=0.26\textwidth]{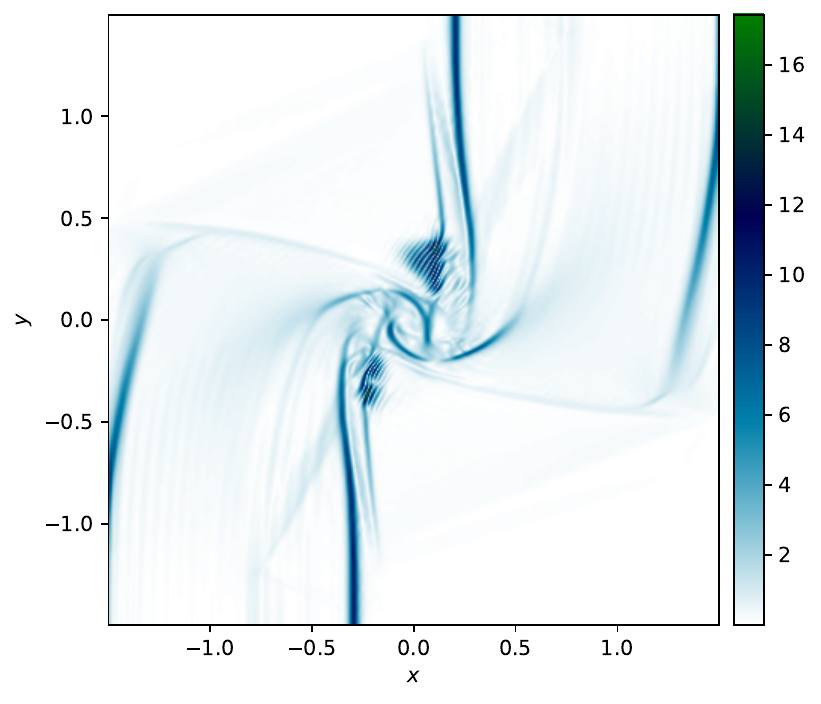}\label{fig:2DRP_by_w_o4}}\\
		\subfigure[$|(\nabla\cdot\B)_{i,j}|$ for $\ote$ scheme for CGL]{\includegraphics[width=0.26\textwidth]{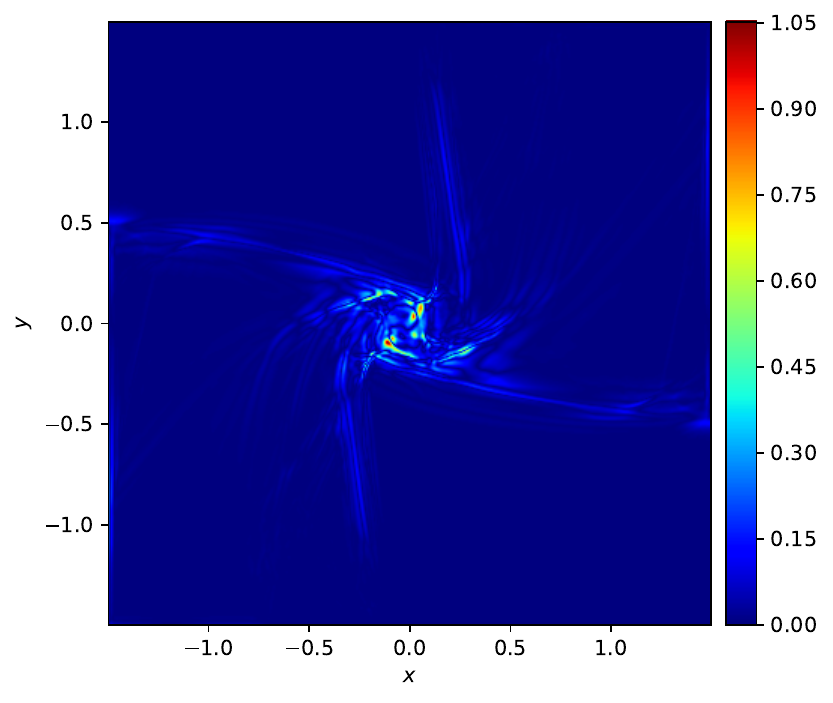}\label{fig:2DRP_db_wo_o2}}~
		\subfigure[$|(\nabla\cdot\B)_{i,j}|$ for $\othe$ scheme for CGL]{\includegraphics[width=0.26\textwidth]{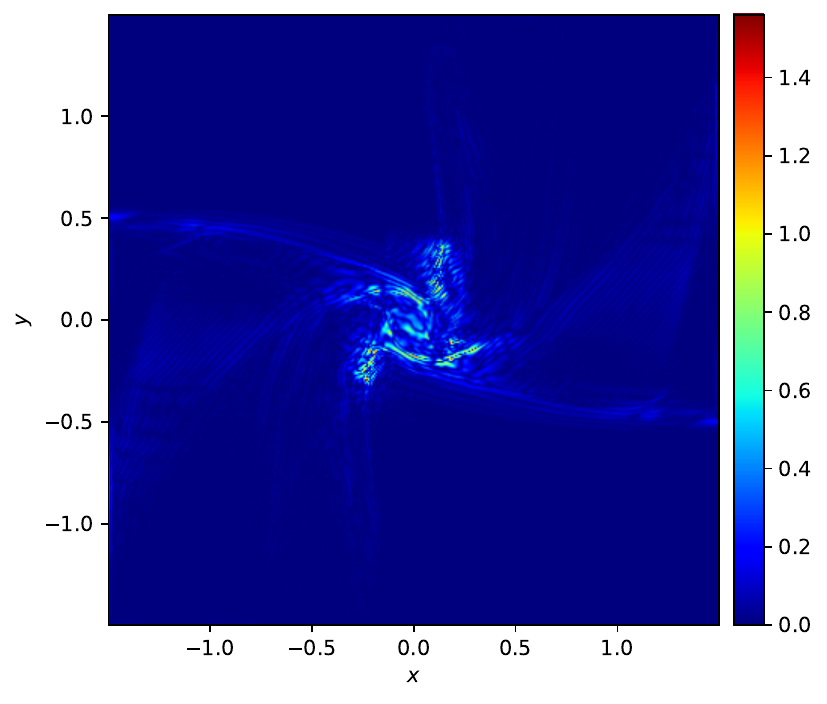}\label{fig:2DRP_db_wo_o3}}~
		\subfigure[$|(\nabla\cdot\B)_{i,j}|$ for $\ofe$ scheme for CGL]{\includegraphics[width=0.26\textwidth]{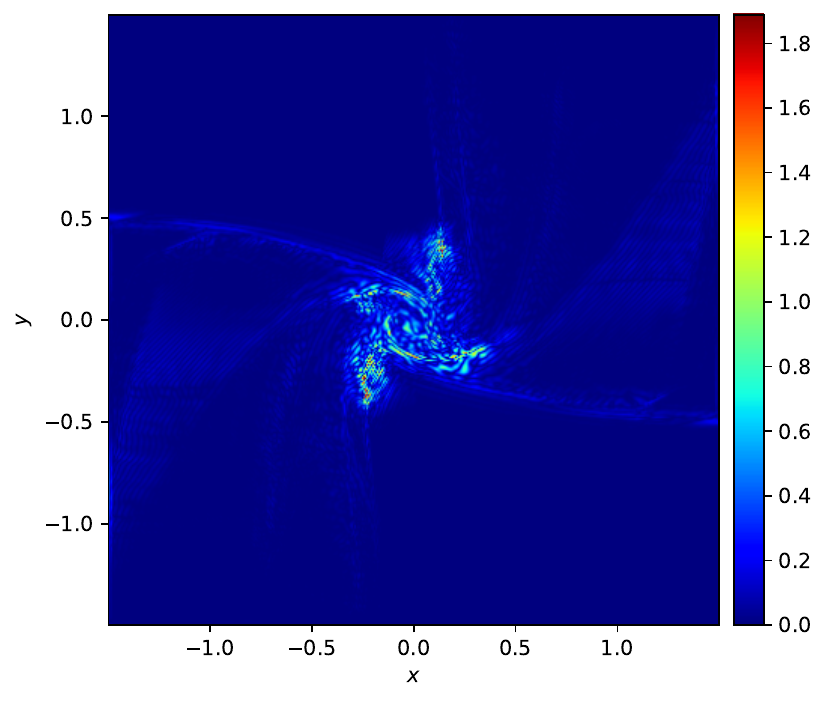}\label{fig:2DRP_db_wo_o4}}\\
		\subfigure[$|(\nabla\cdot\B)_{i,j}|$ for $\ote$ scheme for GLM-CGL]{\includegraphics[width=0.26\textwidth]{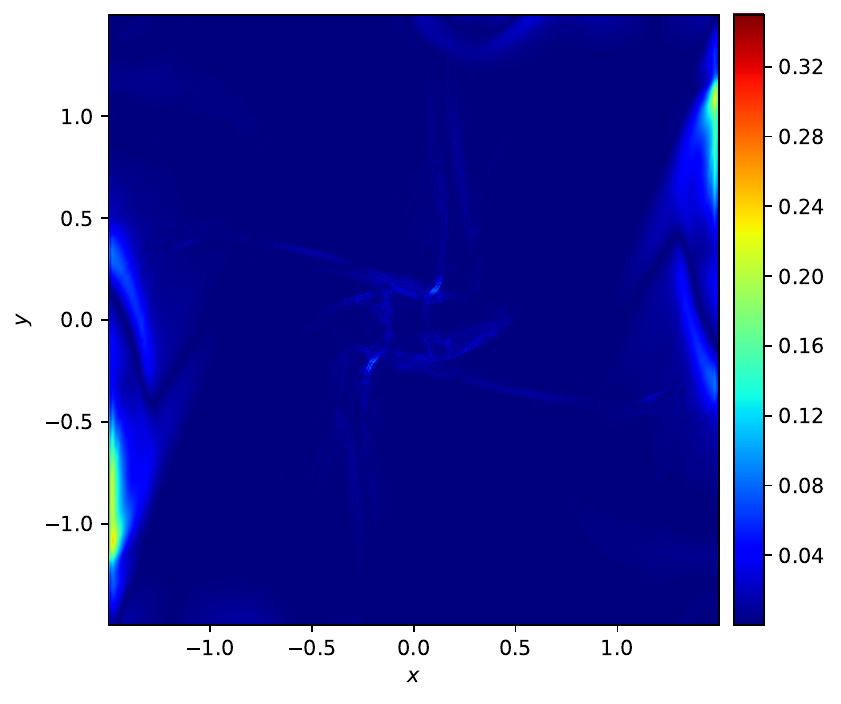}\label{fig:2DRP_db_w_o2}}~
		\subfigure[$|(\nabla\cdot\B)_{i,j}|$ for $\othe$ scheme for GLM-CGL]{\includegraphics[width=0.26\textwidth]{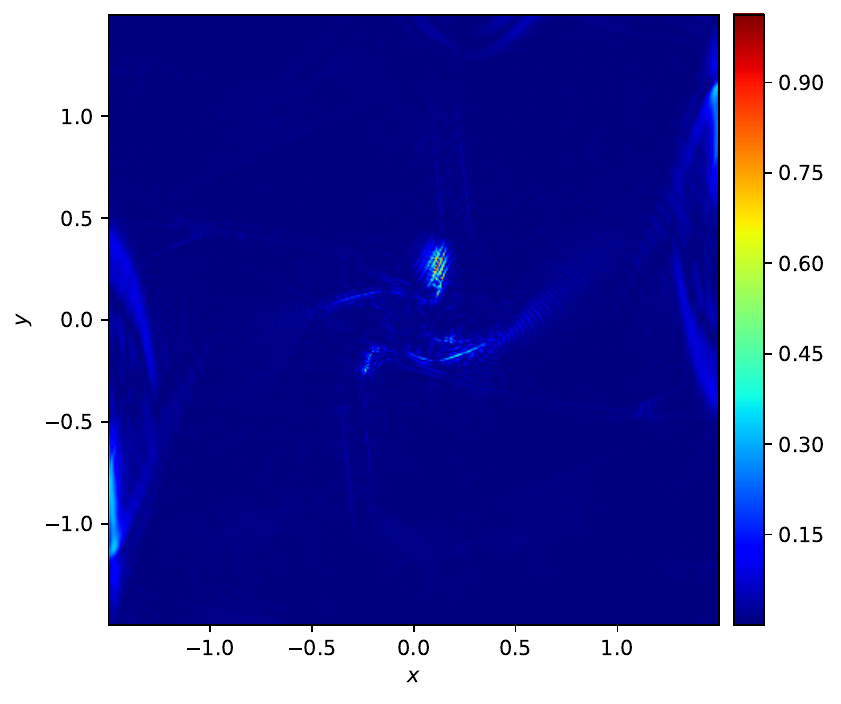}\label{fig:2DRP_db_w_o3}}~
		\subfigure[$|(\nabla\cdot\B)_{i,j}|$ for $\ofe$ scheme for GLM-CGL]{\includegraphics[width=0.26\textwidth]{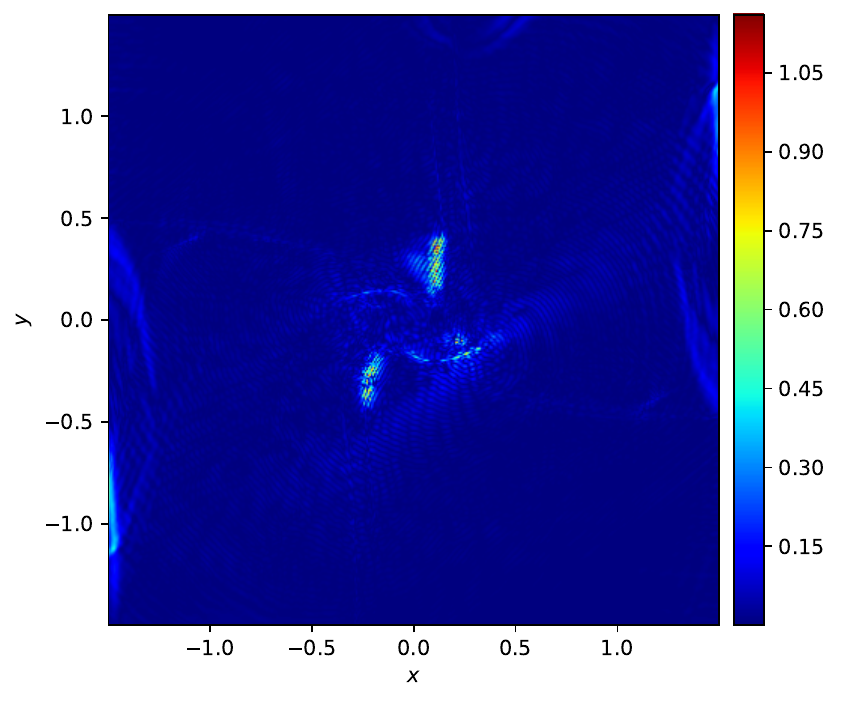}\label{fig:2DRP_db_w_o4}}\\
		\caption{\textbf{\nameref{test:2DRP}}: Schlieren image of $B_y$ and plots of $|(\nabla\cdot\B)_{i,j}|$ for $\ote$, $\othe$ and $\ofe$ schemes for CGL and GLM-CGL at time $t=1.0$.}
		\label{fig:2DRP_cgl_by_divb}
	\end{center}
\end{figure}
\begin{figure}[!htbp]
	\begin{center}
		\subfigure[Schlieren image of $B_y$ for $\oti$ scheme for isotropic CGL]{\includegraphics[width=0.26\textwidth]{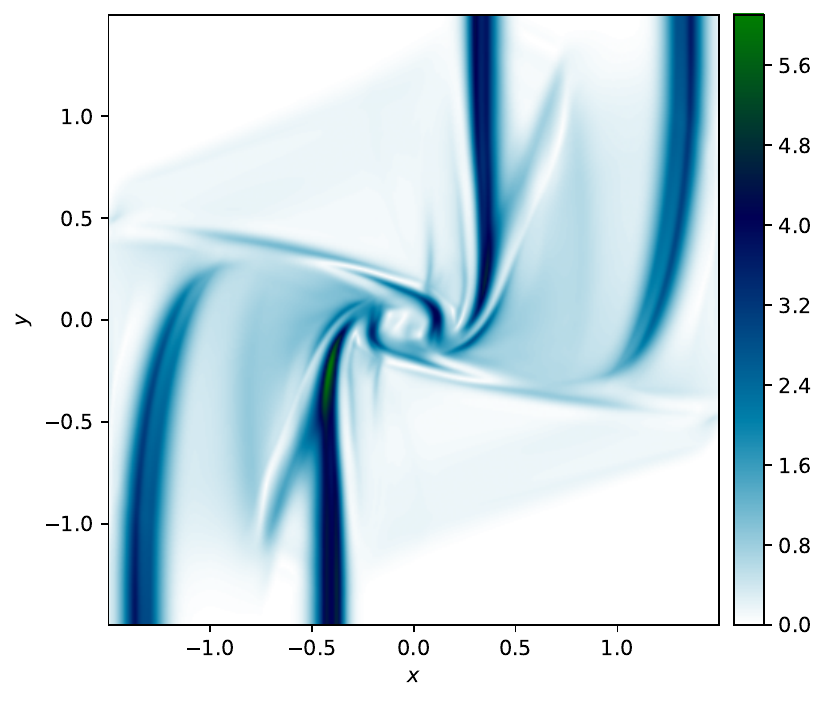}\label{fig:2DRP_by_wo_m_o2}}~
		\subfigure[Schlieren image of $B_y$ for $\othi$ scheme for isotropic CGL]{\includegraphics[width=0.26\textwidth]{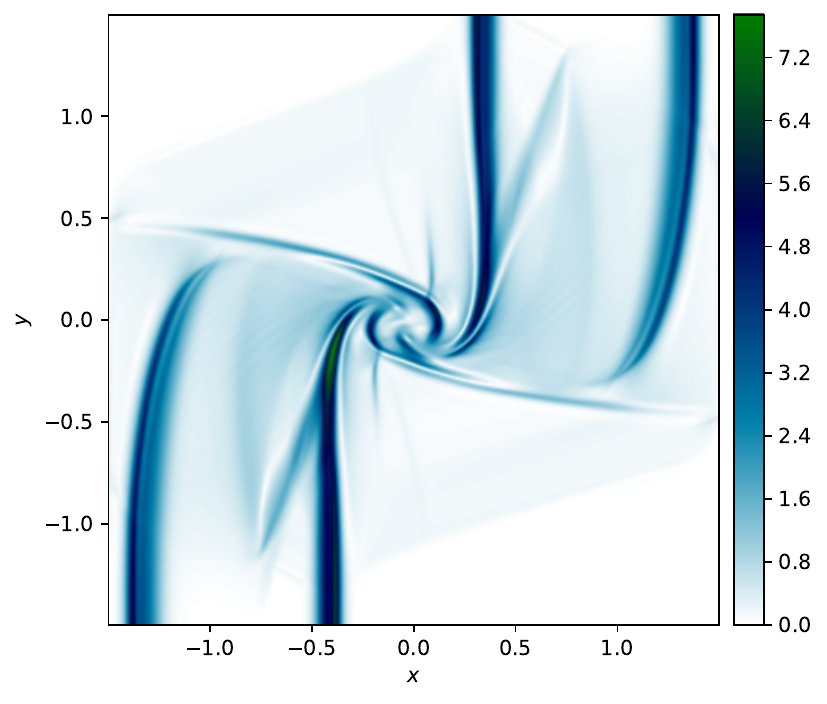}\label{fig:2DRP_by_wo_m_o3}}~
		\subfigure[Schlieren image of $B_y$ for $\ofi$ scheme for isotropic CGL]{\includegraphics[width=0.26\textwidth]{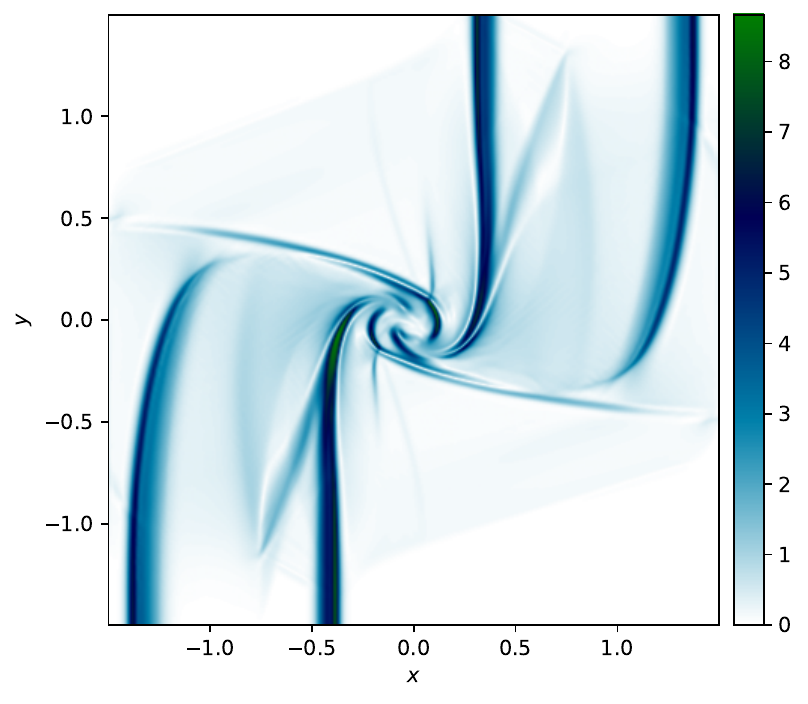}\label{fig:2DRP_by_wo_m_o4}}\\
		\subfigure[Schlieren image of $B_y$ for $\oti$ scheme for isotropic GLM-CGL]{\includegraphics[width=0.26\textwidth]{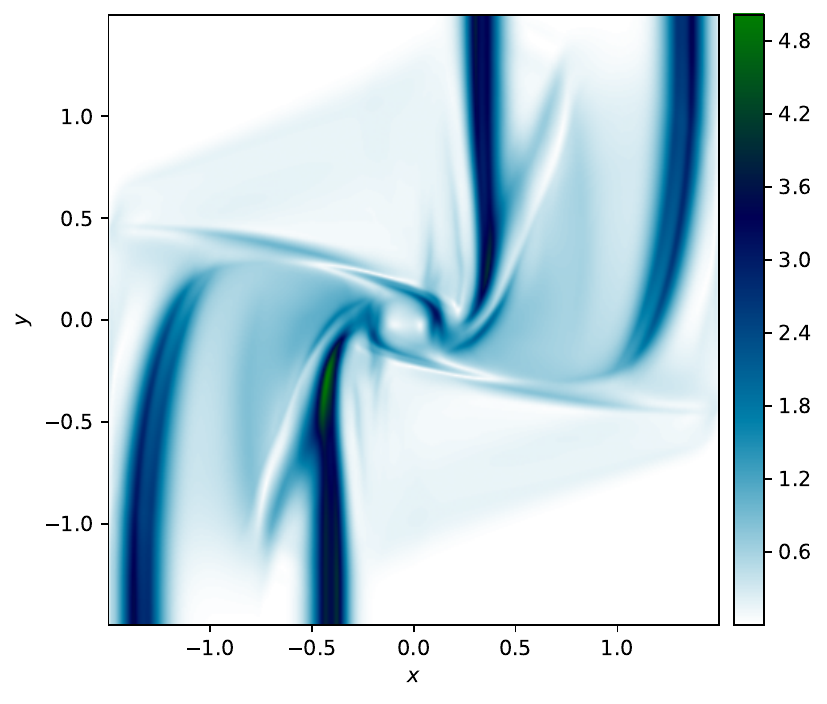}\label{fig:2DRP_by_w_m_o2}}~
		\subfigure[Schlieren image of $B_y$ for $\othi$ scheme for isotropic GLM-CGL]{\includegraphics[width=0.26\textwidth]{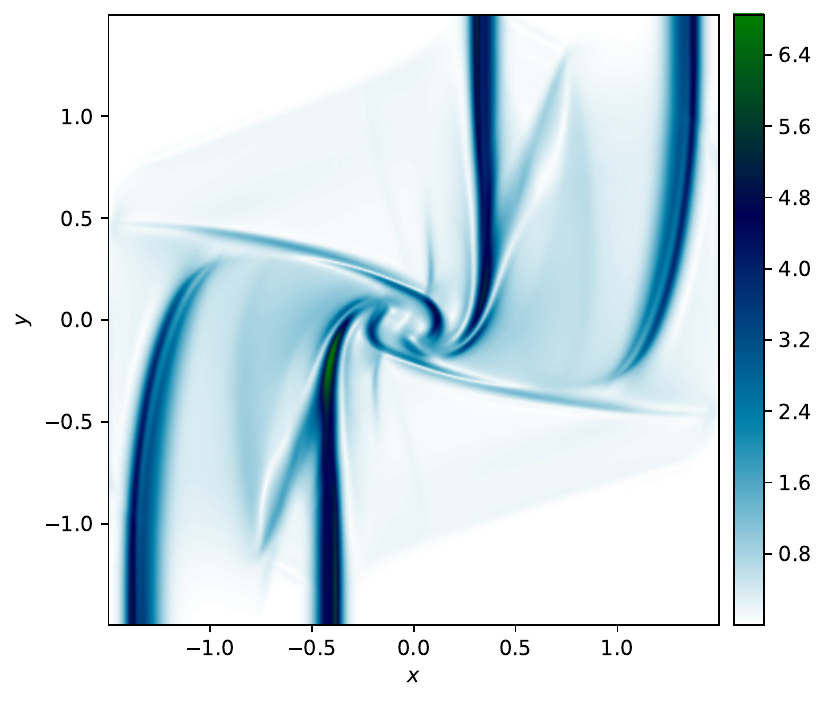}\label{fig:2DRP_by_w_m_o3}}~
		\subfigure[Schlieren image of $B_y$ for $\ofi$ scheme for isotropic GLM-CGL]{\includegraphics[width=0.26\textwidth]{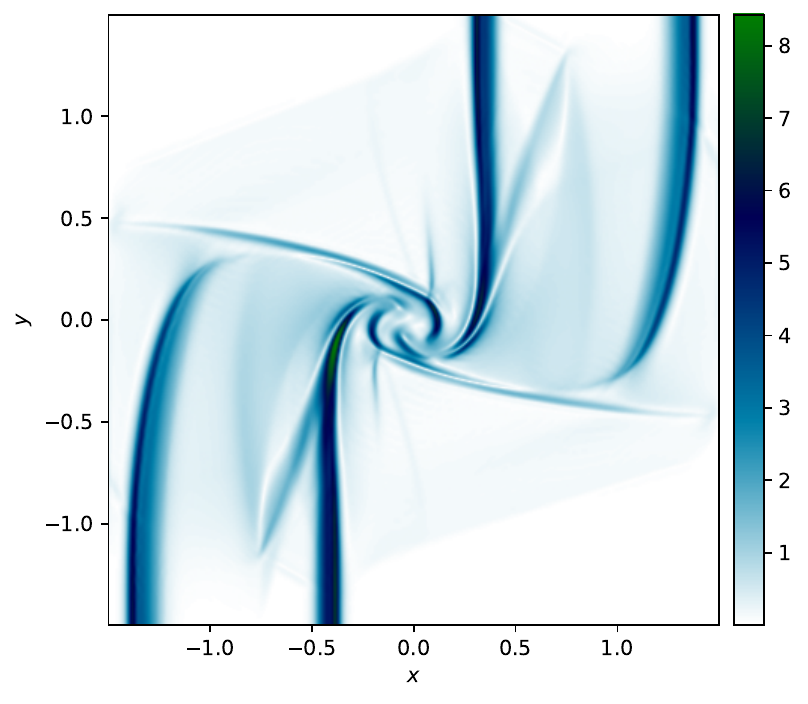}\label{fig:2DRP_by_w_m_o4}}\\
		\subfigure[$|(\nabla\cdot\B)_{i,j}|$ for $\oti$ scheme for isotropic CGL]{\includegraphics[width=0.26\textwidth]{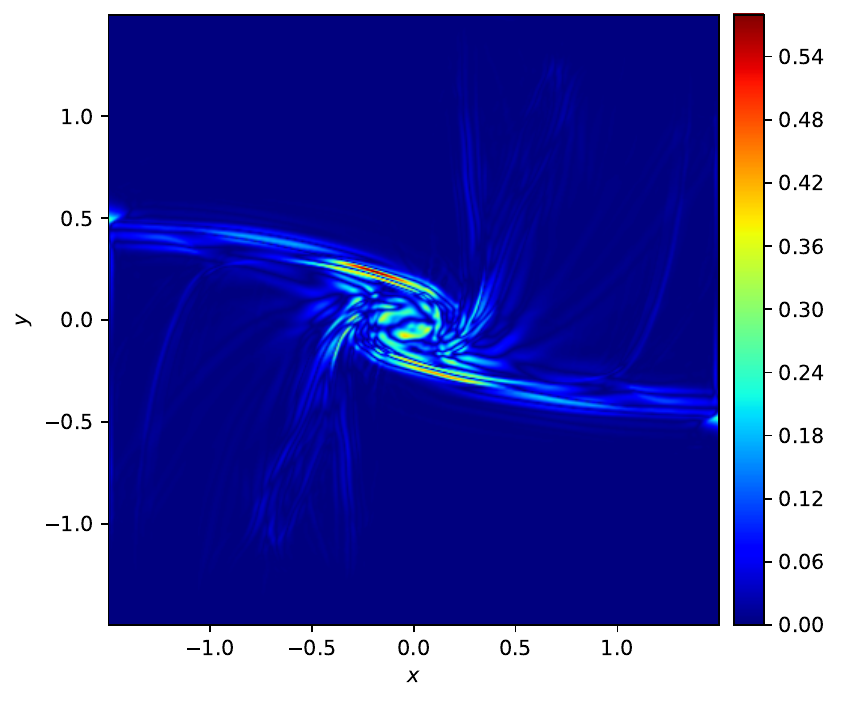}\label{fig:2DRP_db_wo_m_o2}}~
		\subfigure[$|(\nabla\cdot\B)_{i,j}|$ for $\othi$ scheme for isotropic CGL]{\includegraphics[width=0.26\textwidth]{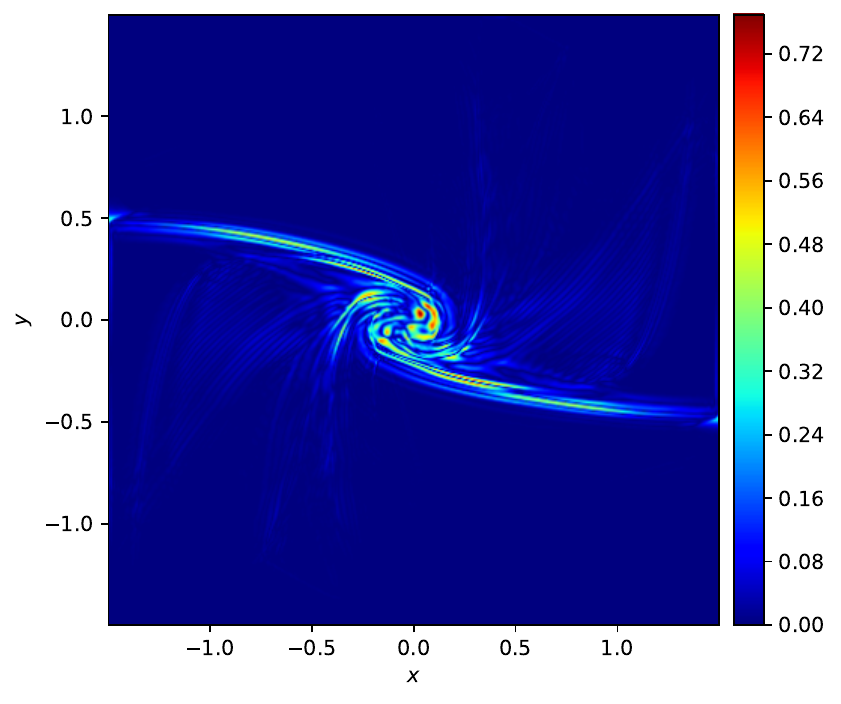}\label{fig:2DRP_db_wo_m_o3}}~
		\subfigure[$|(\nabla\cdot\B)_{i,j}|$ for $\ofi$ for isotropic CGL]{\includegraphics[width=0.26\textwidth]{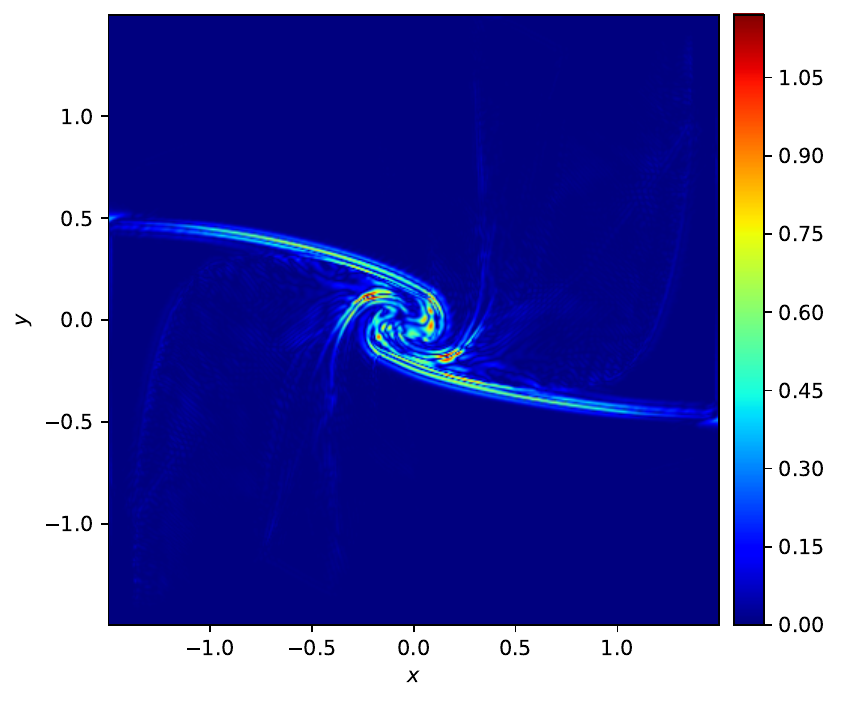}\label{fig:2DRP_db_wo_m_o4}}\\
		\subfigure[$|(\nabla\cdot\B)_{i,j}|$ for $\oti$ scheme for isotropic GLM-CGL]{\includegraphics[width=0.26\textwidth]{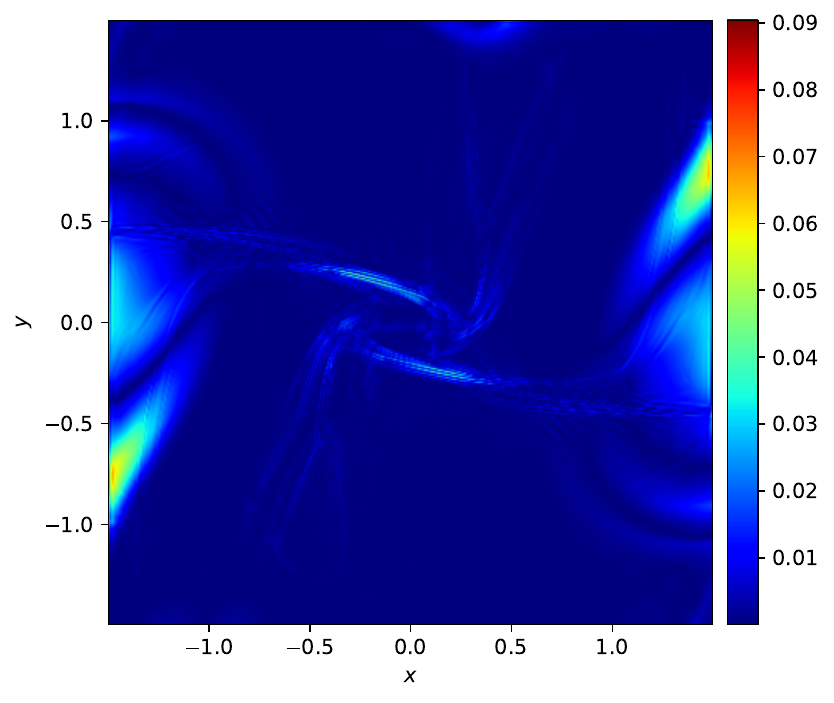}\label{fig:2DRP_db_w_m_o2}}~
		\subfigure[$|(\nabla\cdot\B)_{i,j}|$ for $\othi$ scheme for isotropic GLM-CGL]{\includegraphics[width=0.26\textwidth]{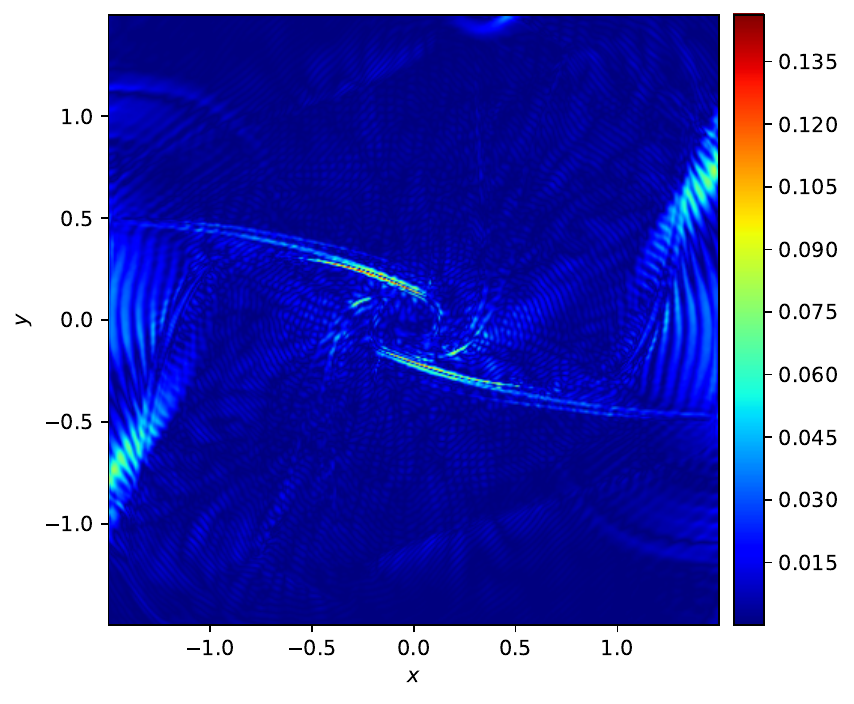}\label{fig:2DRP_db_w_m_o3}}~
		\subfigure[$|(\nabla\cdot\B)_{i,j}|$ for $\ofi$ scheme for isotropic GLM-CGL]{\includegraphics[width=0.26\textwidth]{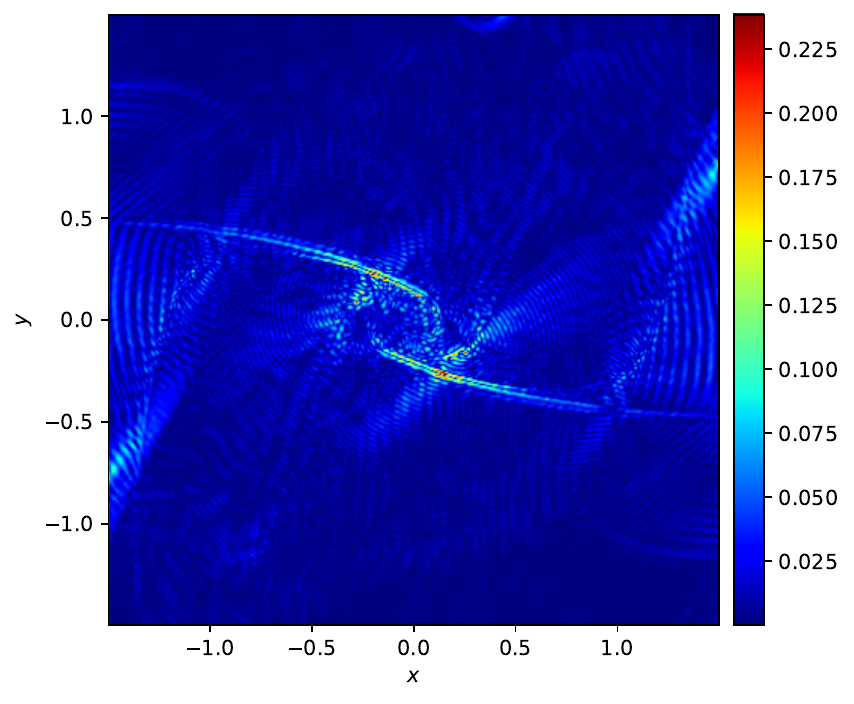}\label{fig:2DRP_db_w_m_o4}}\\
		\caption{\textbf{\nameref{test:2DRP}}: Schlieren image of $B_y$ and plots of $|(\nabla\cdot\B)_{i,j}|$ for $\oti$, $\othi$ and $\ofi$ schemes for isotropic CGL and isotropic GLM-CGL at time $t=1.0$.}
		\label{fig:2DRP_mhd_by_divb}
	\end{center}
\end{figure}
\begin{figure}[!htbp]
	\begin{center}	
		\subfigure[$\|\nabla\cdot\B\|_{1}$ and $\|\nabla\cdot\B\|_{2}$ for $\ote$ scheme for CGL and GLM-CGL ]{\includegraphics[width=0.28\textwidth]{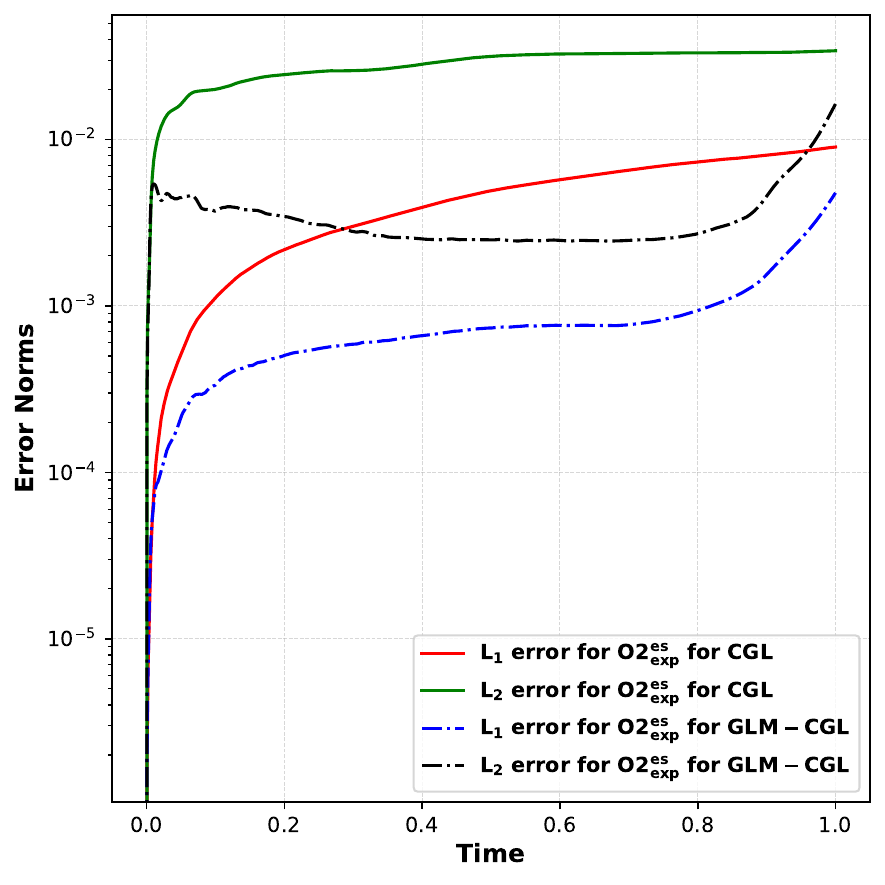}\label{fig:2DRP_error_cgl_o2}}~
		\subfigure[$\|\nabla\cdot\B\|_{1}$ and $\|\nabla\cdot\B\|_{2}$ for $\othe$ scheme for CGL and GLM-CGL ]{\includegraphics[width=0.28\textwidth]{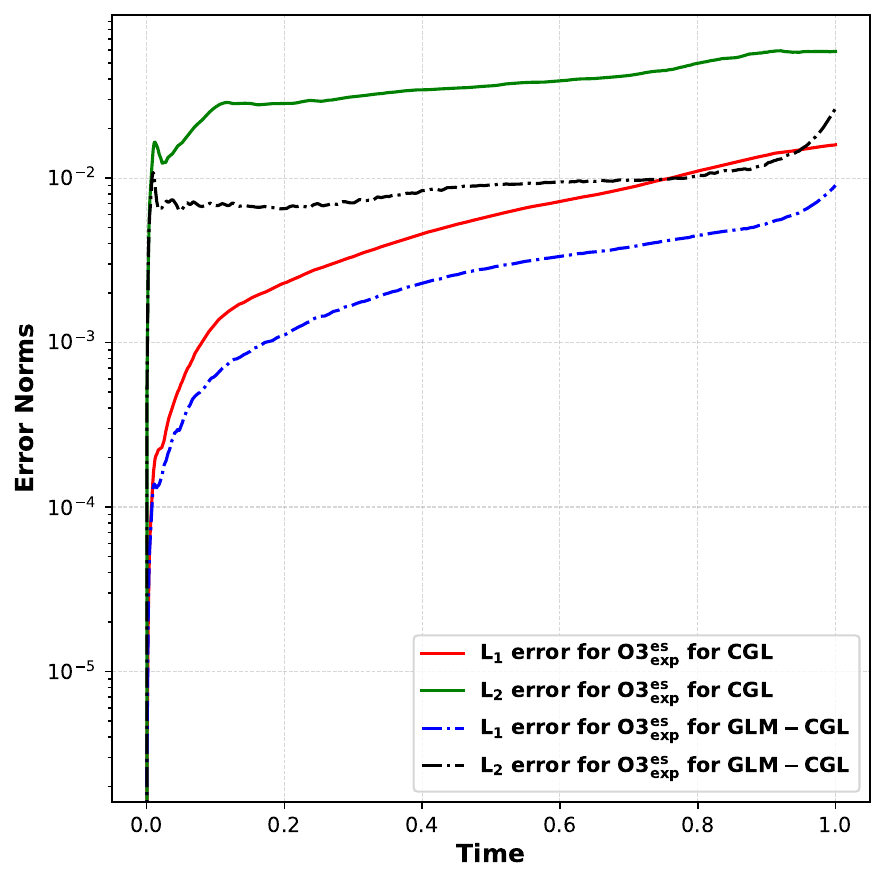}\label{fig:2DRP_error_cgl_o3}}~
		\subfigure[$\|\nabla\cdot\B\|_{1}$ and $\|\nabla\cdot\B\|_{2}$ for $\ofe$ scheme for CGL and GLM-CGL ]{\includegraphics[width=0.28\textwidth]{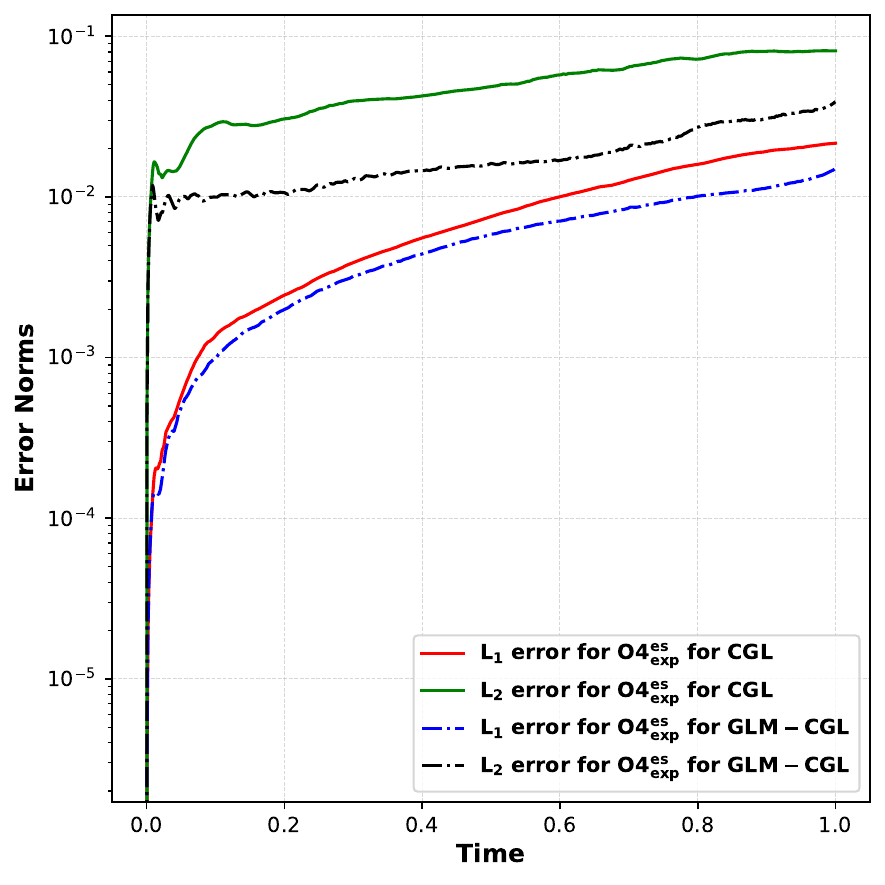}\label{fig:2DRP_error_cgl_o4}}\\
		\subfigure[$\|\nabla\cdot\B\|_{1}$ and $\|\nabla\cdot\B\|_{2}$ for $\oti$ scheme for isotropic CGL and isotropic GLM-CGL ]{\includegraphics[width=0.28\textwidth]{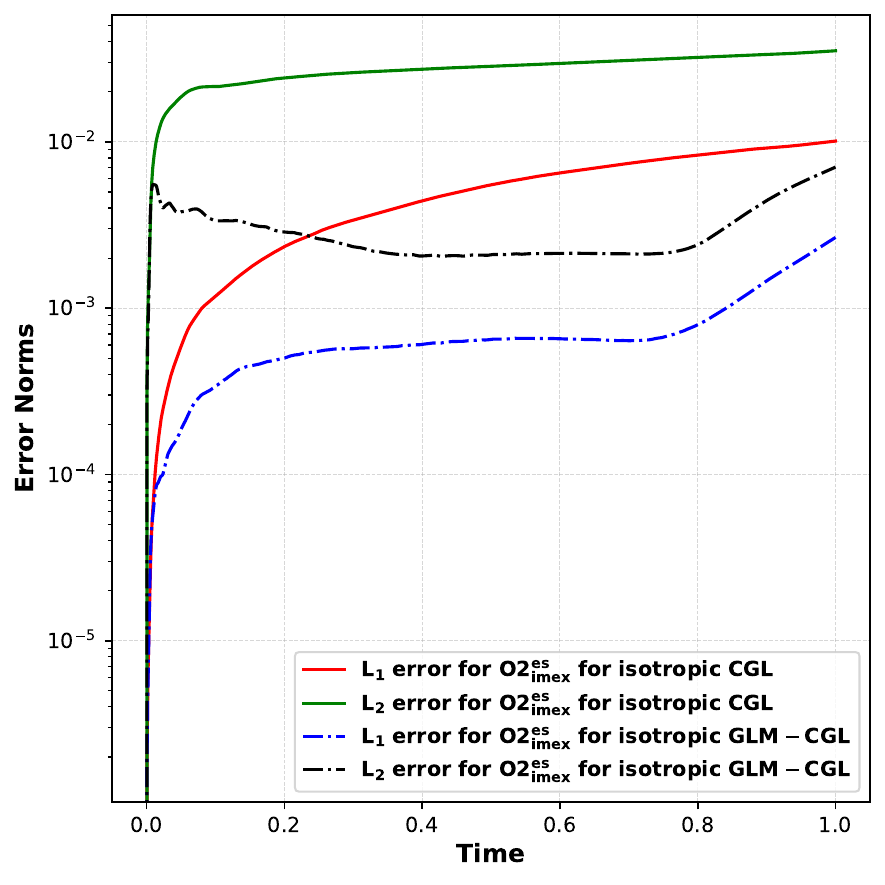}\label{fig:2DRP_error_mhd_o2}}~
		\subfigure[$\|\nabla\cdot\B\|_{1}$ and $\|\nabla\cdot\B\|_{2}$ for $\othi$ scheme for isotropic CGL and isotropic GLM-CGL ]{\includegraphics[width=0.28\textwidth]{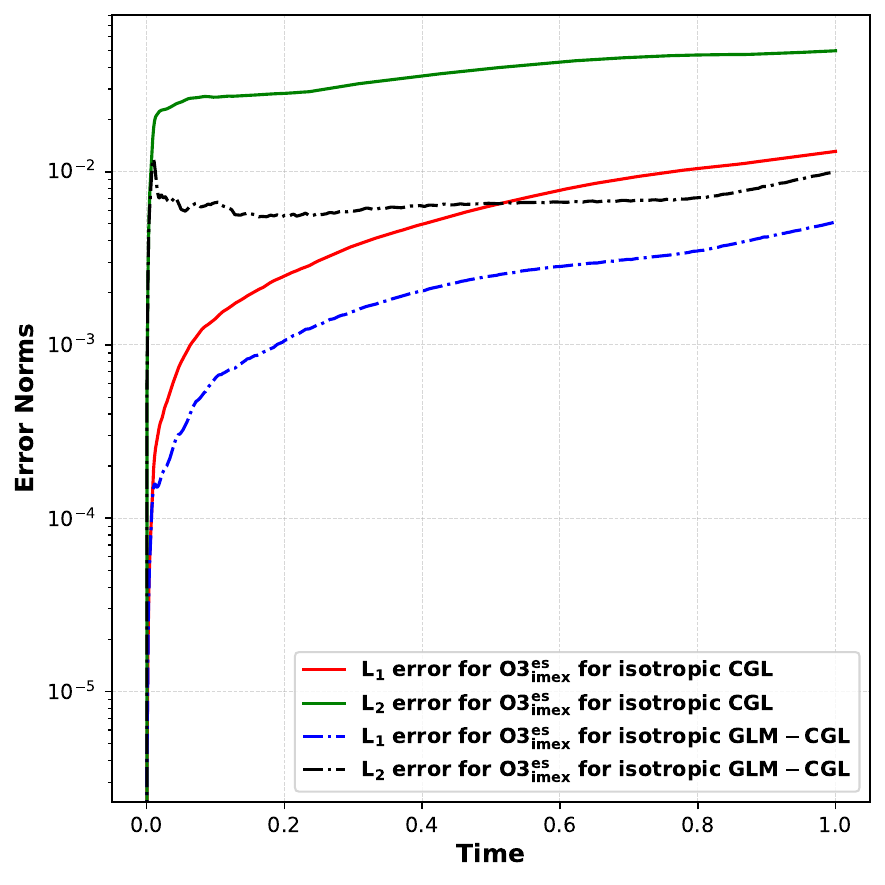}\label{fig:2DRP_error_mhd_o3}}~
		\subfigure[$\|\nabla\cdot\B\|_{1}$ and $\|\nabla\cdot\B\|_{2}$ for $\ofi$ scheme for isotropic CGL and isotropic GLM-CGL ]{\includegraphics[width=0.28\textwidth]{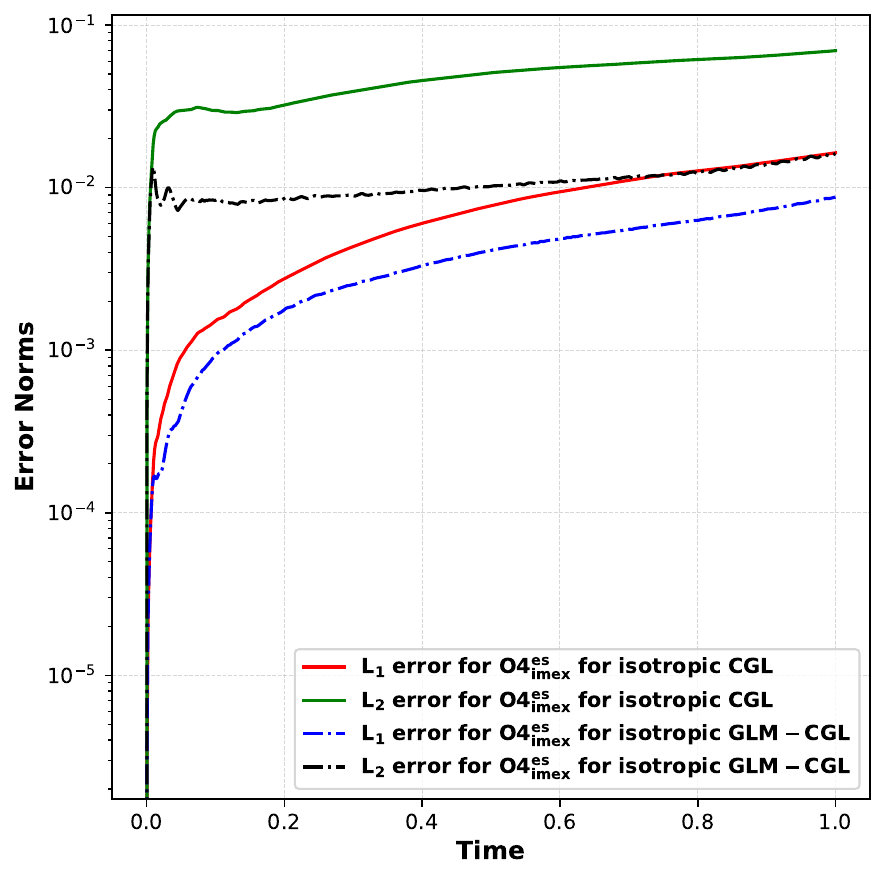}\label{fig:2DRP_error_mhd_o4}}\\
		\caption{\textbf{\nameref{test:2DRP}}: Evolution of the magnetic field divergence constraint errors till time $t=1.0$.}
		\label{fig:2DRP_error}
	\end{center}
\end{figure}

\begin{table}[ht]
\centering
\renewcommand{\arraystretch}{1.4}
\begin{tabular}{|c|c|c|c|c|c|c|}
\hline
\multirow{2}{*}{Scheme} &
\multirow{2}{*}{System} &
$\rho_{\min}$ &
$(p_{\parallel})_{\min}$ &
$(p_{\perp})_{\min}$ &
$\|\nabla\!\cdot\!\mathbf{B}\|_{1}$ &
$\|\nabla\!\cdot\!\mathbf{B}\|_{2}$ \\
&&&&&&\\
\hline

\multirow{2}{*}{$\ote$}
& CGL
& 4.674257E-01 & 1.855371E-01 & 1.636048E-01 &  9.017376E-03 & 3.415954E-02 \\ \cline{2-7}

& GLM--CGL
& 5.035759E-01 & 1.949532E-01 & 1.829844E-01 &  4.781753E-03 & 1.634865E-02  \\  \hline

\multirow{2}{*}{$\othe$}
& CGL
& 3.881367E-01 & 1.302556E-01 & 1.37023E-01 &  1.593041E-02 & 5.892514E-02 \\ \cline{2-7}

& GLM--CGL
& 4.043492E-01 & 1.568324E-01 & 1.458064E-01 & 9.042088E-03 & 2.655514E-02  \\  \hline

\multirow{2}{*}{$\ofe$}
& CGL
& 3.629564E-01 & 1.293211E-01 & 1.290771E-01 &  2.160493E-02 & 8.118646E-02 \\ \cline{2-7}

& GLM--CGL
& 3.717285E-01 & 1.184460E-01 & 1.343949E-01 &  1.500617E-02 & 3.912504E-02  \\  \hline

\multirow{2}{*}{$\oti$}
& Isotropic CGL
& 4.270803E-01 & 1.921454E-01 & 1.922029E-01 &  1.011883E-02 & 3.533485E-02 \\ \cline{2-7}

& Isotropic GLM--CGL
& 4.759654E-01 & 2.168588E-01 & 2.169223E-01 &  2.664113E-03 & 7.053657E-03  \\  \hline

\multirow{2}{*}{$\othi$}
& Isotropic CGL
& 3.410475E-01 & 1.453806E-01 & 1.453785E-01 &  1.309311E-02 & 4.986444E-02  \\ \cline{2-7}

& Isotropic GLM--CGL
& 3.777291E-01 & 1.674162E-01 & 1.674142E-01 &  5.120159E-03 & 9.987299E-03  \\  \hline

\multirow{2}{*}{$\ofi$}
& Isotropic CGL
& 3.113392E-01 & 1.184990E-01 & 1.184986E-01 &  1.637506E-02 & 6.959841E-02  \\ \cline{2-7}

& Isotropic GLM--CGL
& 3.396487E-01  & 1.439546E-01 & 1.439500E-01 &  8.752277E-03 & 1.618990E-02  \\ \hline

\end{tabular}
\caption{\textbf{\nameref{test:2DRP}}: Minimum values of density $\rho$, parallel pressure $p_{\parallel}$, and perpendicular pressure $p_{\perp}$, together with the value of $L_1$ and $L_2$ norm of $\nabla\cdot\B$ at final time.}
\label{tab:div_2DRP}
\end{table}

\revo{\subsection{Magnetized Kelvin-Helmholtz instability}\label{test:KHI}}
\revo{Following~\cite{mignone2010high,rueda2025entropy}, we consider the magnetized Kelvin-Helmholtz instability problem for the CGL model. The computation domain for the test case is $[0,1]\times[-1,1]$ with periodic boundary conditions in the $x$-direction and Neumann boundary conditions at the top and bottom boundaries. The initial condition consists of a single shear layer with velocity in the $x$-direction, together with uniform density, parallel and perpendicular pressures, and a uniform magnetic field lying in the $xz$-plane at an angle $\theta=\frac{\pi}{3}$ with the direction of propagation. A small single-mode perturbation is imposed on the velocity in $y$-direction. The initial conditions are given by 
\[\left(\rho, \pll, \per, \Psi\right) = 
\left(1, \frac{1}{\gamma},\frac{1}{\gamma},0\right)\]
\[v_x=\frac{M}{2}\tanh{\left(\frac{y}{y_0}\right)},~~~v_y=v_{y_0}\sin{\left(2\pi x\right)\exp{\left(-\frac{y^2}{\sigma^2}\right)}},~~~v_z=0\]
and
\[\B=c_a\sqrt{\rho}\left(\cos{\theta},~0,~\sin{\theta}\right)\]
where, $M=1$ is the Mach number, $y_0=\frac{1}{20}$ is the steepness of the shear, $c_a = 0.1$ is the Alfv\'en speed and the perturbation parameters are $v_{y_0}=10^{-2}$, $\sigma=0.1$. }

\revo{The solutions are computed using $256\times512$ cells till the final time of $t=20$. Figure~\eqref{fig:KHI_CGL} shows the density $\rho$ obtained using the $\ofe$ scheme for the CGL and GLM-CGL systems, and the $\ofi$ scheme for the isotropic CGL and isotropic GLM-CGL systems at time $t=7.2$, $12$, and $20$. At $t=7.2$, the initial perturbation grows, and the shear layer starts to roll up. By $t=12$, the mixing layer becomes wider and small-scale structures appear. At the final time $t=20$, the flow exhibits complex vortical structures and fine-scale features in both the density and magnetic field. The CGL and GLM-CGL formulations produce similar flow features, while the GLM-CGL solution is slightly smoother at later times due to the additional GLM divergence-cleaning mechanism.}

\revo{In Figure~\eqref{fig:KHI_MHD}, we present the images of $B_p=\frac{\sqrt{B_x^2+B_y^2}}{B_z}$ for isotropic CGL and isotropic GLM-CGL using $\ofi$ schemes at times $t=5,~8,~12$, and $t=20$. At $t=5$, the initial perturbation grows, and the magnetic field lines roll up into a cat's-eye vortex structure. By $t=8$, secondary structures develop, and the mixing layer becomes wider. As the instability evolves at $t=12$ and $t=20$, the mixing layer continues to broaden, and the magnetic field forms thin filament-like structures. Compared with the isotropic CGL solution, we observe that the isotropic GLM-CGL solution more closely matches the solution presented in~\cite{mignone2010high}.}

\revo{We have also plotted $|(\nabla\cdot\mathbf{B})_{i,j}|$ for the CGL, GLM-CGL, isotropic CGL, and isotropic GLM-CGL systems in Figure.~\eqref{fig:KHI_CGL_div}, and observe that the GLM-CGL formulation produces significantly lower divergence errors. The colorbars in all divergence plots are dynamically scaled using the absolute local minimum and maximum values of $\nabla\cdot\B$ in each plot.}

\revo{To monitor the divergence error over time, we have plotted the evolution of the $L_1$ and $L_2$ errors of the divergence of the magnetic field in Figure~\eqref{fig:KHI_CGL_div}. We note that the errors are stable; however, GLM-CGL errors are significantly less than the CGL errors for both isotropic and anisotropic cases.}

\revo{Table~\eqref{tab:div_KHI} lists the minimum values of density, $\pll$, and $\per$ as well as the value of $L_1$ and $L_2$ norms of $\nabla\cdot\B$ at the final time. All reported values of density and pressure remain positive for all schemes and formulations, indicating that physically admissible solutions are maintained throughout the computation. The table also demonstrates the significant reduction in divergence errors achieved by the GLM formulation. These results indicate that the proposed GLM-CGL formulation remains robust for this challenging Kelvin-Helmholtz instability test.}
\begin{figure}[!htbp]
	\begin{center}
		\subfigure[$\rho$ for CGL at time $t=7.2$]{\includegraphics[width=0.245\textwidth]{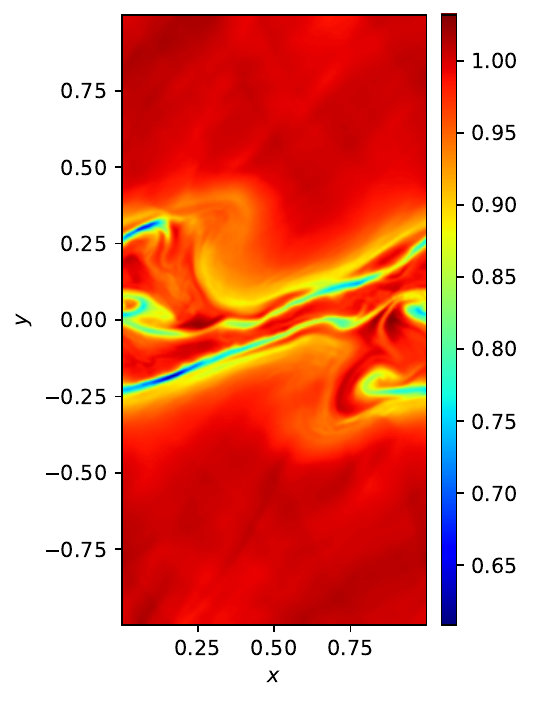}\label{fig:KHI_CGL_rho_wo_7.2}}~
		\subfigure[$\rho$  for GLM-CGL at time $t=7.2$]{\includegraphics[width=0.245\textwidth]{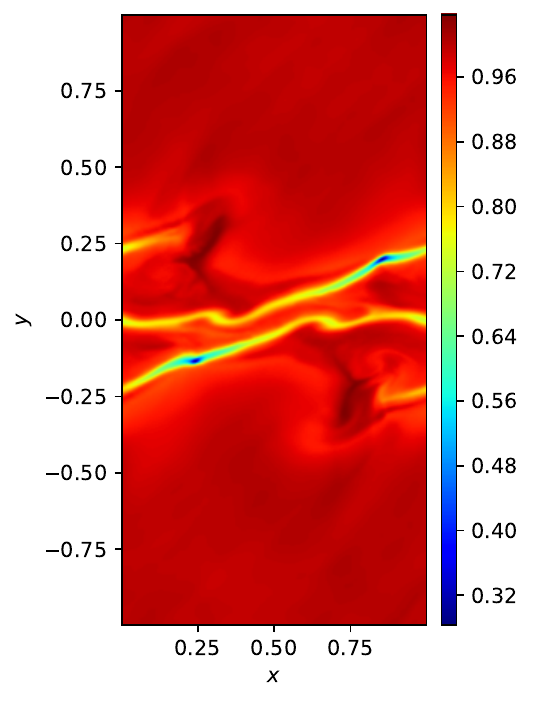}\label{fig:KHI_CGL_rho_w_7.2}}~
		\subfigure[$\rho$ for isotropic CGL at time $t=7.2$]{\includegraphics[width=0.245\textwidth]{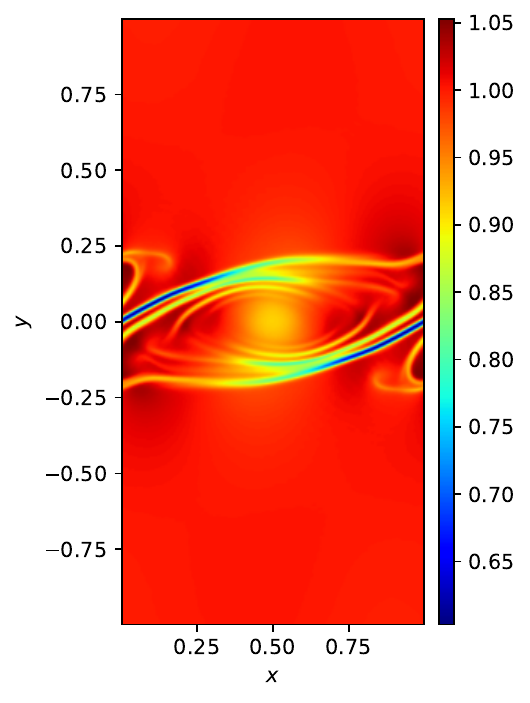}\label{fig:KHI_MHD_rho_wo_7.2}}~
		\subfigure[$\rho$ for isotropic GLM-CGL at time $t=7.2$]{\includegraphics[width=0.245\textwidth]{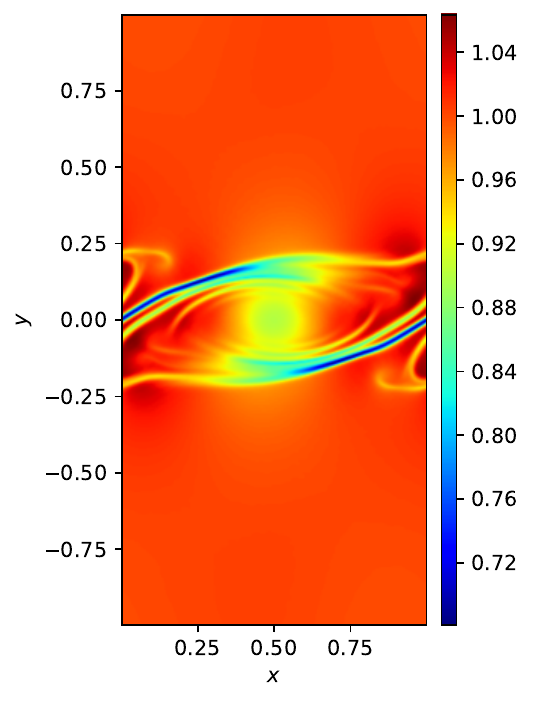}\label{fig:KHI_MHD_rho_w_7.2}}\\
        \subfigure[$\rho$ for CGL at time $t=12$]{\includegraphics[width=0.245\textwidth]{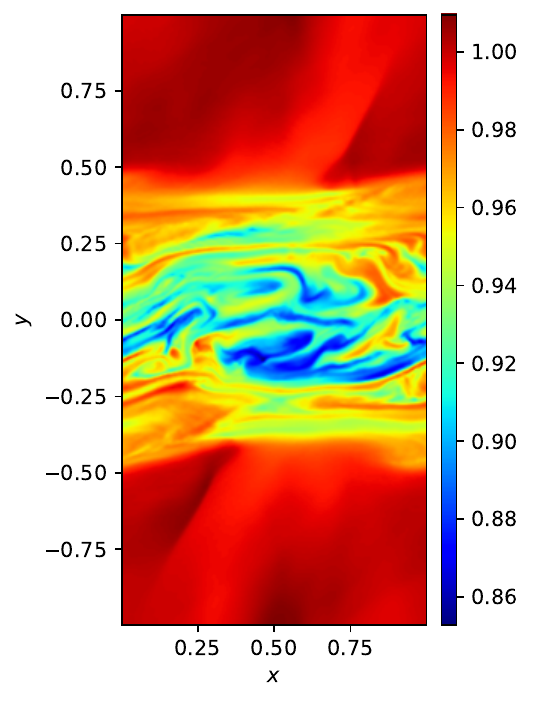}\label{fig:KHI_CGL_rho_wo_12}}~
		\subfigure[$\rho$  for GLM-CGL at time $t=12$]{\includegraphics[width=0.245\textwidth]{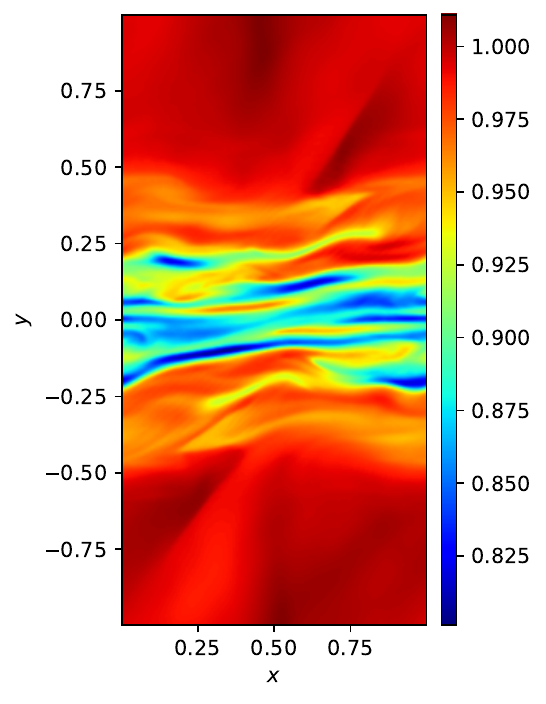}\label{fig:KHI_CGL_rho_w_12}}~
		\subfigure[$\rho$ for isotropic CGL at time $t=12$]{\includegraphics[width=0.245\textwidth]{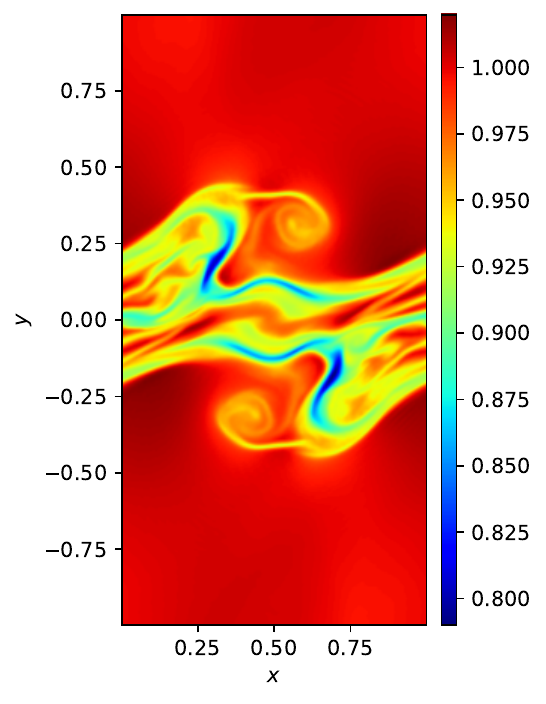}\label{fig:KHI_MHD_rho_wo_12}}~
		\subfigure[$\rho$ for isotropic GLM-CGL at time $t=12$]{\includegraphics[width=0.245\textwidth]{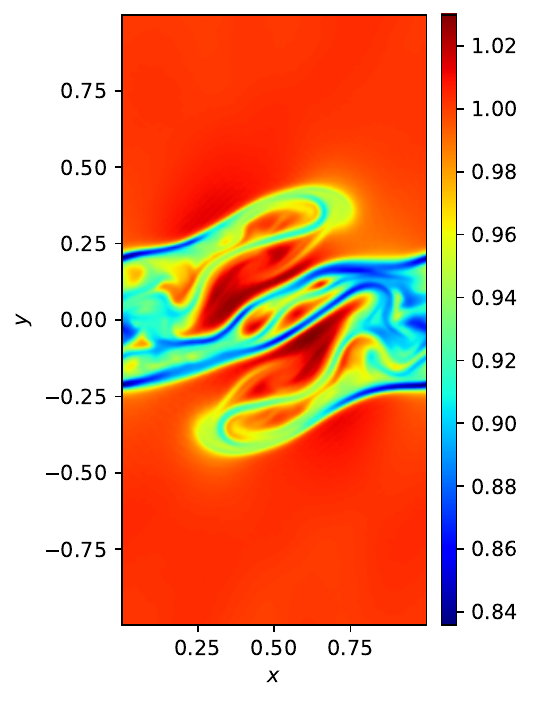}\label{fig:KHI_MHD_rho_w_12}}\\
        \subfigure[$\rho$ for CGL at time $t=20$]{\includegraphics[width=0.245\textwidth]{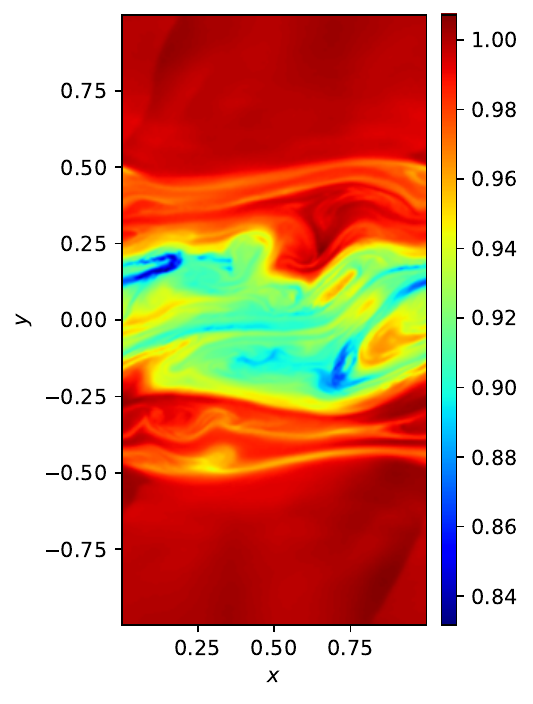}\label{fig:KHI_CGL_rho_wo_20}}~
		\subfigure[$\rho$  for GLM-CGL at time $t=20$]{\includegraphics[width=0.245\textwidth]{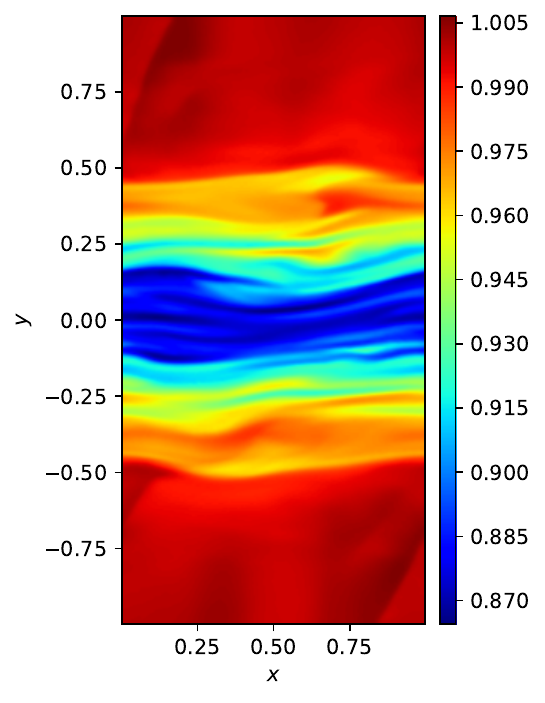}\label{fig:KHI_CGL_rho_w_20}}~
		\subfigure[$\rho$ for isotropic CGL at time $t=20$]{\includegraphics[width=0.245\textwidth]{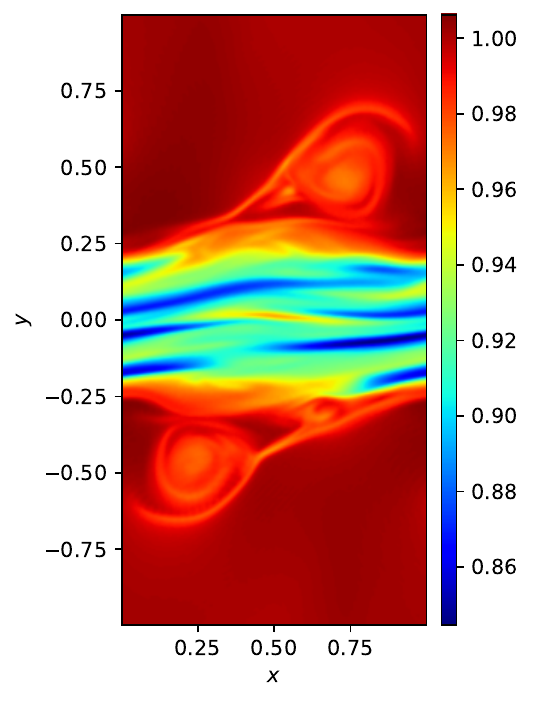}\label{fig:KHI_MHD_rho_wo_20}}~
		\subfigure[$\rho$ for isotropic GLM-CGL at time $t=20$]{\includegraphics[width=0.245\textwidth]{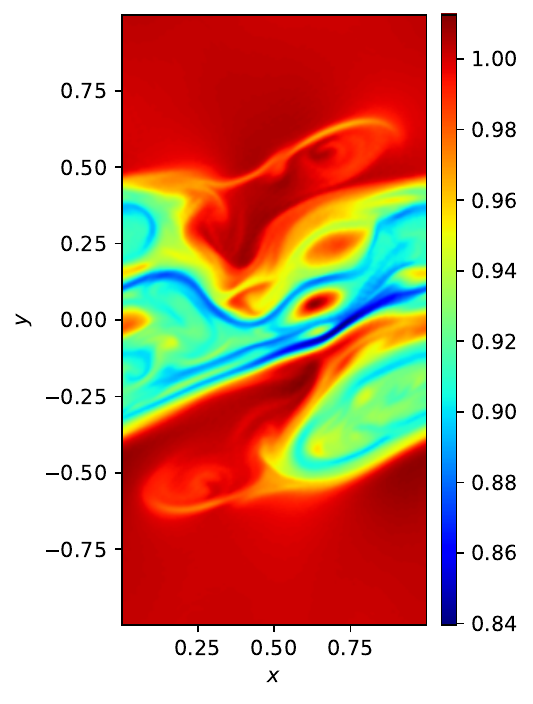}\label{fig:KHI_MHD_rho_w_20}}\\
		\caption{\textbf{\nameref{test:KHI}}: Plots of density  at times $t=7.2,\,12$ and $20$ obtained using the $\ofe$ scheme for the CGL and GLM-CGL, and the $\ofi$ scheme for the isotropic CGL and isotropic GLM-CGL.}
		\label{fig:KHI_CGL}
	\end{center}
\end{figure}

\begin{figure}[!htbp]
	\begin{center}
		\subfigure[$B_p$ for isotropic CGL at time $t=5$]{\includegraphics[width=0.245\textwidth]{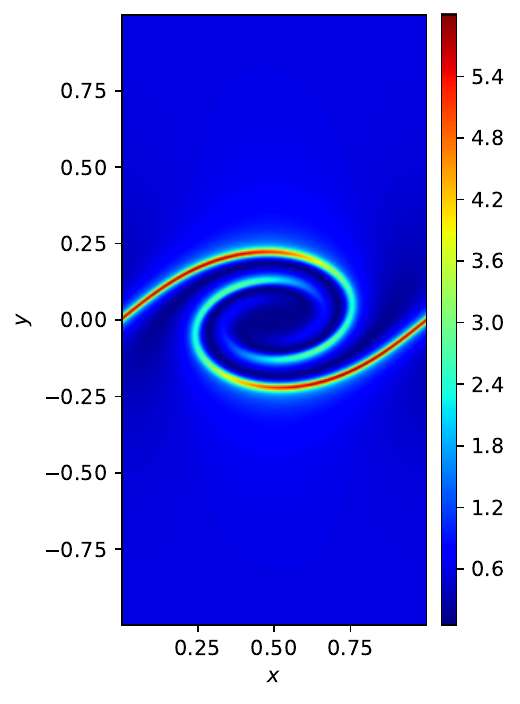}\label{fig:KHI_wo_t_5}}~
		\subfigure[$B_p$  for isotropic CGL at time $t=8$]{\includegraphics[width=0.245\textwidth]{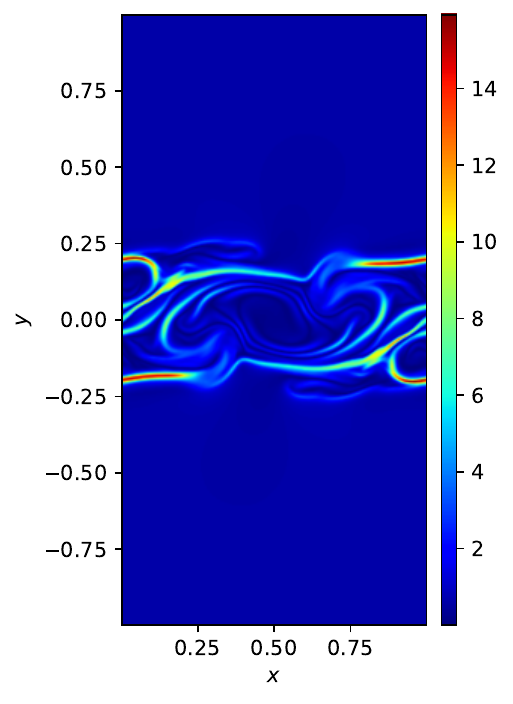}\label{fig:KHI_wo_t_8}}~
		\subfigure[$B_p$ for isotropic CGL at time $t=12$]{\includegraphics[width=0.245\textwidth]{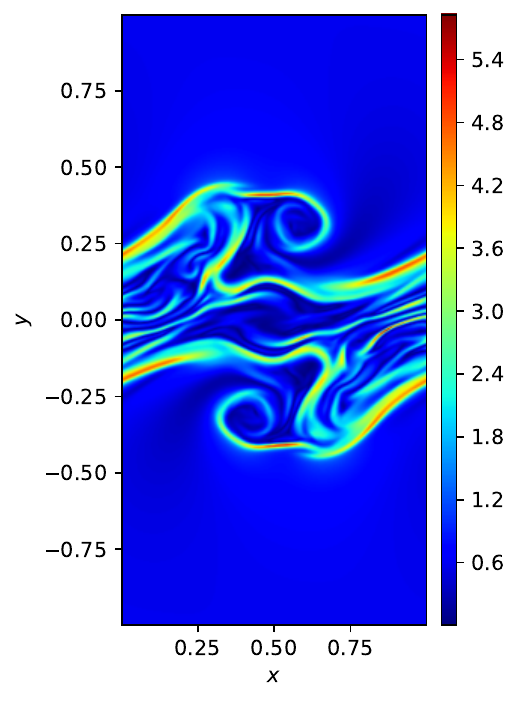}\label{fig:KHI_wo_t_12}}~
		\subfigure[$B_p$ for isotropic CGL at time $t=20$]{\includegraphics[width=0.245\textwidth]{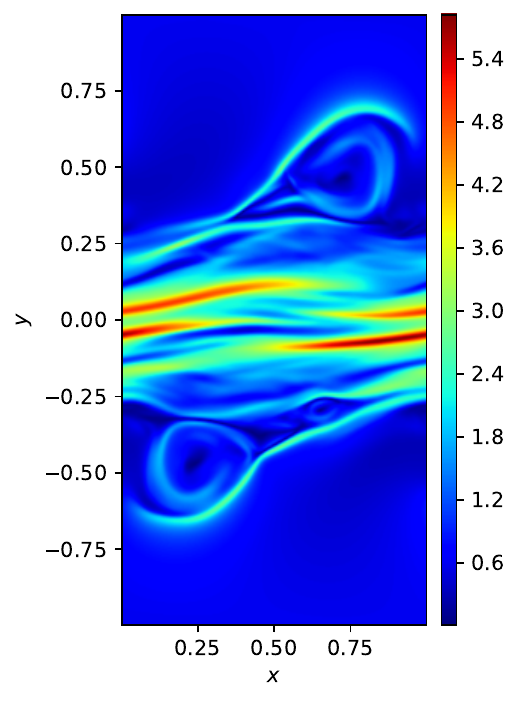}\label{fig:KHI_wo_t_20}}\\
        \subfigure[$B_p$ for isotropic GLM-CGL at time $t=5$]{\includegraphics[width=0.245\textwidth]{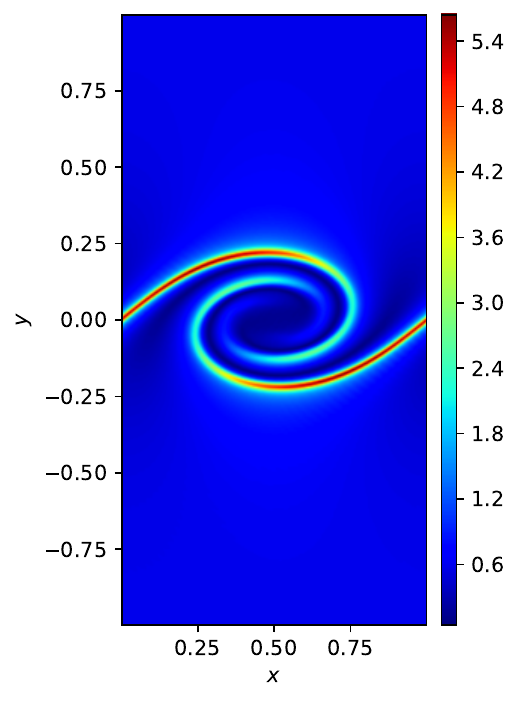}\label{fig:KHI_w_t_5}}~
		\subfigure[$B_p$  for isotropic GLM-CGL at time $t=8$]{\includegraphics[width=0.245\textwidth]{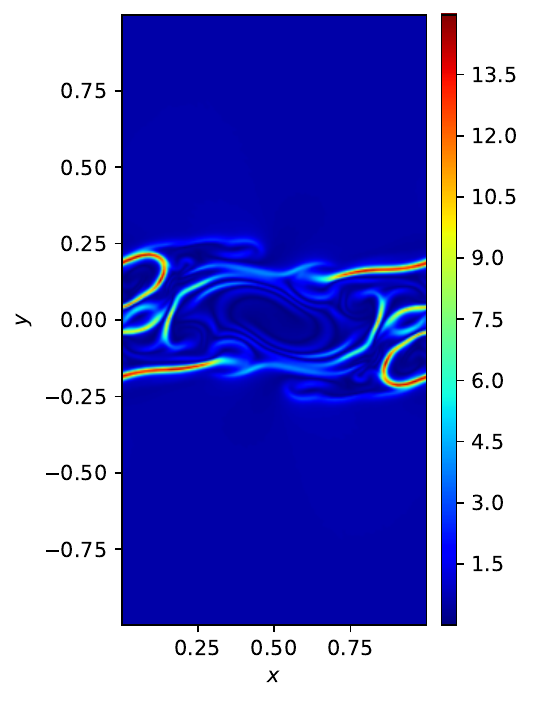}\label{fig:KHI_w_t_8}}~
		\subfigure[$B_p$ for isotropic GLM-CGL at time $t=12$]{\includegraphics[width=0.245\textwidth]{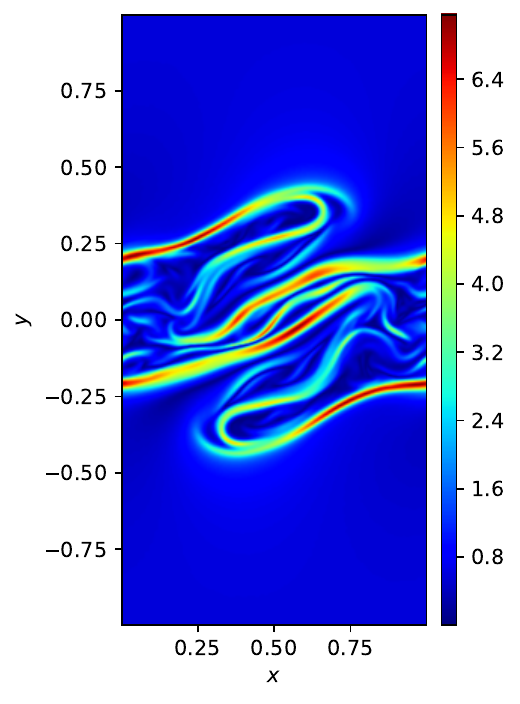}\label{fig:KHI_w_t_12}}~
		\subfigure[$B_p$ for isotropic GLM-CGL at time $t=20$]{\includegraphics[width=0.245\textwidth]{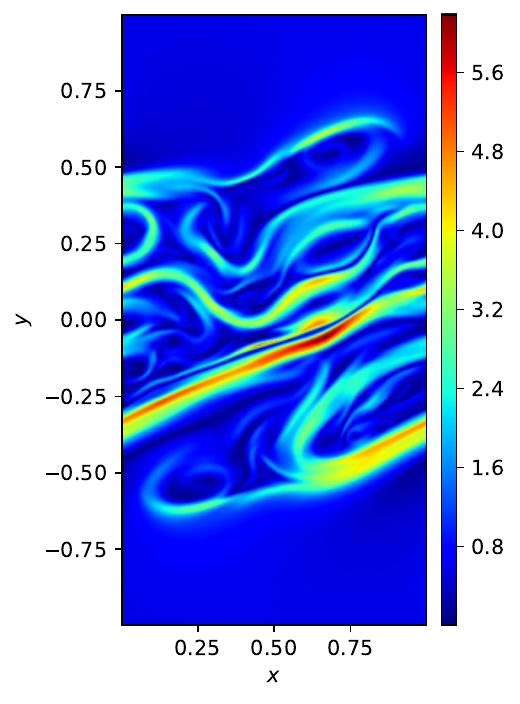}\label{fig:KHI_w_t_20}}\\
		\caption{\textbf{\nameref{test:KHI}}: Plots of $B_p$ at times $t=5,\,8,\,12,$ and $20$ using $\ofi$ scheme for the isotropic CGL and isotropic GLM-CGL.}
		\label{fig:KHI_MHD}
	\end{center}
\end{figure}

\begin{figure}[!htbp]
	\begin{center}
        \subfigure[$|(\nabla\cdot\B)_{i,j}|$ for CGL at time $t=20$]{\includegraphics[width=0.25\textwidth]{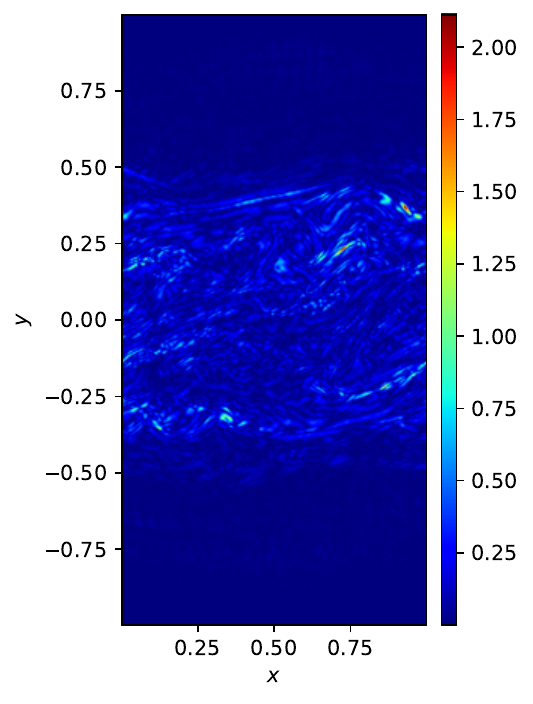}\label{fig:KHI_db_wo_CGL}}~
		\subfigure[$|(\nabla\cdot\B)_{i,j}|$ for GLM-CGL at time $t=20$]{\includegraphics[width=0.25\textwidth]{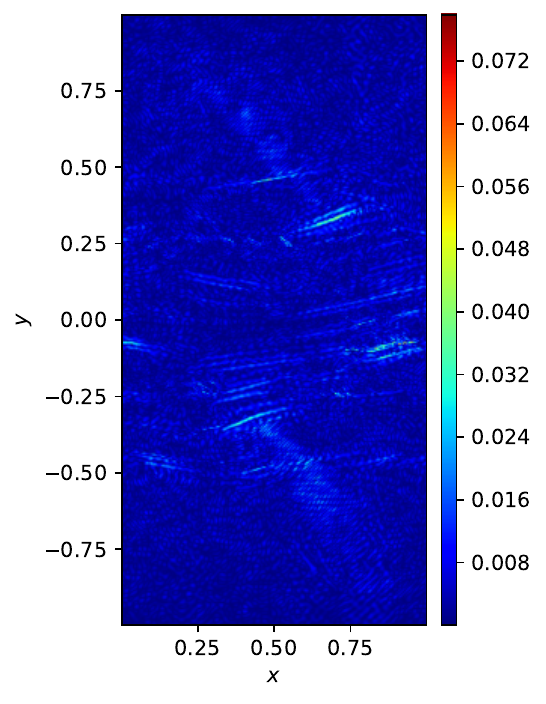}\label{fig:KHI_db_w_CGL}}
        \subfigure[$\|\nabla\cdot\B\|_{1}$ and $\|\nabla\cdot\B\|_{2}$ for CGL and GLM-CGL at time $t=20$]{\includegraphics[width=0.35\textwidth]{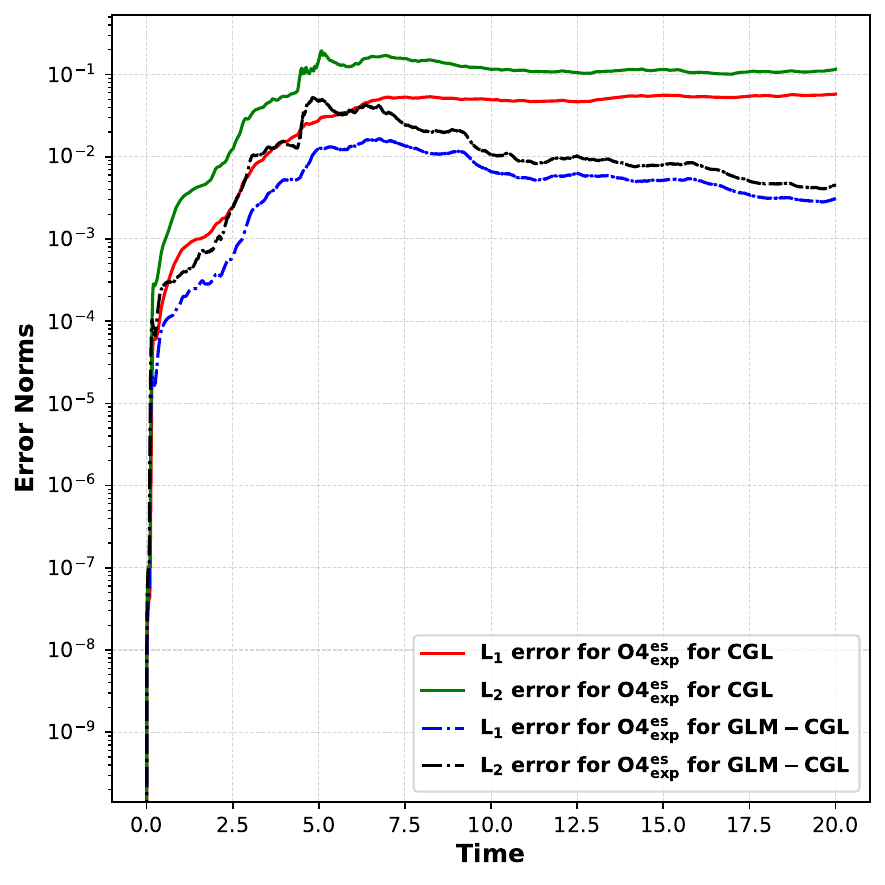}\label{fig:KHI_CGL_error}}\\
         \subfigure[$|(\nabla\cdot\B)_{i,j}|$ for isotropic CGL at time $t=20$]{\includegraphics[width=0.245\textwidth]{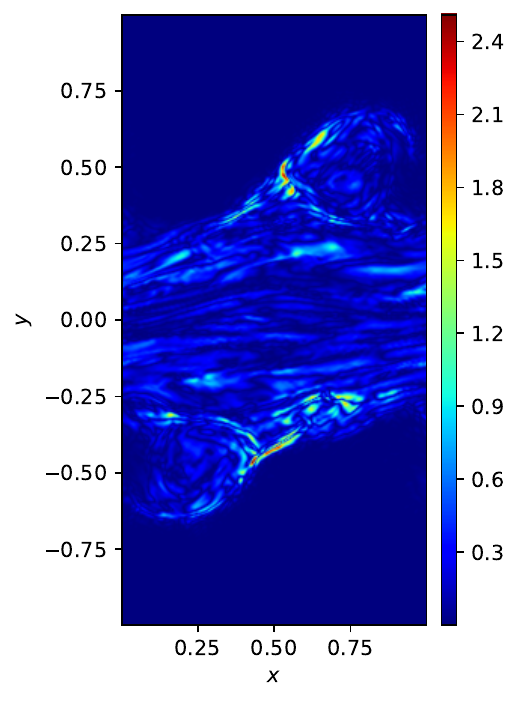}\label{fig:KHI_db_wo}}~
		\subfigure[$|(\nabla\cdot\B)_{i,j}|$ for isotropic GLM-CGL at time $t=20$]{\includegraphics[width=0.245\textwidth]{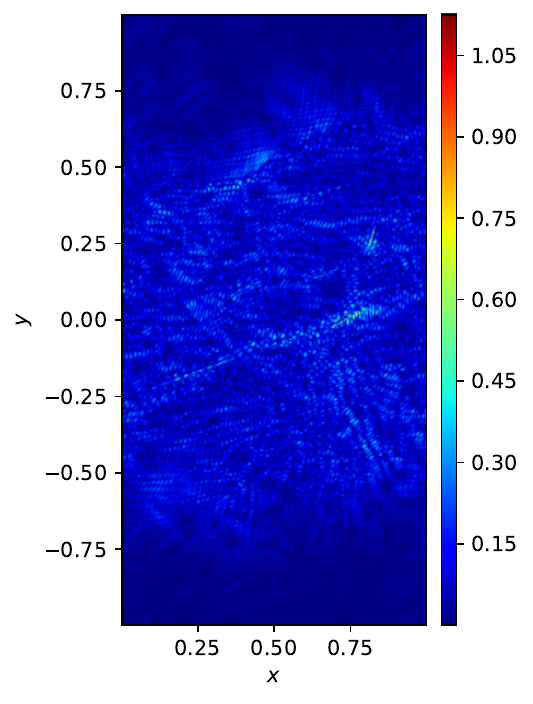}\label{fig:KHI_db_w}}
        \subfigure[$\|\nabla\cdot\B\|_{1}$ and $\|\nabla\cdot\B\|_{2}$ for isotropic CGL and isotropic GLM-CGL at time $t=20$]{\includegraphics[width=0.35\textwidth]{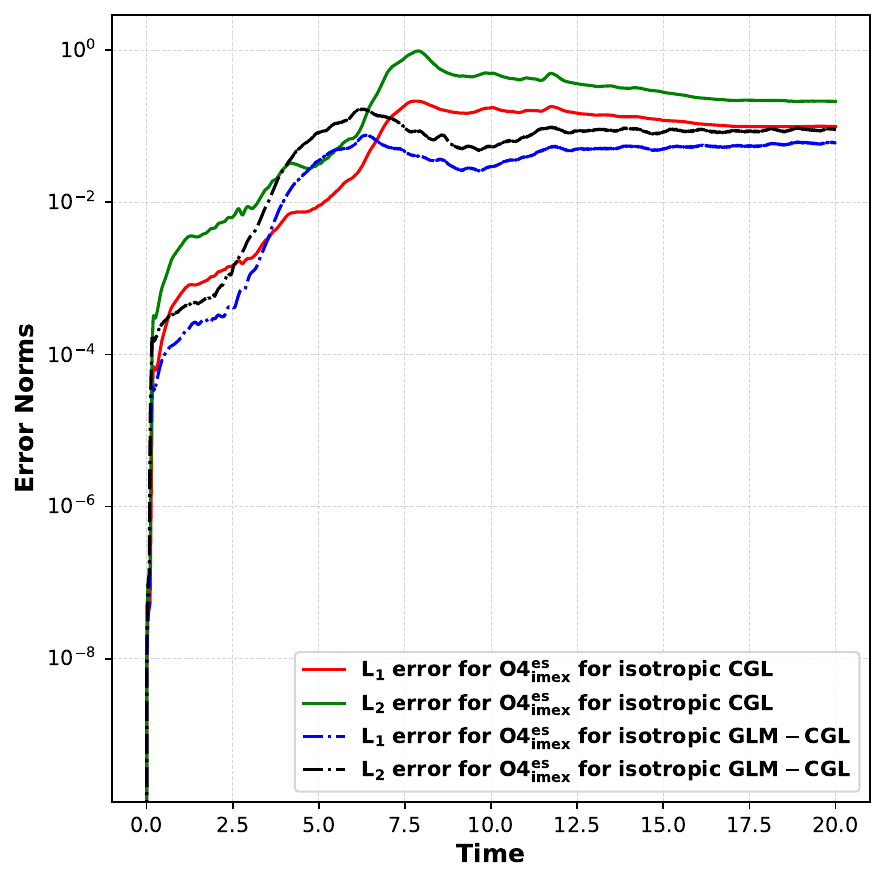}\label{fig:KHI_MHD_error}}\\
		\caption{\textbf{\nameref{test:KHI}}: Plots of $|(\nabla\cdot\B)_{i,j}|$ at times $t=20$, together with the time evolution of the magnetic field divergence constraint errors till time $t=20$, obtained using the $\ofe$ schemes for the CGL and GLM-CGL and the $\ofi$ scheme for the isotropic CGL and isotropic GLM-CGL.}
		\label{fig:KHI_CGL_div}
	\end{center}
\end{figure}

\begin{table}[ht]
\centering
\renewcommand{\arraystretch}{1.4}
\begin{tabular}{|c|c|c|c|c|c|c|}
\hline
\multirow{2}{*}{Scheme} &
\multirow{2}{*}{System} &
$\rho_{\min}$ &
$(p_{\parallel})_{\min}$ &
$(p_{\perp})_{\min}$ &
$\|\nabla\!\cdot\!\mathbf{B}\|_{1}$ &
$\|\nabla\!\cdot\!\mathbf{B}\|_{2}$ \\
&&&&&&\\
\hline
\multirow{2}{*}{$\ofe$}
& CGL
& 8.316871E-01 & 3.821701E-01 & 5.780448E-01 &  5.780362E-02 & 1.158326E-01  \\ \cline{2-7}

& GLM--CGL
& 8.644874E-01 & 4.346404E-01 & 5.901851E-01 & 3.077761E-03 & 4.528463E-03  \\  \hline

\multirow{2}{*}{$\ofi$}
& Isotropic CGL
& 8.446233E-01 & 5.103078E-01 & 5.103083E-01 &   9.878701E-02 & 2.119360E-01   \\ \cline{2-7}

& Isotropic GLM--CGL
&  8.395822E-01 & 5.225705E-01 & 5.225694E-01 &  6.047111E-02 & 9.038694E-02   \\ \hline

\end{tabular}
\caption{\textbf{\nameref{test:KHI}}: Minimum values of density $\rho$, parallel pressure $p_{\parallel}$, and perpendicular pressure $p_{\perp}$, together with the value of $L_1$ and $L_2$ norm of $\nabla\cdot\B$ at final time.}
\label{tab:div_KHI}
\end{table}
\section{Conclusion}
\label{sec:conc}
In this article, we have presented a divergence-diminishing strategy based on the GLM formulation of the CGL equations. This is inspired by a similar strategy for MHD equations. The resulting system is named the GLM-CGL model. We then rewrite the system in such a way that the non-conservative terms do not contribute to entropy production. Furthermore, the conservative part is then symmetrized. The resulting reformulation makes the GLM-CGL system ideal for proposing entropy-stable schemes.

We then design higher-order entropy stable finite difference schemes for the system by designing an entropy conservative flux and a higher-order entropy diffusion operator. The entropy diffusion operator uses the entropy-scaled right eigenvectors. We also consider a source term to compute the isotropic solutions.

For the numerical test cases, we consider several 1D and 2D test cases. We demonstrate that in all the test cases, the proposed schemes have the expected numerical order of accuracy. Also, the isotropic solutions are comparable to the MHD solutions. Furthermore, in each numerical test case, the GLM-CGL formulation is observed to be superior in diminishing the divergence of the magnetic field.


\section*{Acknowledgements}
Dinshaw S. Balsara and Harish Kumar acknowledge support from a Vajra award (VJR/2018/00129).

\printcredits

\section*{Declarations}
\textbf{Conflict of interest} The authors declare that they have no Conflict of interest.

\section*{\revg{Data availability}}
Data will be made available on request.

\bibliographystyle{unsrt}

\bibliography{cas-refs}

\appendix
\section{Eigenvalues and admissible domain of GLM-CGL system}\label{Eiegn_value_admissible domain}
Following~\cite{kato1966propagation,singh2024entropy,singh2024eigen}, the set of eigenvalues of GLM-CGL system~\eqref{eq:glm_cgl_con} in $x$-direction is given below,
\begin{equation}
	\tilde{\bb{\Lambda}}_x=\left\{v_x,~ v_x, ~\frac{1}{2}\left(v_x \pm \sqrt{4 c_h^2 + v_x^2}\right), ~v_x\pm c_a,~v_x\pm c_f,~v_x\pm c_s\right\}
	\label{eq:fulleigen_glm_cgl_x}
\end{equation}
where,
\begin{align*}
	& c_a= \sqrt{{\frac{B_x^2}{\rho}-\frac{(\pll-p_\perp){b}_x^2}{\rho}}},&\\&
	c_f= \frac{1}{\sqrt{2\rho}}\big[|\B|^2+2p_\bot +b_x^2(2\pll-\per) + \{(|\B|^2+2\per +b_x^2(2\pll-\per))^2&\\&~~~~~~+4(\per^2b_x^2(1-b_x^2)-3\pll \per b_x^2(2-b_x^2)+3\pll^2 b_x^4-3B_x^2 \pll)\}^{\frac{1}{2}}\big]^{\frac{1}{2}},&\\&
	c_s=\frac{1}{\sqrt{2\rho}}\big[|\B|^2+2\per +b_x^2(2\pll-\per) - \{(|\B|^2+2\per +b_x^2(2\pll-\per))^2&\\&~~~~~~+4(\per^2b_x^2(1-b_x^2)-3\pll \per b_x^2(2-b_x^2)+3\pll^2 b_x^4-3B_x^2 \pll )\}^{\frac{1}{2}}\big]^{\frac{1}{2}}.
\end{align*}
Here, $v_x$ define the entropy and pressure anisotropy wave and $\frac{1}{2}\left(v_x \pm \sqrt{4 c_h^2 + v_x^2}\right)$ define the right and left going GLM waves. Similarly, $v_x\pm c_a$ defines the right- and left-going Alfv\'en waves. Also, $v_x\pm c_f$ and $v_x\pm c_s$ are the right and left going fast and slow waves. In the GLM-CGL system, we have extra eigenvalues corresponding to the right and left going GLM waves, and all other eigenvalues are equal to the CGL system. In~\cite{kato1966propagation}, the authors show that the CGL system is not always hyperbolic, and additional constraints need to be imposed to ensure hyperbolicity. We define
$$
p_m=\frac{p_{\bot}^2}{6p_{\bot}+3|\B|^2},~~\text{and}~~p_M={|\B|^2}+p_{\bot}.
$$
Then, the admissible domain for the CGL system is given by,
\begin{equation*}
	\Omega_{CGL}=\{\con\in \mathbb{R}^{9} |~\rho>0,~\pll>0,\per>0,~p_{m} \leq \pll \leq p_{M} \}.
\end{equation*}
The solution domain $\Omega$ then must be partitioned because it was observed in \cite{kato1966propagation} that the slow wave is not always slower than the Alfv\'en wave. So, we impose
\begin{enumerate}
	\item $p_{m}\le \pll \le\frac{p_{M}}{4},~~~~~~~~~~\text{if}~~~~~~~ c_s\le c_{a}\le c_f,$
	\item $\frac{p_{M}}{4}\le \pll \le\frac{p_{M}}{4}+\frac{3p_{m}}{4},~~\text{if}~~~~~~ c_s\le c_{a}< c_f,$
	\item $\frac{p_{M}}{4}+\frac{3p_{m}}{4}\le \pll \le p_{M},~~\text{if}~~~~~~ c_a\le c_{s}< c_f.$
\end{enumerate}
With these conditions, the system is hyperbolic. As discussed above, the GLM-CGL system will also be hyperbolic under the same assumptions with the extended admissible domain
\begin{equation}
	\label{eq:cgl_domain_new}
	\Omega_{GLM-CGL}=\{\con\in \mathbb{R}^{10} |~\rho>0,~\pll>0,\per>0,~p_{m} \leq \pll \leq p_{M} \}.
\end{equation}
 Similarly, the eigenvalues in $y$-direction are,
\begin{equation}
	\tilde{\bb{\Lambda}}_y=\left\{v_y,~ v_y, ~\frac{1}{2}\left(v_y \pm \sqrt{4 c_h^2 + v_y^2}\right), ~v_y\pm c_a,~v_y\pm c_f,~v_y\pm c_s\right\}
	\label{eq:fulleigen_glm_cgl_y}
\end{equation}
where,
\begin{align*}
	& c_a= \sqrt{{\frac{B_y^2}{\rho}-\frac{(\pll-p_\perp){b}_y^2}{\rho}}},&\\&
	c_f= \frac{1}{\sqrt{2\rho}}\big[|\B|^2+2p_\bot +b_y^2(2\pll-\per) + \{(|\B|^2+2\per +b_y^2(2\pll-\per))^2&\\&~~~~~~+4(\per^2b_y^2(1-b_y^2)-3\pll \per b_y^2(2-b_y^2)+3\pll^2 b_y^4-3B_y^2 \pll)\}^{\frac{1}{2}}\big]^{\frac{1}{2}},&\\&
	c_s=\frac{1}{\sqrt{2\rho}}\big[|\B|^2+2\per +b_y^2(2\pll-\per) - \{(|\B|^2+2\per +b_y^2(2\pll-\per))^2&\\&~~~~~~+4(\per^2b_y^2(1-b_y^2)-3\pll \per b_y^2(2-b_y^2)+3\pll^2 b_y^4-3B_y^2 \pll )\}^{\frac{1}{2}}\big]^{\frac{1}{2}}.
\end{align*}
Following~\cite{derigs2018ideal}, instead of using the above eigenvalues, we consider eigenvalues of the system~\eqref{eq:glm_cgl_final} for the $x$-direction, which is given below:
\begin{equation}
	\bb{\Lambda_x}=\left\{v_x,~ v_x, v_x \pm c_h, ~v_x\pm c_a,~v_x\pm c_f,~v_x\pm c_s\right\}
	\label{eq:fulleigen_glm_cgl_psi_x}
\end{equation}
where,
\begin{align*}
	& c_a= \sqrt{{\frac{B_x^2}{\rho}-\frac{(\pll-p_\perp){b}_x^2}{\rho}}},&\\&
	c_f= \frac{1}{\sqrt{2\rho}}\big[|\B|^2+2p_\bot +b_x^2(2\pll-\per) + \{(|\B|^2+2\per +b_x^2(2\pll-\per))^2&\\&~~~~~~+4(\per^2b_x^2(1-b_x^2)-3\pll \per b_x^2(2-b_x^2)+3\pll^2 b_x^4-3B_x^2 \pll)\}^{\frac{1}{2}}\big]^{\frac{1}{2}},&\\&
	c_s=\frac{1}{\sqrt{2\rho}}\big[|\B|^2+2\per +b_x^2(2\pll-\per) - \{(|\B|^2+2\per +b_x^2(2\pll-\per))^2&\\&~~~~~~+4(\per^2b_x^2(1-b_x^2)-3\pll \per b_x^2(2-b_x^2)+3\pll^2 b_x^4-3B_x^2 \pll )\}^{\frac{1}{2}}\big]^{\frac{1}{2}}.
\end{align*}
We consider eigenvalues of the system~\eqref{eq:glm_cgl_final} for $y$-direction, which is given below:
\begin{equation}
	\bb{\Lambda_y}=\left\{v_y,~ v_y, ~v_y \pm c_h, ~v_y\pm c_a,~v_y\pm c_f,~v_y\pm c_s\right\}
	\label{eq:fulleigen_glm_cgl_psi_y}
\end{equation}
where,
\begin{align*}
	& c_a= \sqrt{{\frac{B_y^2}{\rho}-\frac{(\pll-p_\perp){b}_y^2}{\rho}}},&\\&
	c_f= \frac{1}{\sqrt{2\rho}}\big[|\B|^2+2p_\bot +b_y^2(2\pll-\per) + \{(|\B|^2+2\per +b_y^2(2\pll-\per))^2&\\&~~~~~~+4(\per^2b_y^2(1-b_y^2)-3\pll \per b_y^2(2-b_y^2)+3\pll^2 b_y^4-3B_y^2 \pll)\}^{\frac{1}{2}}\big]^{\frac{1}{2}},&\\&
	c_s=\frac{1}{\sqrt{2\rho}}\big[|\B|^2+2\per +b_y^2(2\pll-\per) - \{(|\B|^2+2\per +b_y^2(2\pll-\per))^2&\\&~~~~~~+4(\per^2b_y^2(1-b_y^2)-3\pll \per b_y^2(2-b_y^2)+3\pll^2 b_y^4-3B_y^2 \pll )\}^{\frac{1}{2}}\big]^{\frac{1}{2}}.
\end{align*}
\section{Matrices for the Non-Conservative Terms}\label{NC_CGL_Matrices}

The matrix $\textcolor{blue}{\bc_{x}(\con)}$ is given by,
\[\scalebox{0.7}{\textcolor{blue}{$\begin{pmatrix}
			0&0&0&0&0&0&0&0&0&\textcolor{red}{0}\\
			-\frac{b_x^2 |\bu|^2}{2}& b_x^2 v_x& b_x^2 v_y& b_x^2 v_z&\frac{3}{2} b_x^2 & -b_x^2 & b_x^2 B_x+\frac{2\Gamma_x b_x}{|\B|} & b_x^2 B_y - \frac{2\DP b_{yxx}}{|\B|} & b_x^2 B_z-\frac{2\DP b_{zxx}}{|\B|}&\textcolor{red}{b_x^2 \Psi}\\
			-\frac{b_x b_y |\bu|^2}{2} & b_x b_y v_x & b_x b_y v_y & b_x b_y v_z & \frac{3}{2}b_x b_y & -b_x b_y & b_x b_y B_x+\frac{\Gamma_x b_y - \DP b_{yxx}}{|\B|} & b_x b_y B_y+\frac{\Gamma_y b_x - \DP b_{xyy}}{|\B|} & b_x b_y B_z - \frac{2\DP b_{xyz}}{|\B|}&\textcolor{red}{b_x b_y \Psi}\\
			-\frac{b_x b_z |\bu|^2}{2} & b_x b_z v_x & b_x b_z v_y & b_x b_z v_z & \frac{3}{2}b_x b_z & -b_x b_z & b_x b_z B_x+\frac{\Gamma_x b_z - \DP b_{zxx}}{|\B|} & b_x b_z B_y-\frac{2\DP b_{xyz}}{|\B|} & b_x b_z B_z + \frac{\Gamma_z b_x - \DP b_{xzz}}{|\B|}&\textcolor{red}{b_x b_z \Psi}\\
			-\frac{2\pll b_x}{\rho}(\bhat\cdot\bu) & \frac{2\pll b_x}{\rho}b_x & \frac{2\pll b_x}{\rho}b_y & \frac{2\pll b_x}{\rho}b_z & 0 & 0 & 0 & 0 & 0&\textcolor{red}{0}\\
			\Upsilon_1^x & \Upsilon_2^x & \Upsilon_3^x & \Upsilon_4^x & \frac{3}{2} b_x(\bhat\cdot\bu) & -b_x(\bhat\cdot\bu) & b_x(\bhat\cdot\bu)B_x + \Theta^x_1 & b_x(\bhat\cdot\bu)B_y + \Theta^x_2 & b_x(\bhat\cdot\bu)B_z + \Theta^x_3&\textcolor{red}{b_x(\bhat\cdot\bu)\Psi}\\
			0&0&0&0&0&0&0&0&0&\textcolor{red}{0}\\
			0&0&0&0&0&0&0&0&0&\textcolor{red}{0}\\
			0&0&0&0&0&0&0&0&0&\textcolor{red}{0}\\
			0&0&0&0&0&0&0&0&0&\textcolor{red}{0}
		\end{pmatrix}$}}\]
and matrix $\textcolor{blue}{\bc_{y}(\con)}$ is given as,
\[\scalebox{0.7}{\textcolor{blue}{$\begin{pmatrix}
			0&0&0&0&0&0&0&0&0&\textcolor{red}{0}\\
			-\frac{b_x b_y |\bu|^2}{2} & b_x b_y v_x & b_x b_y v_y & b_x b_y v_z & \frac{3}{2}b_x b_y & -b_x b_y & b_x b_y B_x+\frac{\Gamma_x b_y-\DP b_{yxx}}{|\B|} & b_x b_y B_y+\frac{\Gamma_y b_x-\DP b_{xyy}}{|\B|} & b_x b_y B_z -\frac{2\DP b_{xyz}}{|\B|}&\textcolor{red}{b_x b_y \Psi}\\
			-\frac{b_y^2 |\bu|^2}{2} & b_y^2 v_x & b_y^2 v_y & b_y^2 v_z & \frac{3}{2} b_y^2 & -b_y^2 & b_y^2 B_x - \frac{2\DP b_{xyy}}{|\B|} & b_y^2 B_y + \frac{2\Gamma_y b_y}{|\B|} & b_y^2 B_z-\frac{2\DP b_{zyy}}{|\B|}&\textcolor{red}{b_y^2 \Psi} \\
			-\frac{b_y b_z |\bu|^2}{2} & b_y b_z v_x & b_y b_z v_y & b_y b_z v_z & \frac{3}{2}b_y b_z & -b_y b_z & b_y b_z B_x - \frac{2\DP b_{xyz}}{|\B|} & b_y b_z B_y + \frac{\Gamma_y b_z - \DP b_{zyy}}{|\B|} & b_y b_z B_z + \frac{\Gamma_z b_y - \DP b_{yzz}}{|\B|}&\textcolor{red}{b_y b_z \Psi} \\
			-\frac{2\pll b_y}{\rho}(\bhat\cdot\bu) & \frac{2\pll b_y}{\rho}b_x & \frac{2\pll b_y}{\rho}b_y & \frac{2\pll b_y}{\rho}b_z & 0 & 0 & 0 & 0 & 0&\textcolor{red}{0}\\
			\Upsilon_1^y & \Upsilon_2^y & \Upsilon^y_3 & \Upsilon_4^y & \frac{3}{2}b_y(\bhat\cdot\bu) & -b_y(\bhat\cdot\bu) & b_y(\bhat\cdot\bu)B_x + \Theta_1^y & b_y(\bhat\cdot\bu)B_y + \Theta_2^y & b_y(\bhat\cdot\bu)B_z + \Theta_3^y&\textcolor{red}{b_y(\bhat\cdot\bu)\Psi}\\
			0&0&0&0&0&0&0&0&0&\textcolor{red}{0}\\
			0&0&0&0&0&0&0&0&0&\textcolor{red}{0}\\
			0&0&0&0&0&0&0&0&0&\textcolor{red}{0}
		\end{pmatrix}$}}\]

where,\\
\scalebox{0.8}{$\DP = (\pll - \per),~\Gamma_i = \DP (1-b_i^2),~\forall i\in\left\{x,y,z\right\},~b_{lmn}=b_l b_m b_n,~\forall~l,m,n\in\left\{x,y,z\right\}$}\\
\scalebox{0.8}{$\Theta^x_1 = \{(\bhat\cdot\bu)+ b_x v_x\}\left(\frac{\Gamma_x}{|\B|}\right)- \DP\frac{b^2_x b_y v_y}{|\B|} - \DP\frac{b^2_x b_z v_z}{|\B|},~\Theta^x_2 = \Gamma_y \left(\frac{b_x v_y}{|\B|}\right) - \DP\frac{b_x b_y b_z v_z}{|\B|} - \{(\bhat\cdot\bu)+ b_x v_x\}\frac{\DP b_x b_y}{|\B|},$}\\
\scalebox{0.8}{$\Theta^x_3 = \Gamma_z \left(\frac{b_x v_z}{|\B|}\right) - \DP\frac{b_x b_y b_z v_y}{|\B|} - \{(\bhat\cdot\bu)+b_x v_x\}\frac{\DP b_x b_z}{|\B|},~\Theta_1^y = \Gamma_x \left(\frac{b_y v_x}{|\B|}\right) - \DP\frac{b_x b_y b_z v_z}{|\B|} - \{(\bhat\cdot\bu)+ b_y v_y\}\frac{\DP b_x b_y}{|\B|},$}\\
\scalebox{0.8}{$\Theta_2^y = \{(\bhat\cdot\bu) + b_y v_y\}\left(\frac{\Gamma_y}{|\B|}\right)- \DP\frac{b^2_y b_x v_x}{|\B|} - \DP\frac{b^2_y b_z v_z}{|\B|},~\Theta_3^y = \Gamma_z \left(\frac{b_y v_z}{|\B|}\right) - \DP\frac{b_x b_y b_z v_x}{|\B|} - \{(\bhat\cdot\bu)+b_y v_y\}\frac{\DP b_y b_z}{|\B|},$}\\
\scalebox{0.8}{$\Upsilon_1^x = -b_x(\bhat\cdot\bu)\frac{|\bu|^2}{2}-\frac{\DP b_x (\bhat\cdot\bu)}{\rho},~\Upsilon_2^x = b_x(\bhat\cdot\bu)v_x + \frac{\DP b^2_x}{\rho},~\Upsilon_3^x = b_x(\bhat\cdot\bu)v_y +\frac{\DP b_x b_y}{\rho},~\Upsilon_4^x = b_x(\bhat\cdot\bu)v_z+\frac{\DP b_x b_z}{\rho}$}\\
\scalebox{0.8}{$\Upsilon_1^y = -b_y(\bhat\cdot\bu)\frac{|\bu|^2}{2} - \frac{\DP b_y(\bhat\cdot\bu)}{\rho},~\Upsilon_2^y = b_y(\bhat\cdot\bu)v_x+\frac{\DP b_x b_y}{\rho},~\Upsilon_3^y = b_y(\bhat\cdot\bu)v_y+\frac{\DP b_y^2}{\rho}$ and $\Upsilon_4^y = b_y(\bhat\cdot\bu)v_z+\frac{\DP b_y b_z}{\rho}$}.\\

\section{Entropy scaling of right eigenvectors}
\label{scaledrev}
The eigenvectors for~\eqref{eq:glm_cgl_mhd} are detailed in Section~\ref{App:RE}. The entropy-scaled eigenvectors are derived in Section~\ref{App:Barth_process}. We perform all of the analysis using primitive variables $\textbf{w} = \{\rho,v_x,v_y,v_z,\pll,\per,B_x,B_y,B_z,\Psi\}$.

\subsection{Right eigenvectors}\label{App:RE}
The right eigenvectors for the primitive formulation of the system~\eqref{eq:glm_cgl_mhd} with respect to the primitive variables $\textbf{w}$ in $x$-direction for the eigenvalues $K_x$  (define in Section~\ref{sec:sym}) are given as follows:
\begin{gather*}R^{1}_{v_x}=
	\begin{pmatrix}
		1\\
		0\\
		0\\
		0\\
		0\\
		0\\
		0\\
		0\\
		0\\
		0
	\end{pmatrix},~R^{2}_{v_x}=
	\begin{pmatrix}
		0\\
		0\\
		0\\
		0\\
		1\\
		0\\
		0\\
		0\\
		0\\
		0
	\end{pmatrix},~R_{v_x\pm c_h}=
	\begin{pmatrix}
		0\\
		0\\
		0\\
		0\\
		0\\
		0\\
		\pm 1\\
		0\\
		0\\
		1\\
	\end{pmatrix},~R_{v_x\pm v_{ax}}=
	\begin{pmatrix}
		0\\
		0\\
		\pm{\beta_z}\\
		\mp{\beta_y}\\
		0\\
		0\\
		0\\
		-{\beta_z}sgn(B_x)\sqrt{\rho}\\
		{\beta_y}sgn(B_x)\sqrt{\rho}\\
		0\\
	\end{pmatrix},
\end{gather*}
\begin{gather*}R_{v_x\pm c_f}=
	\begin{pmatrix}
		\alpha_f \rho\\
		\pm \alpha_f c_f\\
		\mp \alpha_s c_s \beta_y sgn(B_x)\\
		\mp \alpha_s c_s \beta_z sgn(B_x)\\
		\alpha_f \pll\\
		\alpha_f \rho a^2\\
		0\\
		\alpha_s a\beta_y\sqrt{\rho}\\
		\alpha_s a\beta_z\sqrt{\rho}\\
		0
	\end{pmatrix},~R_{v_x\pm c_s}=
	\begin{pmatrix}
		\alpha_s \rho\\
		\pm \alpha_s c_s\\
		\pm \alpha_f c_f \beta_y sgn(B_x)\\
		\pm \alpha_fc_f \beta_z sgn(B_x)\\
		\alpha_s \pll\\
		\alpha_s \rho a^2\\
		0\\
		-\alpha_f a\beta_y\sqrt{\rho} \\
		-\alpha_f a\beta_z\sqrt{\rho} \\
		0
	\end{pmatrix},
\end{gather*}
where,
\begin{align*}
	& v^2_{ax}= \frac{B_x^2}{\rho},~v^2_a=\frac{|\B|^2}{\rho}, ~a^2=\frac{2P_\perp}{\rho},&\\&
	c^2_{f,s}=\frac{1}{2}\bigg[(v_a^2+a^2)\pm\sqrt{(v_a^2+a^2)^2-4v^2_{ax}a^2}\bigg].&\\&
	\alpha_f^2=\frac{a^2-c_s^2}{c_f^2-c_s^2},~\alpha_s^2=\frac{c_f^2-a^2}{c_f^2-c_s^2},~
	\beta_y=\frac{B_y}{\sqrt{B_y^2+B_z^2}},~\beta_z=\frac{B_z}{\sqrt{B_y^2+B_z^2}}.
\end{align*}
\revg{In addition, to avoid singularities in the definitions of $\beta_y$ and $\beta_z$, we employ the regularization
\begin{equation*}
    \beta_y=\beta_z=\frac{1}{\sqrt{2}},
\end{equation*}
whenever $\sqrt{B_y^2+B_z^2}<10^{-10}$. Otherwise, the standard definitions of $\beta_y$ and $\beta_z$ are used.}

Similarly, for the $y$-direction, eigenvectors corresponding to the eigenvalues $K_y$ (define in Section~\ref{sec:sym}) are given below:
\begin{gather*}R^{1}_{v_y}=
	\begin{pmatrix}
		1\\
		0\\
		0\\
		0\\
		0\\
		0\\
		0\\
		0\\
		0\\
		0
	\end{pmatrix},~R^{2}_{v_y}=
	\begin{pmatrix}
		0\\
		0\\
		0\\
		0\\
		1\\
		0\\
		0\\
		0\\
		0\\
		0
	\end{pmatrix},~R_{v_y\pm c_h}=
	\begin{pmatrix}
		0\\
		0\\
		0\\
		0\\
		0\\
		0\\
		0\\
		\pm 1\\
		0\\
		1\\
	\end{pmatrix},~R_{v_y\pm v_{ay}}=
	\begin{pmatrix}
		0\\
		\pm{\beta_z}\\
		0\\
		\mp{\beta_x}\\
		0\\
		0\\
		-{\beta_z}sgn(B_y)\sqrt{\rho}\\
		0\\
		{\beta_x}sgn(B_y)\sqrt{\rho}\\
		0
	\end{pmatrix},
\end{gather*}
\begin{gather*}R_{v_y\pm c_f}=
	\begin{pmatrix}
		\alpha_f \rho\\
		\mp \alpha_s c_s \beta_x sgn(B_y)\\
		\pm \alpha_f c_f\\
		\mp \alpha_s c_s \beta_z sgn(B_y)\\
		\alpha_f \pll\\
		\alpha_f \rho a^2\\
		\alpha_s a\beta_x\sqrt{\rho}\\
		0\\
		\alpha_s a\beta_z\sqrt{\rho}\\
		0
	\end{pmatrix},~R_{v_y\pm c_s}=
	\begin{pmatrix}
		\alpha_s \rho\\
		\pm \alpha_f c_f \beta_x sgn(B_y)\\
		\pm \alpha_s c_s\\
		\pm \alpha_fc_f \beta_z sgn(B_y)\\
		\alpha_s \pll\\
		\alpha_s \rho a^2\\
		-\alpha_f a\beta_x\sqrt{\rho} \\
		0\\
		-\alpha_f a\beta_z\sqrt{\rho} \\
		0
	\end{pmatrix},
\end{gather*}
where,
\begin{align*}
	& v^2_{ay}= \frac{B_y^2}{\rho},~v^2_a=\frac{|\B|^2}{\rho}, ~a^2=\frac{2P_\perp}{\rho},&\\&
	c^2_{f,s}=\frac{1}{2}\bigg[(v_a^2+a^2)\pm\sqrt{(v_a^2+a^2)^2-4v^2_{ay}a^2}\bigg].&\\&
	\alpha_f^2=\frac{a^2-c_s^2}{c_f^2-c_s^2},~\alpha_s^2=\frac{c_f^2-a^2}{c_f^2-c_s^2},~
	\beta_x=\frac{B_x}{\sqrt{B_x^2+B_z^2}},~\beta_z=\frac{B_z}{\sqrt{B_x^2+B_z^2}}.
\end{align*}
\revg{In addition, to avoid singularities in the definitions of $\beta_x$ and $\beta_z$, we employ the regularization
\begin{equation*}
    \beta_x=\beta_z=\frac{1}{\sqrt{2}},
\end{equation*}
whenever $\sqrt{B_x^2+B_z^2}<10^{-10}$. Otherwise, the standard definitions of $\beta_x$ and $\beta_z$ are used.}

The set of eigenvectors in both directions is linearly independent. 
\subsection{Entropy-scaled right eigenvectors via Barth scaling process}\label{App:Barth_process}
Here we will follow the Barth scaling process~\cite{barth1999numerical} and calculate the entropy-scaled right eigenvectors for $x$-direction. For the system~\eqref{eq:glm_cgl_mhd}, let $R^x$ be the right eigenvector matrix for the  system~\eqref{eq:glm_cgl_mhd} in terms of \revo{working variables} and $R_{\textbf{w}}^x$ be right eigenvectors in terms of the primitive variables described in~\ref{App:RE}. Then we have
\begin{equation*}
	R^x = \frac{\p \con}{\p \textbf{w}} R_{\textbf{w}}^x
\end{equation*}
where,  $\frac{\p \con}{\p \textbf{w}}$ is the Jacobian matrix for the change of variable, given by
\begin{align*}
	\frac{\p \con}{\p \textbf{w}}=
	\begin{pmatrix}
		1 & 0 & 0 & 0 & 0 & 0 & 0 & 0 & 0 & 0\\
		v_x & \rho & 0 & 0 & 0 & 0 & 0 & 0 & 0 & 0\\
		v_y & 0 & \rho & 0 & 0 & 0 & 0 & 0 & 0 & 0\\
		v_z & 0 & 0 & \rho & 0 & 0 & 0 & 0 & 0 & 0\\
		0 & 0 & 0 & 0 & 1 & 0 & 0 & 0 & 0 & 0\\
		\frac{\bu^2}{2} & \rho v_x & \rho v_y & \rho v_z & \frac{1}{2} & 1 & B_x & B_y & B_z & \Psi\\
		0 & 0 & 0 & 0 & 0 & 0 & 1 & 0 & 0 & 0\\
		0 & 0 & 0 & 0 & 0 & 0 & 0 & 1 & 0 & 0\\
		0 & 0 & 0 & 0 & 0 & 0 & 0 & 0 & 1 & 0\\
		0 & 0 & 0 & 0 & 0 & 0 & 0 & 0 & 0 & 1
	\end{pmatrix}.
\end{align*}
and the matrix $R_{\textbf{w}}^x$ is the right eigenvector matrix for the  system~\eqref{eq:glm_cgl_mhd} in terms of primitive variable (see~\ref{App:RE}), which is given below,
\begin{align*}
	R_{\textbf{w}}^x=
	\begin{pmatrix}
		R_{v_x-c_f} & R_{v_x-c_s} & R_{v_x-c_h} & R_{v_x-v_{ax}} & R^{1}_{v_x} & R^{2}_{v_x} & R_{v_x+v_{ax}}  & R_{v_x+c_h} &  R_{v_x+c_s} & R_{v_x+ c_f}\\
	\end{pmatrix}.
\end{align*}
Our aim is to find the scaling matrix $T^x$ such that the scaled right eigenvector matrix $\tilde{R^x}=R^x T^x$ satisfies
\begin{equation}
	\frac{\p \con}{\p \evar} = \tilde{R}^x \tilde{R^x}^{\top}
\end{equation}
where $\evar$ is the entropy variable vector as in~\eqref{eq:envar_glm}. Now we follow~\cite{barth1999numerical} and define the matrix 
\begin{equation*}
	\mathcal{Y}^x = (R^x_{\textbf{w}})^{-1} \frac{\p \textbf{w}}{\p \evar}\left(\frac{\p \con}{\p \textbf{w}_x}\right)^{-\top}(R^x_{\textbf{w}}) ^{-\top}
\end{equation*}
After doing some calculations, we get 
\begin{align*}
	\mathcal{Y}^x = \begin{pmatrix}
		\frac{1}{8\rho} & 0 & 0 & 0 & 0 & 0 & 0 & 0 & 0 & 0\\
		0 & \frac{1}{8\rho} & 0 & 0 & 0 & 0 & 0 & 0 & 0 & 0\\
		0 & 0 & \frac{\per}{4\rho} & 0 & 0 & 0 & 0 & 0 & 0 & 0\\
		0 & 0 & 0 & \frac{\per}{4\rho^2} & 0 & 0 & 0 & 0 & 0 & 0\\
		0 & 0 & 0 & 0 & \frac{\rho}{4} & \frac{\pll}{4} & 0 & 0 & 0 & 0\\
		0 & 0 & 0 & 0 & \frac{\pll}{4} &\frac{5 \pll^2}{4\rho} & 0 & 0 & 0 & 0 \\
		0 & 0 & 0 & 0 & 0 & 0 & \frac{\per}{4\rho^2}& 0 & 0 & 0\\
		0 & 0 & 0 & 0 & 0 & 0 & 0 & \frac{\per}{4\rho} & 0 & 0\\
		0 & 0 & 0 & 0 & 0 & 0 & 0 & 0 & \frac{1}{8\rho} & 0\\ 
		0 & 0 & 0 & 0 & 0 & 0 & 0 & 0 & 0 & \frac{1}{8\rho}
	\end{pmatrix}.
\end{align*}
Then the scaling matrix $T^x$ is the square root of $\mathcal{Y}^x$ is given below
\begin{align*}
	\begin{pmatrix}
		\frac{1}{2\sqrt{2\rho}} & 0 & 0 & 0 & 0 & 0 & 0 & 0 & 0 & 0\\
		0 & \frac{1}{2\sqrt{2\rho}} & 0 & 0 & 0 & 0 & 0 & 0 & 0 & 0\\
		0 & 0 & \frac{\sqrt{\per}}{2\sqrt{\rho}} & 0 & 0 & 0 & 0 & 0 & 0 & 0\\
		0 & 0 & 0 & \frac{\sqrt{\per}}{2\rho} & 0 & 0 & 0 & 0 & 0 & 0\\
		0 & 0 & 0 & 0 & \frac{\sqrt{\rho}(2 \pll + \rho )}{2 \sqrt{5 \pll^{2} + 4 \pll\rho + \rho^2}} &\frac{\pll\sqrt{\rho}}{2 \sqrt{5 \pll^{2} + 4 \pll\rho + \rho^2}} &  0 & 0 & 0 & 0\\
		0 & 0 & 0 & 0 & \frac{\pll\sqrt{\rho}}{2 \sqrt{5 \pll^{2} + 4 \pll\rho + \rho^2}} & \frac{\pll (5 \pll + 2 \rho)}{2 \sqrt{\rho (5 \pll^{2} + 4 \pll\rho + \rho^2)}}& 0 & 0 & 0 & 0\\
		0 & 0 & 0 & 0 & 0 & 0 & \frac{\sqrt{\per}}{2\rho}& 0 & 0 & 0\\
		0 & 0 & 0 & 0 & 0 & 0 & 0 & \frac{\sqrt{\per}}{2\sqrt{\rho}} & 0 & 0\\
		0 & 0 & 0 & 0 & 0 & 0 & 0 & 0 & \frac{1}{2\sqrt{2\rho}} & 0\\ 
		0 & 0 & 0 & 0 & 0 & 0 & 0 & 0 & 0 & \frac{1}{2\sqrt{2\rho}}
	\end{pmatrix}.
\end{align*}


In a similar manner, for the $y$-direction, we obtain the scaling matrix $T^y$.
%
Combining all the discussions, in $x$-direction, the entropy-scaled right eigenvectors matrix in terms of primitive variables is given by,
\begin{equation*}
	\tilde{R^x}_{\textbf{w}} = R^x_{\textbf{w}} T^x\\
\end{equation*}
\begin{align*}
	\tilde{R^x}_{\textbf{w}}=
	\begin{pmatrix}
		\tilde{R}_{v_x-c_f} & \tilde{R}_{v_x-c_s} & \tilde{R}_{v_x-c_h} & \tilde{R}_{v_x-v_{ax}} & \tilde{R}^{1}_{v_x} & \tilde{R}^{2}_{v_x} & \tilde{R}_{v_x+v_{ax}}  & \tilde{R}_{v_x+c_h} &  \tilde{R}_{v_x+c_s} & \tilde{R}_{v_x+ c_f}\\
	\end{pmatrix}.
\end{align*}
where, 
\begin{gather*}\tilde{R}^{1}_{v_x}=
	\begin{pmatrix}
		\frac{\sqrt{\rho}(2 \pll + \rho )}{2 \sqrt{5 \pll^{2} + 4 \pll\rho + \rho^2}}\\
		0\\
		0\\
		0\\
		\frac{\pll\sqrt{\rho}}{2 \sqrt{5 \pll^{2} + 4 \pll\rho + \rho^2}}\\
		0\\
		0\\
		0\\
		0\\
		0
	\end{pmatrix},~\tilde{R}^{2}_{v_x}=
	\begin{pmatrix}
		\frac{\pll\sqrt{\rho}}{2 \sqrt{5 \pll^{2} + 4 \pll\rho + \rho^2}}\\
		0\\
		0\\
		0\\
		\frac{\pll (5 \pll + 2 \rho)}{2 \sqrt{\rho (5 \pll^{2} + 4 \pll\rho + \rho^2)}}\\
		0\\
		0\\
		0\\
		0\\
		0
	\end{pmatrix},~\tilde{R}_{v_x\pm c_h}=
	\frac{1}{2}\begin{pmatrix}
		0\\
		0\\
		0\\
		0\\
		0\\
		0\\
		\pm\sqrt{\frac{\per}{\rho}}\\
		0\\
		0\\
		\sqrt{\frac{\per}{\rho}}
	\end{pmatrix},
\end{gather*}
\begin{gather*}\tilde{R}_{v_x\pm v_{ax}}=
	\frac{1}{2}\begin{pmatrix}
		0\\
		0\\
		\pm\frac{\sqrt{\per}}{\rho}{\beta_z}\\
		\mp\frac{\sqrt{\per}}{\rho}{\beta_y}\\
		0\\
		0\\
		0\\
		-\sqrt{\frac{\per}{\rho}}{\beta_z}\\
		\sqrt{\frac{\per}{\rho}}{\beta_y}\\
		0
	\end{pmatrix},~\tilde{R}_{v_x\pm c_f}=
	\frac{1}{2\sqrt{2}}\begin{pmatrix}
		\alpha_f \sqrt{\rho}\\
		\pm \frac{\alpha_f c_f}{\sqrt{\rho}}\\
		\mp \frac{\alpha_s c_s \beta_y}{\sqrt{\rho}}\\
		\mp \frac{\alpha_s c_s \beta_z}{\sqrt{\rho}}\\
		\alpha_f \frac{\pll}{\sqrt{\rho}}\\
		\alpha_f \sqrt{\rho} a^2\\
		0\\
		\alpha_s a\beta_y\\
		\alpha_s a\beta_z\\
		0
	\end{pmatrix},~\tilde{R}_{v_x\pm c_s}=
	\frac{1}{2\sqrt{2}}\begin{pmatrix}
		\alpha_s \sqrt{\rho}\\
		\pm \frac{\alpha_s c_s}{\sqrt{\rho}}\\
		\pm \frac{\alpha_f c_f \beta_y}{\sqrt{\rho}}\\
		\pm \frac{\alpha_f c_f \beta_z}{\sqrt{\rho}}\\
		\alpha_s \frac{\pll}{\sqrt{\rho}}\\
		\alpha_s \sqrt{\rho} a^2\\
		0\\
		-\alpha_f a\beta_y\\
		-\alpha_f a\beta_z\\
		0
	\end{pmatrix}
\end{gather*}
Similarly, in the $y$-direction, the entropy-scaled right eigenvectors matrix in terms of primitive variables is given by, 
\begin{equation*}
	\tilde{R^y}_{\textbf{w}} = R^y_{\textbf{w}} T^y\\
\end{equation*}
\begin{align*}
	\tilde{R^y}_{\textbf{w}}=
	\begin{pmatrix}
		\tilde{R}_{v_y-c_f} & \tilde{R}_{v_y-c_s} & \tilde{R}_{v_y-c_h} & \tilde{R}_{v_y-v_{ay}} & \tilde{R}^{1}_{v_y} & \tilde{R}^{2}_{v_y} & \tilde{R}_{v_y+v_{ay}}  & \tilde{R}_{v_y+c_h} &  \tilde{R}_{v_y+c_s} & \tilde{R}_{v_y+ c_f}\\
	\end{pmatrix}.
\end{align*}
where,
\begin{gather*}\tilde{R}^{1}_{v_y}=
	\begin{pmatrix}
		\frac{\sqrt{\rho}(2 \pll + \rho )}{2 \sqrt{5 \pll^{2} + 4 \pll\rho + \rho^2}}\\
		0\\
		0\\
		0\\
		\frac{\pll\sqrt{\rho}}{2 \sqrt{5 \pll^{2} + 4 \pll\rho + \rho^2}}\\
		0\\
		0\\
		0\\
		0\\
		0
	\end{pmatrix},~\tilde{R}^{2}_{v_y}=
	\begin{pmatrix}
		\frac{\pll\sqrt{\rho}}{2 \sqrt{5 \pll^{2} + 4 \pll\rho + \rho^2}}\\
		0\\
		0\\
		0\\
		\frac{\pll (5 \pll + 2 \rho)}{2 \sqrt{\rho (5 \pll^{2} + 4 \pll\rho + \rho^2)}}\\
		0\\
		0\\
		0\\
		0\\
		0
	\end{pmatrix},~\tilde{R}_{v_y\pm c_h}=
	\begin{pmatrix}
		0\\
		0\\
		0\\
		0\\
		0\\
		0\\
		0\\
		\pm\sqrt{\frac{\per}{\rho}}\\
		0\\
		\sqrt{\frac{\per}{\rho}}
	\end{pmatrix},
\end{gather*}
\begin{gather*}\tilde{R}_{v_y\pm v_{ay}}=
	\frac{1}{2}\begin{pmatrix}
		0\\
		\pm\frac{\sqrt{\per}}{\rho}{\beta_z}\\
		0\\
		\mp\frac{\sqrt{\per}}{\rho}{\beta_x}\\
		0\\
		0\\
		-\sqrt{\frac{\per}{\rho}}{\beta_z}\\
		0\\
		\sqrt{\frac{\per}{\rho}}{\beta_x}\\
		0
	\end{pmatrix},~\tilde{R}_{v_y\pm c_f}=
	\frac{1}{2\sqrt{2}}\begin{pmatrix}
		\alpha_f \sqrt{\rho}\\
		\mp \frac{\alpha_s c_s \beta_x}{\sqrt{\rho}}\\
		\pm \frac{\alpha_f c_f}{\sqrt{\rho}}\\
		\mp \frac{\alpha_s c_s \beta_z}{\sqrt{\rho}}\\
		\alpha_f \frac{\pll}{\sqrt{\rho}}\\
		\alpha_f \sqrt{\rho} a^2\\
		\alpha_s a\beta_x\\
		0\\
		\alpha_s a\beta_z\\
		0
	\end{pmatrix},~\revg{\tilde{R}_{v_y\pm c_s}}=
	\frac{1}{2\sqrt{2}}\begin{pmatrix}
		\alpha_s \sqrt{\rho}\\
		\pm \frac{\alpha_f c_f \beta_x}{\sqrt{\rho}}\\
		\pm \frac{\alpha_s c_s}{\sqrt{\rho}}\\
		\pm \frac{\alpha_f c_f \beta_z}{\sqrt{\rho}}\\
		\alpha_s \frac{\pll}{\sqrt{\rho}}\\
		\alpha_s \sqrt{\rho} a^2\\
		-\alpha_f a\beta_x\\
		0\\
		-\alpha_f a\beta_z\\
		0
	\end{pmatrix},
\end{gather*}
All the notations are defined in~\ref{App:RE}.
\section{Central Difference formulas for derivatives in non-conservative terms}
\label{appendix:central}
For a grid function $a_{i,j}$, the second-order central difference approximations to approximate the derivatives $\left(\frac{\p a}{\p x}\right)_{i,j}$ and $\left(\frac{\p a}{\p y}\right)_{i,j}$ in $x$ and $y$-directions are given by,
\[\left(\frac{\p a}{\p x}\right)_{i,j}=\frac{a_{i+1,j}-a_{i-1,j}}{2\Delta x}~~\text{and}~~\left(\frac{\p a}{\p y}\right)_{i,j}=\frac{a_{i,j+1}-a_{i,j-1}}{2\Delta y}.\]

Similarly, the fourth-order central difference approximations to approximate the derivatives $\left(\frac{\p a}{\p x}\right)_{i,j}$ and $\left(\frac{\p a}{\p y}\right)_{i,j}$ in $x$ and $y$-directions are given below
\begin{align*}
	\left(\frac{\p a}{\p x}\right)_{i,j}&=\frac{-a_{i+2,j}+8a_{i+1,j}-8a_{i-1,j}+a_{i-2,j}}{12\Delta x},\\
	&~~~~~~~~~~~~~~~~~~~~~~~~\text{and}\\
	\left(\frac{\p a}{\p y}\right)_{i,j}&=\frac{-a_{i,j+2}+8a_{i,j+1}-8a_{i,j-1}+a_{i,j-2}}{12\Delta y}.
\end{align*}

\end{document}